%% file: main.tex
\documentclass[12pt,onecolumn,a4paper]{article}
\input{commands.tex} 
\begin{document}

\title{Fault-Tolerant One-Shot Entanglement Generation\\
 with Constant-Sized Quantum Devices in the Plane}
 \author[1]{Dylan Harley} 
 \author[2,3]{Robert K{\"o}nig}
\affil[1]{Department of Mathematical Sciences, University of Copenhagen, Denmark}
\affil[2]{Department of Mathematics, School of Computation, Information and Technology, Technical University of Munich, Garching, Germany}
\affil[3]{Munich Center for Quantum Science and Technology, Munich, Germany}
\maketitle 
\begin{abstract} 
Consider a rectangular grid of qubits in 2D  with single-qubit and nearest-neighbor two-qubit operations subject to local stochastic Pauli noise. 
At different length scales, this setup describes  both a single  quantum computing device with geometrically limited connectivity between qubits 
arranged on a disc, and planar networks composed of quantum repeater stations of constant size. We give a protocol which robustly generates entanglement between distant qubits in this setup. For noise below a  constant threshold error strength, it generates  
a constant-fidelity Bell pair between qubits separated by an arbitrarily large distance~$R$. 
To generate distance-$R$ entanglement, a rectangular grid of qubits of dimensions $\Theta(R)\times \Theta(\mathsf{poly}(\log R))$ suffices. 
Our protocol applies  quantum operations in one shot, establishing a Bell state in a constant time up to a known Pauli correction. In contrast, existing
entanglement generation protocols either require local quantum devices controlling a number of qubits growing with the targeted distance, 
or are not single-shot, i.e., have  a distance-dependent execution time.  The protocol leverages many-body entanglement in networks and  provides the first example of a short-range entangled state in 2D with long-range localizable entanglement robust to local stochastic Pauli noise. As an immediate corollary, we construct a 2D-local stabilizer Hamiltonian whose Gibbs states possess long-range localizable entanglement at constant positive temperature.
\end{abstract}

\newpage 
\tableofcontents

\section{Introduction}
Generating long-ranged entanglement is a vital requirement for quantum technologies across a wide range of applications; for network communications, entanglement must be consumed between parties in order to transmit states via teleportation or to perform quantum key distribution for cryptography (see e.g.,~\cite{bennett1993teleporting,bennett2014quantum}). Moreover, near-term scalable architectures for fault-tolerant quantum computation will likely consist of spatially separated processors working in parallel~\cite{azuma2023quantum}, between which entanglement must be generated. As such, the question of how to fault-tolerantly  produce such long-range entanglement is a fundamental problem which has attracted a great deal of theoretical and experimental interest (see  Section~\ref{sec:related} below). 

Naively transporting a quantum system over a long distance leaves it exposed to environmental noise for the entire transit time, leading to an accumulation of operation errors. For example, this is the case when motional degrees of trapped ions or neutral atom arrays are used to encode quantum information and these information carriers are physically moved. In quantum optics setups, on the other hand, fiber attenuation leads to photon loss, a major source of 
experimental difficulty. As a consequence,  an exponential decay of the transmission fidelity is generally unavoidable when trying to directly communicate quantum information over a significant distance.

Primarily studied in the setting of fiber-optic quantum point-to-point communication, quantum repeater-based protocols seek to overcome this fundamental challenge~\cite{briegel1998quantum}. The underlying idea is to insert repeater stations at regular intervals along the communication line.
This  breaks up the challenge of generating long-range entanglement into that of creating high-fidelity bipartite entangled states between neighboring pairs of repeaters. The latter are then used as building blocks to create longer-ranged entanglement, e.g., by entanglement swapping. Integrating active error detection and correction steps in such interactive processes, a zoo of protocols have been developed for long-range fault-tolerant entanglement generation along a line of repeaters (see e.g., ~\cite{muralidharan2016optimal,azuma2023quantum} for reviews).

{\em Low-latency entanglement generation protocols.} Here we are interested in  protocols establishing entanglement with low latency, i.e., in a constant time independent of the communication distance.  In more detail, such a protocol generates a long-range entangled Bell pair in a constant time up to a single-qubit unitary correction (generally a Pauli operator) which can be determined from measurement results (``syndromes'') collected in its execution. It typically takes the following form when trying to establish entanglement between two qubits $Q_1,Q_2$:
\begin{enumerate}[(i)]
\item\label{it:stateprepsingleshot}
 a short-range entangled state~$\ket{\Psi}$ shared between the repeater stations is prepared. (We use the term short-range entangled for any state which can be prepared by a constant-depth local circuit from a product state.)  This can be achieved in a constant time  with limited communication between repeaters. The corresponding operations are formally described by a constant-depth circuit~$U$ composed of single-qubit  and nearest-neighbor two-qubit operations applied to an initial product state.
 \item\label{it:measurementsingleshot}
With the exception of the qubits $Q_1Q_2$, all qubits  in the prepared state~$\ket{\Psi}$ are measured simultaneously, i.e., in a one-shot manner, resulting in a measurement result~$z$ and a post-measurement state~$\ket{\Psi(z)}$ on the qubits~$Q_1Q_2$.
\item\label{it:correctionfix}
The measurement results~$z$ are pooled, i.e., communicated to the repeater holding the qubit~$Q_1$ by means of  one-way classical communication.
Then a Pauli $P(z)$ is computed from~$z$ and applied to the post-measurement state, resulting in the state
\begin{align}
(P(z)_{Q_1}\otimes I_{Q_2})\ket{\Psi(z)}\ .
\end{align}
\end{enumerate}
The main `quantum part', i.e., Steps~\eqref{it:stateprepsingleshot} and\eqref{it:measurementsingleshot} of such a protocol establishes long-range entanglement in the form of a state~$\ket{\Psi(z)}$ which (in the absence of errors) is the Bell state~$\ket{\Phi}=\frac{1}{\sqrt{2}}(\ket{00}+\ket{11})$ up to a  (known) single-qubit Pauli~$P(z)$.
We emphasize here that these two steps  only take a constant amount of time to execute. This is because the circuit~$U$ is of constant depth, and the single-qubit measurements can be applied simultaneously: it is not necessary to know the measurement outcomes from other qubits (repeaters). Correspondingly, this kind of entanglement generation protocol is referred to as one-shot.

The Pauli~$P(z)$ is a function of the obtained measurement results~$z$ and determines the `Pauli frame' of subsequent operations on the qubits~$Q_1Q_2$. It can be actively applied as in Step~\eqref{it:correctionfix} to remove it. To maintain the relevant entanglement, only the qubits~$Q_1Q_2$ need to be preserved in a good quantum memory while this communication and computation takes place. The Pauli~$P(z)$ can then be applied to qubit~$Q_1$.
 We note that in some settings, such as when dealing with Clifford circuits, it suffices to keep track of~$P(z)$ ``in software''. Furthermore, as discussed in~\cite{perseguers2008one}, in certain applications such as key distribution, the determination of $P(z)$ can be omitted entirely since any (even Pauli-corrupted) Bell pair suffices.  In these settings, Step~\eqref{it:correctionfix} can be omitted, and the protocol only requires qubit coherence times that are constant: no long-term quantum memory is needed.

{\em Limitations of entanglement generation protocols.} 
The idea of using repeaters along a communication line follows a long and successful tradition in classical signal transmissions, where e.g., submarine amplifiers in long-distance communication were first deployed in TAT-1 (Transatlantic No. 1) in 1956. For quantum communication, the repeater-based approach to entanglement distribution suffers from a key disadvantage: it requires quantum repeaters whose local system dimension (number of qubits) scales with the communication distance, and/or requires protocols taking a distance-dependent amount of time to establish long-range entanglement.
For example, to reach a distance~$R$, the protocol of~\cite{childressetal05}, while only requiring $O(1)$~qubits at each repeater, is realized by a circuit of $O(\mathsf{poly}(R))$ depth.

In contrast, low-latency entanglement generation along a line of repeaters requires large on-site quantum processors at every repeater node. Roughly, this is because the bipartite entanglement established between neighboring repeaters must 
be of a quality scaling inversely with the communication distance: to obtain distance-$R$ entanglement with $N=\Theta(R)$ repeaters, the  infidelity of Bell pairs between neighboring repeaters should not exceed~$O(1/N)=O(1/R)$.
Such high-fidelity Bell pairs can only be created in a constant time using, e.g., an errror-correcting code with a distance scaling non-trivially with~$R$ at each repeater. In turn, this implies that the dimension of the local system needs to scale with~$R$. We refer to
e.g.,~\cite{acin2007entanglement,choe2024long} for rigorous versions of this argument and corresponding statements. The scaling of the local system (repeater) dimension with the targeted distance is a strong requirement to ask of every node/repeater. 

\subsection*{Our contribution: Leveraging the power of networks}
In many settings of practical interest, the line of communication does not exist in isolation: instead, it is a path in a large network of `quantum nodes' (i.e., devices operating  on local quantum systems and connected by quantum communication links). Examples are a prospective quantum internet, or a device with many qubits arranged on a planar piece of material. 
It is natural to ask whether fault-tolerant entanglement generation  can be enhanced by exploiting synergies in such networks, where nodes cooperate to establish long-range entanglement. Specifically, we ask the following main question:
\begin{quote}
{\em Can constant-fidelity entanglement be created in a one-shot manner in planar networks using constant-size on-site devices only?}
\end{quote}
Our main result is an affirmative answer to this question: We propose a one-shot protocol for arbitrary-distance entanglement generation in 2D which only uses repeaters/nodes of constant size (i.e., manipulating a constant number of qubits) independently of the targeted distance. We then show that it creates constant-fidelity Bell pairs at arbitrary distances.

We note that our finding follows a long line of works on entanglement generation and localizable entanglement, see Section~\ref{sec:related} below for an in-depth discussion. In fact, our central question has previously been posed and partially answered in~\cite[Section III.C]{perseguers2008one}, which we became aware of   close to the completion of this work. The main contribution of Ref.~\cite{perseguers2008one} is the proposal and analysis of a  protocol for one-shot entanglement generation in 2D with repeater stations whose number of local qubits is logarithmic in the distance. Complementary to this, the authors raise the question of whether a constant number of qubits per station is in principle sufficient in Section~III.C. In two paragraphs, they sketch a solution to this problem based on the  fault-tolerant quantum computation scheme of~\cite{stephens2007universal}. 

Our construction can be seen as a worked-out variant of the discussion
in~\cite[Section III.C]{perseguers2008one}: We give a fully explicit description of all steps involved in the construction of such a one-shot protocol, along with a threshold proof for general local stochastic Pauli noise. We note that our construction and analysis go beyond what has traditionally been discussed in the context of quantum fault tolerance  (such as in the relevant literature around the time of publication of Ref.~\cite{perseguers2008one}). In fault-tolerant quantum computation, the  focus typically is on approximately reproducing the statistics of measurement outcomes of an ideal quantum circuit. Instead, to get a fault-tolerant protocol for long-range entanglement generation between physical qubits, we need to argue more specifically about noisy encoding and decoding steps. The corresponding full analysis is one of our key technical advances necessary to rigorously execute this program. It results in  explicit threshold estimates (see Eq.~\eqref{eq:pzerooneDresdef} and Theorem~\ref{thm:2Dentanglementgenerationgeneral} below).

Let us formally state our main result. We consider qubits arranged on the vertices 
of a grid graph of dimensions~$\ell_X\times\ell_Y$, i.e., 
the Cartesian product~$\mathsf{P}_{\ell_X}\square \mathsf{P}_{\ell_Y}$ of two  path graphs $\mathsf{P}_{\ell_X}$ and $\mathsf{P}_{\ell_Y}$ of lengths~$\ell_X$ and $\ell_Y$, respectively. We assume that this graph is embedded in~$\mathbb{R}^2$ with neighboring vertices (qubits) at distance~$1$. Our main result can be summarized as follows (see Fig.~\ref{fig:mainresultsketch} for an illustration).

\counterwithout{theorem}{section}
\begin{theorem}[Arbitrary-distance entanglement generation in 2D]\label{thm:2Dentanglementgeneration} 
There is a (constant) threshold error strength~$p_{9/10}>0$ such that the following holds.
For any desired communication distance~$R>0$,  there are 
\begin{enumerate}
\item
parameters
\begin{align}
\ell_X&=\Theta(\mathsf{poly}(\log R))\\
\ell_Y&=\Theta(R)\ ,
\end{align}
\item
two qubits $Q_1,Q_2$ separated by a distance~$\Theta(R)$ on the grid graph~$\mathsf{P}_{\ell_X}\square \mathsf{P}_{\ell_Y}=:(V,E)$, 
\item
 a Clifford circuit~$U$ of constant depth composed of single-qubit and  nearest-neighbor two-qubit gates on~$\mathsf{P}_{\ell_X}\square \mathsf{P}_{\ell_Y}$ and
 \item an (efficiently computable)  post-processing function~$P:\{0,1\}^{V\backslash \{Q_1,Q_2\}}\rightarrow \cP_1$
taking values in the single-qubit Pauli group~$\cP_1$
\end{enumerate}
with the following properties.
Suppose we measure every qubit~$v\in V\backslash \{Q_1,Q_2\}$ 
of the state
\begin{align}
\ket{\Psi}&= U\ket{0^V}\label{eq:centralstateshortrangelongrange}
\end{align}
in the computational basis. Let 
\begin{align}
z&=(z_v)_{v\in V\backslash\{Q_1,Q_2\}}\in \{0,1\}^{V\backslash \{Q_1,Q_2\}}\\
\ket{\Psi(z)}&\propto  \left(I_{Q_1Q_2}\otimes \bra{z}\right)\ket{\Psi}
\end{align}
denote the measurement outcome and  the associated post-measurement state on the qubits~$Q_1Q_2$, respectively. Then
\begin{align}
(P(z)_{Q_1}\otimes I_{Q_2})\ket{\Psi(z)}\propto \ket{\Phi}:=\frac{1}{\sqrt{2}}(\ket{00}+\ket{11})\ .\label{eq:correctalwayswithoutnoise}
\end{align}
Furthermore, if $U$ is replaced by a noisy implementation~$\tilde{U}$ 
with  (circuit-level) local stochastic  Pauli noise~$E$ of strength~$p\leq p_{9/10}$ to instead produce a corrupted state $\ket{\tilde{\Psi}} = \tilde{U} \ket{0^V}$, then 
\begin{align}
\Pr_{z,E}\left[(P(z)_{Q_1}\otimes I_{Q_2})\ket{\tilde{\Psi}(z)}\propto \ket{\Phi}\right]&\geq 9/10\ ,\label{eq:noisyversionqgeneration}
\end{align}
where the probability is taken over the noise realization~$E$ and the  measurement outcome~$z$. 
\end{theorem}
\counterwithin{theorem}{subsection}

\noindent  We remark here that the constant~$9/10$ is arbitrary (see the corresponding generalized statement given as  Theorem~\ref{thm:2Dentanglementgenerationgeneral} below).

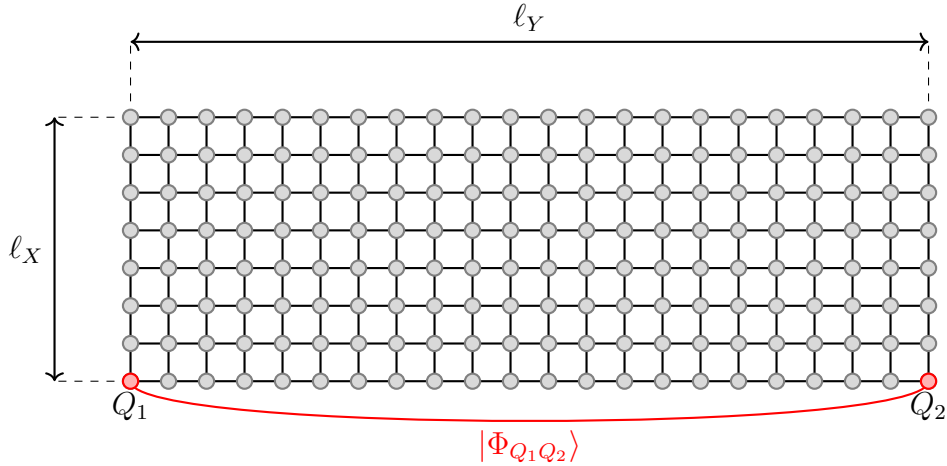
\begin{figure}
    \centering
    \begin{tikzpicture}
        \def\spacing{0.5}
        \def\latwidth{21}
        \def\latheight{7}
        \def\siterad{0.1}
        \def\labelbuffer{1}

        \draw[dashed] (0,0) -- (-\labelbuffer,0);
        \draw[dashed] (0,\latheight*\spacing) -- (-\labelbuffer,\latheight*\spacing);
        \draw[thick,<->] (-\labelbuffer,0) -- (-\labelbuffer,\latheight*\spacing);
        \node[left] at (-\labelbuffer,\latheight*\spacing*0.5) {$\ell_X$};
        
        \draw[dashed] (0,\latheight*\spacing) -- (0,\latheight*\spacing + \labelbuffer);
        \draw[dashed] (\latwidth*\spacing,\latheight*\spacing) -- (\latwidth*\spacing,\latheight*\spacing + \labelbuffer);
        \draw[thick,<->] (0,\latheight*\spacing + \labelbuffer) -- (\latwidth*\spacing,\latheight*\spacing + \labelbuffer);
        \node[above] at (\latwidth*\spacing*0.5,\latheight*\spacing + \labelbuffer) {$\ell_Y$};

        \foreach\x in {0,...,\latwidth} {
        \draw[thick] (\x*\spacing,0) -- (\x*\spacing,\latheight*\spacing);
        }
        \foreach\y in {0,...,\latheight} {
        \draw[thick] (0,\y*\spacing) -- (\latwidth*\spacing,\y*\spacing);
        }
        \foreach\x in {0,...,\latwidth} {
        \foreach\y in {0,...,\latheight} {
        \draw[draw=gray,fill=gray!30,thick] (\x*\spacing,\y*\spacing) circle (\siterad);
        }
        }

        \coordinate (Q1) at (0,0);
        \coordinate (Q2) at (\latwidth*\spacing,0);

        \draw[draw=red,thick] (Q1) to[in=-90,out=-90,distance=20] (Q2);

        \draw[draw=red,fill=red!30,thick] (Q1) circle (\siterad) node[below]{$Q_1$};
        \draw[draw=red,fill=red!30,thick] (Q2) circle (\siterad) node[below]{$Q_2$};

        \node[below,red] at (\latwidth*\spacing*0.5,-0.5) {$\ket{\Phi_{Q_1Q_2}}$};
        
    \end{tikzpicture}
    \caption{Sketch of the setup for Theorem~\ref{thm:2Dentanglementgeneration}. A short-ranged entangled state $\ket{\Psi}$ is prepared on a $2D$ lattice of qubits of dimensions $\ell_X\times \ell_Y$ by a constant-depth circuit $U$, where $\ell_X = \Theta(\poly(\log R))$ and $\ell_Y = \Theta(R)$ for some $R > 0$. Two qubits $Q_1,Q_2$ (red) separated by distance $\ell_Y$ are distinguished, and the remainder of the vertices (gray) are measured in the computational basis. This produces a Bell pair $\ket{\Phi_{Q_1Q_2}}$, up to a classically computable Pauli correction on $Q_1$, and the entire procedure is robust to local stochastic Pauli noise.}
    \label{fig:mainresultsketch}
\end{figure}

 We define the notion of local stochastic errors below (see Section~\ref{sec:localstochasticpaulinoise}).  Roughly speaking, this corresponds to randomly inserting Pauli operators into a quantum circuit such that the probability of a given set of $\ell$~spacetime locations being faulty is bounded by~$O(p^\ell)$, where $p>0$ is the noise strength. In particular, the errors may be correlated across both space and time.

Theorem~\ref{thm:2Dentanglementgeneration} amounts to a one-shot protocol with constant-size (in fact: single-qubit) repeaters at each site: Each  qubit is initialized in the state~$\ket{0}$ (a local operation at each repeater). Subsequently, the circuit~$U$ is applied, which requires local operations  (for single-qubit gates) and quantum communication (for two-qubit gates) along each link. Finally, each repeater qubit is measured. The (combined) result~$z$ determines a Pauli correction~$P(z)$. 
Eq.~\eqref{eq:correctalwayswithoutnoise} expresses the fact that the produced state on the qubits~$Q_1Q_2$ is the Bell state up to the Pauli correction~$P(z)$. Eq.~\eqref{eq:noisyversionqgeneration} states that this is still true with constant probability even if all operations are noisy.

In addition to its immediate operational relevance to long-distance entanglement generation,  our work provides fundamental insight into the nature of entanglement in 2D~networks. To our knowledge, the state~$\ket{\Psi}$ introduced in  Eq.~\eqref{eq:centralstateshortrangelongrange}
is the first example of a short-range entangled state in 2D with long-range localizable entanglement  robust to local stochastic Pauli noise.
In other words, it is a state which
\begin{enumerate}[(a)]
\item\label{it:shortrangeentangled}
can be created by a constant-depth, 2D-local unitary circuit from a product state, and
\item
allows to extract long-range entanglement by single-qubit measurements of ``bulk qubits'' even if it is corrupted by local stochastic Pauli noise.
\end{enumerate}
We also note that Property~\eqref{it:shortrangeentangled}  of being short-ranged entangled means that $\ket{\Psi}$ is the ground state of a local (stabilizer) Hamiltonian~$H$, a fact which may be used to provide alternative protocols to prepare it. This Hamiltonian is of independent interest: In Section~\ref{sec:localizableentanglementfinite}, we show that Gibbs states associated with~$H$ have long-range localizable entanglement below a constant temperature.

\subsection{Prior work}\label{sec:related}
Here we briefly discuss the relationship between our result and existing work on (one-shot) entanglement generation. The latter  can roughly be categorized according to the underlying network, i.e., the spatial arrangement of (repeater) nodes and their connectivity. 

\paragraph{1D: Quantum repeaters.} 
 In 1D,  entanglement swapping~\cite{zukowski1993event} through a chain of intermediate nodes having
 a constant number of qubits each can establish long-range entanglement in a one-shot manner. However, in the presence of noise,
applying this respectively arbitrary one-shot protocols results in a fidelity which decays exponentially with the  distance~\cite{acin2007entanglement}.

This barrier can be circumvented by increasing the sizes of the intermediate systems (maintaining the one-shot property)
 and/or iteratively distilling imperfect entanglement to the desired fidelity (going beyond one-shot protocols). Such approaches are known as quantum repeater protocols (see e.g.,~\cite{briegel1998quantum,muralidharan2016optimal,azuma2023quantum}), and can be categorized (following  the review paper~\cite{munro2015inside}) into three distinct ``generations''. In the first generation (see e.g., ~\cite{briegel1998quantum,sangouard2011quantum}), entanglement is created between adjacent nodes through heralded entanglement generation (HEG), and errors at each level of entanglement swapping are corrected through heralded entanglement purification (HEP). Both HEG and HEP are probabilistic methods which must be repeated over multiple rounds to succeed. In the second generation (see e.g.,~\cite{jiang2009quantum,munro2010quantum,zwerger2014hybrid}), HEP is replaced by using quantum error correction to deterministically correct errors, reducing the number of shots required and decreasing the classical communication overhead. In the third generation (see e.g.,~\cite{fowler2010surface,muralidharan2014ultrafast,muralidharan2017overcoming,munro2012quantum}), HEG is also replaced with techniques from quantum error correction, leading to further asymptotic improvements in time and communication cost, albeit at the expense of more challenging implementation. 
 Such protocols differ from our work in that they either allow for large local systems of qubits (scaling logarithmically with the distance) at each repeaters, or
extensive (e.g., polynomial) execution time (quantum circuit depth). We consider constant-sized network nodes (i.e., the local Hilbert space dimension is independent of the targeted distance) and one-shot protocols. 
The possibility of such one-shot protocols in 1D is ruled out by the no-go results of~\cite{acin2007entanglement} (see also~\cite{choe2024long}).

\paragraph{2D: One-shot protocols and entanglement percolation.} In 2D, Ref.~\cite{acin2007entanglement} demonstrates the possibility of fault-tolerant long-range entanglement generation in a planar lattice starting from 
imperfect bipartite two-qubit entangled states shared along each edge of the lattice. 
They argue that -- provided  the error rate~$p_e$ is below a certain threshold -- bipartite entanglement between distant pairs of qubits can be created with constant probability. Their argument connects this noise threshold to the percolation threshold of the underlying lattice~\cite{grimmett1999percolation}. It heavily relies on the fact that the considered noise model is somewhat restricted: noise is only allowed in the initial state in the form of non-maximally entangled pure states (instead of Bell states) across links,
with the ideal Bell state replaced by such a corrupted state with probability~$p_e$ independently for each link. All subsequent on-site operations are executed noiselessly.
As a consequence, states on individual links can be projected by local projective measurements onto maximally entangled  states. The relation to percolation arises because entanglement swapping can subsequently be used to establish entanglement between two qubits whenever these are connected by a path of such Bell states.
Similar to the general notion of single-shot protocols discussed above, this protocol generates long-range entanglement in a constant time (up to a Pauli correction from entanglement swapping).

 The idea of entanglement percolation has been extended in several directions by follow-up works. For fixed local system dimension, entanglement percolation protocols have extensively been  studied and optimized to achieve thresholds surpassing the classical percolation thresholds of the associated lattice, where links are deleted with probability~$p_e$~\cite{perseguers2008entanglement,lapeyre2009enhancement,lapeyre2012distribution}.
 Further improvements have been achieved by considering  schemes involving multipartite shared entanglement between intermediate nodes~\cite{perseguers2010multipartite}. Common to these works is the consideration of  noise models which only affect  the initial state in a specific way rather than circuit-level noise affecting the protocol. Realizing the measurements required in these protocols using unitaries and single-qubit measurements (potentially requiring auxiliary qubits), these constructions can be interpreted as generating 2D short-range-entangled states with long-range localizable entanglement~\cite{poppetal05}; however, this localizable entanglement is only shown to be robust against the specific noise model considered (and not, e.g., to single-qubit depolarizing noise). In contrast, our work provides an example of a state with localizable entanglement robust to local stochastic Pauli noise.
 
 Progress in the direction of 2D entanglement generation under more general noise models (allowing e.g., mixed initial states) and, in particular, circuit-level noise was made in Ref.~\cite{perseguers2008one} mentioned above. 
 Here the authors design a single-shot protocol which is robust to noisy execution of the on-site operations at each intermediate node. They argue that bit-flip errors can be corrected using a matching algorithm given certain syndrome patterns (akin to error correction in the surface code), whereas phase errors are taken care of by encoding a logical qubit into a quantum error-correcting code at each node. The latter is spanned by two GHZ states with a logarithmic number of qubits. That is, the 
number of qubits at each  intermediate nodes scales logarithmically with the distance.

As mentioned above, in~\cite[Section III.C]{perseguers2008one}, the authors also outline a single-shot scheme akin to ours
using only a constant number of qubits per node. In more detail, they sketch a connection between 1D fault-tolerant circuits and 2D entanglement localization protocols using techniques from measurement-based quantum computing. This is essentially the same idea from which our construction is built. However, the detailed construction and associated fault tolerance threshold proof is novel to our work.

\begin{figure}
    \centering
    \begin{tabular}{l|c|c|c|c|}
        Method &  Ref. & Noise model  & Local \# of qubits   \\
        \hline\hline
        Entanglement percolation &\cite{acin2007entanglement} & non-maximal entanglement& $O(1)$ \\
        Local encoding & \cite{perseguers2008one} & circuit-level i.i.d. & $O(\log R)$ \\
                & This work & local stochastic & $O(1)$ \\
        \hline\hline
                Cluster state-based &\cite{raussendorf2005long} & thermal & $O(1)$ \\
                                           &\cite{perseguers2010fidelity} & i.i.d. & $O(1)$  \\
         & \cite{bravyi2020quantum,choe2024long} & local stochastic & $O(1)$   \\
    \end{tabular}
    \caption{Single-shot entanglement generation: We compare 
    existing one-shot entanglement generation protocols in 2D and 3D in terms of resource requirements. The table gives the number of qubits needed at each node to generate a maximally entangled pair of two physical qubits at distance $R$ up to constant error. 
    In order to fit into this comparison, we must assume
ideal (decoding) operations on two ``boundary'' lattices to move the encoded entangled states to two physical qubits for Refs.~\cite{raussendorf2005long} and~\cite{bravyi2020quantum}. 
We distinguish between different noise models:
\emph{thermal} noise refers to a setting where each qubit
undergoes a phase error with some probability; the corresponding initial state is then a thermal state of the parent (stabilizer) Hamiltonian of the cluster state. 
\emph{i.i.d.} noise models assume independent Pauli errors on each qubit, either only on the shared initial state or at each circuit level.
     \emph{Local stochastic}  model: All operations, including initial entanglement distribution and local measurements, take place with Pauli faults at each circuit location under a local stochastic model (see Section~\ref{sec:localstochasticpaulinoise}). \emph{Non-maximal entanglement} model: The initial bipartite entangled states between adjacent nodes are prepared independently with some probability of error, whilst subsequent local operations occur noiselessly.}\label{fig:comparison}
\end{figure}

\paragraph{3D: Cluster states.} The most direct point of comparison for our work comes from existing protocols for entanglement localization using cluster states~\cite{raussendorf2005long,perseguers2010fidelity,bravyi2020quantum,choe2024long} in $3D$. In pioneering work~\cite{raussendorf2005long}, the authors established that the $3D$ cluster state~\cite{briegel2001persistent} possesses localizable entanglement which is resilient to independent $Z$-noise on each qubit: By measuring qubits in the bulk of (possibly noisy) cluster state on a cubic lattice, a high-fidelity  maximally entangled state encoded in the surface code 
on each of two opposing faces be generated. This conclusion was later shown to hold in the  more general case of local stochastic Pauli noise~\cite{bravyi2020quantum}. Via a modified scheme including additional measurements on the qubits on the target, these techniques were modified to include single-shot decoding of the logical qubits~\cite{perseguers2010fidelity,choe2024long}.
Ref.~\cite{choe2024long} showed robustness against local stochastic Pauli noise, establishing that constant-fidelity entanglement  between a pair of physical qubits  can be generated fault-tolerantly at arbitrary distance for noise strengths below some threshold.  These works imply that the cluster state provides an example of a short-range entangled state with long-range localizable entanglement robust to local stochastic Pauli noise.
As first observed in~\cite{raussendorf2005long}, a lattice of dimensions
$\Theta(\log R)\times \Theta(\log R)\times \Theta(R)$  is sufficient to generate entanglement at distance~$R$. 
The $\Theta(\log^2 R)$-scaling of the number of qubits in each ``slice'' along the ``communication'' direction was shown to be optimal up to a logarithmic factor in~\cite[Section 8.1]{choe2024long}.

Our work provides an example of a short-range entangled state in 2D with robust long-range entanglement. To obtain constant-fidelity entanglement at distance~$R$, our construction requires a lattice of dimensions~$\Theta(\poly(\log R))\times \Theta(R)$.  Unlike the 3D cluster state, our 2D construction is not translation-invariant.

\Cref{fig:comparison} provides a summary and comparison of these techniques with our work.

\subsection{Our contribution}
Our construction can be summarized in terms of the following main steps.

\paragraph{Fault tolerance against local stochastic Pauli noise in 1D for Clifford circuits.}
We argue that qubits arranged on a line with (noisy) one-qubit and nearest-neighbor two-qubit operations can serve as a basis for a quantum fault tolerance scheme.  We refer to Section~\ref{sec:ft1Dlieterature} for a discussion of related works showing how to achieve general fault-tolerant quantum computation with geometrically local operations. Our main result here applies to the implementation of 1D-local Clifford circuits:
 We give fault tolerance threshold theorems showing that such circuits can be implemented fault-tolerantly in the presence of local stochastic Pauli noise of strength below some a certain threshold. 

We first consider prepare-and-measure circuits as commonly considered in quantum computing: Here one is given a unitary circuit~$U$, and 
the goal is to (approximately) sample from the distribution of measurement outcomes when measuring the state~$U\ket{0^n}$ in the computational basis.  We show the following, see Theorem~\ref{thm:ftpathgraph} for further details:
\begin{result}[Fault-tolerant 1D-local implementation of prepare-and-measure Clifford circuits]
\label{res:ftimplementation1Dprepmeas}
There is a constant threshold error strength $p_*\in (0,1)$ such that the following holds. 
Let $\varepsilon>0$ be arbitrary. Let $\cC$ be a $1D$-local prepare-and-measure Clifford-circuit.
Then there is a $1D$-local implementation~$\cC_{FT}$ 
which  produces a sample from a distribution which is $\varepsilon$-close to the output distribution of~$\cC$ in the presence of local stochastic Pauli noise of any strength~$p\leq p_*$. The construction has polylogarithmic qubit and depth overhead.
\end{result}
We emphasize that Step~\ref{res:ftimplementation1Dprepmeas} merely illustrates our techniques, and note that significantly stronger results relevant to quantum computation have been obtained (see the discussion in Section~\ref{sec:ft1Dlieterature}). In particular, these constructions are neither restricted to Clifford circuits nor to Pauli noise. Instead, they consider general circuits based on a universal gate set, and arbitrary local stochastic noise (and/or adversarial noise).

For our purposes, dealing with Clifford circuits is sufficient to analyze entanglement generation.  The restriction to Clifford circuits is, in fact, essential for our approach: It allows to convert unitary subcircuits by gate teleportation (i.e., adaptive measurement-based) circuits, and permits to postpone Pauli corrections (effectively eliminating mid-circuit adaptive operations. Furthermore, the additional restriction to local stochastic Pauli noise allows us to  propagate errors through circuits while controlling their strength.  In particular, this allows us to analyze the effect of substituting subcircuits in a fault tolerance construction.
Such subcircuit substitutions can be used to relate circuits with different (but closely related) geometries, as well as to convert between unitary and measurement-based circuits. 

Step~\ref{res:ftimplementation1Dprepmeas} is obtained by applying this kind of analysis to the fault tolerance construction of Ref.~\cite{stephens2007universal}, which establishes an analogous result for
circuits whose  operations are local on a bilinear array of qubits. We use circuit transformations which relate the bilinear array to a line of qubits, and analyze their effect of local stochastic errors. This establishes Step~\ref{res:ftimplementation1Dprepmeas}.

\paragraph{Robust implementation of circuits.}
To construct fault-tolerant entanglement-generation protocols, statements of the kind given in Step~\ref{res:ftimplementation1Dprepmeas} are insufficient because the ultimate goal is to prepare a quantum state (rather than sample from a certain distribution). This means we have to deviate from the standard paradigm of fault-tolerant quantum computation: We need to ask how to render more general circuits robust to noise, namely circuits producing a quantum output state (i.e., where not every qubit is measured at the end of the circuit). 
Even more generally,
one can consider the problem of fault-tolerantly implementing circuits which realize  quantum channels, i.e., completely positive trace-preserving (CPTP) maps taking a number of input qubits to a certain number of output qubits.  We refer to Section~\ref{sec:noisyencodingdecoding} for a more detailed discussion of this extension of quantum-fault tolerance considerations.

Concretely, suppose we have a circuit~$\cC$ which acts on~$N$ qubits, $N_{\mathsf{in}}\leq N$ of which are ``input qubits''
(taking an arbitrary $N_{\mathsf{in}}$-qubit state),  whereas the remaining $N-N_{\mathsf{in}}$~qubits are initialized in the state~$\ket{0^{N-N_{\mathsf{in}}}}$. At the end of the computation, a state of~$N_{\mathsf{out}}$ qubits is returned (whereas the remaining~$N-N_{\mathsf{out}}$ qubits are discarded). The problem of fault-tolerantly realizing such a circuit has recently attracted some attention (see e.g., Ref.~\cite{christandl2026fault} titled ``Fault-tolerant quantum input/output'').  Unlike in ``standard'' quantum fault tolerance, it should be emphasized that the input and output-qubits in this setting are not encoded, i.e., they are bare physical qubits subject to noise. As a consequence, corresponding accuracy guarantees when considering e.g., success probabilities, necessarily depend on the number of qubits involved (unlike e.g., in Step~\ref{res:ftimplementation1Dprepmeas}). 

We show that the behaviour of a general $N$-qubit circuit (possibly based on universal gate sets) with (any) number $N_{\mathsf{in}},N_{\mathsf{out}}\leq N$ of input, respectively output qubits can be implemented robustly by a concatenated fault tolerance construction involving noisy encoding and decoding maps. 
The corresponding result applies to general local stochastic (not necessarily Pauli) noise and leads to the following result. Here failure
refers to the fault-tolerant circuit not acting identically as the original circuit when seen as a CPTP map. 
\begin{result}[Robust circuit implementation]\label{res:resultrobust}
Let $N\in\mathbb{N}$ and $\varepsilon\in [0,1]$ be arbitrary. 
There is threshold error strength $p_0(N,\varepsilon)=O(\varepsilon/N) $ such that the following holds.
Any  circuit~$\cC$ on~$N$ qubits composed of gates belonging to a certain universal gate set 
can be implemented by circuit~$\cC_{\mathsf{FT}}$ under local stochastic  noise of strength~$p\leq p_0(N,\varepsilon)$ 
 except with failure probability~$\varepsilon$, with 
qubit respectively depth overhead which is polylogarithmic
in~$1/\varepsilon$ and the circuit size of~$\cC$.
\end{result}
\noindent We refer to Theorem~\ref{thm:robustimplementationstochastic} below for the exact form of~$p_0(N,\varepsilon)$. The $O(1/N)$-dependence of ths function here is  optimal e.g., when the number of output qubits is~$N_{\mathsf{out}}=\Omega(N)$.

\paragraph{1D-local robust implementation of Clifford circuits.}
Combining our construction of robust circuit implementation 
with the fault tolerance gadgets of  Ref.~\cite{stephens2007universal}, and our analysis of circuit transformations for changing the locality of circuits, we obtain the following ``geometrically local'' version of Step~\ref{res:resultrobust}. It applies to circuits composed of Clifford unitaries and local stochastic Pauli noise only.
\begin{result}[Robust $1D$-local Clifford circuit implementation]\label{res:resultft1dlocal}
Let $N\in\mathbb{N}$ and $\varepsilon\in [0,1]$ be arbitrary. 
There is threshold error strength $p_0(N,\varepsilon)=O((\varepsilon/N)^{\Theta(1)})$ such that the following holds.
Any $1D$-local (unitary) Clifford circuit~$\cC$ on~$N$ qubits
can be implemented by a $1D$-local circuit~$\cC^{\mathsf{1D}}$ under local stochastic Pauli noise of strength~$p\leq p_0(N,\varepsilon)$, 
 except with failure probability~$\varepsilon$, with polylogarithmic  
qubit respectively depth overhead.
\end{result}
\noindent A detailed statement including explicit bounds on $p_0(N,\varepsilon)$ is given in Theorem~\ref{thm:implementation1D}. Importantly, this threshold error strength depends polynomially on $\varepsilon$ and  inversely polynomially  on~$N$. Such a dependence is unavoidable because  the constructed circuit 
takes as input physical (i.e., unencoded) qubits, and outputs unencoded qubits. For the goal of entanglement generation, we will consider very simple circuits~$\cC$ with $N=O(1)$.

\paragraph{$1D$-local robust implementation with  constant qubit lifespan.}
As an intermediate next step, we modify
the circuit~$\cC^{\mathsf{1D}}$ given in Step~\ref{res:resultft1dlocal}
by replacing Clifford gates by  gate teleportation circuits, and postponing Pauli corrections to the end of the circuit, see Section~\ref{sec:qubitresets}. This results in a new circuit~$\cC^{\mathsf{1D,res}}$ with the special property that each qubit is active only for a constant amount of time before being reset (i.e., reinitialized in the state~$\ket{0}$). 
In more detail, let the lifespan of a qubit
be the maximum number of gate layers between between two successive  (re)initializations of the qubit. The qubit lifespan of the circuit refers to the maximum lifespan of any qubit. In this terminology, we establish the following (see Theorem~\ref{thm:1Dresets} for details).
\begin{result}[Robust $1D$-implementation of Clifford circuits with mid-circuit resets]
Let $N\in\mathbb{N}$ and $\varepsilon\in [0,1]$ be arbitrary. 
There is threshold error strength $p_0(N,\varepsilon)=O((\varepsilon/N)^{\Theta(1)})$ such that the following holds.
Any $1D$-local (unitary) Clifford circuit~$\cC$ on~$N$ qubits
can be implemented by a $1D$-local circuit~$\cC^{\mathsf{1D,res}}$ under local stochastic Pauli noise of strength~$p\leq p_0(N,\varepsilon)$, 
 except with failure probability~$\varepsilon$, with polylogarithmic  
qubit respectively depth overhead. Furthermore, the circuit~$\cC^{\mathsf{1D,res}}$ has constant qubit lifespan.
\end{result}
The circuit~$\cC^{\mathsf{1D,res}}$ can be seen as hybrid
between measurement-based computation (MBC) and the quantum circuit model: gates are implemented by measurements on a resource state, but unlike the MBC, new resource states are prepared on the fly throughout the computation.

\paragraph{$2D$-local constant-depth implementation by a space-time transformation.}
As a final step, we parallelize the $1D$-local circuit~$\cC^{\mathsf{1D,res}}$ by replacing each reinitializing of a  qubit by the introduction of a new physical qubit in the state~$\ket{0}$, see Section~\ref{sec:2dsimulation}. The resulting circuit~$\cC^{\mathsf{2D}}$ has the property that qubits can be arranged in a rectangular grid in such a way that the circuit becomes $2D$-local. Furthermore, it is of constant depth as the measurements can be applied simultaneously. 
The $2D$-grid is of linear dimensions $\ell_X\times \ell_Y$, where 
$\ell_X$ is linear in the number $N(\cC^{\mathsf{1D,res}})$ of qubits of the circuit~$\cC^{\mathsf{1D,res}}$, whereas~$\ell_Y$ is linear in the circuit depth~$\mathsf{depth}(\cC^{\mathsf{1D,res}})$.  In other words,  the construction of~$\cC^{\mathsf{2D}}$ ``unwraps'' the circuit~$\cC^{\mathsf{1D,res}}$ in time, transforming time (circuit depth) into space (qubits). We obtain the following, see Theorem~\ref{thm:2Dimplementation} for details.
\begin{result}[Robust constant-depth $2D$-local Clifford circuit implementation]\label{re:result2Dconstantdepth}
Let $N\in\mathbb{N}$ and $\varepsilon>0$ be arbitrary.
Then there is a  threshold error strength $p_0^{\mathsf{2D}}(N,\varepsilon)=O((\varepsilon/N)^{\Theta(1)})$ such that the following holds.
Let $\cC$ be a $1D$-local (unitary) Clifford circuit on~$N$ qubits, and let
$\upperboundL\geq |\mathsf{Loc}(\cC)|)$ be an upper bound on the number of circuit locations of~$\cC$. Then $\cC$ can be implemented by a $2D$-local constant-depth circuit~$\cC^{\mathsf{2D}}$ under local stochastic Pauli noise of strength~$p\leq p_0^{\mathsf{2D}}(N,\varepsilon)$, 
 except with failure probability~$\varepsilon$.
The circuit~$\cC^{\mathsf{2D}}$ acts on a rectangular array of~$\ell_X\times \ell_Y$ qubits, where 
\begin{align}
\ell_X &=N(\cC)\cdot \Theta\left(\mathsf{poly}(\log 1/\varepsilon, \log \upperboundL\right)\\
\ell_Y &=\mathsf{depth}(\cC)\cdot \Theta\left(\mathsf{poly}(\log 1/\varepsilon, \log \upperboundL)\right)\ .
\end{align}
Its input and output qubits are located at opposite sides of the array, and are spaced equidistantly with separation
\begin{align}
\mathsf{dist}=\Theta\left(\mathsf{poly}(\log 1/\varepsilon, \log \upperboundL)\right)\ .
\end{align}
\end{result}
The result expressed by Step~\ref{re:result2Dconstantdepth} immediately gives rise to fault-tolerant one-shot entanglement generation protocols 
in~$2D$. For instance, we can apply the construction to a simple  two-qubit Clifford circuit~$\cC^{\mathsf{prep}}$ which prepares a Bell state from a product state. 
Setting $\upperboundL:=e^{R}$, this gives  a fault-tolerant  entanglement generation protocol generating a Bell pair on qubits separated by a distance~$\Theta(R)$ on an (essentially) square grid, see Lemma~\ref{lem:squaregridstategeneration}.

Our main result, expressed as Theorem~\ref{thm:2Dentanglementgeneration}, gives an entanglement generation protocol for distance-$\Theta(R)$ entanglement requiring a rectangular lattice of linear dimensions $\Theta(\mathsf{poly}(\log R))\times\Theta(R)$ only. It is obtained by constructing a circuit~$\cC^{2D}$ (using Step~\ref{re:result2Dconstantdepth}) realizing the identity channel between two distant qubits; this circuit is applied to the output of a (local) Bell pair preparation circuit. We refer to Theorem~\ref{thm:2Dentanglementgenerationgeneral} for details.

\section*{Conclusions and open problems}
To our knowledge, our one-shot $2D$-local protocol for entanglement generation is the first fully worked out proposal of this kind with an analytically proven fault tolerance threshold error strength against local stochastic Pauli noise (see Theorem~\ref{thm:2Dentanglementgeneration}). The established threshold value is expressed  in terms of the threshold~$p_*$ of the fault tolerance construction for fault-tolerant quantum computing based on a bilinear qubit array (see Ref.~\cite{stephens2007universal}), as well as combinatorial constants $(\Lambda_{\mathsf{1D,res}},\lambda_{\mathsf{1D,res}})$ depending
on the circuit transformations we use  (see Eq.~\eqref{eq:pzerooneDresdef} and Theorem~\ref{thm:2Dentanglementgenerationgeneral} below). Estimates on the value of~$p_*$ (for closely related settings) can be found in Ref.~\cite{stephens2007universal}, whereas the constants $(\Lambda_{\mathsf{1D,res}},\lambda_{\mathsf{1D,res}})$ 
could easily be computed from the expressions we provide. However, our focus is here is not on specific numbers, but rather on the problem of  establishing a threshold theorem in the first place. Indeed, in many places, we choose to use overly conservative bounds to simplify the analysis. As is common in fault tolerance considerations of this kind, analytical thresholds are typically orders of magitude smaller than actual thresholds; here follow-up work could provide more insights e.g., by Monte-Carlo simulations.

In terms of fault tolerance and the use of entanglement generation protocols, several problems remain. For example, while our construction is based on an $L$-fold  concatenated code and associated gadgets, other codes could provide benefits such as lower thresholds, a larger entanglement yield, or optimal width of the necessary $2D$ lattice (i.e., scaling as $\log R$ rather than $\poly(\log R)$, matching the lower bound of Ref.~\cite{choe2024long}). For example, a new scheme for fault-tolerant quantum computation in~$1D$ was recently proposed in Ref.~\cite{bergamaschi2024fault}. Inspired by~\cite{yamasaki2024time} which was based on  concatenated Hamming codes, the proposal of Ref.~\cite{bergamaschi2024fault} 
yields a constant encoding rate (i.e., a constant qubit overhead) as well as a quasi-polylogarithmic time overhead. It is natural to ask to what extent our approach can be applied to this construction. A new resource analysis would be needed to cover such constructions. First steps in this direction (but without considering locality) are taken in~\cite{belzig2026constant}. We note, however, that some of the usual metrics in quantum fault tolerance, such as constant qubit or depth/time overhead,
are not necessarily what is  needed in our setting: Indeed,  the depth and qubit overhead can translate into the targeted distance of the generated long-range entanglement (hence need not necessarily be minimized in our setup).

Extending and applying our work beyond entanglement distribution poses fascinating challenges. For example, a natural problem which may be accessible with further fault-tolerant analysis would be to establish whether our protocol can be parallelized, i.e., whether many EPR pairs can be produced simultaneously up to local stochastic noise acting on their joint system.  Such a parallel repetition theorem may be used to render general fault tolerance schemes geometrically local in~$2D$ following an approach similar to that used in Ref.~\cite{choe2025fault}. Furthermore,  adapting the techniques of Ref.~\cite{bravyi2020quantum} from $3D$ to $2D$, our construction could also serve as a building block for establishing that noisy $2D$-local constant-depth quantum circuits are computationally more powerful that comparable (non-local) classical circuits.

On the level of many-body physics, our work provides the first  examples of a short-range entangled state in $2D$ with long-range localizable entanglement robust to local stochastic noise. While the state belongs to the trivial, i.e.,  product state phase when 
classifying many-body states in terms of equivalence under local constant-depth circuits (or constant-time evolution under a local Hamiltonian, see e.g.,~\cite{PRBChenetal10}), it nevertheless acts as a resource for generating long-range entanglement by local operations. Furthermore, this property is preserved under local stochastic Pauli noise. In a similar vein, we show that thermal states of a certain $2D$-local Hamiltonian at sufficiently low (constant) temperature exhibit long-range localizable entanglement. This is akin to the cluster-state construction in 3D pioneered by Raussendorf, Bravyi and Harrington~\cite{raussendorf2005long}. 

\paragraph{Outline}
Sections~\ref{sec:backgroundquantumcircuitsnoise} 
and~\ref{sec:fault tolerance} provide background material and introduce basic concepts and notation: Section~\ref{sec:backgroundquantumcircuitsnoise} discusses adaptive and non-adaptive quantum circuits, and different noise models. In Section~\ref{sec:fault tolerance}  we give a high-level introduction to quantum fault tolerance using code concatenation. We also discuss quantum fault tolerance based on geometrically constrained two-qubit interactions.

We then derive our main results: In Section~\ref{sec:noisyencodingdecoding}, we extend the standard framework of quantum fault tolerance to include (concatenated) encoding and decoding circuits subject to noise.   In Section~\ref{sec:errorpropagation}, we discuss how Pauli errors propagate through (possibly adaptive) Clifford circuits. 
In Section~\ref{sec:ftimplementcliff}, we show how to implement a $1D$-local Clifford circuit fault-tolerantly in $1D$.  
In Section~\ref{sec:qubitresets}, we give a $1D$ fault-tolerant implementation of Clifford circuits which only uses qubits of limited lifespan.  We then ``unravel'' this construction in time, giving a fault-tolerant constant-depth implementation  of a $1D$-local Clifford circuit based on a $2D$ array of qubits. 
Finally, in Section~\ref{sec:longdistancegeneration} we apply this result to construct one-shot entanglement generation protocols.

\section{Quantum circuits and noise\label{sec:backgroundquantumcircuitsnoise}}
Here we introduce basic notions and notation. In Section~\ref{sec:quantumcircuitsaschannels} we discuss quantum circuits. Contrary to commonly used terminology (where circuits typically only refer to sequences of unitary gates), we describe circuits in terms of quantum channels, and include both state preparation and measurement steps in our description.
This is particularly useful when discussing mid-circuit measurements and subsequent adaptive operations which depend on the measurement results (see Section~\ref{sec:adaptivequantumcirc}). 

We then discuss noisy implementations of circuits. We introduce noise models of increasing generality, starting with Pauli noise in Section~\ref{sec:paulinoisecircuits}. Local stochastic Pauli noise is discussed in Section~\ref{sec:localstochasticpaulinoise}.
In Section~\ref{sec:locfaults} we review the notion of fault-locations and the definition of a quantum comb.  Finally, in Section~\ref{sec:localstochasticnoise}, we discuss the definition and basic properties of  general local stochastic noise.

\subsection{Quantum circuits as quantum channels\label{sec:quantumcircuitsaschannels}}
It will be convenient to describe quantum circuits in terms of quantum channels and quantum instruments.  Preparation amounts to the quantum channel (i.e., the completely positive trace-preserving (CPTP) map)~$\cE_{\mathsf{prep}}(\rho)=\tr(\rho)\proj{0}$, and application of a unitary~$U$  corresponds to the quantum channel~$\cU(\rho)=U\rho U^\dagger$. A measurement is described by a quantum instrument, i.e., a family $\{\cM^x\}_{x\in \cX}$ of completely positive (CP) maps indexed by the set~$\cX$ of possible measurement results such that $\sum_{x\in \cX}\cM^x$ is trace-preserving. When applying such a measurement to a state~$\rho$, the outcome~$x\in \cX$ is observed with probability~$p_x:=\tr(\cM^x(\rho))$, and the post-measurement state upon observing~$x$ is $\rho_x:=\cM^x(\rho)/p_x$.  Specifically, measurement of a qubit in the orthonormal basis~$\{\ket{0},\ket{1}\}$ is described by the instrument~$\cM_{\mathsf{meas}}:=\{\cM^x\}_{x\in \{0,1\}}$ where~$\cM^x(\rho)=\langle x,\rho x\rangle \cdot \proj{x}$.

The composition of a quantum channel given by a CPTP map~$\cE$ followed by a measurement specified by an instrument~$\{\cM^x\}_{x\in \cX}$  is the instrument~$\{\cM^x\circ \cE\}_{x\in\cX}$. Similarly,  a CPTP map~$\cE$ applied after an instrument~$\{\cM^x\}_{x\in \cX}$ corresponds to the instrument~$\{\cE\circ \cM^x\}_{x\in\cX}$. Applying a CPTP map~$\cE$ to the first factor a bipartite system and an instrument~$\{\cM^x\}_{x\in\cX}$ to the second factor amounts to applying the instrument~$\{\cE\bigotimes\cM^x\}_{x\in\cX}$.
By definition, we may think of a CPTP map~$\cE$ as an instrument with an outcome set~$\cX$ of size~$|\cX|=1$. We denote any such outcome set as $\cX=\{\bot\}$ and often drop this symbol altogether (e.g., identifying $\{0,1\}\times\{\bot\}$ with $\{0,1\}$). This is consistent with the following general composition properties of instruments: 
Two instruments~$\{\cM^x\}_{x\in \cX}$ and $\{\cN^y\}_{y\in \cY}$ applied consecutively amount to the instrument~$\{\cN^y\circ\cM^x\}_{(x,y)\in \cX\times \cY}$. The application of two instruments~$\{\cM^x\}_{x\in \cX}$ and $\{\cN^y\}_{y\in \cY}$ on each factor of a bipartite system is described by the instrument~$\{\cM^x\otimes\cN^y\}_{(x,y)\in \cX\times \cY}$.

We sometimes consider circuits where measurement results are used to 
determine subsequent actions. A classically controlled quantum channel is a family~$\cE:=\{\cE^x\}_{x\in \cX}$ of quantum channels. The composition of a quantum instrument~$\cM=\{\cM^x\}_{x\in\cX}$ with a classically controlled channel $\cE:=\{\cE^x\}_{x\in \cX}$ is the instrument
\begin{align}
\cE\circ\cM:=\{\cE^x\circ \cM^x\}_{x\in\cX}\ .
\end{align}
We typically assume that (a description of) the map~$x\mapsto \cE^x$  can be efficiently computed. An example is the case where the measurement extracts the syndrome of an error correcting code, and the associated classically controlled operation is a correction. The composition of the syndrome measurement and the classically controlled correction  is called a decoder.

Very often, $\cX=\{0,1\}^m$ consists of $m$-bit strings. In this situation, any channel can be classically controlled by a subset~$S:=\{i_1<\ldots<i_r\}\subset [m]:=\{1,\ldots,m\}$ of bits. Here the classically controlled channel is of the form~$\cE:=\{\cE^z\}_{z\in \{0,1\}^r}$ and gives rise to the instrument
\begin{align}
\{\cE^{\pi_S(x)}\circ\cM^x\}_{x\in \{0,1\}^m}\label{eq:compositionpis}
\end{align}
where $\pi_S(x)=(x_{i_1},\ldots,x_{i_r})$ projects on the $r$~control bits. Slightly abusing notation, we often simply write~$\cE\circ\cM$ for the instrument defined by Eq.~\eqref{eq:compositionpis} suppressing the explicit dependence on the subset~$S$. Similar reasoning applies for permutations of these control bits.

\subsection{Adaptive and non-adaptive quantum circuits\label{sec:adaptivequantumcirc}}
We typically consider circuits composed of a set~$\cO$ of basic operations. For example, in the case of Clifford circuits,~$\cO$ could consist of Clifford unitaries, single-qubit stabilizer state preparations and computational basis measurements.
Locality restrictions (e.g., the ability to only use nearest-neighbor two-qubit operations on some graph) can then be encoded into the choice of~$\cO$.
\begin{definition}[Valid set of operations]
A set of operations~$\cO$ for $n$~qubits~$Q_1\cdots Q_n\cong (\mathbb{C}^2)^{\otimes n}$
is a set of instruments and classically controlled channels of the form~$\cO_S$ where $S\subseteq [n]$ specifies  a subset
 $Q_S:=\bigcup_{s\in S}\{Q_s\}$ of qubits and $\cO_S$ is an instrument acting on the qubits~$Q_S$. It is called valid if it contains the identity~$\mathsf{id}_{\cB(Q_S)}$ for any non-empty subset~$S\subset[n]$, i.e., any subset of qubits.
\end{definition}

\begin{definition}[Circuits, circuit depth and operation layers]\label{def:nonadaptive circuit}
Let $\cO$ be  a valid set of operations. 
A  depth-$1$ circuit~$\mathsf{C}$ composed of operations belonging to~$\cO$ (or simply ``composed of~$\cO$'') is defined as a tensor product
\begin{align}
\mathsf{C}&=\bigotimes_{j=1}^r \cO_{S_j}
\end{align}
where the sets~$S_1,\ldots,S_j$ form a disjoint partition $[n]=\bigcup_{j=1}^r S_j$ of $[n]:=\{1,\ldots,n\}$ and $\cO_{S_j}\in \cO$ for every $j\in [n]$. A depth-$D$ circuit~$\cC$ is a composition 
\begin{align}
\cC=\cC_D\circ\cdots \circ\cC_1\label{eq:compositiondepthD}
\end{align} of~$D$ depth-$1$ circuits~$\cC_1,\ldots,\cC_D\in \mathsf{Circ}^{(1)}(\cO)$ we refer to as (operation) layers of the circuit.  
We also write~$\mathsf{depth}(\cC)=D$ for a circuit composed of $D$ operation layers.
\end{definition}

\begin{definition}[Non-adaptive and adaptive circuits]
We call a circuit~$\cC$ as in  Definition~\ref{def:nonadaptive circuit} non-adaptive if none of the operations~$\cO_{S_j}$ used (and thus none of the operations~$\cC_1,\ldots,\cC_D$) are classically controlled quantum channels. 
In case $\cC$ contains classically controlled quantum channels and these appear in Eq.~\eqref{eq:compositiondepthD}, the corresponding composition has to be consistent with the temporal order and is left implicit (see the remark after Eq.~\eqref{eq:compositionpis}). In this case we refer to~$\cC$ as an adaptive circuit. 
\end{definition}
Consider a non-adaptive circuit~$\cC$. We note that in Definition~\ref{def:nonadaptive circuit}, each operation~$\cE_{S_j}$  is an instrument
which may be a measurement. In particular,  for each~$t\in [D]$, the layer~$\cC_t$ of operations is an instrument
\begin{align}
\cC_t=\{\cC^{x_t}_t\}_{x_t\in \cX_t}\qquad\textrm{ where }\qquad \cX_t:=\begin{cases}
\{0,1\}^{m_t}\qquad &\textrm{ if }m_t>0\\
\{\bot\}&\textrm{ if }m_t=0 \ ,
\end{cases}
\end{align}
and where $m_t$ is the number of single-qubit measurements applied in the gate layer~$\cC_t$.

The following makes the notion of temporal order more explicit: An adaptive circuit~$\cC$ is an instrument $\cC=\{\cC^x\}_{x\in \cX}$ where $\cX:=\cX_1\times\cdots\times\cX_D$ and
where for each ``history'' $x=(x_1,\ldots,x_D)\in \cX_1\times\cdots\times\cX_D$ of measurement results, the CP map~$\cC^x$ is the composition
\begin{align}
\cC^x&=\cC^{x^{(D)}}\circ\cdots\circ \cC^{x^{(1)}}\ .
\end{align}
Here we write $x^{(t)}=(x_1,\ldots,x_t)$ for the ``history'' up to time~$t\in [D]$.

\subsection{Pauli noise on circuits\label{sec:paulinoisecircuits}}
It is convenient to represent a circuit by a diagram where boxes are used for individual operations, and input and output qubit systems are represented by wires. 
We call~$\cW_{\cC}:=[D]\times [n]$ the set of space-time qubit locations/wires of~$\cC$: each element $(t,j)\in\cW_{\cC}$ specifies a ``wire'' representing the qubit~$j$ after time~$t$ (i.e., after the gate layer~$\cC_t$). 
We now consider noisy implementations where each such wire is affected by a Pauli error. Let $\pauli_n$ denote the Pauli group on $n$~qubits generated by single-qubit Pauli-$X$ and Pauli-$Z$-operators.
\begin{definition}
Let $\cC=\cC_D\circ\cdots \circ\cC_1$ be a depth-$D$ circuit on $n$~qubits. A Pauli error/Pauli fault~$F$ on~$\cC$ is a Pauli-valued function~$F:\cW_{\cC}\rightarrow \pauli_1$ associating a single-qubit Pauli~$F(w)$ to every wire~$w\in \cW_{\cC}$.
The  support of~$F$ is the set of wires
\begin{align}
\supp(F):=\left\{w\in \cW_{\cC}\ |\ F(w)\not\propto I\right\}
\end{align}
on which the error acts non-trivially. We write 
\begin{align}
F\bowtie\cC:=(\cF_D\circ\cC_D)\circ\cdots\circ(\cF_1\circ\cC_1)
\end{align}
for the alternating composition of each layer~$\cC_t$ with the $n$-qubit quantum channel 
\begin{align}
\cF_t(\rho)&= \left(\bigotimes_{j=1}^n \cF(t,j)\right)(\rho)
\end{align}
where $\cF(t,j)(\rho)=F(t,j)\rho F(t,j)^\dagger$ is the single-qubit channel conjugating by the Pauli~$F(t,j)$.  We call~$F\bowtie\cC$ a noisy execution of~$\cC$ with Pauli error~$F$.
\end{definition}
Observe that this definition applies to both  adaptive as well as non-adaptive circuits.

\subsection{Local stochastic Pauli noise\label{sec:localstochasticpaulinoise}}
We formulate our claims in terms of local stochastic Pauli noise, a paradigmatic noise model 
introduced by Gottesman~\cite{gottesman2014fault} which commonly studied in quantum error correction. We note, however, that the fault tolerance threshold theorems we rely on apply more generally to adversarial local stochastic noise of possibly non-Pauli form. 
Indeed, the discussion about noisy decoding can also easily be adapted to such a more general scenario.

We note, however, that the reduction from adaptive to non-adaptive Clifford circuits requires noise of Pauli form. As this step is essential in our construction, it is a priori unclear how to obtain our main results without the assumption of Pauli noise. 
\begin{definition}[Local stochastic Pauli noise]
Let $\mathsf{C}=\cC_D\circ \cdots \circ\cC_1$ be a depth-$D$ circuit on $n$~qubits.
Let $\cW_{\mathsf{C}}:=[D]\times [n]$ be the set of wires of~$\mathsf{C}$. Stochastic Pauli noise on~$\mathsf{C}$ is a random variable~$F$ which takes as values elements in~$\pauli_1^{\cW_{\mathsf{C}}}$, i.e., functions~$f:\cW_{\mathsf{C}}\rightarrow \pauli_1$. In other words, it associates a random single-qubit Pauli~$F(w)$ to every wire~$w\in \cW_{\cC}$ according to some joint distribution. 
The (random) CPTP map $F\bowtie\cC$ is called a noisy implementation of~$\mathsf{C}$ with stochastic noise~$F$.

Stochastic Pauli noise~$F$ is called local of strength~$p\in [0,1]$ if
\begin{align}
\Pr_F\left[W\subseteq \supp(F)\right]&\leq p^{|W|}\qquad\textrm{ for any subset }\qquad W\subseteq\cW_{\mathsf{C}}
\end{align}
of wires.  We write $F\sim \cN^{\pauli}_{\cC}(p)$ if $F$ is local stochastic Pauli noise of strength~$p$ on the circuit~$\mathsf{C}$. 
\end{definition}
It will be convenient to introduce the following definitions: If $F_1$ and $F_2$ are stochastic Pauli noise on a circuit~$\mathsf{C}$, then $F_1\cdot F_2$ is the local stochastic Pauli noise defined in terms of $(F_1,F_2)$ as
\begin{align}
(F_1\cdot F_2)(w)&:=F_1(w)F_2(w)\qquad\textrm{ for all }\qquad w\in \cW_{\mathsf{C}}\ .
\end{align}

Local stochastic Pauli noise satisfies a number of basic facts which will be useful for us.
The following is an adaptation of~\cite[Lemma 11]{bravyi2020quantum}   to the ``spacetime'' setting of local stochastic Pauli noise acting on a circuit~$\mathsf{C}$ (i.e., on a set~$\cW_{\mathsf{C}}$ of wires).
\begin{lemma}[Properties of local stochastic Pauli noise]\label{lem:pauli noise properties}
Let $\mathsf{C}=\cC_D\circ \cdots \circ\cC_1$ be a depth-$D$ circuit on $n$~qubits.
Let $\cW_{\mathsf{C}}:=[D]\times [n]$ be the set of wires of~$\mathsf{C}$. 
    \begin{enumerate}[(i)]
        \item\label{item:pauli noise property 1} Suppose $F \sim \cN_{\mathsf{C}}^{\pauli}(p)$, and $F^\prime$ is stochastic Pauli noise on~$\mathsf{C}$ 
         such that $\supp(F^\prime) \subseteq \supp(F)$ with certainty. Then $F^\prime \sim \cN_{\mathsf{C}}^{\pauli}(p)$.
         \item\label{item:permutationclaimb}
         Let $F\sim \cN_{\mathsf{C}}^{\pauli}(p)$ be local stochastic Pauli noise on~$\cC$. Let $\pi:\cW_\cC\rightarrow \cW_\cC$ be a permutation. Then $F\circ \pi\sim \cN_{\cC}^{\pauli}(p)$. 
         
         \item\label{item:pauli noise property 3} Suppose $F_i \sim \cN_{\mathsf{C}}^{\pauli}(p_i)$, $1\leq i \leq r$ are random Paulis which may be dependent. Then $F_1\cdot \dots \cdot F_r \sim \cN^{\pauli}_{\mathsf{C}}(p^\prime)$, where $p^\prime = r \max_i\{p_i^{1/r}\}$.
    \end{enumerate}
\end{lemma}
\begin{proof}
Claims~\eqref{item:pauli noise property 1} and~\eqref{item:permutationclaimb} immediately follow from the definition.

 For Claim~\eqref{item:pauli noise property 3}, consider a fixed subset $W \subseteq \mathcal{W}_{\mathsf{C}}$ of wires.  
 Summing over all disjoint partitions~$W=\bigcup_{j=1}^r W_j$ into (possibly empty, ordered) subsets $W_1,\ldots,W_r$, we have
     \begin{align}
        \prob_{F_1,\ldots,F_r}[W\subseteq \supp(F_1\cdot \cdots \cdot F_r)] &
               \leq \sum_{W=\bigcup_{j=1}^r W_j}\Pr\left[W_k\subseteq \supp(F_k)\textrm{ for all }k\in [r]\right]\\
               &\leq  \sum_{W=\bigcup_{j=1}^r W_j} \min_{k\in [r]} \Pr\left[W_k\subseteq \supp(F_k)\right]\\
               &\leq   \sum_{W=\bigcup_{j=1}^r W_j} \min_{k\in [r]} p_k^{|W_k|}\ .\label{eq:minpkwks}
               \end{align}
             We note that
             for any disjoint partition~$W=\bigcup_{j=1}^r W_j$ we have 
             $\min_{k\in [r]} p_k^{|W_k|}\leq p_{k_*}^{|W_{k_*}|}$ for any fixed $k_*\in [r]$. In particular,
             using $|W|=\sum_{k=1}^r |W_k|$, we conclude that there is some $k_*$ with $|W_{k^*}|\geq |W|/r$. It follows that 
             \begin{align}
             \min_{k\in [r]} p_k^{|W_k|}\leq p_{k_*}^{|W|/r}\leq \left(\max_{k\in [r]}p_k\right)^{|W|/r}\ .\label{eq:minpkwk}
             \end{align} Inserting Eq.~\eqref{eq:minpkwk} into Eq.~\eqref{eq:minpkwks}
             and using that the number of partitions of~$W$ into $r$~(possibly empty) subsets is  $r^{|W|}$, we conclude that
             \begin{align}
                     \prob_{F_1,\ldots,F_r}[W\subseteq \supp(F_1\cdot \cdots \cdot F_r)] &\leq  \left(r\left(\max_{k\in [r]}p_k\right)^{1/r}\right)^{|W|}\ .
             \end{align}
 As $W$ was arbitrary, this gives the claim~\eqref{item:pauli noise property 3}.
\end{proof}

\subsection{Locations, faults and the causal structure of errors\label{sec:locfaults}}
In the following (see Section~\ref{sec:localstochasticnoise}) we will consider general (possibly non-Pauli) local stochastic noise. To introduce the corresponding notion, we need a few additional definitions discussed here.

The circuit $\cC$ can be organized into $\mathsf{depth}(\cC)$ layers of operations, each consisting of operations acting on disjoint subsets of qubits. In this view, we follow standard terminology and refer to each level-$0$ gadget, i.e., each physical one- or two-qubit gate, one-qubit preparation, or one-qubit measurement, as a circuit location (or simply location) of~$\mathsf{C}$. We note that ``wait locations'' where qubits are idle in a particular layer (i.e., the identity gate is applied) are also considered as circuit locations. In contrast, operations on classical data used to post-process measurement results are omitted in this definition as we will assume classical computations to be error-free. We denote the set of all circuit locations of a circuit~$\mathsf{C}$ by~$\mathsf{Loc}(\cC)$. Observe that for any circuit $\cC$ we have the upper bound 
\begin{align}
|\mathsf{Loc}(\cC)| &\leq \mathsf{depth}(\cC) \cdot N(\cC) \ \label{eq:upperboundciruitsizedep}
\end{align}
on the number of circuit locations. 

Consider a fixed subset $F\subseteq \mathsf{Loc}(\cC)$ (see Fig.~\ref{fig:originalcircuit} for an illustration)
and an error affecting exactly all locations~$F$ (and no other locations). In this situation, we refer to~$F$ as a set of fault locations. In the terminology of~\cite{gottesmantutorial}, the error is then modeled as a quantum operation~$\cE$ which is consistent with the causal structure of~$F$. Following~\cite{Chiribella2008}, we refer to such an operation as a quantum comb associated with~$F$, see Fig.~\ref{fig:noisycircuit} for an illustration, and write $\supp(\cE) = F$ for the support of a comb $\cE$. Diagrammatically, it is an object which attaches to input and output wires associated with a circuit with the locations~$F$ removed (itself a comb). This defines a circuit we denote~$\cE\bowtie\mathsf{C}$ and refer to as a noisy execution of~$\mathsf{C}$ with noise~$\cE$ inserted at the locations~$F$. We denote the set of all combs defined by a subset of locations~$F$ by~$\mathsf{Comb}_{\mathsf{C}}(F)$.

In more detail, every comb~$\cE\in\mathsf{Comb}_{\mathsf{C}}(F)$ is a sequence of completely positive trace-preserving maps (quantum channels)~$\cW_1,\ldots,\cW_K$  where $K\leq \mathsf{depth}(\cC)$, see Fig.~\ref{fig:internalstructurecomb}.  Let $\{t_1<\cdots <t_K\}\subseteq \{1,\ldots,\mathsf{depth}(\cC)\}$ be the set of layers (times) of~$\mathsf{C}$ containing a location belonging to~$F$, and let $F(t_j)\subseteq F$ be the corresponding locations in the layer~$t_j$. We assume that~$\cE$ uses an auxiliary system~$R$ of arbitrary dimension initialized in a fixed state~$\ket{0}$. Then the comb~$\cE$ applies a quantum channel~$\cW_{j}$ at time $t_j$ for each $j\in \{1,\ldots,K\}$ in succession, where
\begin{align}
\cW_j:\cB\left(R\otimes \left(\bigotimes_{r\in F(t_j)}\cH^{\mathsf{in}}_r\right)\right)\rightarrow \cB\left(R\otimes \left(\bigotimes_{r\in F(t_j)}\cH^{\mathsf{out}}_r\right)\right)
\end{align}
maps a tensor product of input Hilbert spaces~$\cH^{\mathsf{in}}_r$ (each either isomorphic to~$\mathbb{C}^2$ or~$(\mathbb{C}^2)^{\otimes 2}$) associated with location~$r\in F(t_j)$ to a tensor product of corresponding output Hilbert spaces~$\cH^{\mathsf{out}}_r$ while also acting on the reference system~$R$. This definition ensures that the comb~$\cE$ is consistent with the causal structure of the fault locations~$F$. It also ensures that -- as discussed above, a comb~$\cE\in\mathsf{Comb}_{\mathsf{C}}(F)$ can be inserted into a modified circuit obtained by removing each operation associated with a fault location~$f\in F$ of~$\mathsf{C}$. This results in a circuit (completely positive trace-preserving map) which we denote by~$\cE\bowtie\mathsf{C}$, see Fig.~\ref{fig:noisycircuit} for an illustration.

For our following proofs, it will be useful to define a way to localize the support of combs on $\cC$ to a given subcircuit (i.e., a subset of locations $L\subseteq \mathsf{Loc}_{\cC}$). Given $F \subseteq \mathsf{Loc}_{\cC}$ and $\cE \in \mathsf{Comb}_{\cC}(F)$, this procedure should lead to an associated comb $\cE^{\localization L} \in \mathsf{Comb}_{\cC}(F\cap L)$. There are several ways to define such a localization, but we opt for the simple approach (which is sufficient for our purposes) of ignoring the action of $\cE$ whenever it is not contained within $L$. That is, we define the localization of $\cE$ to $L$ by
\begin{align}
    \cE^{\localization L} := \left\{\begin{array}{ll}
        \cE & \text{if $F \subseteq L$} \\
        \cI & \text{otherwise.}
    \end{array} \right.\label{eq:combrestriction}
\end{align}
Here $\cI$ is understood to mean the trivial comb, whose support is the empty set (so that $\cI \bowtie \cC = \cC$). It is clear from this definition that the support of $\cE^{\localization L}$ is contained within $L$.

\subsection{Local stochastic noise\label{sec:localstochasticnoise}}
In the following, we define local stochastic noise on a circuit~$\mathsf{C}$; this generalizes the notion of local stochastic Pauli noise. 
For simplicity, we consider a local stochastic noise model where the error strength is parameterized by a single parameter~$p\in [0,1]$. It is straightforward to generalize e.g., to the case of location-dependent error strengths as considered e.g., in~\cite{gottesmantutorial}.

The following definitions will be convenient. For a circuit~$\mathsf{C}$, let
\begin{align}
\mathsf{comb}_{\mathsf{C}}:=\bigcup_{F\subseteq \mathsf{Loc}(\cC)} \mathsf{comb}_{\mathsf{C}}(F)
\end{align}
be the disjoint union of the set of combs over all possible subsets~$F$ of fault locations. For an element~$\cE\in\mathsf{comb}_{\mathsf{C}}$, let us write~$\supp(\cE)\subseteq \mathsf{Loc}(\cC)$ for the set of locations~$\cE$ acts on, i.e., the set defined by~$\cE\in \mathsf{comb}_{\mathsf{C}}(\supp(\cE))$.

A random variable~$\cE$ taking on values in the set of combs~$\mathsf{comb}_{\mathsf{C}}$
is called local stochastic noise of strength~$p$ on~$\mathsf{C}$ if 
\begin{align}
\Pr\left[F\subseteq \supp(\cE)\right]&\leq p^{|F|}\qquad\textrm{ for any subset }F\subseteq\mathsf{Loc}(\cC)\textrm{ of circuit locations}\ .
\end{align}

\begin{figure}
\centering
\begin{subfigure}{0.45\textwidth}
\centering
\includegraphics[width=9cm]{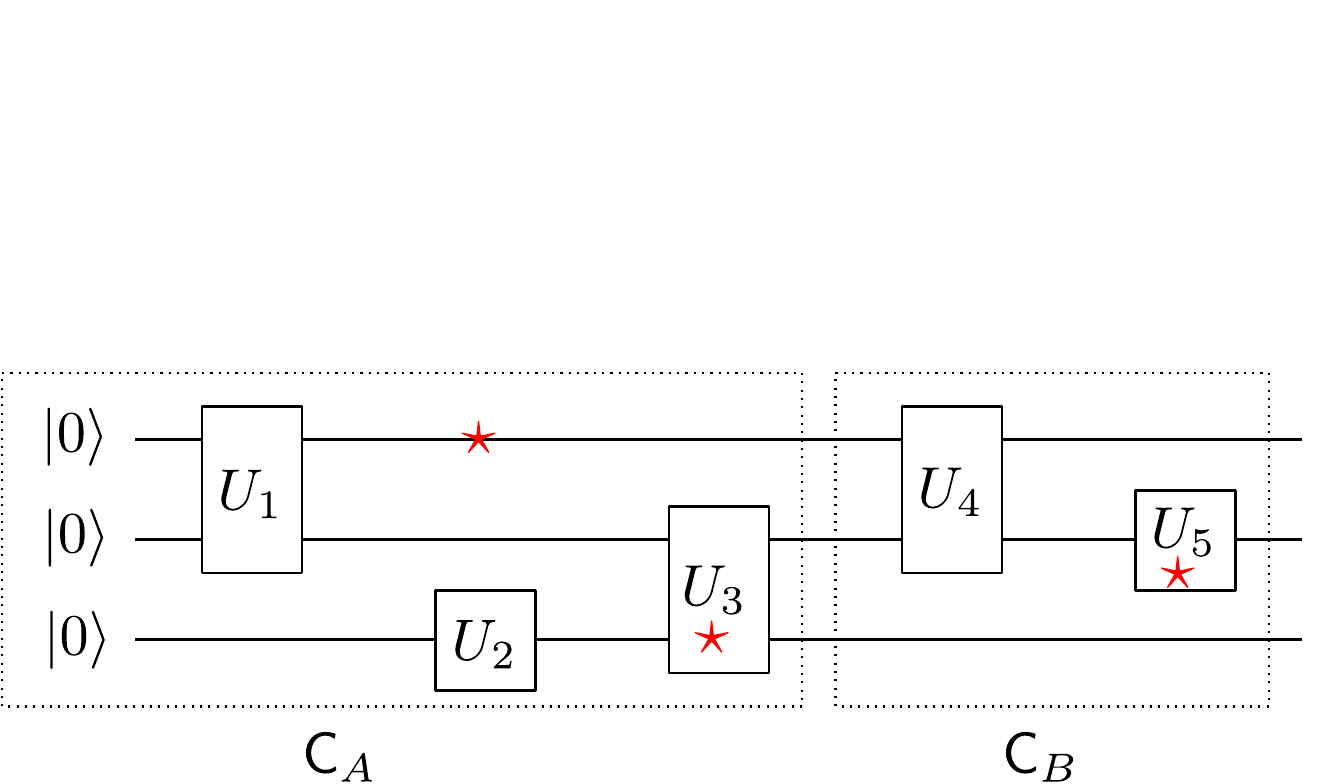}
\caption{The original circuit~$\mathsf{C}=\mathsf{C}_B\circ \mathsf{C}_A$ with a subset~$F\subset \mathsf{Loc}(\cC)$ of fault locations  (including a wait location)  marked by red stars\label{fig:originalcircuit}}
\end{subfigure}\\
\begin{subfigure}{0.45\textwidth}
\includegraphics[width=9cm]{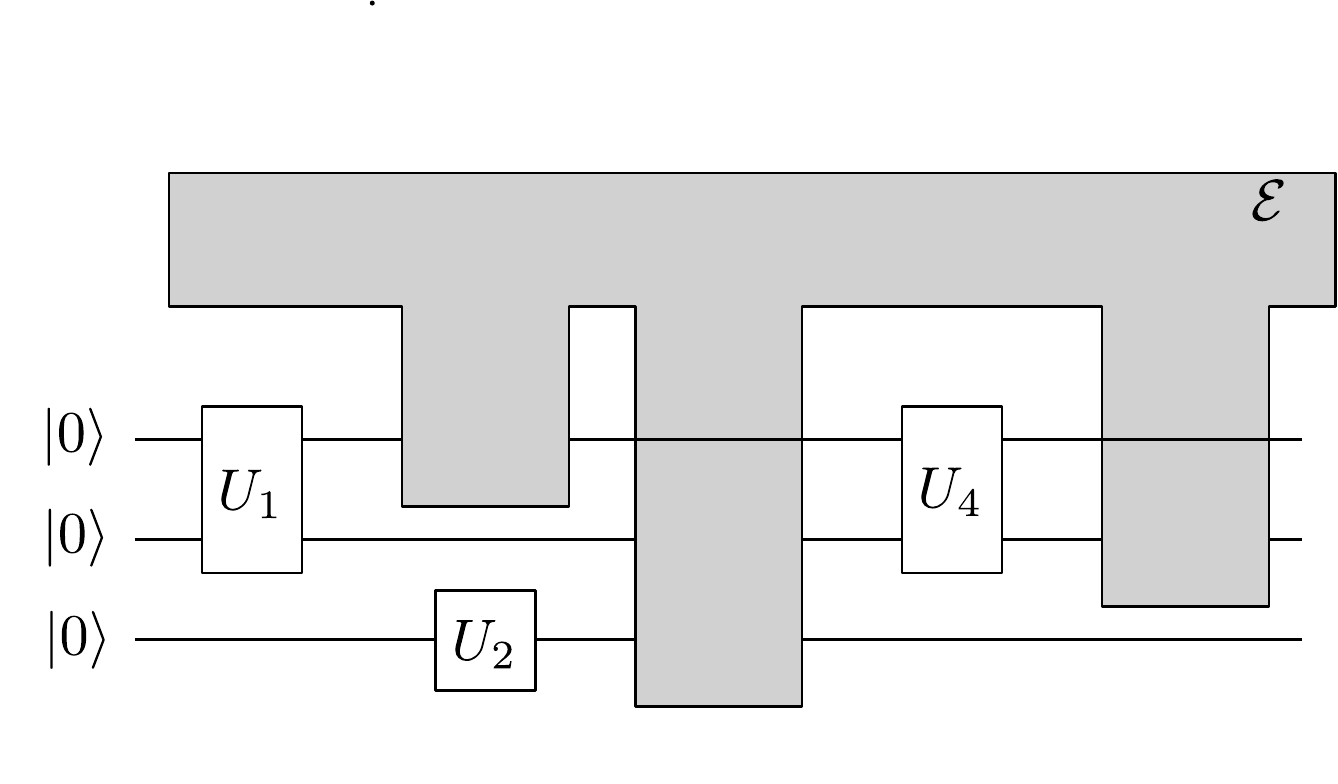}
\caption{The circuit~$\cE\bowtie \mathsf{C}$ for a particular comb~$\cE\in\mathsf{comb}_{\mathsf{C}}(F)$\label{fig:noisycircuit}}
\end{subfigure}\\
\begin{subfigure}{0.95\textwidth}
\centering
\includegraphics[width=9cm]{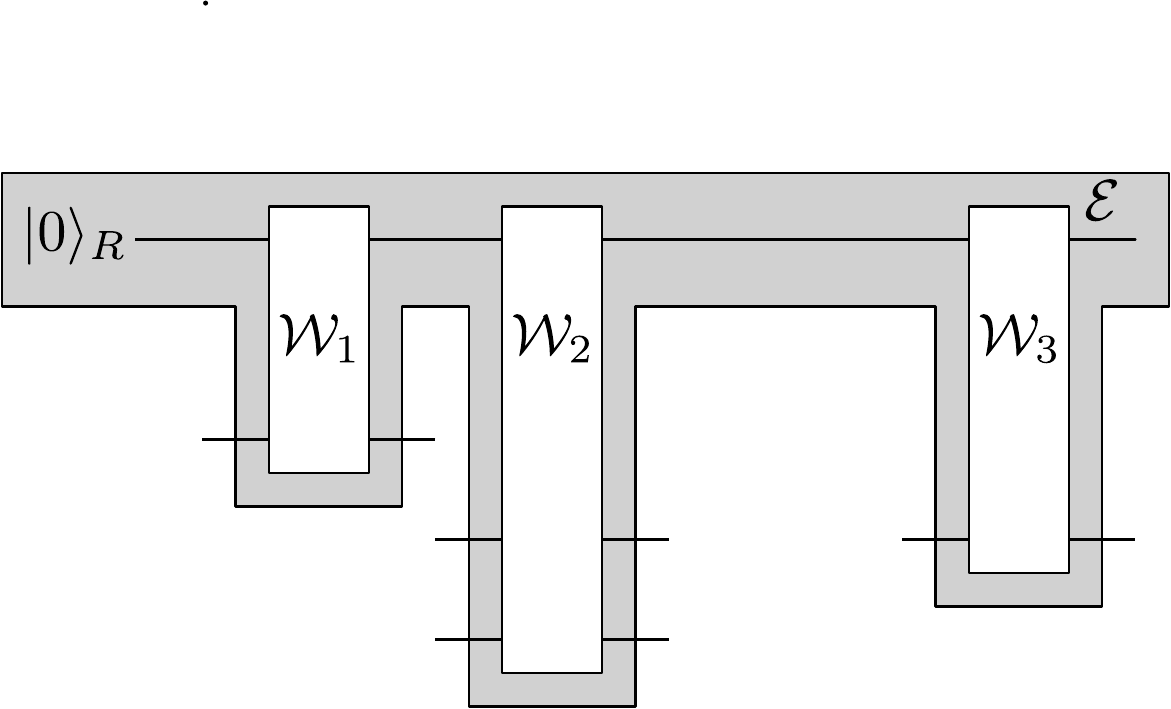}
\caption{The internal structure of the comb~$\cE$\label{fig:internalstructurecomb}}
\end{subfigure}\\
\caption{Noise defined by a comb $\cE$ affecting a subset~$F\subset\mathsf{Loc}(\mathsf{C})$ of (fault) locations of a circuit~$\mathsf{C}$}
\label{fig:localnoise}
\end{figure}
For  a circuit~$\mathsf{C}$ preparing an $N$-qubit quantum state~$\ket{\Psi}=\mathsf{C}$ and local stochastic noise~$\cE$ on $\mathsf{C}$,
the noisy implementation~$\cE\bowtie\mathsf{C}$ is a probabilistic ensemble of $N$-qubit states. 
 For a general circuit~$\cC$, a noisy implementation~$\cE\bowtie\cC$ is an ensemble of CPTP maps. 
It has the following properties:
\begin{lemma}[Local stochastic noise and circuits with quantum input and output]\label{lem:localstochasticnoiseqoutput}
We have the following:
\begin{enumerate}[(i)]
\item\label{it:circnoisyclaim}
Let $\cE$ be local stochastic noise of strength~$p$ on a circuit~$\mathsf{C}$. Then the probability that $\cE$ acts non-trivially on $\cC$ is upper bounded by
\begin{align}
    \prob[\cE \bowtie \cC \neq \cC] \leq |\mathsf{Loc}(\cC)| \cdot p\ .
\end{align}

\item\label{it:circnoisycomposedclaim}
Let $\cE$ be local stochastic noise of strength~$p$ on a circuit~$\mathsf{C}$.
Let $\cE^{\localization L}$ denote the localization of~$\cE$ to the circuit locations~$L\subseteq\mathsf{Loc}(\mathsf{C})$. 
Then
\begin{align}
\Pr\left[\cE \bowtie \cC \neq \cE^{\localization L} \bowtie \cC \right]&\leq |\mathsf{Loc}(\cC) \setminus L|\cdot p\ .
\end{align}

\item\label{it:thirdnoisydecomposedclaim}
Let $\cE,\tilde{\cE}$ be local stochastic noise of strength~$p$ on a circuit~$\mathsf{C}$.
Suppose the localizations $\cE^{\localization L}$  and $\tilde{\cE}^{\localization L}$  of~$\cE$ and $\tilde{\cE}$ to the circuit locations~$L\subseteq\mathsf{Loc}(\mathsf{C})$ satisfy $\cE^{\localization L}=\tilde{\cE}^{\localization L}$. Then 
\begin{align}
\Pr\left[\cE \bowtie \cC \neq \tilde{\cE} \bowtie \cC\right]
&\leq 2 |\mathsf{Loc}(\cC) \setminus L|\cdot p\ .
\end{align}
\end{enumerate}
\end{lemma}

\begin{proof}
For local stochastic noise~$\cE$ of strength~$p$, we have
\begin{align}
\Pr\left[\cE \neq \mathsf{id} \right]&\leq \Pr\left[\exists \ell\in \mathsf{Loc}_{\mathsf{C}}\textrm{ with } \ell\in \supp(\cE)\right]\\
&\leq \sum_{\ell\in \mathsf{Loc}(\mathsf{C})} \Pr\left[\{\ell\}\subseteq \supp(\cE)\right]\\
&\leq  |\mathsf{Loc}(\mathsf{C})|\cdot p\ .
\end{align}
Here we used the union bound in the penultimate step and the local stochasticity of~$\cE$ in the last step. This establishes Claim~\eqref{it:circnoisyclaim}. 

Claim~\eqref{it:circnoisycomposedclaim} is shown in an analogous manner. Note that by definition, we have $\cE = \cE^{\localization L}$ if and only if the support of $\cE$ is contained within $L$. Hence we obtain
\begin{align}
    \prob[\cE \neq \cE^{\localization L}] &= \prob[\supp (\cE) \nsubseteq L] \\
    &\leq \sum_{l \in \mathsf{Loc}(\cC) \setminus L} \prob[l \in \supp(\cE)] \\
    &\leq |\mathsf{Loc}(\cC) \setminus L| \cdot p\ 
\end{align}
by the union bound.

To prove Claim~\eqref{it:thirdnoisydecomposedclaim}, observe that
Claim~\eqref{it:circnoisycomposedclaim} implies that 
\begin{align}
    \prob[\cE \neq \cE^{\localization L}] &\leq |\mathsf{Loc}(\cC) \setminus L| \cdot p \\
    \prob[\tilde{\cE} \neq \tilde{\cE}^{\localization L}] &\leq |\mathsf{Loc}(\cC)\setminus L|\cdot p\ .
\end{align}
The union bound implies  $\Pr[X\neq Y]\leq \Pr[X\neq Z]+\Pr[Y\neq Z]$ for any triple $(X,Y,Z)$ of random variables. Using this and the assumption that $\cE^{\localization L} = \tilde{\cE}^{\localization L}$, the claim follows.
\end{proof}

\section{Fault tolerance against local stochastic noise}\label{sec:fault tolerance}
Here we discuss standard quantum fault tolerance constructions based on concatenated codes, closely following the approaches of~\cite{aliferis2005quantum,aliferis2007level}, see~\cite{gottesman2024surviving} of a modern introduction. 
Traditional quantum fault tolerance typically
focuses on the (approximate) realization of  an (ideal)  circuit~$\mathsf{C}$ on $N$ qubits composed of the following steps.
\begin{enumerate}[(i)]
\item\label{it:prepstepgeneralcircuit} 
Each  of the $N$ qubits  is prepared in the computational basis state~$\ket{0}$.
\item\label{it:evolutionstepgeneralcircuit}
A sequence $U_1,\ldots,U_s$ of one- and two-qubit gates (unitaries) from a specified gate set is applied.   
\item All qubits are measured in the computational basis, resulting in a sample $z$ from the distribution
\begin{align}
p(z)&=|\langle z,(U_s\cdots U_1)0^N\rangle|^2\qquad\textrm{ for }\qquad  z\in \{0,1\}^{N}\ . \label{eq:samplingfromdistribution}
\end{align}
\label{it:measurementstepgeneralcircuit}
\end{enumerate}
In the following, we refer a circuit~$\cC$ consisting of the steps~\eqref{it:prepstepgeneralcircuit}--\eqref{it:measurementstepgeneralcircuit} as a prepare-and-measure circuit.

The goal of quantum fault tolerance is to recompile the circuit~$\mathsf{C}$ into a circuit~$\mathsf{C}_{\mathsf{FT}}$
which approximately simulates the action of~$\mathsf{C}$ even when its implementation is imperfect: Even under noise, the circuit~$\cC_{\mathsf{FT}}$ should produce a sample~$z$ from a distribution~$p_{\mathsf{FT}}$ which is close in variational distance to the distribution~$p$ (see Eq.~\eqref{eq:samplingfromdistribution}).

 We view the circuit~$\mathsf{C}$ consisting of the three steps~\eqref{it:prepstepgeneralcircuit}--\eqref{it:measurementstepgeneralcircuit} and its simulation $\mathsf{C}^{\mathsf{sim}}$ as quantum channels, i.e., a completely positive trace-preserving (CPTP) maps. A noisy implementation of $\mathsf{C}^{\mathsf{sim}}$ is then given by a probabilistic ensemble over CPTP maps (corresponding to the different possible actions of the noise), and the relevant notion of approximate simulation is that this map is equal to the CPTP map realized by $\cC$ with high probability.
 
 We proceed as follows: In Section~\ref{sec:ftgadgetsdiscussion}, we discuss fault tolerance gadgets and their properties. In Section~\ref{sec:FTgagdetssimulation}, we 
 review the construction of the level-$L$ simulation~$\cC^{(L)}$ 
 of a circuit~$\cC$ based on a concatenated code. 
 In Section~\ref{sec:levelreductionftthresholdthm}, we explain how level reduction implies the fault tolerance threshold theorem, and discuss how fault tolerance can be achieved under locality constraints.

\subsection{Fault tolerance gadgets\label{sec:ftgadgetsdiscussion}}
Consider 
  an $[[n,1,d]]$-quantum error correcting code~$\cL\subset (\mathbb{C}^2)^{\otimes n}$ encoding a single logical qubit into~$n$ physical qubits with distance~$d$. 
We first define a family of operations called gadgets following established terminology, see e.g.,~\cite{gottesmantutorial}. 

  Figs.~\ref{fig:ftgadgetcomponentsnotions} and~\ref{fig:ftgadgetcomponents} illustrate the diagrammatic formalism for the definition of gadgets.
Two of the involved objects (see Fig.~\ref{fig:ftgadgetcomponentsnotions})  are only used to define
properties of gadgets and reason about them: an $r$-filter projects onto 
the subspace spanned by the union of subspaces~$E\cL$, where 
$E$ is any Pauli error of weight less than or equal to~$r$. The ideal decoder, in the following denoted~$\idec:\cB\left((\mathbb{C}^2)^{\otimes N}\right)\rightarrow \cB\left(\mathbb{C}^2\right)$, is a CPTP map which optimally performs error correction:  it satisfies
\begin{align}
\idec\left(E\proj{\Psi}E^\dagger\right)&\propto \proj{\Psi}\quad\textrm{ for any } \ket{\Psi}\in\cL\textrm{ and every correctable Pauli error } E\ ,\label{eq:conditionerrorrecovery}
\end{align}
i.e.,  any Pauli error~$E$ of weight less than or equal to~$t$.
Here we identify $\cL$ with $\mathbb{C}^2$ using a suitable isometry. 
We note that Eq.~\eqref{eq:conditionerrorrecovery} does not uniquely determine a  CPTP map  since the action on the subspace spanned by code states corrupted by uncorrectable errors is not prescribed. This is sufficient in the current context as the ideal decoder is only used formally: here we do not care about actually realizing~$\idec$ by a quantum circuit. Corresponding circuit realizations will only be considered below when we consider noisy encoding and decoding procedures, see Section~\ref{sec:noisyencodingdecoding}.

 Fig.~\ref{fig:ft-definitionsgadgets} states the necessary properties of level-$1$ gadgets required to obtain fault tolerance.
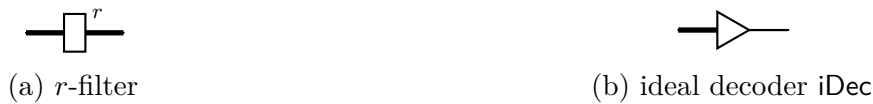
\begin{figure}[htbp]
    \centering
    \begin{subfigure}[b]{0.4\textwidth}
        \centering
        \begin{quantikz} \qw & \filtr & \qw 
        \arrow[from=1-1,to=1-3,black,line width=\encwidth,-]{}
        \end{quantikz}
        \caption{$r$-filter}
    \end{subfigure}
    \qquad\qquad
    \begin{subfigure}[b]{0.4\textwidth}
        \centering
        \begin{quantikz}
            \qw & \decgate{} & \qw
            \arrow[from=1-1,to=1-2,black,line width=\encwidth,-]{}
        \end{quantikz}
        \caption{ideal decoder $\idec$}
    \end{subfigure}
        \caption{Diagrammatic notation for notions used in the definition fault tolerance $1$-gadgets (see~\cite{gottesmantutorial} for details).  Thick (thin) lines represent encoded (decoded) qubits. 
Both the $r$-filter and the (ideal) decoder are not actually applied in 
a recompiled circuit, but are used in the analysis.}
    \label{fig:ftgadgetcomponentsnotions}
\end{figure}

We use the following standard terminology.  
First,  we refer to each elementary physical operations as used in~\eqref{it:prepstepgeneralcircuit}--\eqref{it:measurementstepgeneralcircuit} (i.e., single-qubit state preparation of~$\ket{0}$, any one- or two-qubit gate from the considered gate set and any single-qubit computational basis measurement) as a level-$0$ gadget or simply $0$-Ga.

A level-$1$ error correction gadget (or simply $1$-EC) is an adaptive circuit (i.e., CPTP map) on~$(\mathbb{C}^2)^{\otimes n}$ (i.e., acting on one ``code block'') composed of $0$-Gas (i.e., elementary physical operations). It performs error recovery: if no faults occur during its execution, it maps any state supported on $\cL$ (i.e., any encoded state) to itself, and any corrupted encoded state to some encoded state. The fault tolerance construction also assumes that error propagation is limited as long as the number of faults, i.e., the sum of the number of preexisting errors and the number of faults inside the gadget, is upper bounded by $t=\lfloor (d-1)/2\rfloor$ (the maximal weight of a correctable error), see Fig.~\ref{it:ECcorrectionprop}.

The fault tolerance construction also requires a level-$1$ gadget~$G'$ for each level-$0$ gadget~$G$. Here $G'$ is a (possibly adaptive) circuit of constant size which implements an encoded version of~$G$, where each of the $N$~qubits is encoded in the code~$\cL$, i.e., $G'$ acts on  $Nn$~physical qubits encoding $N$~logical qubits in $\cL^{\otimes N}\subset ((\mathbb{C}^2)^{\otimes n})^{\otimes N}$. Each gadget~$G'$ needs to be designed in such a way as to implement the desired (logical) functionality defined by~$G$, and limit error propagation when combined with error correction, see Figs.~\ref{it:measurementftprop}, \ref{it:preparationftprop} and~\ref{it:gateftpropreq}.

\begin{figure}[htbp]
    \centering
    \begin{subfigure}[b]{0.2\textwidth}
        \centering
        \begin{quantikz}
            \prepgate["s"{font=\scriptsize,xshift=13,yshift=5}]{} & \qw
            \arrow[from=1-1,to=1-2,black,line width=\encwidth,-]{}
        \end{quantikz}
        \caption{Preparation}
    \end{subfigure}
    \qquad
    \begin{subfigure}[b]{0.2\textwidth}
        \centering
        \begin{quantikz}
            \qw & \measgate["s"{font=\scriptsize,yshift=5}]{}
            \arrow[from=1-1,to=1-2,black,line width=\encwidth,-]{}
        \end{quantikz}
        \caption{Measurement}
    \end{subfigure}
    \qquad
    \begin{subfigure}[b]{0.2\textwidth}
        \centering
        \begin{quantikz} \qw & \ftgates{U} & \qw 
        \arrow[from=1-1,to=1-3,black,line width=\encwidth,-]{}
        \end{quantikz}
        \caption{Gate}
    \end{subfigure}
    \qquad
    \begin{subfigure}[b]{0.2\textwidth}
        \centering
        \begin{quantikz} 
        \qw & \ECs & \qw 
        \arrow[from=1-1,to=1-3,black,line width=\encwidth,-]{}
        \end{quantikz}
        \caption{Error Correction\label{it:ECftprop}}
    \end{subfigure}
   
    \caption{Diagrammatic notation for fault tolerance $1$-gadgets (see~\cite{gottesmantutorial} for details). Superscripts denote the number of fault locations inside the gadget. }
    \label{fig:ftgadgetcomponents}
\end{figure}
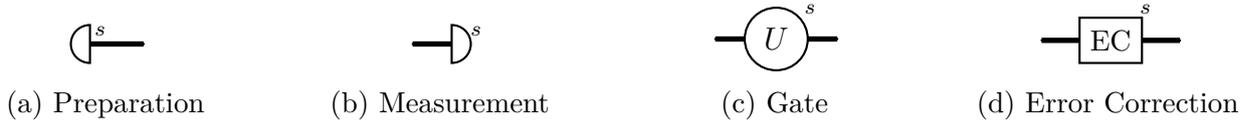

\begin{figure}[htbp]
    \centering
    
    \begin{subfigure}[b]{0.4\textwidth}
        \centering
        $\begin{aligned}
        \begin{quantikz}[column sep=0.4cm]
            \qw & \ECs & \qw
            \arrow[from=1-1,to=1-3,black,line width=\encwidth,-]{}
        \end{quantikz}
        &\;=\;
        \begin{quantikz}[column sep=0.4cm]
            \qw & \ECs & \filts & \qw
            \arrow[from=1-1,to=1-4,black,line width=\encwidth,-]{}
        \end{quantikz}
        \\[1em]
        \begin{quantikz}[column sep=0.4cm]
            \qw & \filtr & \ECs & \decgate["s"{font=\scriptsize,yshift=4,xshift=8}]{} & \qw
            \arrow[from=1-1,to=1-4,black,line width=\encwidth,-]{}
        \end{quantikz}
        &\;=\;
        \begin{quantikz}[column sep=0.4cm]
            \qw & \filtr & \decgate{} & \qw
            \arrow[from=1-1,to=1-3,black,line width=\encwidth,-]{}
        \end{quantikz}
        \end{aligned}$
        \caption{Error correction ($s \leq t$; $r+s \leq t$)\label{it:ECcorrectionprop}}
    \end{subfigure}\qquad \qquad
        \begin{subfigure}[b]{0.4\textwidth}
        \centering
        $\begin{aligned}
        \begin{quantikz}[column sep=0.4cm]
            \qw & \filtr & \measgate["s"{font=\scriptsize,yshift=7,xshift=0}]{}
            \arrow[from=1-1,to=1-3,black,line width=\encwidth,-]{}
        \end{quantikz}
        &\;=\;
        \begin{quantikz}[column sep=0.4cm]
            \qw & \filtr & \decgate{} & \measgate{}
            \arrow[from=1-1,to=1-3,black,line width=\encwidth,-]{}
        \end{quantikz}
        \end{aligned}$
        \caption{Measurement ($r+s \leq t$)\label{it:measurementftprop}}
    \end{subfigure}
    \vspace{1.5em}
     \begin{subfigure}[b]{0.4\textwidth}
        \centering
        $\begin{aligned}
        \begin{quantikz}
        \prepgate["s"{font=\scriptsize,yshift=7,xshift=13}]{} & \qw
        \arrow[from=1-1,to=1-2,black,line width=\encwidth,-]{}
        \end{quantikz}
        &\;=\;
        \begin{quantikz}
            \prepgate["s"{font=\scriptsize,yshift=7,xshift=13}]{} & \filts & \qw
            \arrow[from=1-1,to=1-3,black,line width=\encwidth,-]{}
        \end{quantikz}\\
        \begin{quantikz}
            \prepgate["s"{font=\scriptsize,yshift=7,xshift=13}]{} & \decgate{} & \qw
            \arrow[from=1-1,to=1-2,black,line width=\encwidth,-]{}
        \end{quantikz}
        &\;=\;
        \begin{quantikz}
            \prepgate{} & \qw
        \end{quantikz}
        \end{aligned}$
        \caption{Preparation ($s \leq t$)\label{it:preparationftprop}}
    \end{subfigure}\qquad\qquad 
    \begin{subfigure}[b]{0.4\textwidth}
        \centering
        $\begin{aligned}
        \begin{quantikz}[column sep=0.4cm]
            \qw & \filtri & \ftgates{U} & \qw
            \arrow[from=1-1,to=1-4,black,line width=\encwidth,-]{}
        \end{quantikz}
        &\;=\;
        \begin{quantikz}[column sep=0.4cm]
            \qw & \filtri & \ftgates{U} & \filtsum & \qw
            \arrow[from=1-1,to=1-5,black,line width=\encwidth,-]{}
        \end{quantikz}
        \\[1em]
        \begin{quantikz}[column sep=0.4cm]
            \qw & \filtri & \ftgates{U} & \decgate{} & \qw
            \arrow[from=1-1,to=1-4,black,line width=\encwidth,-]{}
        \end{quantikz}
        &\;=\;
        \begin{quantikz}[column sep=0.4cm]
            \qw & \filtri & \decgate{} & \ftgate{U} & \qw
            \arrow[from=1-1,to=1-3,black,line width=\encwidth,-]{}
        \end{quantikz}
        \end{aligned}$
        \caption{Gate ($s + \sum_i r_i \leq t$)\label{it:gateftpropreq}.}
    \end{subfigure}
     \caption{Conditions for fault tolerance of $1$-gadgets for an $[[n,1,d]]$-code, see~\cite{gottesmantutorial}. Here $t=\lfloor (d-1)/2\rfloor$.}
    \label{fig:ft-definitionsgadgets}
\end{figure}
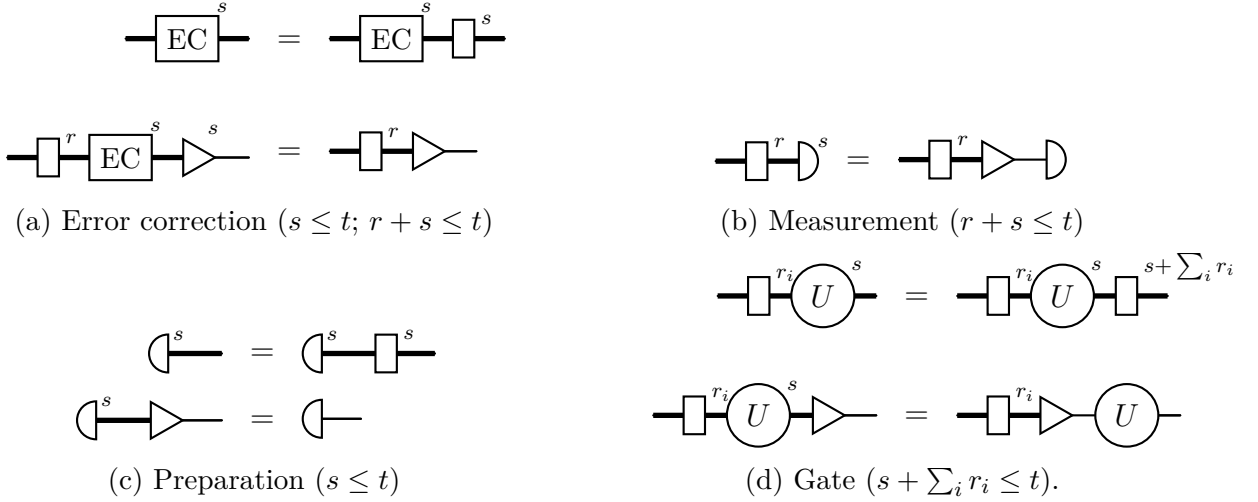
\subsection{Level-$L$ simulation of a circuit\label{sec:FTgagdetssimulation}}
Here we describe how the fault-tolerant circuit~$\mathsf{C}_{\mathsf{FT}}$ is constructed from (a description of) the ideal circuit~$\mathsf{C}$. We note that -- unlike the original circuit~$\mathsf{C}$ -- the circuit~$\mathsf{C}_{\mathsf{FT}}$ generally is  an adaptive circuit  involving mid-circuit measurements and subsequent operations depending on the measurement results. Indeed, a typical building block is an error-recovery step where the syndrome of an error-correcting code is measured and a corresponding correction is applied subsequently.

For concreteness we describe the established construction of~$\mathsf{C}_{\mathsf{FT}}$ based on concatenating an $[[n,1,d]]$-code~$\cL$ encoding a single logical qubit into~$n$ physical qubits. In fact, this results in a family of circuits~$\{\mathsf{C}^{(L)}\}_{L\in \mathbb{N}}$ where each circuit~$\mathsf{C}^{(L)}$ simulates~$\mathsf{C}$. We refer to~$\mathsf{C}^{(L)}$ as the level-$L$ (ideal) simulation of~$\mathsf{C}$, and the integer $L$ as the concatenation level. As explained below, it determines the degree of robustness of the simulation~$\mathsf{C}^{(L)}$ to noise.
 
 The family~$\{\mathsf{C}^{(L)}\}_{L\in \mathbb{N}}$ is constructed recursively using certain substitution rules involving the fault tolerance gadgets discussed in Section~\ref{sec:ftgadgetsdiscussion}. Given the definition of $0$-Gas and $1$-Gas,  the level-$1$ simulation~$\mathsf{C}^{(1)}$ of~$\mathsf{C}$ is defined by applying the following substitution rules to the $0$-Gas constituting~$\mathsf{C}$:
\begin{enumerate}[(I)]
\item\label{it:substitutionruleprep}
For each preparation of~$\ket{0}$, i.e., each preparation $0$-Ga, replace the $0$-Ga by a preparation $1$-Ga followed by a $1$-EC on the corresponding block,
\item 
For each unitary $U_j$, i.e., each gate $0$-Ga, replace the $0$-Ga by the corresponding gate $1$-Ga followed by a $1$-EC on the corresponding block (if $U_j$ is a single-qubit gate) or two $1$-ECs on the corresponding blocks (if $U_j$ is a two-qubit gate).
\item\label{it:substitutionrulemeasurement}
For each measurement of a qubit, i.e., each measurement $0$-Ga, replace the measurement $0$-Ga by a measurement $1$-Ga.
\end{enumerate}
This defines a circuit~$\mathsf{C}^{(1)}$ on $((\mathbb{C}^2)^{\otimes n})^{\otimes N}$ (i.e., $nN$) physical qubits with $nN_{\mathsf{in}}$-qubit input and $nN_{\mathsf{out}}$-qubit output, where every $n$-tuple of physical qubits corresponds to a code block encoding one logical qubit. We note that this circuit is adaptive in general because of the use of the $1$-EC gadget: it uses mid-circuit measurements and adaptive corrections (which are $0$-Gas). 

While the substitution rules~\eqref{it:substitutionruleprep}--\eqref{it:substitutionrulemeasurement}
constitute the most commonly used construction, it will be convenient for our analysis to  slightly modify~\eqref{it:substitutionrulemeasurement}: We will include an additional $1$-EC before the measurement. With the $1$-EC introduced by the preceding preparation or gate application, this results in two $1$-ECs in succession before each measurement. It is easy to check that this does not affect the fault tolerance properties of the resulting circuit; the substitution simply facilitates our resource accounting (see Lemma~\ref{lem:qubitdepthoverheadsimulation} below). 

 For later use, let us also point out that the substitution rules~\eqref{it:substitutionruleprep}--\eqref{it:substitutionrulemeasurement} can also be applied to any adaptive circuit~$\mathsf{C}$ that measures qubits and applies unitaries from the gate set (i.e., $0$-Gas) which are classically controlled on the measurement outcomes. In addition, we can define a simulation of a circuit~$\mathsf{C}$ which includes tracing out qubits by adding the substitution rule
 \begin{enumerate}[(I)]\setcounter{enumi}{3}
 \item\label{it:substitutionrulepartialtrace}
 Replace each partial trace over a qubit by a partial trace over the corresponding block consisting of $n$~qubits. 
 \end{enumerate}

The family~$\{\mathsf{C}^{(L)}\}_{L\in\mathbb{N}}$ of 
simulations is then obtained from~$\mathsf{C}$ by recursively applying this definition, as follows. We set $\mathsf{C}^{(0)}:=\mathsf{C}$ and call this the level-$0$ simulation of~$\mathsf{C}$. Given a level-$j$ simulation~$\mathsf{C}^{(j)}$ of~$\mathsf{C}$ for $j\in \mathbb{N}_0=\mathbb{N}\cup \{0\}$, the level-$(j+1)$ simulation~$\mathsf{C}^{(j+1)}$ is obtained by applying the substitution rules~\eqref{it:substitutionruleprep}--\eqref{it:substitutionrulemeasurement} to the  circuit~$\mathsf{C}^{(j)}$.

Given a set of fault tolerance $1$-Gas as defined above,  define the constants
\begin{align}
\begin{matrix}
\nmax &:=&\max_{G\ \textrm{$1$-Ga}} N(G)\\
\dmin &:=&\mathsf{depth}(\mathop{EC})+\min_{\substack{G\ \textrm{$1$-Ga}\\ G\neq EC}}\mathsf{depth}(G)\\
\dmax &:=&\mathsf{depth}(\mathop{EC})+
\max_{\substack{G\ \textrm{$1$-Ga}\\ G\neq \mathop{EC}}}\mathsf{depth}(G)
\end{matrix}\ .\label{eq:dminmaxdef}
\end{align}
We note that $\nmax\geq n$ in general since $1$-Gas may use auxiliary qubits (in addition to the logical qubit encoded in $n$~physical qubits). 

In terms of these constants, we can give the following estimates on the number~$N(\cC^{(L)})$ of qubits and the depth~$\mathsf{depth}(\cC^{(L)})$ of the level-$L$ simulation~$\cC^{(L)}$.
\begin{lemma}[Qubit and depth overhead of level-$L$ simulation]
\label{lem:qubitdepthoverheadsimulation}
Consider the concatenated code construction based on an $[[n,1,d]]$-code~$\cL$.  
 Let $\nmax, \dmin, \dmax$ be the gadget-dependent constants defined in Eq.~\eqref{eq:dminmaxdef}.
 Then the number of qubits of the circuit~$\cC^{(L)}$ and its depth are bounded as
\begin{align}
\begin{matrix}
n^L&\leq& N(\cC^{(L)})/N(\cC) &\leq \nmax^L\\
\dmin^L&\leq& \mathsf{depth}(\cC^{(L)})/\mathsf{depth}(\cC) &\leq \dmax^L
\end{matrix}\ .\label{eq:gminmaxupperbound}
\end{align}
\end{lemma}
\noindent In particular, this lemma implies that the overhead in both the number of qubits and the depth of~$\cC^{(L)}$ scales as $e^{\Theta(L)}$. 

We note that in standard fault tolerance considerations, the goal is typically to minimize the overhead. Correspondingly, only upper bounds on quantities such as $N(\cC^{(L)})/N(\cC)$ are of interest. However, in our setup, lower bounds are also relevant since the quantities in Eq.~\eqref{eq:gminmaxupperbound}
determine the linear dimensions  of the $2D$~system used in our one-shot entanglement generation protocols. In particular, these overheads also determine the distance at which long-range entanglement is established, see Section~\ref{sec:2dsimulation}.

\begin{proof}
The bound on $\mathsf{depth}(\cC^{(L)})/\mathsf{depth}(\cC)$ in Eq.~\eqref{eq:gminmaxupperbound} 
 follows by recursion from the definition of $\dmin,\dmax$ and the fact
that when applying the substitution rules~\eqref{it:substitutionruleprep}--\eqref{it:substitutionrulemeasurement} (with the mentioned modification of~\eqref{it:substitutionrulemeasurement}, each $0$-Ga is substituted by the corresponding~$1$-Ga and a $1$-EC. The upper bound on $N(\cC^{(L)})/N(\cC)$ follows similarly, whereas the lower bound is a consequence of the fact that an $[[n,1,d]]$-code is concatenated $L$~times. 
\end{proof}
If the original circuit~$\cC$ has a certain structure, this 
is also shared by its level-$L$ simulation~$\cC^{(L)}$.  Specifically, this applies to product circuits and compositions of subcircuits. We state this as a Lemma. It follows immediately 
from  the substitution rules~\eqref{it:substitutionruleprep}--\eqref{it:substitutionrulepartialtrace}. 
\begin{lemma}[Level-$L$ simulations, product- and subcircuits]\label{lem:productandsubcircuitssimulation}
We have the following for any $L\in\mathbb{N}$. 
\begin{enumerate}[(i)]
\item
Let~$\mathsf{C}=\mathsf{C}_R\otimes \mathsf{C}_S$ be a circuit acting separately on disjoint subsets $R,S$ of qubits.
Then the level-$L$ simulation~$\mathsf{C}^{(L)}$ is itself of tensor product form
\begin{align}
\mathsf{C}^{(L)}&=\mathsf{C}^{(L)}_R\otimes \mathsf{C}^{(L)}_S\ .\label{eq:tensorproductcrs}
\end{align}
\item\label{it:subcircuitpropertylem}
Suppose that $\mathsf{C}=\mathsf{C}_B\circ \mathsf{C}_A$ is the composition of two circuits~$\mathsf{C}_A$, $\mathsf{C}_B$, each of which may be adaptive. Then it is easy to check that the level-$L$ simulation~$\mathsf{C}^{(L)}$ of $\mathsf{C}$  decomposes as
\begin{align}
\mathsf{C}^{(L)}&=\mathsf{C}^{(L)}_B\circ \mathsf{C}^{(L)}_A\label{eq:compositioncrs}
\end{align}
into a composition of the corresponding level-$L$ simulations. 
\end{enumerate}
\end{lemma}
We often consider subcircuits corresponding to a partial execution of a circuit~$\mathsf{C}$. This is the main setting where we use Property~\eqref{it:subcircuitpropertylem}.

\subsection{Level reduction and the fault tolerance threshold theorem\label{sec:levelreductionftthresholdthm}}
As described in Section~\ref{sec:fault tolerance},
in typical quantum fault tolerance considerations, one considers a circuit which applies a unitary circuit to the product state~$|0^N\rangle$ and then measures in the computational basis. To analyze the effect of errors, 
one can  focus on the (potentially corrupted) state before the measurement, and consider the effect of an ideal decoder.
Formally, this follows because of the condition  for a fault-tolerant measurement gadget illustrated in 
Fig.~\ref{it:measurementftprop}.
For ease of presentation, we therefore omit the final measurements in the following. We call~$\mathsf{C}$ the (ideal) circuit obtained by applying a unitary circuit to~$|0^N\rangle$. In other words, we assume that $\cC$ consists of single-qubit state preparation as well as one- and two-qubit unitaries.

As argued in~\cite{aliferis2005quantum,aliferis2007level}, the key to establishing a fault tolerance threshold theorem is
the setting of concatenated codes is the following ``level reduction'' result. It considers the effect of one level of decoding 
to a level-$1$ simulation~$\cC^{(1)}$ of $\cC$. The ideal decoder $(\idec^{(1)})^{\otimes N}$ is explicitly included in the statement of the lemma.
Its proof relies on the 
fault tolerance properties of gadgets shown in Fig.~\ref{fig:ft-definitionsgadgets}. Using these, it is argued that the ideal decoder can be pulled through (i.e., commuted to the left of)  the level-$L$ simulation  to obtain a (less) noisy level-$(L-1)$ simulation. 
We reproduce
the result of~\cite{aliferis2007level} for the special case of local stochastic noise here, see Fig.~\ref{fig:levelreductionillustrated} for an illustration.
\begin{lemma}[Level reduction, Lemma~3 of~\cite{aliferis2007level} paraphrased]\label{lem:level reductionoriginal}
Fix an $[[n,1,d]]$-code $\cL$ and associated level-0 and level-1 gadgets satisfying the fault tolerance properties given in Fig.~\ref{fig:ft-definitionsgadgets}. Let $t := \lfloor (d-1)/2\rfloor$. Let $\idec$  be an  ideal decoder for the code $\cL$.
There is a constant $c > 0$ (depending only on the code and gadgets used) such that the following holds
for any $N$-qubit prepare-and-measure circuit $\cC$  defined by~$s$ one- and two-qubit gates $U_s,\ldots, U_1$ and producing a sample $z\in \{0,1\}^N$ according to $p(z)$ (see Eq.~\eqref{eq:samplingfromdistribution}). For any local stochastic noise $\cE_1$ on $\cC$ of strength $p_1 = p$, there is local stochastic noise $\cE_0$ of strength $p_0 = cp_1^{t+1}$ such that
\begin{align}
    \idec^{\otimes N} \circ \left( \cE_1 \bowtie \cC^{(1)} \right)  = \cE_0 \bowtie \cC\ .
\end{align}
Here  $\cC^{(1)}$ denotes the level-1 simulation of $\cC$. 
\end{lemma} 
\noindent We refer to~\cite{aliferis2007level} for further details; see Fig.~\ref{fig:levelreductionbasic} for a schematic illustration of some of the basic steps used to establish this result. 
\begin{figure}
\centering
\begin{subfigure}[b]{0.9\textwidth}
    \centering
        \begin{quantikz}[column sep=0.4cm]
            \prepgate{} & \EC & \ftgate{U^{(1)}} & \EC & \measgate{}
            \arrow[from=1-1,to=1-5,black,line width=\encwidth,-]{}
        \end{quantikz}
        $\;=\;$
        \begin{quantikz}[column sep=0.4cm]
             \prepgate{} & \EC & \ftgate{U^{(1)}} & \EC & \decgate{} & \measgate{}
             \arrow[from=1-1,to=1-5,black,line width=\encwidth,-]{}
        \end{quantikz}
    \caption{First step: an ideal decoder is inserted before the final measurement step. This will be propagated backwards through the circuit, transforming 1-Gas into their corresponding logical operations on the physical space.\label{fig:levelreductionbasicA}}
\end{subfigure}\\
\hfill
\begin{subfigure}{0.9\textwidth}
    \centering
    $\begin{aligned}
        \begin{quantikz}[column sep=0.4cm]
            \prepgate{} & \EC & \decgate{} & \ftgate{U} & \measgate{}
            \arrow[from=1-1,to=1-3,black,line width=\encwidth,-]{}
        \end{quantikz}
        &\;=\;
        \begin{quantikz}[column sep=0.4cm]
            \prepgate{} & \ftgate{U} & \measgate{}
        \end{quantikz}
    \end{aligned}$
    \caption{Last step: the ideal decoder meets the state preparation 1-Ga at the beginning of the circuit. This has equivalent action to the preparation of the corresponding physical (unencoded) state.\label{fig:levelreductionbasicB}}
\end{subfigure}\qquad 
\caption{Illustration of basic steps in the proof of Lemma~\ref{lem:level reductionoriginal}.
\label{fig:levelreductionbasic}}
\end{figure}

For any $j\geq 2$, let $\mathsf{iDec}^{(j)}$ be defined recursively as the result of applying the substitution rules~\eqref{it:substitutionruleprep}--\eqref{it:substitutionrulepartialtrace} to~$\mathsf{iDec}^{(j-1)}$. 
Then $\mathsf{iDec}^{(L)}$ is a map taking~$n^L$ physical qubits to~$n^{L-1}$ qubits  which changes the concatenation level from~$L$ to $L-1$. We note that the composed map
\begin{align}
\underset{L\rightarrow 0}{\mathsf{iDec}}:=\mathsf{iDec}^{(1)}\circ \cdots \circ \mathsf{iDec}^{(L)} \label{eq:composeddecoder}
\end{align} 
then implements -- in the absence of faults -- an ideal ``level-by-level'' decoder taking a logical qubit encoded in the $L$-fold concatenated code~$\cL^{\circ L}\subset (\mathbb{C}^2)^{\otimes n^L}$ associated with~$\cL$ to a single physical qubit,  see Fig.~\ref{fig:idealdecoderconcatenatedthree} for an illustration.
\begin{figure}[htbp]
\centering
\begin{tikzpicture}[scale=0.9]
    \def\boxwidth{1.4}
    \def\boxheight{5.0}
    \def\wirelen{1.2}
    \def\leaflen{0.8}
    \def\bendx{0.3}
    
    \def\boxAx{0}
    \def\boxBx{\boxwidth + \wirelen}
    \def\boxCx{2*\boxwidth + 2*\wirelen}
    
    \def\boxAheight{5.0}
    \def\boxBheight{3.5}
    \def\boxCheight{2.0}
    
    \draw[fill=white] (\boxAx, -\boxAheight/2) rectangle (\boxAx + \boxwidth, \boxAheight/2);
    \node at (\boxAx + \boxwidth/2, 0) {$\mathsf{iDec}^{(3)}$};
    
    \foreach \i in {0,...,26} {
        \pgfmathsetmacro{\ypos}{\boxAheight/2 - (\i + 0.5) * \boxAheight/27}
        \draw ({-\leaflen}, \ypos) -- (\boxAx, \ypos);
    }
    
    \draw[fill=white] (\boxBx, -\boxBheight/2) rectangle (\boxBx + \boxwidth, \boxBheight/2);
    \node at (\boxBx + \boxwidth/2, 0) {$\mathsf{iDec}^{(2)}$};
    
    \foreach \i in {0,...,8} {
        \pgfmathsetmacro{\yposA}{\boxAheight/2 - (\i + 0.5) * \boxAheight/9}
        \pgfmathsetmacro{\yposB}{\boxBheight/2 - (\i + 0.5) * \boxBheight/9}
        \draw (\boxAx + \boxwidth, \yposA) -- ({\boxAx + \boxwidth + \bendx}, \yposA) 
              -- ({\boxBx - \bendx}, \yposB) -- (\boxBx, \yposB);
    }
    
    \draw[fill=white] (\boxCx, -\boxCheight/2) rectangle (\boxCx + \boxwidth, \boxCheight/2);
    \node at (\boxCx + \boxwidth/2, 0) {$\mathsf{iDec}^{(1)}$};
    
    \foreach \i in {0,...,2} {
        \pgfmathsetmacro{\yposB}{\boxBheight/2 - (\i + 0.5) * \boxBheight/3}
        \pgfmathsetmacro{\yposC}{\boxCheight/2 - (\i + 0.5) * \boxCheight/3}
        \draw (\boxBx + \boxwidth, \yposB) -- ({\boxBx + \boxwidth + \bendx}, \yposB) 
              -- ({\boxCx - \bendx}, \yposC) -- (\boxCx, \yposC);
    }
    
    \draw (\boxCx + \boxwidth, 0) -- (\boxCx + \boxwidth + \leaflen, 0);
    
\end{tikzpicture}
\caption{The decoder $\underset{L\rightarrow 0}{\mathsf{iDec}}$  for $\cL^{\circ L}$, for $L=3$ concatenation levels associated with an $[[n,1,d]]$-code~$\cL$. For our illustrations, we use $n=3$ throughout but the constructions actually require larger codes. Each line represents a physical qubit.\label{fig:idealdecoderconcatenatedthree}}
\end{figure}
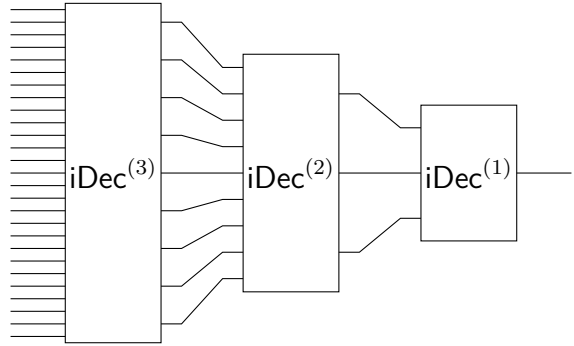

We note that by Eqs.~\eqref{eq:tensorproductcrs} and~\eqref{eq:compositioncrs}, we have 
\begin{align}
\left(\underset{L\rightarrow 0}{\mathsf{iDec}}\right)^{\otimes N}&=(\mathsf{iDec}^{(1)})^{\otimes N}\circ\cdots\circ(\mathsf{iDec}^{(L)})^{\otimes N}\ .\label{eq:parallelconcdec}
\end{align}

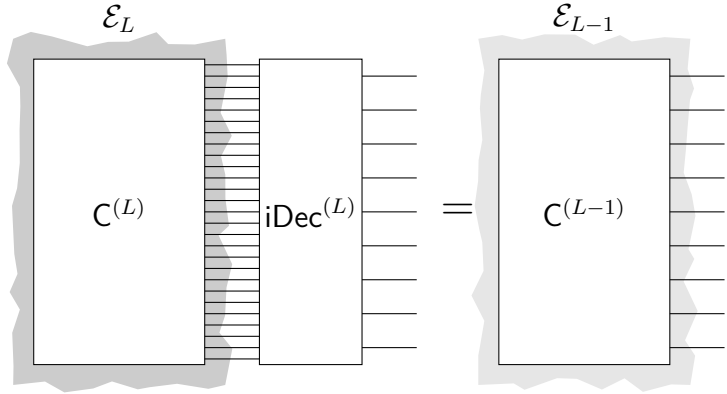
\begin{figure}[htbp]
\centering

\begin{tikzpicture}[scale=0.9]
\def\boxwidth{1.5}
\def\grayboxwidth{2.5}  
\def\boxheight{4.5}
\def\wirelen{0.8}
\def\bendx{0.25}
\def\outlen{0.8}
\def\noisemargin{0.3}  
\def\grayboxleft{0}

\fill[gray!40, decorate, decoration={random steps, segment length=3mm, amplitude=1mm}] 
  (\grayboxleft - \noisemargin, -\boxheight/2 - \noisemargin) rectangle 
  (\grayboxleft + \grayboxwidth + \noisemargin, \boxheight/2 + \noisemargin);

\fill[white] (\grayboxleft, -\boxheight/2) rectangle (\grayboxleft + \grayboxwidth, \boxheight/2);
\draw (\grayboxleft, -\boxheight/2) rectangle (\grayboxleft + \grayboxwidth, \boxheight/2);
\node at (\grayboxleft + \grayboxwidth/2, 0) {$\mathsf{C}^{(L)}$};

\node at (\grayboxleft + \grayboxwidth/2, \boxheight/2 + \noisemargin + 0.3) {$\cE_{L}$};

\def\decboxleft{\grayboxwidth + \wirelen}
\def\decboxheight{4.5}
\draw[fill=white] (\decboxleft, -\decboxheight/2) rectangle (\decboxleft + \boxwidth, \decboxheight/2);
\node at (\decboxleft + \boxwidth/2, 0) {$\mathsf{iDec}^{(L)}$};

\foreach \i in {0,...,26} {
  \pgfmathsetmacro{\yposGray}{\boxheight/2 - (\i + 0.5) * \boxheight/27}
  \pgfmathsetmacro{\yposDec}{\decboxheight/2 - (\i + 0.5) * \decboxheight/27}
  \draw (\grayboxleft + \grayboxwidth, \yposGray) -- ({\grayboxleft + \grayboxwidth + \bendx}, \yposGray)
    -- ({\decboxleft - \bendx}, \yposDec) -- (\decboxleft, \yposDec);
}

\def\decboxright{\decboxleft + \boxwidth}
\foreach \i in {0,...,8} {
  \pgfmathsetmacro{\ypos}{\decboxheight/2 - (\i + 0.5) * \decboxheight/9}
  \draw (\decboxright, \ypos) -- (\decboxright + \outlen, \ypos);
}

\def\eqpos{\decboxright + \outlen + 0.6}
\node at (\eqpos, 0) {\Large $=$};

\def\rightboxleft{\eqpos + 0.6}
\def\rightboxheight{4.5}

\fill[gray!20, decorate, decoration={random steps, segment length=3mm, amplitude=1mm}] 
  (\rightboxleft - \noisemargin, -\rightboxheight/2 - \noisemargin) rectangle 
  (\rightboxleft + \grayboxwidth + \noisemargin, \rightboxheight/2 + \noisemargin);

\fill[white] (\rightboxleft, -\rightboxheight/2) rectangle (\rightboxleft + \grayboxwidth, \rightboxheight/2);
\draw (\rightboxleft, -\rightboxheight/2) rectangle (\rightboxleft + \grayboxwidth, \rightboxheight/2);
\node at (\rightboxleft + \grayboxwidth/2, 0) {$\mathsf{C}^{(L-1)}$};

\node at (\rightboxleft + \grayboxwidth/2, \rightboxheight/2 + \noisemargin + 0.3) {$\cE_{L-1}$};

\def\rightboxright{\rightboxleft + \grayboxwidth}
\foreach \i in {0,...,8} {
  \pgfmathsetmacro{\ypos}{\rightboxheight/2 - (\i + 0.5) * \rightboxheight/9}
  \draw (\rightboxright, \ypos) -- (\rightboxright + \outlen, \ypos);
}
\end{tikzpicture}

\caption{An illustration of level reduction (see Lemma~\ref{lem:level reduction}) for $N=1$ encoded qubit. Gray colors indicate components affected by local stochastic noise, with darker shading indicating larger noise strength.
The noisy execution~$\cE_L\bowtie\mathsf{C}^{(L)}$ of the level-$L$ simulation~$\mathsf{C}^{(L)}$ of $\mathsf{C}$ is illustrated by a gray box (labeled with the local stochastic noise~$\cE_L$) enclosing the ideal circuit~$\mathsf{C}^{(L)}$ on the left (and similarly on the right). For a noisy level-$L$ simulation of a circuit~$\mathsf{C}$ with local stochastic noise~$\cE_L$ of strength~$p$ followed by an ideal decoder to level~$L-1$, there is a noisy level-$(L-1)$ simulation~$\cE_{L-1}\bowtie \mathsf{C}^{(L-1)}$ of~$\mathsf{C}$  with local stochastic noise~$\cE_{L-1}$ of strength~$p_{L-1}$ related to $p_L$. \label{fig:levelreductionillustrated}}
\end{figure}

The following result is an immediate consequence of level reduction (i.e., Lemma~\ref{lem:level reduction}). 
 It is illustrated in Fig.~\ref{fig:levelreductioniterated}. 
\begin{corollary}[Level reduction, iterated -- \cite{aliferis2007level}, Lemma 3, paraphrased] \label{cor:levelreductioniterated} 
Fix an $[[n,1,d]]$-code~$\cL$ and associated level-$0$ and level-$1$ gadgets
satisfying the fault tolerance properties given in Fig.~\ref{fig:ft-definitionsgadgets}. Let $t:=\lfloor (d-1)/2\rfloor$.
Let $L\in\mathbb{N}$ and let  
$\underset{L\rightarrow 0}{\mathsf{iDec}}$ be ideal decoding the map defined in Eq.~\eqref{eq:composeddecoder} decoding~$\cL^{\circ L}$ to a physical qubit.  Then there is a threshold error strength~$p_*>0$ depending only on the code and gadgets used such that the following holds:
Let $N\in\mathbb{N}$ be arbitrary and let $\mathsf{C}$ be an arbitrary $N$-qubit circuit.  Let $\mathsf{C}^{(L)}$ be the level-$L$ simulation of~$\mathsf{C}$. Then for any local stochastic noise~$\cE_L$ of strength~$p_L=p$, there is local stochastic noise~$\cE_{0}$ of strength
    \begin{align}\label{eq:level reduction noise}
        p_{0} = p_* (p/p_*)^{(t+1)^L}\ ,
    \end{align}
such that
    \begin{align}
        \left(\underset{L\rightarrow 0}{\mathsf{iDec}}\right)^{\otimes N} \circ (\cE_L\bowtie \mathsf{C}^{(L)}) = \cE_{0} \bowtie \mathsf{C}\  .
    \end{align}
    \end{corollary}
\begin{proof}
Follows immediately by iterating the relation~$p_{j-1}\leq c p_{j}^{t+1}$ from Lemma~\ref{lem:level reduction} with $p_L=p$ 
 and setting $p_*=c^{-1/t}$, cf.~\cite[Section 3.4]{aliferis2007level}.
\end{proof}

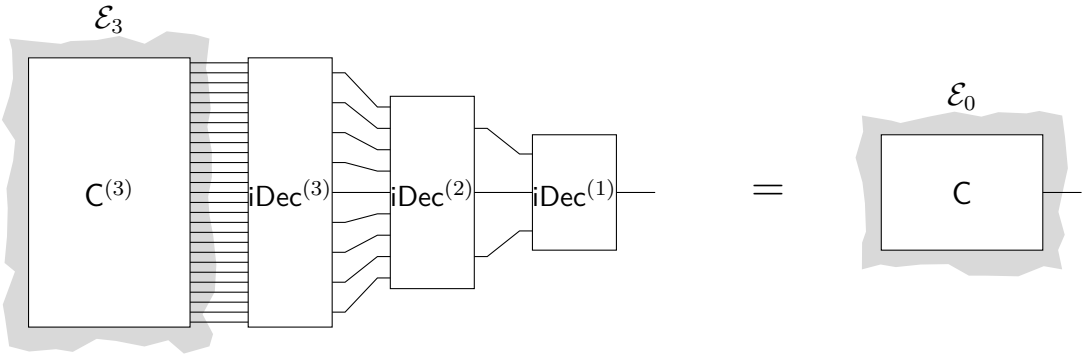
\begin{figure}
\centering

\begin{tikzpicture}[scale=0.85]
\def\boxwidth{1.3}
\def\grayboxwidth{2.5} 
\def\boxheight{4.2}
\def\wirelen{0.9}
\def\bendx{0.2}
\def\outlen{0.6}
\def\circuitgap{3.5}
\def\noisemargin{0.3} 
\def\boxAheight{4.2}
\def\boxBheight{3.0}
\def\boxCheight{1.8}
\def\grayboxleft{0}

\fill[gray!30, decorate, decoration={random steps, segment length=3mm, amplitude=1mm}] 
  (\grayboxleft - \noisemargin, -\boxAheight/2 - \noisemargin) rectangle 
  (\grayboxleft + \grayboxwidth + \noisemargin, \boxAheight/2 + \noisemargin);

\fill[white] (\grayboxleft, -\boxAheight/2) rectangle (\grayboxleft + \grayboxwidth, \boxAheight/2);
\draw (\grayboxleft, -\boxAheight/2) rectangle (\grayboxleft + \grayboxwidth, \boxAheight/2);
\node at (\grayboxleft + \grayboxwidth/2, 0) {$ \mathsf{C}^{(3)}$};

\node at (\grayboxleft + \grayboxwidth/2, \boxAheight/2 + \noisemargin + 0.3) {$\cE_3$};

\def\decAxleft{\grayboxwidth + \wirelen}
\draw[fill=white] (\decAxleft, -\boxAheight/2) rectangle (\decAxleft + \boxwidth, \boxAheight/2);
\node at (\decAxleft + \boxwidth/2, 0) {\small $\mathsf{iDec}^{(3)}$};
\foreach \i in {0,...,26} {
\pgfmathsetmacro{\ypos}{\boxAheight/2 - (\i + 0.5) * \boxAheight/27}
\draw (\grayboxleft + \grayboxwidth, \ypos) -- (\decAxleft, \ypos);
}
\def\decBxleft{\grayboxwidth + \boxwidth + 2*\wirelen}
\draw[fill=white] (\decBxleft, -\boxBheight/2) rectangle (\decBxleft + \boxwidth, \boxBheight/2);
\node at (\decBxleft + \boxwidth/2, 0) {\small \small $\mathsf{iDec}^{(2)}$};
\foreach \i in {0,...,8} {
\pgfmathsetmacro{\yposA}{\boxAheight/2 - (\i + 0.5) * \boxAheight/9}
\pgfmathsetmacro{\yposB}{\boxBheight/2 - (\i + 0.5) * \boxBheight/9}
\draw (\decAxleft + \boxwidth, \yposA) -- +(\bendx,0) -- ({\decBxleft - \bendx}, \yposB) -- (\decBxleft, \yposB);
}
\def\decCxleft{\grayboxwidth + 2*\boxwidth + 3*\wirelen}
\draw[fill=white] (\decCxleft, -\boxCheight/2) rectangle (\decCxleft + \boxwidth, \boxCheight/2);
\node at (\decCxleft + \boxwidth/2, 0) {\small \small $\mathsf{iDec}^{(1)}$};
\foreach \i in {0,...,2} {
\pgfmathsetmacro{\yposB}{\boxBheight/2 - (\i + 0.5) * \boxBheight/3}
\pgfmathsetmacro{\yposC}{\boxCheight/2 - (\i + 0.5) * \boxCheight/3}
\draw (\decBxleft + \boxwidth, \yposB) -- +(\bendx,0) -- ({\decCxleft - \bendx}, \yposC) -- (\decCxleft, \yposC);
}
\draw (\decCxleft + \boxwidth, 0) -- +(\outlen, 0);
\def\leftend{\decCxleft + \boxwidth + \outlen}
\node at ({\leftend + \circuitgap/2}, 0) {\Large $=$};
\def\rightshift{\leftend + \circuitgap}

\fill[gray!30, decorate, decoration={random steps, segment length=3mm, amplitude=1mm}] 
  (\rightshift - \noisemargin, -\boxCheight/2 - \noisemargin) rectangle 
  (\rightshift + \grayboxwidth + \noisemargin, \boxCheight/2 + \noisemargin);

\fill[white] (\rightshift, -\boxCheight/2) rectangle (\rightshift + \grayboxwidth, \boxCheight/2);
\draw (\rightshift, -\boxCheight/2) rectangle (\rightshift + \grayboxwidth, \boxCheight/2);
\node at (\rightshift + \grayboxwidth/2, 0) {$\mathsf{C}$};

\node at (\rightshift + \grayboxwidth/2, \boxCheight/2 + \noisemargin + 0.3) {$\cE_0$};

\draw (\rightshift + \grayboxwidth, 0) -- +(\outlen, 0);
\end{tikzpicture}
      \caption{The statement of Corollary~\ref{cor:levelreductioniterated} illustrated for $N=1$ encoded qubits and $L=3$: Applying the ideal decoder~$\underset{L\rightarrow 0}{\mathsf{Dec}}$ after a strength-$p$ noisy execution~$\cE_{L}\bowtie\cC^{(L)}$ of the level-$L$ simulation~$\cC^{(L)}$ is equivalent to a noisy execution~$\cE_0\bowtie\cC$ of~$\mathsf{C}$ with smaller error strength $p_0$ related to~$p$.
      \label{fig:levelreductioniterated}}
      \end{figure}

Iterated level reduction is the basis for the following fundamental result. 
\begin{theorem}[Fault tolerance threshold theorem for local stochastic noise~\cite{aliferis2007level,aliferis2005quantum} paraphrased]\label{thm:mainthresholdtheoremagp}
There is a constant threshold $p_*>0$ such that the following holds.
Suppose $\cC$ is a prepare-and-measure circuit 
using  one- and two-qubit operations belonging to a certain (universal) set, and let $p(z)$, $z\in \{0,1\}^N$ be
the associated output distribution (see Eq.~\eqref{eq:samplingfromdistribution}). Let $\varepsilon>0$. 
Then there is fault-tolerant implementation~$\cC_{FT}$ of~$\cC$
with the following properties. The number of qubits and the depth of~$\cC_{FT}$ are bounded as 
\begin{align}
N(\cC_{FT})&\leq N(\cC)\cdot \mathsf{poly}(\log (|\mathsf{Loc}(\cC)|/\varepsilon))\\
\mathsf{depth}(\cC_{FT})&\leq \mathsf{depth}(\cC)\cdot 
\mathsf{poly}(\log(| \mathsf{Loc}(\cC)|/\varepsilon))\ .
\end{align}
Furthermore,
if $\cE$ is local stochastic noise on~$\cC_{FT}$  of strength $p\leq p_*$, then the output distribution~$\tilde{p}$ of 
$\cE\bowtie \cC_{FT}$ is $\varepsilon$-close in variational distance to~$p$. 
\end{theorem}
\noindent Summarizing Theorem~\ref{thm:mainthresholdtheoremagp}, we say that $\cC_{FT}$ fault-tolerantly implements 
$\cC$ with polylogarithmic overhead below the threshold error strength~$p_*$.

\subsubsection{Fault tolerance in 1D\label{sec:ft1Dlieterature}}
 A key ingredient of our construction is the ability to design fault-tolerant quantum circuits which only require connectivity between nearest neighbors on a line. Such fault-tolerant schemes have been described as early as Ref.~\cite{aharonov1997fault} (whilst more sophisticated methods which yield constant-rate encoding have been discovered recently \cite{gidney2025constant}). The idea is to simply add SWAP gates between circuit layers, to transport any qubits which must participate in a 2-qubit gate onto adjacent vertices. This leads to a polynomially enlarged circuit which requires only fault-tolerant gate execution between nearest neighbors. The standard approach of level reduction \cite{aliferis2007level,gottesman2024surviving} can be applied in this situation, using code concatenation \cite{knill1996concatenated} to iteratively design  circuits with lower effective noise rate, however there is some subtlety in 1D: at the lowest level of concatenation, a single faulty SWAP gate implementation will propagate errors onto two physical qubits --- this is incompatible with the requirements of fault-tolerant simulation \cite{gottesman2024surviving}. This issue can be overcome either by using a code which can correct 2-qubit errors, or by upgrading the geometry to a bilinear array \cite{gottesman2000fault} on which fault-tolerant swap operations can be realized. We opt for the latter approach here, using the work of Ref.~\cite{stephens2007universal}, in which the authors describe concrete gadgets for the 7-qubit Steane code \cite{steane1996error} to realize fault-tolerant computation in a bilinear array.

\subsubsection{Fault tolerance in a bilinear qubit array\label{sec:bilineararrayfaulttolerance}}
In Ref.~\cite{stephens2007universal}, Stephens, Fowler and Hollenberg demonstrate how to use a quasi-1D array of qubits with nearest-neighbor interactions to realize universal fault-tolerant quantum computation with a distance-$3$ code. The interaction graph is the bilinear array~$B_{2r}$ consisting of two parallel lines of $r$~qubits, see Fig.~\ref{fig:bilineararray}. Their construction relies 
on concatenating the Steane code~\cite{Steane7qubit}, an quantum $[[7,1,3]]$-code, with itself.

 Syndrome extraction circuits of Steane type~\cite{SteaneEC97}
can be used in the construction of the error correction gadget. These require the use of encoded ancillary states.
As proposed in~\cite{stephens2007universal}, verification and postselection routines associated with fault-tolerant preparation of encoded ancillas can be sidestepped using the
``ancilla decoding'' method of DiVincenzo and Aliferis~\cite{divincenzo2007effective}. Here an adaptive  Pauli correction is applied to the data qubits depending on measurement results obtained by measuring the ancilla qubits after their use.  One can easily check that the corresponding adaptive circuits are linear and read-once (see Section~\ref{sec:adaptivevsnonadaptivepauli}). 

As a CSS-code the $7$-qubit code  allows for a transversal implementation of the $CNOT$ gate. It also has  transversal implementations of the single-qubit Clifford gates~$\{X, Z, H, S\}$. One of the main challenges overcome by~\cite{stephens2007universal} is to realize fault tolerance gadgets (see Section~\ref{sec:FTgagdetssimulation}) 
for the $SWAP$ and $T$ gates with nearest-neighbor operations on the bilinear array. As shown in~\cite{stephens2007universal}, these components can then be used to recompile
a geometrically local circuit on a line (a linear nearest-neighbor  circuit in the terminology of~\cite{stephens2007universal}) into a  fault-tolerant circuit 
on a bilinear array. The provided threshold proof relies on the extended rectangles framework~\cite{aliferis2005quantum} and provides explicit threshold estimates.

Using these gadgets, it is argued in~\cite{stephens2007universal}, that  a geometrically $1D$-local circuit~$\cC$ can be recompiled into a fault-tolerant circuit which only nearest-neighbor operations on the bilinear array. Here we call a circuit~$\cC$ on $N$~qubits geometrically $1D$-local  if it is $\pathgraph_{N}$-local 
on the path graph $\pathgraph_N$ (see Fig.~\ref{fig:linegraphfigure}), i.e., if it only uses nearest-neighbor operations when the qubits are arranged on a line.  The main finding of~\cite{stephens2007universal}
 (combined with 
Theorem~\ref{thm:mainthresholdtheoremagp})
can then be summarized as follows.
\begin{theorem}[Fault tolerance on the bilinear qubit array~\cite{stephens2007universal}, paraphrased informally]\label{thm:ftbilineararrayconstruct}
There is a constant threshold error strength~$p_*>0$ such that the following holds.  Let $\varepsilon>0$ and  $\cC$ be a $1D$-local prepare-and-measure circuit~$\cC$ composed of a certain universal gate set. Then there a circuit~$\cC^{\mathsf{bilinear}}$
whose output distribution is $\varepsilon$-close in variational distance 
to that of $\cC$ in the presence of local stochastic noise of strength $p\leq p_*$. The circuit~$\cC^{\mathsf{bilinear}}$ is  local on the bilinear qubit array. This construction has a polylogarithmic overhead. 
\end{theorem}

\section{Robust implementation of circuits\label{sec:noisyencodingdecoding}}
Here we extend the standard quantum fault tolerance framework
outlined in Section~\ref{sec:fault tolerance} to apply more generally in situations where the circuit
acts on (quantum) input states, and may produce (not only a sample from a distribution, but) a general quantum output state. Such extensions were recently discussed in~\cite{christandl2026fault,belzig2026constant}.

Our approach is particularly simple: We argue that the fault tolerance  framework pioneered in~\cite{aliferis2005quantum,aliferis2007level} already essentially treats this setup when a few minor modifications are applied. In particular, noisy encoding and decoding circuits are added and analyzed. Such ideas appeared as early as Ref.~\cite{knill1996concatenated}, and recent works~\cite{christandl2026fault,he2025composable} have given alternative general formulations of the problem.

In Section~\ref{sec:robustimplementation}, we introduce the relevant notion  of a robust implementation of a circuit.  In Section~\ref{sec:ftgadgetsencoding} we discuss encoding and decoding gadgets, and 
then apply this to concatenated codes in Section~\ref{sec:concatenatedencoding}.
In Section~\ref{sec:levelreductionnoisy}, we argue that the level-reduction
technique of~\cite{aliferis2005quantum}  applies to our generalized setting with noisy encoders and decoders. 
In Section~\ref{sec:robustimplementationconcatenated} we summarize our main results showing that robust implementations can be obtained by fault tolerance constructions with concatenated codes.

\subsection{Robust implementations of circuits: Definition}\label{sec:robustimplementation}
Concretely, we consider 
the problem of rendering 
the following more general circuits fault-tolerant (cf.~Section~\ref{sec:FTgagdetssimulation}). We consider a circuit~$\mathsf{C}$ on $N$ (``logical'') qubits, with $N_{\mathsf{in}}$-qubit input and $N_{\mathsf{out}}$-qubit output consisting of the following steps:
\begin{enumerate}[(i)]
\item\label{it:prepstepgeneralcircuitgeneralized} A system of $N_{\mathsf{in}}$ qubits is taken as input.
\item The computational basis state~$\ket{0}^{\otimes (N - N_{\mathsf{in}})}$ is prepared on the remaining $N-N_{\mathsf{in}}$ qubits.\label{it:prepmvb}
\item\label{it:evolutionstepgeneralcircuitgeneralized}
A sequence $U_1,\ldots,U_s$ of one- and two-qubit gates (unitaries) from a specified gate set is applied.
\item A subset of $N-N_{\mathsf{out}}$ qubits are measured, producing a  sample $z \in \{0,1\}^{N - N_{\mathsf{out}}}$.\label{it:measurementstepgeneralizmv}
\item The remaining~$N_{\mathsf{out}}$ qubits are taken as output.\label{it:measurementstepgeneralcircuitgeneralized}
\end{enumerate}
The following definition will be convenient for a circuit~$\cC$ 
consisting of the steps~\eqref{it:prepstepgeneralcircuitgeneralized}--\eqref{it:measurementstepgeneralcircuitgeneralized}.
\begin{definition}[Robust circuit implementation]\label{def:ftimplementation}
Consider a circuit $\cC$ on $N$ qubits, with $N_{\mathsf{in}}$-qubit input and $N_{\mathsf{out}}$-qubit output. Let $\cC_{FT}$ be another circuit with the same number of input and output qubits. We say that $\cC_{FT}$ fault-tolerantly implements $\cC$ with robustness $f : [0,1] \rightarrow \RR$ under local stochastic (Pauli) noise if the following holds for any $p\in [0,1]$. For any local stochastic (Pauli) noise $\cE$ on $\cC_{FT}$, we have
\begin{align}
    \prob[\cE \bowtie \cC_{FT} \neq \cC] \leq f(p)\ .
\end{align}
In other words, the noisy realization of $\cC_{FT}$ has the same action as a CPTP map as $\cC$, except with probability $f(p)$. For brevity, we refer to $\cC_{FT}$ as an  $f$-robust implementation of $\cC$ in this case.
\end{definition}

\subsection{Fault tolerance gadgets for encoding and decoding\label{sec:ftgadgetsencoding}}

To obtain a robust circuit implementation~$\cC_{FT}$ starting from a circuit~$\cC$ as specified by~\eqref{it:prepstepgeneralcircuitgeneralized}--\eqref{it:measurementstepgeneralcircuitgeneralized}, we follow the gadget approach for concatenated fault tolerance: 
 we first replace each operation in~\eqref{it:prepmvb}--\eqref{it:measurementstepgeneralizmv} 
by the associated gadget,
recursively defining simulations at different levels following th procedure described in Section~\ref{sec:FTgagdetssimulation}.  The  resulting level-$L$ circuit $\cC^{(L)}$ acts on $N_{\mathsf{in}}$ logical qubits each represented by $n^L$ physical qubits via the concatenated code $\cL^{\circ L}$ (and similarly for the output systems).  For our purposes, however,
 it will be important to consider a circuit mapping directly between unencoded (physical) qubits, acting by encoding into the level-$L$ logical space, applying $\cC^{(L)}$, and then decoding the logical information.

In order to include inputs and outputs in our circuit $\cC$, it will  be necessary to define level-$1$ encoding and decoding gadgets  to be inserted on the input and output wires.
These should be contrasted to the use of ``ideal decoders'':  In \cite{aliferis2007level} for example, the ideal decoders are viewed as mathematical proof tools which are not physically implemented, 
see the discussion in Section~\ref{sec:ftgadgetsdiscussion}.
Instead, here we   view both the encoder~$\enc$ as well as the decoder~$\dec$  as concrete (adaptive) circuits (themselves composed of $0$-Gas). We will ilustrate these by the diagrams shown in Fig.~\ref{fig:ftencodingdecodinggadgetsrep}.
\begin{figure}[htbp]
    \centering
    \begin{subfigure}[b]{0.2\textwidth}
        \centering
        \begin{quantikz}
            \qw & \decgate[fill=gray!30]{} & \qw
            \arrow[from=1-1,to=1-2,black,line width=\encwidth,-]{}
        \end{quantikz}
        \caption{Decoder gadget}
    \end{subfigure}
   \qquad\qquad 
    \begin{subfigure}[b]{0.2\textwidth}
        \centering 
        \begin{quantikz}
            \qw &\encgate[fill=gray!30]{} & \qw
            \arrow[from=1-2,to=1-3,black,line width=\encwidth,-]{}
        \end{quantikz}
        \caption{Encoder gadget}
   \end{subfigure}
        \caption{Diagrammatic notation for decoder and encoding gadgets. As before, thick (thin) lines represent encoded (decoded) qubits. 
These gadgets should be contrasted with the ideal decoder (see Fig.~\ref{fig:ftgadgetcomponentsnotions}).
}
    \label{fig:ftencodingdecodinggadgetsrep}
\end{figure}
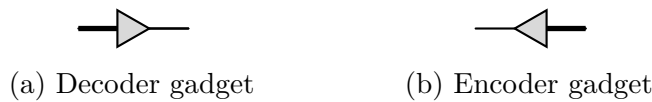

The decoder gadget~$\dec$ is a circuit which --- in the absence of faults --- maps from $n$ physical qubits into $1$ logical qubit, decoding from the code space~$\cL$ and correcting up to $t$~errors. 
In other words, the decoder gadget~$\dec$ is a particular realization 
of the ideal decoder~$\idec$ defined by Eq.~\eqref{eq:conditionerrorrecovery}.  That is, as a CPTP map, the decoder gadget~$\dec:\cB((\mathbb{C}^2)^{\otimes n})\rightarrow\cB(\mathbb{C}^2)$
has the same action on states of the form~$E\ket{\Psi}$, where~$\ket{\Psi}\in\cL$ belongs to the code space and~$E$ is a correctable Pauli error.  This means that  the decoder gadget~$\dec$ can be used as a substitute for the ideal decoder~$\idec$ in the conditions defining various gadgets (see Fig.~\ref{fig:ft-definitionsgadgets}). We illustrate this in Fig.~\ref{fig:decodergadgetprop}.
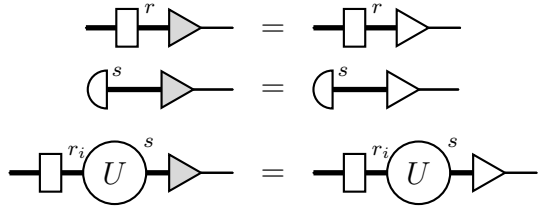
\begin{figure}[htbp]
        \centering
        $\begin{aligned}
                \begin{quantikz}[column sep=0.4cm]
            \qw & \filtr & \decgate[fill=gray!30]{} & \qw
            \arrow[from=1-1,to=1-3,black,line width=\encwidth,-]{}
        \end{quantikz}
&\;=\;
        \begin{quantikz}[column sep=0.4cm]
            \qw & \filtr & \decgate{} & \qw
            \arrow[from=1-1,to=1-3,black,line width=\encwidth,-]{}
        \end{quantikz}\\
        \begin{quantikz}
        \prepgate["s"{font=\scriptsize,yshift=7,xshift=13}]{}  & \decgate[fill=gray!30]{} & \qw
        \arrow[from=1-1,to=1-2,black,line width=\encwidth,-]{}
        \end{quantikz} 
        &\;=\;
        \begin{quantikz}
            \prepgate["s"{font=\scriptsize,yshift=7,xshift=13}]{}  & \decgate{} & \qw
        \arrow[from=1-1,to=1-2,black,line width=\encwidth,-]{}
        \end{quantikz}\\
        \begin{quantikz}[column sep=0.4cm]
            \qw & \filtri & \ftgates{U} & \decgate[fill=gray!30]{} & \qw
            \arrow[from=1-1,to=1-4,black,line width=\encwidth,-]{}
        \end{quantikz}
        &\;=\;
      \begin{quantikz}[column sep=0.4cm]
            \qw & \filtri & \ftgates{U} & \decgate{} & \qw
            \arrow[from=1-1,to=1-4,black,line width=\encwidth,-]{}
        \end{quantikz}
        \end{aligned}$
        \caption{Properties of the decoder gadget ($r\leq t$, $s\leq t$ and $s+\sum_i r_i\leq t$)\label{fig:decodergadgetprop}}
\end{figure}

The  encoder gadget $\enc$, meanwhile, is a circuit which acts as a CPTP map 
$\enc:\cB((\mathbb{C}^2)\rightarrow \cB((\mathbb{C}^2)^{\otimes n})$. 
In the absence of faults, it isometrically maps $\mathbb{C}^2$ into the code space~$\cL$. Furthermore, we assume that its restriction to the code space is inverted by the ideal decoder, i.e., we have  
\begin{align}\label{eq:decenccomp}
    \idec \circ \enc = \mathsf{id}\ .
\end{align}
These conditions are illustrated in Fig.~\ref{fig:encodingdecodinggadgetcorrectness}.
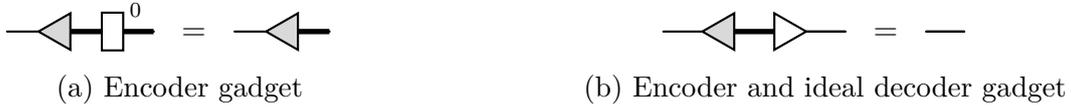
\begin{figure}[htbp]
\centering
        \begin{subfigure}[b]{0.3\textwidth}
        \centering
        $\begin{aligned}\begin{quantikz}[column sep=0.4cm]
\qw & \encgate[fill=gray!30]{} & \filtrzero & \qw
\arrow[from=1-2,to=1-4,black,line width=\encwidth,-]{}
\end{quantikz}
&\;=\;
\begin{quantikz}[column sep=0.4cm]
\qw & \encgate[fill=gray!30]{} & \qw
\arrow[from=1-2,to=1-3,black,line width=\encwidth,-]{}
\end{quantikz}
\end{aligned}$
        \caption{Encoder gadget\label{it:encodinggadgetLspace}}
    \end{subfigure}\qquad \qquad\qquad
 \begin{subfigure}[b]{0.4\textwidth}
        \centering
        $\begin{aligned}
        \begin{quantikz}
            \qw & \encgate[fill=gray!30]{} & \decgate{} & \qw
            \arrow[from=1-2,to=1-3,black,line width=\encwidth,-]{}
        \end{quantikz}
    &\;=\;
    \begin{quantikz}
        \qw & \qw
    \end{quantikz}
\end{aligned}$
        \caption{Encoder  and  ideal decoder gadget\label{it:encodingdecodinggadget}}
    \end{subfigure}
        \caption{Conditions for fault tolerance of encoding and decoding $1$-gadgets for an $[[n,1,d]]$-code, with $t=\lfloor (d-1)/2\rfloor$.
In these conditions, we assume that the encoding and decoding gadgets (gray) are executed without errors.
Condition~\eqref{it:encodinggadgetLspace}  states  that the encoding gadget maps into the code space.
The diagram~\eqref{it:encodingdecodinggadget} 
expresses the condition given by Eq.~\eqref{eq:decenccomp}: the ideal decoder inverts the encoding map.
    \label{fig:encodingdecodinggadgetcorrectness}}
\end{figure}

For an $[[n,1,d]]$-stabilizer code~$\cL$, the encoding and decoding gadgets can be constructed as follows: There is a Clifford unitary~$U$ which maps~$\mathbb{C}^2\otimes (\mathbb{C}\ket{0^{n-1}})$ isometrically into~$\cL$.
Decomposing $U$ into  Clifford gates belonging to the considered gate set, we can define a circuit~$\enc$ which consists in preparing~$\ket{0^{n-1}}$ on $n-1$~qubits, and applying~$U$. In other words, its action is given by~$\enc(\rho)=U(\rho\otimes\proj{0^{n-1}})U^\dagger$. 
An associateding decoding circuit~$\dec$ is obtained by applying the inverse of~$U$ (again by applying Clifford gates), measuring  and tracing out $n-1$~qubits in the computational basis, and applying a suitably chosen Pauli correction~$C(s)$
depending on the measurement outcome~$s\in \{0,1\}^{n-1}$. That is, the decoding map acts as
\begin{align}
\dec(\rho)&=\sum_{s\in \{0,1\}^{n-1}} 
 C(s)(I\otimes \bra{s})(U^\dagger \rho U)(I\otimes \ket{s})C(s)^\dagger\ .
\end{align}
The circuits~$(\enc,\dec)$ here are specified in terms of level-$0$ gadgets. Correspondingly, we write $\enc^{(1)} = \enc$ and $\dec^{(1)}=\dec$, referring to these circuits as the level-$1$ encoding and decoding gadgets.

\subsection{Concatenated encoding and decoding\label{sec:concatenatedencoding}}
Clearly, as the final $N_{\mathsf{out}}$-qubit state prepared this way is not encoded, the resulting state will at best be corrupted by noise of strength of the order of the physical noise strength (and similar errors will occur during the input phase). We show that this is in fact achievable, and has useful practical applications.

For any $j\geq 2$, we again define a level-$j$ to level-$(j-1)$ decoder. 
The map $\mathsf{Dec}^{(j)}$ is defined recursively as the result of applying the substitution rules~\eqref{it:substitutionruleprep}--\eqref{it:substitutionrulepartialtrace} to~$\mathsf{Dec}^{(j-1)}$. 
 We denote the composed map defined by applying these maps in succession (see  Eq.~\eqref{eq:composeddecoder}) as 
$\underset{L\rightarrow 0}{\mathsf{Dec}}$.

For any $j\geq 2$, we define $\enc^{(j)}$ recursively by applying the substitution rules~\eqref{it:substitutionruleprep}-\eqref{it:substitutionrulemeasurement} to $\enc^{(j-1)}$. This is a map taking $n^{j-1}$ physical qubits to $n^j$ qubits, which increases the concatenation level from $j-1$ to $j$. Note that by the decomposition rule given by Eq.~\eqref{eq:compositioncrs}, it is clear that $\dec^{(j)}\circ \enc^{(j)}$ acts as the identity on the $(j-1)$-level logical space by construction.  We thus define the ideal level-by-level encoder by
\begin{align}
    \cenc{0\rightarrow L} := \enc^{(L)} \circ \dots \circ \enc^{(1)}\ ,\label{eq:composedencoder}
\end{align}
which takes an unencoded physical qubit into a corresponding logical state in the $L$-fold concatenated code. In analogy with Eq.~\eqref{eq:parallelconcdec}, we have
\begin{align}
    \left( \cenc{0\rightarrow L} \right)^{\otimes N} = (\enc^{(L)})^{\otimes N}\circ \dots\circ (\enc^{(1)})^{\otimes N}\ ,
\end{align}
and by iterating our previous observation we see that
\begin{align}
    \cdec{L\rightarrow 0}\circ \cenc{0\rightarrow L} = \mathsf{id}\ .
\end{align}

We will be thus be studying the circuit
\begin{align}
\left( \cdec{L\rightarrow 0} \right)^{\otimes N_{\mathsf{out}}} \circ \cC^{(L)} \circ \left(\cenc{0\rightarrow L}\right)^{\otimes N_{\mathsf{in}}}
\end{align}
as well as noisy executions thereof. The number of qubits and the depth of the encoding and decoding steps can be bounded as follows by adapting the definition the gadget-dependent constants $\nmax,\dmin,\dmax$.
\begin{lemma}[Qubit and depth overhead of encoding and decoding]\label{lem:qubitdepthoverheadsimulationb}
Let $\nmax, \dmin, \dmax$ be the gadget-dependent constants defined in Eq.~\eqref{eq:dminmaxdef},
where the level-$1$ encoding and 
decoding gadgets~$\mathsf{Enc}^{(1)}$ and $\mathsf{Dec}^{(1)}$ are included in the set of $1$-Gas.
The depth and qubit count of the encoding and decoding circuits can be bounded by
\begin{align}
\begin{matrix}
n^L  &\leq  & N\left(\cenc{0\rightarrow L}\right) &\leq \nmax^{L+1}\\
\dmin^{L-1}  &\leq   &\mathsf{depth}\left(\cenc{0\rightarrow L}\right) &\leq \dmax^{L+1} 
\end{matrix}\quad\textrm{ and }\quad
\begin{matrix}
\nmin^{L-1} &\leq  & N\left(\cdec{L\rightarrow 0}\right) &\leq & \nmax^{L+1}\\
\dmin^{L-1}  &\leq  & \mathsf{depth}\left(\cdec{L\rightarrow 0}\right) &\leq &\dmax^{L+1}
\end{matrix}\ .
\end{align}
In particular, the number of qubits and the circuit depth of~$\cenc{0\rightarrow L}$ and~$\cdec{L\rightarrow 0}$ scale as $e^{\Theta(L)}$.
\end{lemma}
\begin{proof} As an immediate consequence of Lemma~\ref{lem:qubitdepthoverheadsimulation} we have 
\begin{align}
N(\enc^{(j)})&\leq \dmax^{j-1} N(\enc^{(1)}) \leq \nmax^j\qquad\textrm{ for every }\qquad j\in [L]\ .
\end{align}
Hence we may bound
\begin{align}
 N\left(\cenc{0\rightarrow L} \right) \leq \sum_{j=1}^L N(\enc^{(j)}) \leq \sum_{j=1}^L \nmax^j \leq \nmax^{L+1}\ .
\end{align}
It is clear that $n^L\leq 
 N\left(\cenc{0\rightarrow L} \right)$.

The corresponding results for number of qubits and the circuit depth for the decoding circuit follow by identical calculations.
\end{proof}

\subsection{Level reduction with noisy encoding and decoding\label{sec:levelreductionnoisy}}
  We are interested in the effect of noise on  concatenated encoders and decoders, respectively. As we argue below in detail, the dominant 
contribution to the errors comes from the lowest layers of encoding and decoding, and  the overall error can be bounded independently of the concatenation level.

 The first step towards this result is to establish that the effective noise strength can be reduced in the level-$1$ simulation of a circuit --- this boils down to the level reduction lemma below, illustrated in Fig.~\ref{fig:levelreductionillustrated}, which we adapt from \cite{aliferis2007level}. Subsequently, Lemma~\ref{lem:levelreductionencdec} establishes a version of this result with noisy encoders and decoders, whilst Lemma~\ref{lem:itlevelreductionencdec} iterates this to obtain a statement about level-$L$ circuit simulation for any $L$. Similar arguments were made as far back as~\cite{knill1996concatenated}, and appear implicitly in other works such as~\cite{gottesman2014fault}.

\begin{figure}
\centering
\begin{subfigure}[b]{0.9\textwidth}
    \centering
        \begin{quantikz}[column sep=0.4cm]
            \qw & \encgate[fill=gray!30]{} & \EC & \ftgate{U^{(1)}} & \EC & \decgate[fill=gray!30]{} & \qw
            \arrow[from=1-2,to=1-6,black,line width=\encwidth,-]{}
        \end{quantikz}
        $\;=\;$
        \begin{quantikz}[column sep=0.4cm]
            \qw & \encgate[fill=gray!30]{} & \EC & \ftgate{U^{(1)}} & \EC & \decgate{} & \qw
            \arrow[from=1-2,to=1-6,black,line width=\encwidth,-]{}
        \end{quantikz}
    \caption{First step: after the error correction step projects onto the codespace, the action of the decoder gadget can identified with the ideal decoder $\idec$, see Fig.~\ref{fig:decodergadgetprop}. This can be propagated backwards through the circuit as in Fig.~\ref{fig:levelreductionbasic}.}
\end{subfigure}
\begin{subfigure}[b]{0.9\textwidth}
    \centering
    \begin{quantikz}
        \qw & \encgate[fill=gray!30]{} & \EC & \decgate{} & \ftgate{U} & \qw
        \arrow[from=1-2,to=1-4,black,line width=\encwidth,-]{}
    \end{quantikz}
    $\;=\;$
    \begin{quantikz}
        \qw & \ftgate{U} & \qw
    \end{quantikz}
    \caption{Last step: the ideal decoder meets an encoder gadget at the beginning of the circuit. In the absence of faults, their composition acts as the identity and leaves only the desired logical circuit acting on the physical space, see Fig.~\ref{fig:encodingdecodinggadgetcorrectness}.}
\end{subfigure}
\caption{Illustration of the modifications
used to establish Lemma~\ref{lem:level reduction}. This should be compared to Fig.~\ref{fig:levelreductionbasic}.
\label{fig:levelreductionadvanced}}
\end{figure}

\begin{lemma}[Level reduction with ideal encoding and decoding, Lemma~3 of~\cite{aliferis2007level} paraphrased]\label{lem:level reduction}
Fix an $[[n,1,d]]$-code $\cL$ and associated level-0 and level-1 gadgets satisfying the fault tolerance properties given in Fig.~\ref{fig:ft-definitionsgadgets}. Let $t := \lfloor (d-1)/2\rfloor$. Let $\dec$ and $\enc$ be decoding and encoding circuits respectively for the code $\cL$, and let $\cC$ be an arbitrary $N$-qubit circuit with $N_{\mathsf{in}}$ input qubits and $N_{\mathsf{out}}$ output qubits. Let $\cC^{(1)}$ denote the level-1 simulation of $\cC$. Then there is a constant $c > 0$ depending only on the code and the gadgets used such that the following holds: for any local stochastic noise $\cE_1$ on $\cC$ of strength $p_1 = p$, there is local stochastic noise $\cE_0$ of strength $p_0 = cp_1^{t+1}$ such that
\begin{align}
    \dec^{\otimes N_{\mathsf{out}}} \circ \left( \cE_1 \bowtie \cC^{(1)} \right) \circ \enc^{\otimes N_{\mathsf{in}}} = \cE_0 \bowtie \cC\ .
\end{align}

\end{lemma}
\begin{proof}
    The formulation and proof of this result follow almost exactly the proof of Lemma~3 in~\cite{aliferis2007level}  (see Lemma~\ref{lem:level reductionoriginal} in Section~\ref{sec:levelreductionftthresholdthm}) with only a slight modifications for the input and output qubits (in~\cite{aliferis2007level}, circuits beginning with state preparations and ending with measurements are considered).
    Fig.~\ref{fig:levelreductionadvanced} illustrates some of the proof compontents.
         Using the property of measurement gadgets given in Fig.~\ref{it:measurementftprop}, an ideal decoder is introduced which is then ``pulled through'' the circuit from right (last time step) to left. 
    
    In our formulation, we consider a circuit~$\cC$ which outputs a quantum state. The ideal decoder on the output wires $\dec^{\otimes N_{\mathsf{out}}}$ is explicitly included in the statement of the lemma, and can be pulled through the remainder of the level-1 simulation as argued in~\cite{aliferis2005quantum,aliferis2007level} to obtain a noisy implementation of $\cC$. When the ideal decoders are pulled all the way back to the input wires, they meet the ideal encoders $\enc^{\otimes N_{\mathsf{in}}}$, which they compose with to give identity maps by their defining property (Eq.~\eqref{eq:decenccomp}). We refer to~\cite{aliferis2007level} for further details. 
    \end{proof}

    Below we state and prove a version of Lemma~\ref{lem:level reduction} which accounts for noise also acting in the encoding and decoding steps. The conclusion is essentially the same up to some probability of error proportional to the size of the encoding and decoding circuits. The proof idea is illustrated in Figure \ref{fig:sketch level reduction inout}.

\begin{lemma}[Level reduction with noisy encoding and decoding]\label{lem:levelreductionencdec}
    Fix an error-correcting code and associated fault-tolerant gadgets as in Lemma~\ref{lem:level reduction}.
 Let $\nmax, \dmin, \dmax$ be the gadget-dependent constants defined in Eq.~\eqref{eq:dminmaxdef}.
 Let $\cC$ be an arbitrary $N$-qubit circuit with $N_{\mathsf{in}}$ input qubits and $N_{\mathsf{out}}$ output qubits. Let $\cC^{(1)}$ denote the level-1 simulation of $\cC$. Then there exists a constant $c > 0$ depending only on the code and the gadgets used such that the following holds: for any local stochastic noise $\cE_1$ of strength $p = p_1$ on the composite circuit $\dec^{\otimes N_{\mathsf{out}}}\circ \cC^{(1)} \circ \enc^{\otimes N_{\mathsf{in}}}$, there is local stochastic noise $\cE_0$ of strength $p_0 = cp_1^{t+1}$ such that
    \begin{align}
        \prob\left[\cE_1 \bowtie \left(\dec^{\otimes N_{\mathsf{out}}} \circ \cC^{(1)} \circ \enc^{\otimes N_{\mathsf{in}}}\right) \neq \cE_0 \bowtie \cC \right] \leq p\nmax\dmax (N_{\mathsf{out}} + N_{\mathsf{in}})\ .
    \end{align}
\end{lemma}

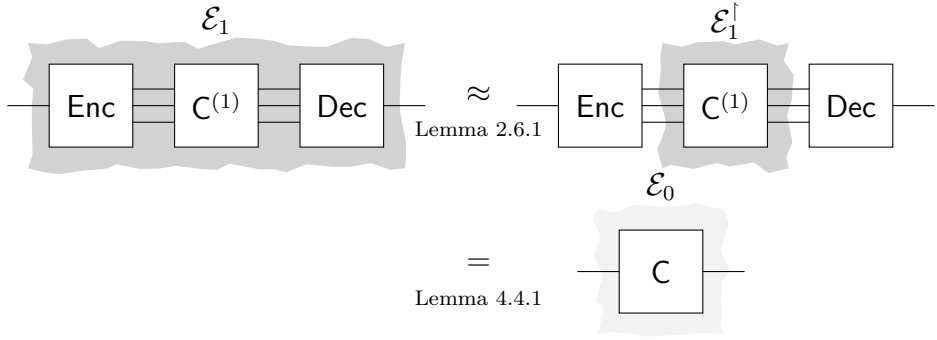
\begin{figure}
    \centering
    \begin{tikzpicture}
\def\bw{1.1} 
\def\bh{1.1}  
\def\gap{0.55}
\def\wl{0.55}  
\def\nm{0.28} 
\def\nwires{3} 
\tikzset{
  noisefill/.style={fill=gray!35,
    decorate,
    decoration={random steps, segment length=2.5mm, amplitude=0.8mm}},
  box/.style={draw, fill=white, minimum width=\bw cm,
              minimum height=\bh cm, inner sep=2pt}
}
\tikzset{
  weaknoisefill/.style={fill=gray!10,
    decorate,
    decoration={random steps, segment length=2.5mm, amplitude=0.8mm}},
  box/.style={draw, fill=white, minimum width=\bw cm,
              minimum height=\bh cm, inner sep=2pt}
}

\def\EncX{0}
\def\CX{\EncX + \bw + \gap}
\def\DecX{\CX  + \bw + \gap}

\def\blobL{\EncX - \bw/2 - \nm}
\def\blobR{\DecX + \bw/2 + \nm}
\def\blobB{-\bh/2 - \nm}
\def\blobT{ \bh/2 + \nm}

\fill[noisefill]
  (\blobL, \blobB) rectangle (\blobR, \blobT);

\pgfmathsetmacro{\leftBlobCentreX}{(\EncX + \DecX) / 2}
\node at (\leftBlobCentreX, \bh/2 + \nm + 0.3) {$\mathcal{E}_1$};

\node[box] (Enc) at (\EncX, 0) {$\mathsf{Enc}$};
\node[box] (C1L) at (\CX,   0) {$\mathsf{C}^{(1)}$};
\node[box] (Dec) at (\DecX, 0) {$\mathsf{Dec}$};

\foreach \i in {1,...,\nwires}{
  \pgfmathsetmacro{\yy}{(\i - (\nwires+1)/2)*0.22}
  \draw (\EncX+\bw/2, \yy) -- (\CX-\bw/2, \yy);
}
\foreach \i in {1,...,\nwires}{
  \pgfmathsetmacro{\yy}{(\i - (\nwires+1)/2)*0.22}
  \draw (\CX+\bw/2, \yy) -- (\DecX-\bw/2, \yy);
}
\draw (\EncX-\bw/2, 0) -- (\EncX-\bw/2-\wl, 0);
\draw (\DecX+\bw/2, 0) -- (\DecX+\bw/2+\wl, 0);

\def\relX{\DecX + \bw/2 + \wl + 0.7}
\node at (\relX, 0.15)   {$\approx$};
\node at (\relX, -0.3) {\scriptsize Lemma \ref{lem:localstochasticnoiseqoutput}};

\def\REncX{\relX + 1.6}
\def\RCX{\REncX + \bw + \gap}
\def\RDecX{\RCX  + \bw + \gap}

\fill[noisefill]
  (\RCX-\bw/2-\nm, -\bh/2-\nm) rectangle (\RCX+\bw/2+\nm, \bh/2+\nm);

\node at (\RCX, \bh/2 + \nm + 0.3) {$\mathcal{E}_1^{\upharpoonright}$};

\node[box] (REnc) at (\REncX, 0) {$\mathsf{Enc}$};
\node[box] (RC1)  at (\RCX,   0) {$\mathsf{C}^{(1)}$};
\node[box] (RDec) at (\RDecX, 0) {$\mathsf{Dec}$};

\foreach \i in {1,...,\nwires}{
  \pgfmathsetmacro{\yy}{(\i - (\nwires+1)/2)*0.22}
  \draw (\REncX+\bw/2, \yy) -- (\RCX-\bw/2,  \yy);
  \draw (\RCX+\bw/2,   \yy) -- (\RDecX-\bw/2, \yy);
}
\draw (\REncX-\bw/2, 0) -- (\REncX-\bw/2-\wl, 0);
\draw (\RDecX+\bw/2, 0) -- (\RDecX+\bw/2+\wl, 0);

\def\rowY{-2.2}

\node at (\relX, \rowY+0.1)   {$=$};
\node at (\relX, \rowY-0.35) {\scriptsize Lemma \ref{lem:level reduction}};

\def\CbX{\relX + 2.4}
\fill[weaknoisefill]
  (\CbX-\bw/2-\nm, \rowY-\bh/2-\nm) rectangle (\CbX+\bw/2+\nm, \rowY+\bh/2+\nm);

\node at (\CbX, \rowY + \bh/2 + \nm + 0.3) {$\mathcal{E}_0$};

\node[box] at (\CbX, \rowY) {$\mathsf{C}$};

\draw (\CbX-\bw/2, \rowY) -- (\CbX-\bw/2-\wl, \rowY);
\draw (\CbX+\bw/2, \rowY) -- (\CbX+\bw/2+\wl, \rowY);

\end{tikzpicture}
    
    \caption{Proof of Lemma~\ref{lem:levelreductionencdec}, illustrated for a circuit $\cC^{(1)}$ with $N_{\mathsf{in}} = N_{\mathsf{out}} = 1$. Given a noisy implementation of $\dec \circ \cC^{(1)}\circ \enc$, we can apply Lemma \ref{lem:localstochasticnoiseqoutput} to show that the noise can be localized to the subcircuit $\cC^{(1)}$, up to some small probability of error. Applying Lemma \ref{lem:level reduction} establishes that this is then equivalent to an implementation of $\cC$ under weaker noise.}
    \label{fig:sketch level reduction inout}
\end{figure}

\begin{proof}
    Define $\cE_1^{\localization}$ to be the localization of $\cE_1$ to $\cC^{(1)}$, i.e., without any errors on the encoding or decoding steps. By Lemma~\ref{lem:localstochasticnoiseqoutput}\eqref{it:circnoisycomposedclaim}, we see that
    \begin{align}\label{eq:levelreductionrestriction}
        \prob [ \cE_1 \neq \cE_1^{\localization}] &\leq p\cdot \left( |\mathsf{Loc}(\dec^{\otimes N_{\mathsf{out}}})| + |\mathsf{Loc} (\enc^{\otimes N_{\mathsf{in}}})|\right) \\
        &\leq p \nmax\dmax (N_{\mathsf{out}} + N_{\mathsf{in}})\ ,
    \end{align}
    where we have used the fact that
    \begin{align}
        |\mathsf{Loc}(\dec)| \leq N(\dec) \cdot \mathsf{depth}(\dec) \leq \nmax\dmax\ ,
    \end{align}
    and similarly for $\enc$. Furthermore, the definition of $\cE_1^{\localization}$ via Eq.~\eqref{eq:combrestriction} immediately implies that $\cE_1^{\localization}$ is also local stochastic noise of strength $p$. Hence, by Lemma~\ref{lem:level reduction}, there exists local stochastic noise $\cE_0$ of strength $p_0 = cp_1^{t+1}$ such that
    \begin{align}\label{eq:levelreductionusinthelemma}
        \dec^{\otimes N_{\mathsf{out}}} \circ \left( \cE_1^{\localization} \bowtie \cC^{(1)} \right)\circ \enc^{\otimes N_{\mathsf{in}}} = \cE_0 \bowtie \cC\ .
    \end{align}
    Combining Eq.~\eqref{eq:levelreductionrestriction} with Eq.~\eqref{eq:levelreductionusinthelemma} yields the desired result.
\end{proof}

Inductively applying Lemma~\ref{lem:levelreductionencdec} $L$ times, we arrive at the following result (which is essentially a version of \cite[Lemma 3]{aliferis2007level}). The proof of this lemma is sketched in Fig.~\ref{fig:sketch iterated level reduction inout}.

\begin{lemma}[Iterated level reduction with noisy encoding and decoding]\label{lem:itlevelreductionencdec}
    Fix an error-correcting code and associated fault-tolerant gadgets as in Lemma~\ref{lem:level reduction}. 
Let $\nmax, \dmin, \dmax$ be the gadget-dependent constants defined in Eq.~\eqref{eq:dminmaxdef},
where the level-$1$ encoding and 
decoding gadgets~$\mathsf{Enc}^{(1)}$ and $\mathsf{Dec}^{(1)}$ are included in the set of $1$-Gas. 
Let $\cC$ be an arbitrary $N$-qubit circuit with $N_{\mathsf{in}}$ input qubits and $N_{\mathsf{out}}$ output qubits. Fix $L \geq 1$, and let $\cC^{(L)}$ denote the level-$L$ simulation of $\cC$. Then there is a constant $p_\ast > 0$ depending only on the code and gadgets used such that the following holds: for any stochastic noise $\cE_L$ of strength $p = p_L < p_\ast/2$ on the composite circuit $(\cdec{L\rightarrow 0})^{\otimes N_{\mathsf{out}}} \circ \cC^{(L)} \circ (\cenc{0 \rightarrow L})^{\otimes N_{\mathsf{in}}}$, there is local stochastic noise $\cE_0$ on $\cC$ of strength
    \begin{align}
        p_0 = p_\ast (p/p_\ast)^{(t+1)^L}\ ,
    \end{align}
    such that
    \begin{align}
        \prob \left[\cE_L \bowtie \left((\cdec{L\rightarrow 0})^{\otimes N_{\mathsf{out}}} \circ \cC^{(L)} \circ (\cenc{0\rightarrow L})^{\otimes N_{\mathsf{in}}} \right) \neq \cE_0 \bowtie \cC \right] \leq p\cdot 2 \nmax\dmax(N_{\mathsf{out}} + N_{\mathsf{in}})\ .
    \end{align}
\end{lemma}

\begin{figure}
    \centering
    \begin{tikzpicture}
\def\bw{1.3}   
\def\tallbh{1.5}
\def\smallbh{1.0}
\def\smallgap{0.45}
\def\largegap{0.65} 
\def\wl{0.55}        
\def\nm{0.28}      
\def\capoff{0.32}  
\def\spreadnine{0.14} 
\def\spreadthree{0.22}

\tikzset{
  noisefill/.style={fill=gray!35, decorate,
    decoration={random steps, segment length=2.5mm, amplitude=0.8mm}},
  tallbox/.style={draw, fill=white, minimum width=\bw cm,
              minimum height=\tallbh cm, inner sep=3pt, align=center},
  smallbox/.style={draw, fill=white, minimum width=\bw cm,
              minimum height=\smallbh cm, inner sep=3pt, align=center}
}
\tikzset{
  weakernoisefill/.style={fill=gray!20, decorate,
    decoration={random steps, segment length=2.5mm, amplitude=0.8mm}},
  tallbox/.style={draw, fill=white, minimum width=\bw cm,
              minimum height=\tallbh cm, inner sep=3pt, align=center},
  smallbox/.style={draw, fill=white, minimum width=\bw cm,
              minimum height=\smallbh cm, inner sep=3pt, align=center}
}
\tikzset{
  weakestnoisefill/.style={fill=gray!10, decorate,
    decoration={random steps, segment length=2.5mm, amplitude=0.8mm}},
  tallbox/.style={draw, fill=white, minimum width=\bw cm,
              minimum height=\tallbh cm, inner sep=3pt, align=center},
  smallbox/.style={draw, fill=white, minimum width=\bw cm,
              minimum height=\smallbh cm, inner sep=3pt, align=center}
}

\newcommand{\wires}[5]{%
  \foreach \ii in {1,...,#1}{
    \pgfmathsetmacro{\wy}{#4 + (\ii - (#1+1)/2)*#5}
    \draw (#2, \wy) -- (#3, \wy);
  }
}

\def\RowOneY{0}
\def\LAEncX{0}
\def\LACx{\LAEncX + \bw + \largegap}
\def\LADecX{\LACx  + \bw + \largegap}

\fill[noisefill]
  (\LAEncX-\bw/2-\nm, \RowOneY-\tallbh/2-\nm)
  rectangle
  (\LADecX+\bw/2+\nm, \RowOneY+\tallbh/2+\nm);
\node at ({(\LAEncX+\LADecX)/2}, \RowOneY+\tallbh/2+\nm+\capoff)
  {$\mathcal{E}^{(L)}$};

\node[tallbox] at (\LAEncX, \RowOneY)
  {$\cenc{0\rightarrow L}$};
\node[tallbox] at (\LACx,   \RowOneY) {$\mathsf{C}^{(L)}$};
\node[tallbox] at (\LADecX, \RowOneY)
  {$\cdec{L\rightarrow 0}$};

\wires{9}{\LAEncX+\bw/2}{\LACx-\bw/2}{\RowOneY}{\spreadnine}
\wires{9}{\LACx+\bw/2}{\LADecX-\bw/2}{\RowOneY}{\spreadnine}
\draw (\LAEncX-\bw/2, \RowOneY) -- (\LAEncX-\bw/2-\wl, \RowOneY);
\draw (\LADecX+\bw/2, \RowOneY) -- (\LADecX+\bw/2+\wl, \RowOneY);

\def\EqOneX{\LADecX+\bw/2+\wl+0.65}
\node[align=center] at (\EqOneX, \RowOneY)
  {$=$};

\def\ROutEncX{\EqOneX+1.7}
\def\RBEncX{\ROutEncX+\bw+\smallgap}
\def\RCx{\RBEncX+\bw+\largegap}
\def\RBDecX{\RCx+\bw+\largegap}
\def\ROutDecX{\RBDecX+\bw+\smallgap}

\fill[noisefill]
  (\ROutEncX-\bw/2-\nm, \RowOneY-\tallbh/2-\nm)
  rectangle
  (\ROutDecX+\bw/2+\nm, \RowOneY+\tallbh/2+\nm);
\node at ({(\ROutEncX+\ROutDecX)/2}, \RowOneY+\tallbh/2+\nm+\capoff)
  {$\mathcal{E}^{(L)}$};

\node[smallbox] at (\ROutEncX, \RowOneY) {$\mathsf{Enc}$};
\node[tallbox]  at (\RBEncX,   \RowOneY)
  {$\cenc{1\rightarrow L}$};
\node[tallbox]  at (\RCx,      \RowOneY) {$\mathsf{C}^{(L)}$};
\node[tallbox]  at (\RBDecX,   \RowOneY)
  {$\cdec{L \rightarrow 1}$};
\node[smallbox] at (\ROutDecX, \RowOneY) {$\mathsf{Dec}$};

\draw (\ROutEncX-\bw/2, \RowOneY) -- (\ROutEncX-\bw/2-\wl, \RowOneY);
\draw (\ROutDecX+\bw/2, \RowOneY) -- (\ROutDecX+\bw/2+\wl, \RowOneY);

\wires{3}{\ROutEncX+\bw/2}{\RBEncX-\bw/2}{\RowOneY}{\spreadthree}
\wires{9}{\RBEncX+\bw/2}{\RCx-\bw/2}{\RowOneY}{\spreadnine}
\wires{9}{\RCx+\bw/2}{\RBDecX-\bw/2}{\RowOneY}{\spreadnine}
\wires{3}{\RBDecX+\bw/2}{\ROutDecX-\bw/2}{\RowOneY}{\spreadthree}

\def\RowTwoY{-3.2}
\def\RelTwoX{\EqOneX}

\node[align=center] at (\RelTwoX, \RowTwoY)
  {$\approx$\\ \scriptsize Lemma \ref{lem:levelreductionencdec}};

\def\TEncX{\RBEncX}
\def\TCx{\RCx}
\def\TDecX{\RBDecX}

\fill[weakernoisefill]
  (\TEncX-\bw/2-\nm, \RowTwoY-\smallbh/2-\nm)
  rectangle
  (\TDecX+\bw/2+\nm, \RowTwoY+\smallbh/2+\nm);
\node at ({(\TEncX+\TDecX)/2}, \RowTwoY+\smallbh/2+\nm+\capoff)
  {$\mathcal{E}^{(L-1)}$};

\node[smallbox] at (\TEncX, \RowTwoY)
  {$\cenc{0 \rightarrow L-1}$};
\node[smallbox] at (\TCx,   \RowTwoY) {$\mathsf{C}^{(L-1)}$};
\node[smallbox] at (\TDecX, \RowTwoY)
  {$\cdec{L-1\rightarrow 0}$};

\wires{3}{\TEncX+\bw/2}{\TCx-\bw/2}{\RowTwoY}{\spreadthree}
\wires{3}{\TCx+\bw/2}{\TDecX-\bw/2}{\RowTwoY}{\spreadthree}
\draw (\TEncX-\bw/2, \RowTwoY) -- (\TEncX-\bw/2-\wl, \RowTwoY);
\draw (\TDecX+\bw/2, \RowTwoY) -- (\TDecX+\bw/2+\wl, \RowTwoY);

\def\RowThreeY{-5.8}

\node at (\RelTwoX, \RowThreeY) {$\approx\ \cdots\ \approx$};

\def\CbX{\TCx}
\def\Cbw{1.1}
\def\Cbh{1.1}
\fill[weakestnoisefill]
  (\CbX-\Cbw/2-\nm, \RowThreeY-\Cbh/2-\nm)
  rectangle
  (\CbX+\Cbw/2+\nm, \RowThreeY+\Cbh/2+\nm);
\node at (\CbX, \RowThreeY+\Cbh/2+\nm+\capoff) {$\mathcal{E}^{(0)}$};

\node[draw, fill=white, minimum width=\Cbw cm, minimum height=\Cbh cm]
  at (\CbX, \RowThreeY) {$\mathsf{C}$};

\draw (\CbX-\Cbw/2, \RowThreeY) -- (\CbX-\Cbw/2-\wl, \RowThreeY);
\draw (\CbX+\Cbw/2, \RowThreeY) -- (\CbX+\Cbw/2+\wl, \RowThreeY);

\end{tikzpicture}
    
    \caption{Proof of Lemma \ref{lem:itlevelreductionencdec}. The noisy implementation of $\cdec{L\rightarrow 0} \circ \cC^{(L)} \circ \cenc{0\rightarrow L}$ is viewed as a noisy implementation of the level-1 simulation of $\cdec{0 \rightarrow L-1} \circ \cC^{(L-1)}\circ \cenc{0 \rightarrow L-1}$, sandwiched with the ideal encoding and decoding maps. Lemma \ref{lem:levelreductionencdec} can be applied to establish that, with high probability, this has the same action as an implementation of the latter circuit at a lower noise rate. This argument can be iterated $L$ times to establish that, except with some bounded probability of error, the circuit has identical action to $\cC$ under very weak noise.}
    \label{fig:sketch iterated level reduction inout}
\end{figure}
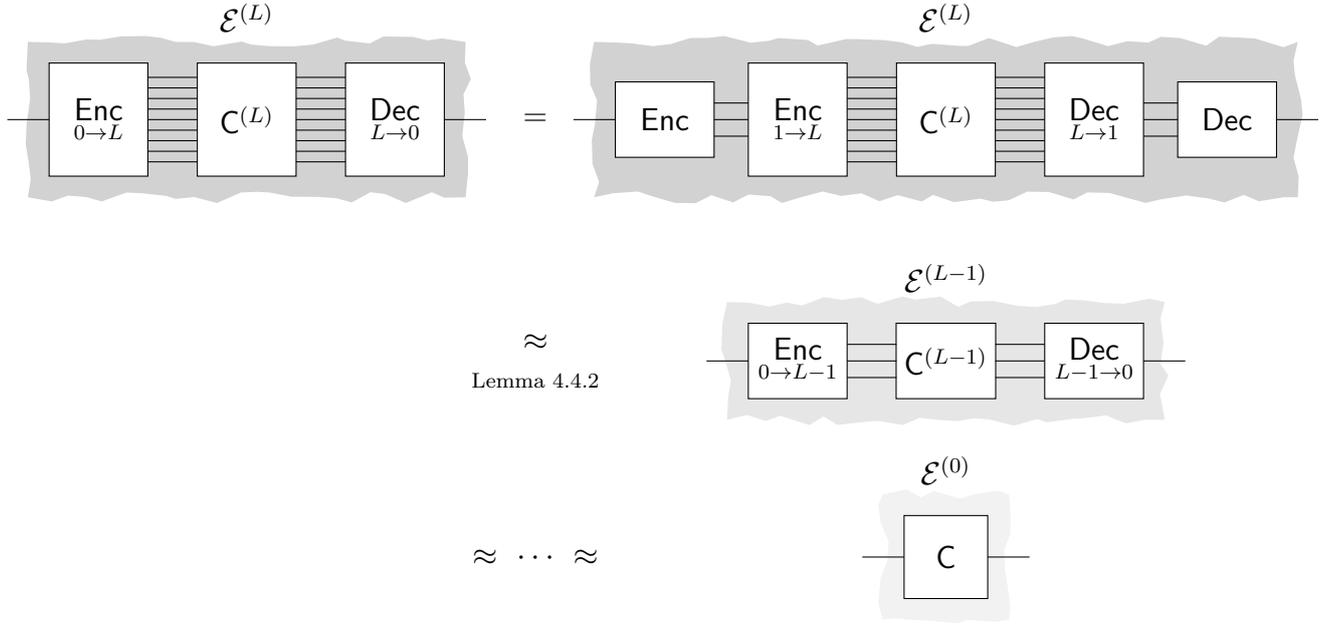

\begin{proof}
    Let $p = p_L > 0$. For notational convenience, we define the circuit $\tilde{\cC}_j$ by
    \begin{align}
        \tilde{\cC}_j := \left( \cdec{j\rightarrow 0}\right)^{\otimes N_{\mathsf{out}}} \circ \cC^{(j)} \circ \left(\cenc{0\rightarrow j} \right)^{\otimes N_{\mathsf{in}}}
    \end{align}
    for $j=0,\dots,L$. Note that
    
    \begin{align}
        \tilde{\cC}_L = \dec^{\otimes N_{\mathsf{out}}} \circ \left[\left( \cdec{L\rightarrow 1}\right)^{\otimes N_{\mathsf{out}}} \circ \cC^{(L)} \circ \left(\cenc{1\rightarrow L} \right)^{\otimes N_{\mathsf{in}}} \right] \circ \enc^{\otimes N_{\mathsf{in}}}\ .
    \end{align}
    The term in the square brackets is a level-1 simulation of $\tilde{\cC}_{L-1}$, so we may readily apply Lemma~\ref{lem:levelreductionencdec} to obtain local stochastic noise $\cE_{L-1}$ of strength $p_{L-1} = c p_L^{t+1}$ such that
    \begin{align}
        \prob\left[\cE_L \bowtie \tilde{\cC}_L \neq \cE_{L-1} \bowtie \tilde{\cC}_{L-1} \right] \leq p_L \nmax\dmax (N_{\mathsf{out}} + N_{\mathsf{in}})\ ,
    \end{align}
    for some constant $c> 0$ as given by Lemma~\ref{lem:levelreductionencdec}. We may apply this procedure inductively, at each step obtaining local stochastic noise $\cE_{j-1}$ of strength $p_{j-1} = cp_j^{t+1}$ such that
    \begin{align}
        \prob\left[\cE_j \bowtie \tilde{\cC}_j \neq \cE_{j-1} \bowtie \tilde{\cC}_{j-1} \right] \leq p_j \nmax\dmax (N_{\mathsf{out}} + N_{\mathsf{in}})\ .
    \end{align}
    Noting that $\tilde{\cC}_0 = \cC$, and applying a union bound, we can compute
    \begin{align}
        \prob\left[ \cE_L \bowtie \tilde{\cC}_L \neq \cE_0 \bowtie \cC \right] &\leq \sum_{j=1}^L \prob\left[ \cE_j \bowtie \tilde{\cC}_j \neq \cE_{j-1} \bowtie \tilde{\cC}_{j-1}\right] \\
        &\leq \sum_{j=1}^L p_j \nmax\dmax (N_{\mathsf{out}} + N_{\mathsf{in}})\ .
    \end{align}
    By the inductive definition, we have $p_j = p_\ast(p/p_\ast)^{(t+1)^{L-j}}$, where $p_\ast := c^{-1/t}$. Hence the result follows by noting that
    \begin{align}
        \sum_{j=1}^L p_j \leq p_\ast \sum_{k=0}^\infty \left(\frac{p}{p_\ast}\right)^{(t+1)^k} \leq p_\ast \left(\frac{p}{p_\ast} \right) \sum_{k=0}^\infty \left( \frac{p}{p_\ast}\right)^k \leq p\sum_{k=0}^\infty 2^{-k} = 2p\ .
    \end{align}
\end{proof}

Note that the statement of Lemma~\ref{lem:itlevelreductionencdec} is 
qualitatively different from
the results obtained in Refs.~\cite{christandl2026fault,he2025composable,belzig2026constant}, which also discuss circuits with quantum inputs and outputs. Ref.~\cite{christandl2026fault} gives an upper bound on 
the diamond-norm distance between a noisy implementation of the level-$L$ simulation of~$\mathsf{C}$ followed by decoding
and $\cF\circ \mathsf{C}$, where $\cF$ is local stochastic noise of a certain strength (depending on~$p$). In other words, the comparison is to an ideal implementation of~$\mathsf{C}$ {\em followed} by local stochastic noise. In contrast, our result compares the noisy encoded circuit $\tilde{\cC}_L$ to an imperfect implementation~$\cE_0\bowtie \mathsf{C}$
of~$\mathsf{C}$ with local stochastic noise~$\cE_0$ of doubly exponentially weaker strength. This formulation is thus closer in spirit to the idea of level reduction as in the fault tolerance construction~\cite{aliferis2005quantum,aliferis2007level}. This is also similar to the approach of \cite{he2025composable}, whose Proposition~5.2 bounds the total variation distance between the outputs of $\cE_L\bowtie\tilde{\cC}_L$ and the noiseless logical circuit $\cC$, albeit without specializing to the case of concatenated error correction.

One advantage of our formulation is that it immediately applies to any construction where a level-reduction lemma analogous to Lemma~\ref{lem:level reduction} can be established. Indeed, level reduction is our only building block and used in a black-box manner. In particular,  it applies to the case of distance-$3$ codes such as the $7$-qubit Steane code we use in our work.
We note that adapting~\cite{christandl2026fault} to this setting would require additional work: Appendix~C of~\cite{christandl2026fault} only discusses the framework of~\cite{aharonov1997fault} based on rectangles (not extended rectangles). This reasoning only applies to codes correcting at least two errors. 

\subsection{Robust implementations with concatenated codes\label{sec:robustimplementationconcatenated}}
We next show Lemma~\ref{lem:failureprobencdec}, which establishes that the failure probability of a level-$L$ circuit (including possibly noisy encoding and decoding steps) can be suppressed independently of the circuit size by increasing $L$: Using Lemma~\ref{lem:itlevelreductionencdec}, we arrive at an explicit bound on the error probability.
\begin{lemma}[Level-$L$ simulation with noisy encoding and decoding]\label{lem:failureprobencdec}
Let $t = \lfloor (d-1) / 2\rfloor$.
  Let $\nmax, \dmax$ be the gadget-dependent constants defined in Eq.~\eqref{eq:dminmaxdef}. Let $p_{\ast}>0$ be the fault tolerance error threshold strength  (see Theorem~\ref{thm:mainthresholdtheoremagp}).
Define
\begin{align}\label{eq:flprobust}
    f_L(p) := \left\{ \begin{array}{ll}
        p\cdot 2\nmax\dmax (N_{\mathsf{out}} + N_{\mathsf{in}}) + p_\ast (p / p_\ast)^{(t+1)^L} |\mathsf{Loc}(\cC)| & \text{for $p \leq p_\ast / 2$} \\
        1& \text{otherwise}
    \end{array}\right.\ .
\end{align}
Then the following holds. Let $\cC$ be a  circuit on $N$ qubits with $N_{\mathsf{in}}$-qubit input and $N_{\mathsf{out}}$-qubit output. Let $\cC^{(L)}$ denote the level-$L$ simulation of $\cC$ constructed using an error-correcting code~$\cL$ of distance~$d$ and associated gadgets, see Section~\ref{sec:FTgagdetssimulation}. 
 Then the following holds for the circuit (cf.~Lemma~\ref{lem:itlevelreductionencdec})
\begin{align}\label{eq:cftconstruction}
    \cC_{FT} = \left( \cdec{L\rightarrow 0} \right)^{\otimes N_{\mathsf{out}}} \circ \cC^{(L)} \circ \left( \cenc{0 \rightarrow L} \right)^{\otimes N_{\mathsf{in}}}\ .
\end{align}
The circuit~$\cC_{FT}$ $f_L$-robustly implements~$\cC$.  Furthermore, its number of qubits and circuit depth are bounded as
\begin{align}
\begin{matrix}
  n^LN& \leq & N(\cC_{FT})& \leq& \nmax^L N + \nmax^{L+1} (N_{\mathsf{out}} + N_{\mathsf{in}})\\
  \dmin^L\mathsf{depth}(\cC)+ 2\dmin^{L-1}&\leq&  \mathsf{depth}(\cC_{FT}) &\leq &\dmax^L\mathsf{depth}(\cC) + 2\dmax^{L+1}\ .
\end{matrix}
\end{align}
In particular, we have $N(\cC_{FT})/N=e^{\Theta(L)}$ and $\mathsf{depth}(\cC_{FT})/\mathsf{depth}(\cC)=e^{\Theta(L)}$.
\end{lemma}
\begin{proof}
To show that~$\cC_{FT}$ $f_L$-robustly implements~$\cC$, it suffices to show that 
    \begin{align}\label{eq:upperbounderrorprobabilitypzerm}
        \prob\left[ \cE_L \bowtie \cC_{FT}\neq \cC\right] &\leq p\cdot 2\nmax\dmax (N_{\mathsf{out}} + N_{\mathsf{in}}) + p_\ast (p/p_\ast)^{(t+1)^L} |\mathsf{Loc}(\cC)|\ 
\end{align}
for any local stochastic noice~$\cE_L$ on $\cC_{FT}$ of  strength $p < p_\ast/2$. This follows immediately from Lemma~\ref{lem:itlevelreductionencdec}, using that $\cE_0\bowtie \cC = \cC$ except with probability $p_0|\mathsf{Loc}(\cC)|$ by Lemma~\ref{lem:localstochasticnoiseqoutput}\eqref{it:circnoisyclaim}.

The bounds on $N(\cC_{FT})$ and $\mathsf{depth}(\cC_{FT})$  follow from  Lemma~\ref{lem:qubitdepthoverheadsimulation} and Lemma~\ref{lem:qubitdepthoverheadsimulationb}.
\end{proof}
We note that the upper bound 
on the failure probability provided by the function~$f_L(p)$  is at least of order~$O(p)$ (even for constant~$N$), i.e., no improvement is made over the physical error strength by the level-$L$ encoding.  This may suggest that the construction and result of Lemma~\ref{lem:failureprobencdec}  are useless for fault tolerance, where the general spirit is to gain accuracy. Nevertheless, they are immediately applicable to long-range entanglement generation as we argue below.

Combining Lemma~\ref{lem:failureprobencdec}  with the $[[7,1,3]]$-code
and associated fault tolerance gadgets (see e.g., Ref.~\cite{stephens2007universal} and  Section~\ref{sec:bilineararrayfaulttolerance}), we obtain the following result.
\begin{theorem}[Robust implementation of circuits against local stochastic noise]\label{thm:robustimplementationstochastic}
There is a constant~$\mathsf{g}\geq 49$ and $p_*>0$ such that the following holds. 
For $N\in\mathbb{N}$ and $\varepsilon>0$, set 
\begin{align}
p_0(N,\varepsilon)&:=\min \left\{p_*/2, \textfrac{\varepsilon}{(8\mathsf{g}N)}\right\}\ .
\end{align}
Then the following holds for any circuit~$\cC$ composed of operations belonging to a certain universal gate set. 
Let 
\begin{align}
\upperboundL &\geq |\mathsf{Loc}(\cC)|
\end{align}
be an upper bound on the number of circuit locations of~$\cC$. 
There is a circuit~$\cC_{\mathsf{FT}}$ 
which realizes~$\cC$ except with probability~$\varepsilon$ in the presence of local stochastic noise of any strength~$p<p_0(N,\varepsilon)$.
Its number of qubits and circuit depth are bounded as 
\begin{align}
N(\cC_{\mathsf{FT}}) &= N(\cC)\cdot \Theta\left(\mathsf{poly}\left(\log 1/\varepsilon, \log \upperboundL\right)\right)\label{eq:numberofqubitslocn}\\
\mathsf{depth}(\cC_{\mathsf{FT}}) &= \mathsf{depth}(\cC)\cdot\Theta\left(\mathsf{poly}\left(\log 1/\varepsilon, \log\upperboundL\right)\right)\ .\label{eq:circuitdepthlocnb}
\end{align}
\end{theorem}
\noindent It may not be clear at this point why using an upper bound~$\upperboundL$ on~$|\mathsf{Loc}(\cC)|$ is helpful. However, as we argue below (see Section~\ref{sec:2dsimulation}), the two quantities $N(\cC_{\mathsf{FT}})$ and $\mathsf{depth}(\cC_{\mathsf{FT}})$ translate into a geometry in certain settings, and making circuits larger can in fact be useful to create long-range entanglement. 

\begin{proof}
Specializing to the $[[7,1,3]]$-code, we have $t=1$.  
Since this is a $7$-qubit code, we have 
$\nmax\geq 7$, $\dmax\geq 7$ (cf. Lemma~\ref{lem:qubitdepthoverheadsimulation}), see e.g., the explicit constructions in~\cite{stephens2007universal}. We set $\mathsf{g}:=\nmax\cdot\dmax$ such that $\mathsf{g}\geq 49$.

Let us write $p_0=p_0(N,\varepsilon)$ for brevity. 
Using Eq.~\eqref{eq:upperboundfrobust1D} 
and the fact that $N_{\mathsf{out}}\leq N$ and $N_{\mathsf{in}}\leq N$, 
we conclude that
\begin{align}
f_L(p_0)&\leq \varepsilon/2+p_* |\mathsf{Loc}(\cC)|\cdot \left(1/2\right)^{2^L}\\
&\leq \varepsilon/2+p_* \upperboundL\cdot \left(1/2\right)^{2^L}
\end{align}
In particular, we have 
\begin{align}
f_L(p_0)&\leq\varepsilon\qquad\textrm{ whenever }\qquad 
\log\left(\log (2p_*/\varepsilon)+\log \upperboundL\right)\leq L\ ,
\end{align}
where $\log$ denotes the logarithm to base~$2$. Since $p_*<1/2$, this is satisfied if 
\begin{align} 
\log \left(\log 1/\varepsilon+\log \upperboundL\right)&\leq L\ , 
\end{align}
i.e., for the choice
\begin{align}
L_0:=\lceil\log \left(\log 1/\varepsilon+\log \upperboundL\right) \rceil
\end{align}
The claim now follows from Lemma~\ref{lem:failureprobencdec}, the monotonicity of the function $f_L$, and the fact that   for any constant $g>0$, we have 
\begin{align}
g^{L_0}&=\left(\log 1/\varepsilon+\log\upperboundL\right)^{\log g}\\
&=\Theta\left(\mathsf{poly}\left(\log 1/\varepsilon, \log \upperboundL\right)\right)
\end{align}
where we used that $g^{\log x}=x^{\log g}$ for $x>0$ and $g\geq 1$. 
\end{proof}

\section{Pauli errors and Clifford circuits }\label{sec:errorpropagation}
Here we analyze the effect of Pauli errors on general Clifford circuits (including mid-circuit measurements and corresponding corrections). 
The main goal is to argue that  the fault tolerance construction for the bilinear qubit array of
Ref.~\cite{stephens2007universal} gives rise to a $1D$-local quantum fault tolerance scheme for Clifford circuits, which is robust to local stochastic Pauli noise, see Section~\ref{ft:ftimplementationprepmeasure}.

We proceed in the following steps: Section~\ref{sec:adaptivevsnonadaptivepauli} discusses adaptivity  in Clifford circuits, and reviews (standard) arguments showing that adaptive Pauli corrections can be postponed to the end of the circuit. 
We also discuss how space-time regions of the circuit can be ``cleaned'' of Pauli noise by applying corresponding commutation relations. This is the basis for the circuit transformation rules we discuss in Section~\ref{sec:noisepreservingtransformations}: Here we show that certain modifications of a circuit are compatible with local stochastic Pauli noise. In particular, in Section~\ref{sec:circuitlocalitytransformations}, we show how to relate circuits associated with different architectures (locality constraints) while keeping control of the noise. This is the basis for the result on $1D$-local fault tolerance established in Section~\ref{ft:ftimplementationprepmeasure}.

\subsection{Adaptive versus  non-adaptive Clifford circuits and Pauli noise\label{sec:adaptivevsnonadaptivepauli}}
Throughout, we consider Clifford circuits composed of the following operations:
\begin{enumerate}[(i)]
\item preparation of a single qubit in the computational basis state $|0\rangle$.\label{it:prepcliff}
\item single- and two-qubit Clifford operations (including the identity, whose inclusion is convenient for formal purposes). Specifically, we use the single-qubit Pauli gates $\{X,Y,Z\}$, the single qubit Hadamard gate~$H$ and the phase gate~$S=\mathsf{diag}(1,i)$, as well as the two-qubit  $\CNOT$ gate.\label{it:unitarycliff}
\item\label{it:measurementcliff}
measurement of any qubit in the computational basis. This provides a measurement result $x \in\{0,1\}$ (which we assume is recorded).
\item\label{it:adaptivecliff}
classically controlled Pauli operators (of any weight) applied to any subset of qubits.
\end{enumerate}
Let us denote all operations defined by~\eqref{it:prepcliff}--\eqref{it:measurementcliff} as well as the identity operations associated with all subsets of qubits by~$\mathsf{Cliff}_{\mathsf{na}}$.
The set of all operations~\eqref{it:prepcliff}--\eqref{it:adaptivecliff} (i.e., including adaptive Pauli operations) will be denoted by~$\mathsf{Cliff}$.

A widely used statement is the fact that an adaptive Clifford circuit~$\mathsf{C}$  (i.e., composed of operations belonging to~$\mathsf{Cliff}$) is equivalent to a non-adaptive Clifford circuit~$\tilde{\mathsf{C}}$ (i.e., composed only of operations belonging to~$\mathsf{Cliff}_{\mathsf{na}}$) up to a Pauli correction which is efficiently computable from the measurement results.  The following lemma formalizes a folklore result stating that this transformation is robust to Pauli noise. We note that the assumption that the noise is of Pauli type is essential for its  proof. 

Without loss of generality, we consider a circuit whose last layer of operations consists in (possibly adaptive) Pauli operations. 

\begin{lemma}[Adaptive and non-adaptive Clifford circuits]\label{lem:adaptivenonadaptiveClifford}
Let $\mathsf{C}=\cC_D\circ \cdots \circ\cC_1$ be a depth-$D$ circuit composed of~$\mathsf{Cliff}$.
Suppose that the last layer~$C_D$ only consists of single-qubit Paulis and adaptive operations, i.e., operations belonging to $\mathsf{Cliff}\backslash \mathsf{Cliff}_{\mathsf{na}}$. 
Then there is a depth-$D$ circuit $\tilde{\mathsf{C}}=\tilde{\mathsf{C}}_D\circ\cdots\circ\tilde{\mathsf{\mathsf{C}}}_1$
composed of~$\mathsf{Cliff}$ 
which (has an efficient description that) is efficiently computable from (an efficient description of) $\mathsf{C}$ and has  the following properties:
\begin{enumerate}[(i)]
\item for each $t\in [D-1]$, the  layer~$\tilde{\mathsf{C}}_t$ is composed of~$\mathsf{Cliff}_{\mathsf{na}}$, i.e., is non-adpative.
\item
the circuit~$\mathsf{\tilde{C}}$ has the same interaction graph as~$\mathsf{C}$,
\item
For any Pauli error~$\tilde{F}$ on~$\tilde{\mathsf{C}}$, there is a Pauli error $F$~on~$\mathsf{C}$ with the following properties:
We have 
\begin{align}
F\bowtie\mathsf{C}&=\tilde{F}\bowtie\tilde{\mathsf{C}}\label{eq:paulipropagatedidentical}
\end{align}
and 
\begin{align}
\supp(F)\subseteq \supp(\tilde{F})\ .
\end{align}
\end{enumerate}
\end{lemma}
Eq.~\eqref{eq:paulipropagatedidentical} implies that~$\tilde{\mathsf{C}}$ and $\mathsf{C}$ implement the identical functionality in the absence of noise, and a noisy execution  of~$\tilde{\mathsf{C}}$ followed by 
 an efficiently computable Pauli correction produces identical output as a noisy execution of the original circuit~$\mathsf{C}$ with related Pauli noise.
\begin{proof}
For a depth-$D$ circuit~$\mathsf{C}=\cO_D\circ\cdots\circ \cO_1$, let $N_{\leq D-1}(\mathsf{C})$ denote the number of adaptive Pauli operations in the operation layers~$\cO_1,\ldots,\cO_{D-1}$. 

We iteratively construct a sequence~$\{\tilde{\mathsf{C}}^{(j)}\}_{j=0}^{K-1}$ of intermediate depth-$D$ circuits such that  the last circuit~$\tilde{\mathsf{C}}:=\tilde{\mathsf{C}}^{(K-1)}$ has the required properties. We  set $\tilde{\mathsf{C}}^{(0)}:=\mathsf{C}$ and show  that  for each $j\in \{0,\ldots,K-1\}$, the following holds:
\begin{enumerate}[(i)]
\item\label{it:firstpropertytoprovea}
Compared to~$\tilde{\mathsf{C}}^{(j)}$, the circuit~$\tilde{\mathsf{C}}^{(j+1)}$
either has $N_{\leq D-1}(\tilde{\mathsf{C}}^{(j+1)})=N_{\leq D-1}(\tilde{\mathsf{C}}^{(j)})-1$
or has one adaptive Pauli correction ``moved to the right'', i.e., moved to a later layer.
\item\label{it:firstpropertytoproveb}
The circuit~$\tilde{\mathsf{C}}^{(j+1)}$ has the same interaction graph as $\tilde{\mathsf{C}}^{(j)}$
\item\label{it:firstpropertytoprovec}
For any Pauli error~$\tilde{F}$ on~$\tilde{\mathsf{C}}^{(j+1)}$, there is a
Pauli error~$F$ on~$\tilde{\mathsf{C}}^{(j)}$ 
such that
$F\bowtie \tilde{\mathsf{C}}^{(j)}=\tilde{F}\bowtie \tilde{\mathsf{C}}^{(j+1)}$.
\end{enumerate}
Note that
 if $N_{\leq D-1}(\mathsf{C}^{(j)})=0$, we are done.
 
Assuming that  $N_{\leq D-1}(\tilde{\mathsf{C}}^{(j)})>0$, we construct~$\tilde{\mathsf{C}}^{(j+1)}$ as follows from the circuit~$\tilde{\mathsf{C}}^{(j)}=\cO_D^{(j)}\circ\cdots\circ \cO_1^{(j)}$.
 Let $t\in \{1,\ldots,D-1\}$ be the maximal time before the last time step~$D$ (operation layer) such that the layer~$\cO_t^{(j)}$ contains an adaptive Pauli correction. Fix one such adaptive Pauli correction.
We define~$\tilde{\mathsf{C}}^{(j+1)}$ by commuting the adaptive Pauli correction ``past'' the operation in layer~$\cO_{t+1}^{(j)}$ following it according to the rules specified in Fig.~\ref{fig:all-transformations}. These are chosen in such a way as to satisfy properties~\eqref{it:firstpropertytoprovea}--\eqref{it:firstpropertytoprovec}. 

\begin{figure}
\centering
\includegraphics{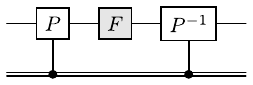} \qquad $=$\qquad\includegraphics{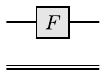}
\caption{A key identity: Conjugating a Pauli~$F$ by a classically controlled Pauli~$\mathsf{ctrl}P$ results in~$F$ up to (possibly) a sign (if $F$ and $P$ anticommute), but the sign can be ignored as the controlling register is classical.~\label{fig:controlPidentity}}
\end{figure}

\begin{figure}
\centering
\renewcommand{\thesubfigure}{a\arabic{subfigure}}
\setcounter{subfigure}{0}
\begin{minipage}{\textwidth}
\centering
\begin{subfigure}{0.45\textwidth}
\centering
\includegraphics{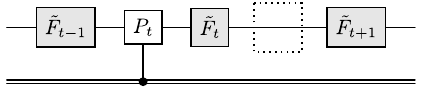}
\caption{Circuit~$\tilde{\mathsf{C}}^{(j)}$ with  noise~$(F_{t-1},F_t,F_{t+1}) = (\tilde{F}_{t-1},\tilde{F}_t,\tilde{F}_{t+1})$}
\end{subfigure}\hfill
\begin{subfigure}{0.45\textwidth}
\centering
\includegraphics{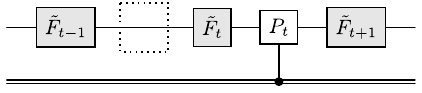}
\caption{Definition of the circuit~$\tilde{\mathsf{C}}^{(j+1)}$.}
\end{subfigure}
\renewcommand{\thesubfigure}{\alph{subfigure}}
\setcounter{subfigure}{0}
\subcaption{Wait location: 
Here the circuit~$\tilde{\mathsf{C}}^{(j+1)}$ is obtained by commuting the adaptive Pauli forward.
The fact that this has the same action even with errors is a consequence of the identity shown in Fig.~\ref{fig:controlPidentity}.}
\label{fig:waitcliff}
\end{minipage}

\vspace{1em}

\renewcommand{\thesubfigure}{b\arabic{subfigure}}
\setcounter{subfigure}{0}
\begin{minipage}{\textwidth}
\centering
\begin{subfigure}{0.45\textwidth}
\centering
\includegraphics{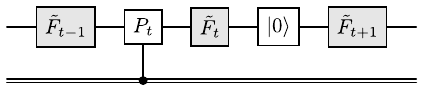}
\caption{Circuit~$\tilde{\mathsf{C}}^{(j)}$ with  noise~$(F_{t-1},F_t,F_{t+1}) = (\tilde{F}_{t-1},\tilde{F}_t,\tilde{F}_{t+1})$}
\end{subfigure}\hfill
\begin{subfigure}{0.45\textwidth}
\centering
\includegraphics{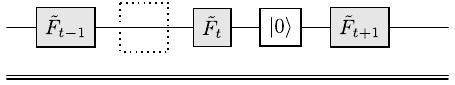}
\caption{Definition of the circuit~$\tilde{\mathsf{C}}^{(j+1)}$.}
\end{subfigure}
\renewcommand{\thesubfigure}{\alph{subfigure}}
\setcounter{subfigure}{1}
\subcaption{Preparation: 
Here the circuit~$\tilde{\mathsf{C}}^{(j+1)}$ is obtained by omitting the controlled Pauli in layer~$t$. This gives the same action as the controlled Pauli has no effect.}
\label{fig:preparationcliff}
\end{minipage}

\vspace{1em}

\renewcommand{\thesubfigure}{c\arabic{subfigure}}
\setcounter{subfigure}{0}
\begin{minipage}{\textwidth}
\centering
\begin{subfigure}{0.45\textwidth}
\centering
\includegraphics{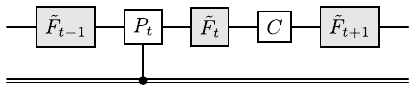}
\caption{Circuit~$\tilde{\mathsf{C}}^{(j)}$ with  noise~$(F_{t-1},F_t,F_{t+1}) = (\tilde{F}_{t-1},\tilde{F}_t,\tilde{F}_{t+1})$}
\end{subfigure}\hfill
\begin{subfigure}{0.45\textwidth}
\centering
\includegraphics{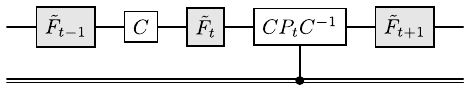}
\caption{Definition of the circuit~$\tilde{\mathsf{C}}^{(j+1)}$.}
\end{subfigure}
\renewcommand{\thesubfigure}{\alph{subfigure}}
\setcounter{subfigure}{2}
\subcaption{Clifford: 
Here the circuit~$\tilde{\mathsf{C}}^{(j+1)}$ is obtained by 
commuting the controlled-Pauli forward through the circuit. We note that $P':=CP_tC^{-1}$ is
a one- or two-qubit Pauli in general (we only illustrate the case of one qubit). The fact that these have the identical action under the (specified) Pauli noise is again a consequence of the identity shown in Fig.~\ref{fig:controlPidentity}
applied to $E_t$ and $P'$.}
\label{fig:cliff}
\end{minipage}

\vspace{1em}

\renewcommand{\thesubfigure}{d\arabic{subfigure}}
\setcounter{subfigure}{0}
\begin{minipage}{\textwidth}
\centering
\begin{subfigure}{0.45\textwidth}
\centering
\includegraphics{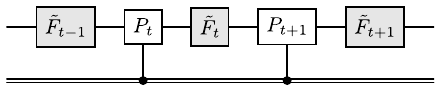}
\caption{Circuit~$\tilde{\mathsf{C}}^{(j)}$ with  noise~$(F_{t-1},F_t,F_{t+1}) = (\tilde{F}_{t-1},\tilde{F}_t,\tilde{F}_{t+1})$}
\end{subfigure}\hfill
\begin{subfigure}{0.45\textwidth}
\centering
\includegraphics{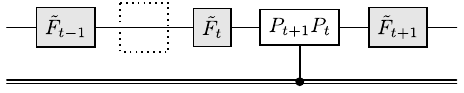}
\caption{Definition of the circuit~$\tilde{\mathsf{C}}^{(j+1)}$.}
\end{subfigure}
\renewcommand{\thesubfigure}{\alph{subfigure}}
\setcounter{subfigure}{3}
\subcaption{Adaptive Pauli: 
Here the circuit~$\tilde{\mathsf{C}}^{(j+1)}$ is obtained by combining the adaptive Paulis into one. The fact that these have the identical action under the (specified) Pauli noise is again a consequence of the identity shown in Fig.~\ref{fig:controlPidentity}.}
\label{fig:adaptivepauli}
\end{minipage}

\vspace{1em}

\renewcommand{\thesubfigure}{e\arabic{subfigure}}
\setcounter{subfigure}{0}
\begin{minipage}{\textwidth}
\centering
\begin{subfigure}{0.45\textwidth}
\centering
\includegraphics{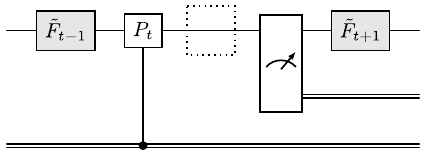}
\caption{Circuit~$\tilde{\mathsf{C}}^{(j)}$ with  noise~$(F_{t-1},F_t,F_{t+1}) = (\tilde{F}_{t-1},I,\tilde{F}_{t+1})$}
\end{subfigure}\hfill
\begin{subfigure}{0.45\textwidth}
\centering
\includegraphics{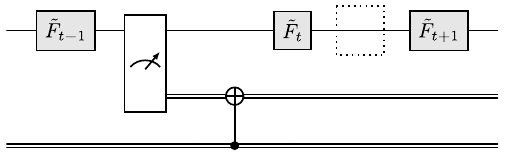}
\caption{Definition of the circuit~$\tilde{\mathsf{C}}^{(j+1)}$.}
\end{subfigure}
\renewcommand{\thesubfigure}{\alph{subfigure}}
\setcounter{subfigure}{4}
\subcaption{Measurement: 
Here the circuit~$\tilde{\mathsf{C}}^{(j+1)}$ is obtained by 
commuting the controlled-Pauli forward through the circuit. 
The CNOT-like symbol signifies that the definition of all subsequent classically controlled Paulis has to be updated by replacing the measurement result~$x_t$ by $1-x_t$ if the control is such that $P_t\in \{X,Y\}$.}
\label{fig:measurementcliff}
\end{minipage}

\caption{Circuit transformations for propagating adaptive Paulis forward past a subsequent operation. The circuit~$\tilde{\mathsf{C}}^{(j+1)}$ (right column)
with Pauli errors~$(\tilde{F}_{t-1},\tilde{F}_t,\tilde{F}_{t+1})$ produces the same output as
the circuit~$\tilde{\mathsf{C}}^{(j)}$ with the ``derived'' Pauli noise~$(F_{t-1},F_t,F_{t+1})$
 (left column)}
\label{fig:all-transformations}
\end{figure}
\end{proof}

In the following, we consider propagation of noise through circuits. 
Consider first the case of a unitary circuit. Here we have the following (again following~\cite[Lemma 11]{bravyi2020quantum}).

\begin{lemma}[Pauli error propagation through a Clifford circuit]\label{lem:unitarycliffordcircuit}
Let $\mathsf{C}$ be a depth-$D$ Clifford circuit consisting only of one- and two-qubit Clifford unitaries. 
Let $F\sim \cN^{\pauli}_{\mathsf{C}}(p)$ be local stochastic Pauli noise on~$\mathsf{C}$. Then there is local  stochastic Pauli noise~$F'\sim \cN^{\pauli}_{\mathsf{C}}(p)$
of strength 
\begin{align}
p'= D\cdot \left(2p^{2^{-(D-1)}}\right)^{1/D}\ .\label{eq:pddeltadefinition}
\end{align} with the following properties:
\begin{enumerate}[(i)]
\item\label{it:supportpropagatedtoend}
$\supp(F')\subseteq \{D\}\times [n]$. 
\item
$F\bowtie \mathsf{C}=F'\bowtie\mathsf{C}$. 
\end{enumerate}
\end{lemma}
Statement~\eqref{it:supportpropagatedtoend} means that the error~$F'$ only acts at the end of the circuit. It is the result of propagating~$F$ to the end of the circuit.
\begin{proof}
Let us write $\mathsf{C}=\cC_D\circ\cdots \circ\cC_1$ where each operation layer~$\cC_t$ is a tensor product of one- and two-qubit Clifford operations.
We similarly factor the stochastic Pauli error as $F=F_D\cdot \cdots F_1$ where for each $t\in [D]$, the local stochastic Pauli error~$F_t$ only includes the errors acting after the layer~$t$, i.e., 
\begin{align}
\supp(F_t)=\supp(F)\cap \left(\{t\}\times [n]\right)\qquad\textrm{ and }\qquad F_t(w)=F(w)\textrm{ for every }w\in \supp(F_t)\ .
\end{align}
By Lemma~\ref{lem:pauli noise properties}\eqref{item:pauli noise property 1}, each $F_t$ is local stochastic with strength~$p$. Let us think of~$F_t$ as a (random) $n$-qubit Pauli~$\tilde{F}_t$ (acting at time~$t$) instead of a ``space-time'' Pauli.
Then we can define an $n$-qubit Pauli
\begin{align}
F' = F_D' F_{D-1}' \dots F_1'\label{eq:ftildedef}
\end{align}
where
\begin{align}
F_1'&=(\cC_2\cdots \cC_D)^\dagger F_1(\cC_2\cdots \cC_D)\\
F_2'&=(\cC_3\cdots \cC_D)^\dagger F_2(\cC_3\cdots \cC_D)\\
\vdots\nonumber\\
F_D' &= F_D
\end{align}
are the result of commuting each of these Paulis to the right.
Then, identifying $F'$ with the Pauli noise which acts as the unitary $F'$ at the last layer, it is clear that
\begin{align}
F\bowtie\cC&= F'\bowtie\cC\ .
\end{align}
It remains to show that $F'$ is local stochastic  with parameter as stated. To do so, we show the following.
\begin{claim}\label{claim:paulinoiseprop} Define the function $g(x):=\sqrt{2x}$. Then the following holds. The 
error $F_{D-k}'$ (interpreted as a spacetime Pauli with support after layer $D-k$)
is local stochastic of strength $g^{\circ k}(p)$. 
\end{claim}
\begin{proof}
In~\cite[Lemma 11]{bravyi2020quantum}, it is shown that conjugating a strength~$p$ local stochastic Pauli by a depth-$1$ Clifford circuit composed of one- and two-qubit Cliffords results in a local stochastic Pauli  of strength~$g(p)$.

This immediately implies the induction base case $k=1$: Since $F_{D-1}$ is a local stochastic Pauli error of strength~$p$, it follows that $F_{D-1}'$ is a local stochastic Pauli error of strength~$g(p)$.

For $k>1$, observe that $F'_{D-k}$ is obtained by commuting~$F_{D-k}$ through $k$ depth-$1$ circuits. This implies the claim by $k$-fold application of~\cite[Lemma 11]{bravyi2020quantum}. 
\end{proof}

By Eq.~\eqref{eq:ftildedef}, we see that $F'=F'_D F_{D-1}' \dots F_1'$
is a product of $D$ local stochastic Pauli errors (all acting after the last layer) with error strengths~$p$ (for $F_D'$) and $g^{\circ k}(p)$ with $k\in [D-1]$. 
By Lemma~\ref{lem:pauli noise properties}~\eqref{item:pauli noise property 3}, it follows that $F'$ is local stochastic with strength
\begin{align}
p'&=D\cdot \left(\max\left\{p,\max_{k\in [D-1]}g^{\circ k}(p)\right\}\right)^{1/D}\ .
\end{align}
This implies the claim as $g^{\circ k}(p)=2^{1-1/2^k}p^{1/2^k}\leq 2p^{1/2^k}$.
\end{proof}

In the following, we use a relatively straightforward generalization of Lemma~\ref{lem:unitarycliffordcircuit}
which considers subcircuits of~$\mathsf{C}$. We show that local stochastic Pauli noise can essentially be cleaned out of these subcircuits. This is a useful proof tool we use in the following.

\begin{figure}
\centering
\includegraphics[width=9cm]{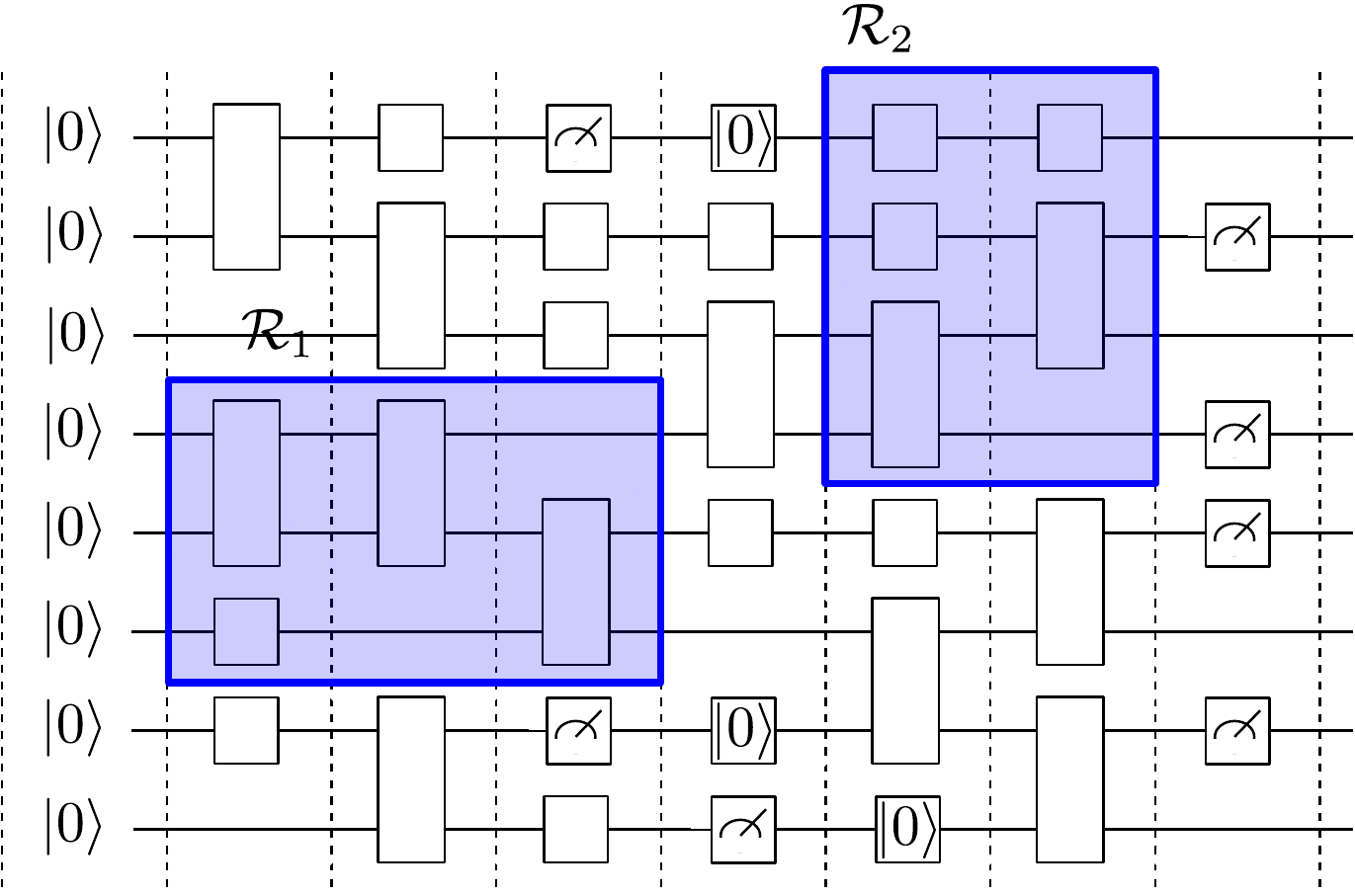}
\caption{Two non-intersecting valid rectangles~$\cR_1$, $\cR_2$ associated with unitary subcircuits of a circuit~$\cC$ \label{fig:rectangles}}
\end{figure}

To formally state this property, consider a depth-$D$ Clifford circuit~$\cC=\cC_D\circ\cdots \circ \cC_1$ with set of wires~$\cW_{\cC}=[D]\times [n]$. We call a subset~$\cR=\cR(\Omega,t,\Delta)\subseteq \cW_{\cC}$ of the form
\begin{align}
\cR=\{t,t+1,\ldots,t+\Delta\}\times \Omega\ 
\end{align}
for some subset~$\Omega\subset [n]$ of qubits and integers~$t,\Delta\in [D]$ a rectangle, see Fig.~\ref{fig:rectangles}. 
The interior $\mathsf{int}(\cR)$ of~$\cR$ is defined as the set of wires
\begin{align}
\mathsf{int}(\cR)=\{t+1,t+2,\ldots,t+\Delta-1\}\times\Omega\ .
\end{align}
We also define the left- and right-boundaries of~$\cR$ as 
\begin{align}
\partial_\ell \cR&=\{t\}\times \Omega\\
\partial_r \cR&=\{t+\Delta\}\times \Omega\ .
\end{align}
The parameter~$\Delta$ is called the depth of~$\cR$. A rectangle~$\cR$ is called valid for~$\cC$ if every operation in the layers~$\cC_{t+1},\ldots,\cC_{t+\Delta}$ of $\mathsf{C}$ has all input and output wires belonging to~$\cR$ (i.e., there are no two-qubit operations crossing the boundary of~$\cR$). A valid rectangle defines a subcircuit~$\cC^{\cR}$ of $\cC$ of $\cR$ acting only on the qubits~$\Omega$ and having depth~$\Delta$. 
We call the subcircuit~$\cC^{\cR}$ unitary if it only contains one- and two-qubit unitaries. We then have the following.
\begin{lemma}[Circuit error cleaning for unitary subcircuits]\label{lem:cleaninglemma}
Let $\cC$ be a depth-$D$ Clifford circuit.
Let $\cR_1,\ldots,\cR_m$ be pairwise disjoint valid rectangles for~$\cC$ such that every 
subcircuit~$\cC^{\cR_j}$, $j\in [m]$ is unitary. 
Let $\Delta:=\max_{j\in [m]}\Delta(\cR_j)$ be the maximal depth of a rectangle. 
Let $F\sim\cN^{\pauli}_\cC(p)$. Then 
there is local stochastic Pauli noise~$F^\prime\sim \cN^{\pauli}_{\cC}(p')$ of strength
\begin{align}
p'&= 2 \Delta^{1/2}\cdot \left(2p^{2^{-(\Delta-1)}}\right)^{1/2 \Delta}\label{eq:pprimestrengthmv}
\end{align}
 with the following properties: We have 
\begin{align}
F\bowtie\cC&=F'\bowtie\cC\label{eq:fbowtieceqa}
\end{align}
and 
\begin{align}
\supp(F')\cap \left(\partial_\ell\cR_j\cap \mathsf{int}(\cR_j)\right)&=\emptyset\ \textrm{ for each }j\in [m]\ \label{eq:disjointfprimercj}
\end{align}
with certainty.
\end{lemma}
\begin{proof}
Let $\mathsf{int}=\bigcup_{j\in [m]}\mathsf{int}(\cR_j)$ denote the union of the interiors of all rectangles, and let $\mathsf{int}^c:=\cW_\cC\backslash \mathsf{int}$ denote the complement.
We factor the noise~$F$ as 
\begin{align}
F&=F_{\mathsf{int}}\cdot F_{\mathsf{int}^c}\ 
\end{align}
where the factors are the restriction of $F$ to these two sets. 
We note that $F_{\mathsf{int}}$ is a union of components acting inside each rectangle~$\cR_j$. 
 By reasoning analogous to Lemma~\ref{lem:unitarycliffordcircuit}, all components of the error~$F_{\mathsf{int}} $
 can simultaneously be propagated ``outside'' of each rectangle, resulting in an error~$F'_{\partial_r}$
 supported  on the union $\partial_r:=\bigcup_{j\in [m]}\cR_j$ of right boundaries of each rectangle~$\cR_j$, 
 with error strength
 \begin{align}
p''= \Delta\cdot \left(2p^{2^{-(\Delta-1)}}\right)^{1/\Delta}\ .\label{eq:Deltapprimeupperbound}
\end{align}
see Eq.~\eqref{eq:pddeltadefinition}. (Here we used that the quantity in Eq.~\eqref{eq:Deltapprimeupperbound} is monotonically increasing in~$\Delta$, hence rectangles of smaller depth are also covered by this bound.)
Since $F_{\mathsf{int}^c}$ is a local stochastic Pauli error of strength~$p$ by Lemma~\ref{lem:pauli noise properties}~\eqref{item:pauli noise property 1} 
and $p< p''$, it follows from Lemma~\ref{lem:pauli noise properties}~\eqref{item:pauli noise property 3} 
that the error
\begin{align}
F'&=F_{\mathsf{int}^c}\cdot F'_{\partial_r}
\end{align}
is local stochastic of strength~$p'$ as given in Eq.~\eqref{eq:pprimestrengthmv}. Clearly, Eq.~\eqref{eq:fbowtieceqa} is satisfied by definition of~$F'$, hence the claim follows. The measurement result~$z$ from the qubit(s) $\Omega_1$ classically control a Pauli~$P(z)$ on the qubit(s)~$\Omega_2$. 
\end{proof}

\begin{figure}
\centering
\includegraphics[width=7cm]{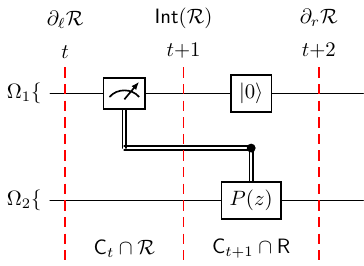}
\caption{An adaptive rectangle~$\cR$\label{fig:adaptiverectangle}}
\end{figure}
We also need a cleaning lemma for certain  rectangles 
containing adaptive Paulis. Specifically, let us define an adaptive rectangle~$\cR=\cR(\Omega,t,\Delta)\subset \cW_{\cC}$ as a rectangle with the following properties, see Fig.~\ref{fig:adaptiverectangle} for an illustration:
\begin{enumerate}
\item
The depth of~$\cR$ is $\Delta=2$.
\item
There is a partition $\Omega=\Omega_{1}\cup \Omega_2$ of the set $\Omega$ of qubits belonging to~$\cR$ such that 
the following holds: The depth-$1$ subcircuit~$\cC_t\cap \Omega_1$ consists of single-qubit measurements on all qubits belonging to~$\Omega_1$, yielding a measurement result $z\in \{0,1\}^{\Omega_1}$. Furthermore, the depth-$1$ subcircuit $\cC_{t+1}\cap \Omega_2$ consists in adaptive Pauli correction~$P(z)\in \pauli^{\Omega_2}$ applied to the qubits belonging to~$\Omega_2$ which depends on this measurement result, and  state preparation of the state~$\ket{0}$ on every qubit belonging to~$\Omega_1$.
\end{enumerate}
Define the ``left'' or ``in''-boundary of an adaptive rectangle~$\cR$ as 
\begin{align}
\partial_{\ell}\cR:=\{ (t,q)\ |\ (t,q)\in \cR \}\ 
\end{align}
and similarly the ``right'' or ``out''-boundary as 
\begin{align}
\partial_{r}\cR:=\{ (t+2,q)\ |\ (t+2,q)\in \cR \}\ .
\end{align}
The interior of~$\cR$ consists of the wires
\begin{align}
\mathsf{Int}(\cR):=\{(t+1,q)\ |\ (t+1,q)\in\cR\}\ .
\end{align}

We say that the adaptive rectangle~$\cR$ is ``read-once'' in the circuit~$\cC$ if the measurement results~$z\in \{0,1\}^{\Omega_1}$ are not used elsewhere in the execution of~$\cC$ (i.e., if they only determine the Pauli~$P(z)$ in the layer~$\cC_{t+1}$). 

Let us call the adaptive Pauli correction~$P(z)$ linear if there are linear maps~$A,B:\mathbb{F}_2^{\Omega_1}\rightarrow\mathbb{F}_2^{\Omega_2}$ such that 
\begin{align}
P(z)&=X(Az) Z(Bz)\qquad\textrm{ for all }\qquad z\in \mathbb{F}_2^{\Omega_1}\ ,
\end{align}
where we write~$X(y)=\prod_{q\in \Omega_2}X_q^{y_q}$ for $y\in \mathbb{F}_2^{\Omega_1}$ and similarly define $Z(y)$. We call an adaptive rectangle~$\cR$ linear if the corresponding Pauli correction is linear.
\begin{lemma}[Circuit error cleaning for adaptive read-once rectangles]\label{lem:cleaninglemmaadaptive}
Let $\cC$ be a depth-$D$ Clifford circuit.
Let $\cR_1,\ldots,\cR_m$ be pairwise disjoint adaptive  rectangles for~$\cC$ where each $\cR_j$, $j\in [m]$ is linear and read-once in~$\cC$.  Let 
\begin{align}
\begin{matrix}
U&:=&\max_{j\in [m]}|\Omega_1(\cR_j)|\\
V&:=&\max_{j\in [m]}|\Omega_2(\cR_j)|
\end{matrix}\label{eq:UVdefinition}
\end{align}
be the maximal number of single-qubit measurement performed in a rectangle, and the maximal number of qubits on which a Pauli is performed.

Let $F\sim\cN^{\pauli}_\cC(p)$. Then 
there is local stochastic Pauli noise of strength
\begin{align}
p'&=5(Up)^{1/(5V)}\label{eq:pprimeadaptivepropagated }
\end{align}
on~$\cC$ with the following properties: We have 
\begin{align}
F\bowtie\cC&=F'\bowtie\cC\label{eq:fbowtieceqaadaptive}
\end{align}
and 
\begin{align}
\supp(F')\cap\left(\mathsf{Int}(\cR_j)\cap\partial_{\ell}\cR_j\right)=\emptyset\ \textrm{ for each }j\in [m]\ \label{eq:adapmvadaptive}
\end{align}
with certainty.
\end{lemma}
\begin{proof}
The proof is similar to the proof of~\cite[Lemma A.2]{choe2025fault}. Consider a fixed adaptive, linear read-once rectangle~$\cR_{j}$ with $\Omega_j=\Omega_{j,1}\cup \Omega_{j,2}$. 
Let~$F\in \pauli^{\cW_\cC}$ be arbitrary, and let 
 $F_j=F^\ell_j\cdot F^{\mathsf{Int}}_j\cdot F^r_j$  denote the restriction of~$F$ to~$\cR_j$, where  
 \begin{align}
 F^\ell_{j}&=F|_{\partial_\ell \cR_j}\\
 F^{\mathsf{Int}}_{j}&=F|_{\mathsf{Int}(\mathsf{\cR_j})}\\
 F^r_j&=F|_{\partial_r \cR_j}
 \end{align} are the restrictions of $F$ to the left (input),
  interior (middle) and right (output) wires of~$\cR_j$, respectively. 
Each of these can further be partitioned according to the partition~$\Omega_j$, i.e.,
\begin{align}
F^{\ell}_{j}&=F^{\ell}_{j,1}\otimes F^{\ell}_{j,2}\\
 F^{\mathsf{Int}}_{j}&= F^{\mathsf{Int}}_{j,1}\otimes  F^{\mathsf{Int}}_{j,2}\\
F^{r}_{j}&=F^{r}_{j,1}\otimes F^{r}_{j,2}
\end{align}
where $F^{\ell}_{j,1}$ acts on the qubits~$\Omega_{j,1}$ before they are measured, whereas $F^{\mathsf{Int}}_{j,1}$ and $F^{r}_{j,1}$ act on these qubits after the measurement. Similarly,  $F^{\ell}_{j,2}$ and $F^{\mathsf{Int}}_{j,2}$ are the errors which act on the remaining qubits (belonging to~$\Omega_{j,2}$) either before the Pauli operator has been applied, and $F^{r}_{j,2}$ are the errors applied after the Pauli. 
We define an associated error~$F'_j\in \pauli^{\cW_{\cC}}$ by propagating~$F^\ell_j$ through the subcircuit~$\cC\cap \cR_j$.
Assume that the measurement yielded the outcome~$z_j=(z_{j,q})_{q\in \Omega_{j,1}}\in \{0,1\}^{\Omega_{j,1}}$. Then it is easy to check that 
replacing~$F_j$ locally by~$F'_j$ defined by 
\begin{align}
(F'_j)|_{\partial_r\cR_j} &=F^{r}_{j,1}\otimes F^{r}_{j,2}F^{\mathsf{Int}}_{j,2}F^{\ell}_{j,2} P_j(z_j\oplus e_j)P_j(z_j)\label{eq:tildeFjealldef}
\end{align}
(and identities elsewhere) 
gives the identical circuit action (ignoring the classical output~$z$). Here the vector~$e_j\in \{0,1\}^{\Omega_{j,1}}$  defined as 
\begin{align}
(e_j)_q&=\begin{cases}
1 \qquad&\textrm{ if } (F^{\ell}_{j,1})_q\in\{X,Y\}\\
0&\textrm{ otherwise }
\end{cases}\label{eq:ejqdefinition}
\end{align}
for each $q\in \Omega_{j,1}$ indicates whether the error $F^{\ell}_1$ flipped the measurement outcome~$z_{q,j}$ of the  qubit~$q$. This follows from the read-once property which implies that the classical outcome~$z$ can be discarded. The error $F^{\ell}_1$ may change the measurement results, but this amounts to applying the erroneous adaptive Pauli~$P_j(z_j\oplus e_j)$ (instead of~$P_j(z_j)$), and is taken into account by the definition of~$\tilde{F}^r_j$.  The error $F^{\mathsf{int}}_{j,1}$ occurs right before the state repreparation step and can thus be omitted.

Because of the assumption that~$\cR_j$ is linear, 
the term~$P_j(z_j\oplus e_j)P_j(z_j)=P(e_j)$ in Eq.~\eqref{eq:tildeFjealldef} does not depend on $z_j$, i.e., we have 
\begin{align}
\begin{matrix}
(F'_j)|_{\partial_r\cR_j} &=&\left(F^{r}_{j,1}\otimes F^{r}_{j,2}F^{\mathsf{Int}}_{j,2}F^{\ell}_{j,2} P_j(e_j)\right)& \qquad\textrm{ where }\qquad e_j=e_j(F^\ell_{j,1})\ ,
\end{matrix}\label{eq:tildefelljsubstition}
\end{align}
see Eq.~\eqref{eq:ejqdefinition}.

Let $F\sim\cN^{\pauli}_\cC(p)$. Define $F'$ by substituting~$F|_{\cR_j}$ by $F'_j$  (see Eq.~\eqref{eq:tildefelljsubstition}) for each $j\in [m]$. Then $F'$ satisfies~\eqref{eq:fbowtieceqaadaptive} (i.e., has the same overall effect on $\cC$ as $F$) and has the support property expressed by Eq.~\eqref{eq:adapmvadaptive}.
It remains to show that $F'$ is local stochastic Pauli noise of strength~$p'$.

We can factor
each $F'_j|_{{\partial_r\cR}}$ as
\begin{align}
F'_j=G_j\cdot H_j\cdot K_j\cdot L_j
\end{align}
where the latter Paulis are defined by 
\begin{align}
\left(G_j\right)|_{\partial_r \cR_j}&=F^r_{j,1}\otimes F^{r}_{j,2}=F^r_j\\
\left(H_j\right)|_{\partial_r \cR_j}&=I\otimes F^{\mathsf{Int}}_{j,2}\\
\left(K_j\right)|_{\partial_r \cR_j}&=I\otimes F^{\ell}_{j,2}\\
\left(L_j\right)|_{\partial_r\cR_j} &=I\otimes P_j(e_j)\ .
\end{align}
Define 
\begin{align}
G=\prod_{j=1}^m G_j\qquad H =\prod_{j=1}^m H_j\qquad K =\prod_{j=1}^m K_j\qquad 
L=\prod_{j=1}^m L_j\ .
\end{align}
Factoring $F=F_{\cR}\cdot F_{\cR^c}$ where $\cR=\bigcup_{j=1}^m \cR_j$, 
and where $F_{\cR^c}$ denotes the restriction of~$F$ to $\cW\backslash \cR$, 
we then have 
\begin{align}
F'&=F_{\cR}\cdot F_{\cR^c}\\
&=G\cdot H\cdot K\cdot L\cdot F_{\cR^c}\ .\label{eq:ghlfcrcm}
\end{align}
Clearly, $F_{\cR^c}$, $G$, $H$ and $L$ are (permuted) restrictions of~$F$, hence local stochastic Pauli noise of strength~$p$.
 We show the following about the stochastic Pauli noise~$L$.
\begin{claim}\label{claim:localstochasticlnpc}
We have $L\sim \cN^{\pauli}_{\cC}(p'')$
where 
\begin{align}
p''&=(Up)^{1/V}\ .
\end{align}
\end{claim}
\begin{proof}Since $\supp(L)\subseteq 
\bigcup_{j=1}^m \{t_j+1\}\times \Omega_{j,2}$ by definition, it suffices to consider a subset~$W=\bigcup_{j=1}^m W_j$ where $W_j\subseteq \{t+1\}\times \Omega_{j,2}$.
Observe that for $W_j\neq \emptyset$, $W_j\subseteq \supp(L)$ implies that $e_j\neq 0$, which in turn implies that $F_{j,1}^\ell\neq I$. In other words, at least one wire~$w_j\in \{t\}\times \Omega_{j,1}$ must belong to the support of~$F$. 
Let $\{j_1,\ldots,j_s\}=\{j\in [m]\ |\ W_j\neq \emptyset\}$. It follows that
\begin{align}
\Pr\left[W\subseteq \supp(L)\right]&\leq \Pr\left[ \exists 
(w_{j_1},\ldots,w_{j_s})\in \Omega_{j_1,1}\times \cdots\times  \Omega_{j_s,1}
\textrm{ with } \{w_{j_1},\ldots,w_{j_s}\}\subseteq \supp(F)\right]\\
&\leq \left(\prod_{r=1}^s |\Omega_{j_r,1}|\right)\cdot p^{s}\\
&\leq (Mp)^s
\end{align}
where we used the union bound, the fact that $F$ is local stochastic and the definition of~$J$.
Also observe that
\begin{align}
|W|&\leq \sum_{r=1}^s |W_{j_r}|\leq s\cdot V
\end{align}
It follows that $s\geq |W|/V$ and thus
\begin{align}
\Pr\left[W\subseteq \supp(L)\right]&\leq \left((Up)^{1/V}\right)^{|W|}\ .
\end{align}
This is the claim.
\end{proof}

Combining Claim~\ref{claim:localstochasticlnpc}
and the fact that $G,H,K,F_{\cR^c}\sim \cN^{\pauli}_{\cC}(p)$
with Eq.~\eqref{eq:ghlfcrcm}, 
we can apply 
Lemma~\ref{lem:pauli noise properties}~\eqref{item:pauli noise property 3}  to obtain $F\sim \cN^{\pauli}_{\cC}(p')$
where
\begin{align}
p'&=5 \max\{p,p''\}^{1/5}\\
&=5 (p'')^{1/5}\  .
\end{align}
This concludes the proof.
\end{proof}

\subsection{Circuit transformations preserving local stochastic Pauli noise\label{sec:noisepreservingtransformations}}
Let $\cC$ be a given Clifford circuit.  In this section, we discuss various transformations defining a new circuit~$\tilde{\mathsf{C}}$ derived from~$\cC$. We show that these transformations preserve local stochastic Pauli noise. More precisely, we argue that a noisy implementation~$\tilde{F}\bowtie \tilde{\mathsf{C}}$ of $\tilde{\mathsf{C}}$ with local stochastic Pauli noise~$\tilde{F}\sim\cN^{\pauli}_{\cC}(\tilde{p})$
is equivalent to a noisy implementation~$F\bowtie\mathsf{C}$ of the original circuit~$\mathsf{C}$ with local stochastic Pauli noise~$F\sim \cN^{\pauli}_{\cC}(p)$ of noise strength $p=p(\tilde{p})$ related to $\tilde{p}$. Specifically, we often show that $p=\Lambda\tilde{p}^\lambda$. The following definition will be convenient.
\begin{definition}
Let $\cC$ and $\tilde{\cC}$ be two circuits.
We write
\begin{align}
\tilde{\cC} &\gtrsim_{(\Lambda,\lambda)}\cC
\end{align}
if the following holds for some positive constants~$(\Lambda,\lambda)$.
For any $\tilde{p}\in [0,1]$ 
and local stochastic Pauli noise~$\tilde{F}\sim\cN^{\pauli}_{\tilde{\cC}}(\tilde{p})$, there is local stochastic Pauli noise~$F\sim \cN^{\pauli}_{\cC}(p)$ with $p =\Lambda\tilde{p}^{\lambda}$ 
such  that 
\begin{align}
F\bowtie \cC&=\tilde{F}\bowtie\tilde{\cC}\ .
\end{align}
\end{definition}
We often omit the constants~$(\Lambda,\lambda)$, simply writing~$\tilde{\cC} \gtrsim\cC$. The significance of this relationship is that it guarantees that the (family of) circuit(s) $\tilde{\cC}$ inherit(s) fault tolerance properties from $\cC$ in the following sense. Suppose that the circuit(s) $\cC$ are robust to noise below some threshold rate $p_\ast > 0$. Then any noisy implementation of $\tilde{\cC}$ with noise rate $\tilde{p}$ is equivalent to a noisy implementation of $\cC$ with noise rate $p = \Lambda\tilde{p}^\lambda$, which is robust to noise if $p < p_\ast$. Informally speaking, this implies that the noisy implementation of $\tilde{\cC}$ is robust to noise if $\tilde{p} < \tilde{p}_{\ast}$ is below
the (smaller) threshold error strength $\tilde{p}_{\ast} = (p_\ast / \Lambda)^{1/\lambda} >0$.

In fact, $\gtrsim$ defines a preorder on circuits, as formalized below.

\begin{lemma}\label{lem:circuitpreorder}
    Let $\cC_1$, $\cC_2$, and $\cC_3$ be circuits such that $\cC_2 \gtrsim \cC_1$ and $\cC_3 \gtrsim \cC_2$. Then $\cC_3 \gtrsim \cC_1$. 
\end{lemma}

\begin{proof}
    Let $p_3 \in [0,1]$ and let $F_3 \sim \cN_{\cC_3}^{\pauli}(p_3)$ be local stochastic Pauli noise on $\cC_3$. Since $\cC_3 \gtrsim \cC_2$, there exist constants $B,b>0$ such that
    \begin{align}\label{eq:F3 and F2}
        F_3 \bowtie \cC_3 = F_2 \bowtie \cC_2\ ,
    \end{align}
    where $F_2 \sim \cN_{\cC_2}^{\pauli}(p_2)$, and $p_2 = Bp_3^b$. But since $\cC_2 \gtrsim \cC_1$, there exist some constants $C,c>0$ such that
    \begin{align}\label{eq:F2 and F1}
        F_2 \bowtie \cC_2 = F_1 \bowtie \cC_1\ ,
    \end{align}
    where $F_1 \sim \cN_{\cC_1}^{\pauli}(p_1)$, and $p_1 = Cp_2^c = (CB^c) p_1^{bc}$. Combining Eq.~\eqref{eq:F3 and F2} with Eq.~\eqref{eq:F2 and F1} implies the claim.
\end{proof}

Thus we can apply several such transformations in succession: If 
\begin{align}
\cC_r\gtrsim \cC_{r-1}\gtrsim\cdots \gtrsim \cC_1\gtrsim \cC\ 
\end{align}
for some constant~$r=O(1)$ and $\cC$ has a certain threshold property, then the same 
is the case for $\cC_r$.
In this way, the problem of establishing 
a threshold for a circuit~$\tilde{\cC}$ is reduced to showing that it can be obtained by transforming a circuit~$\cC$ with a known threshold by a  (constant-length) sequence of transformation steps or ``moves''.

We first discuss the transformation of an adaptive circuit~$\mathsf{C}$ into an essentially non-adaptive circuit~$\tilde{\mathsf{C}}$. 
\begin{lemma}[Adaptive and non-adaptive Clifford circuits and local stochastic Pauli noise]\label{lem:adaptivenonadaptivetransform}
Let $\cC$ be a depth-$D$ Clifford circuit. 
Then there is a depth-$D$ Clifford circuit~$\tilde{\mathsf{C}}$
 with the same interaction graph as~$\mathsf{C}$ and adaptivity only in its last layer 
 such that 
 \begin{align}
 \tilde{\cC}\gtrsim_{(1,1)}\cC\ .
 \end{align}
 \end{lemma}
 
 \begin{proof}
  Let $\tilde{F}\sim\cN^{\pauli}_{\tilde{\mathsf{C}}}(p)$ be local stochastic Pauli noise  of strength~$p$. We need to show that there is local stochastic Pauli noise $F\sim\cN^{\pauli}_{\mathsf{C}}(p)$ such that 
 \begin{align}
 F\bowtie \mathsf{C}&=\tilde{F}\bowtie \tilde{\mathsf{C}}\ .
 \end{align}
  This follows immediately from
 Lemma~\ref{lem:adaptivenonadaptiveClifford}
 and Lemma~\ref{lem:pauli noise properties}~\eqref{item:pauli noise property 1}.
\end{proof}

Next we discuss an ``inflation move'' which inserts layers of identity operations, see   Fig.~\ref{fig:inflation} for an illustration.

\begin{figure}
\centering
\begin{subfigure}{0.95\textwidth}
\centering
\includegraphics[width=8cm]{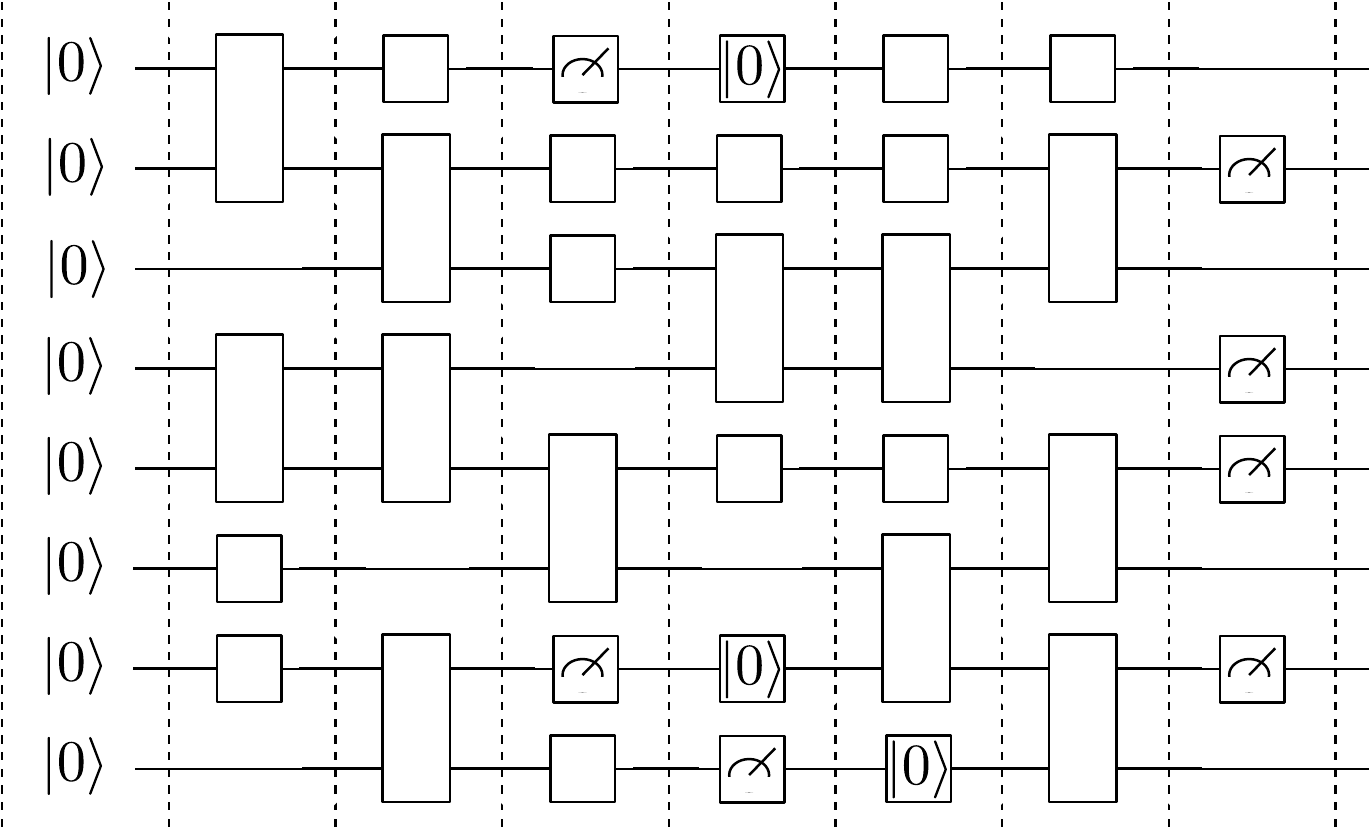}
\caption{The original circuit~$\cC$ \label{fig:beforemanipulation}}
\end{subfigure}\\
\begin{subfigure}{0.95\textwidth}
\centering
\includegraphics[width=0.9\textwidth]{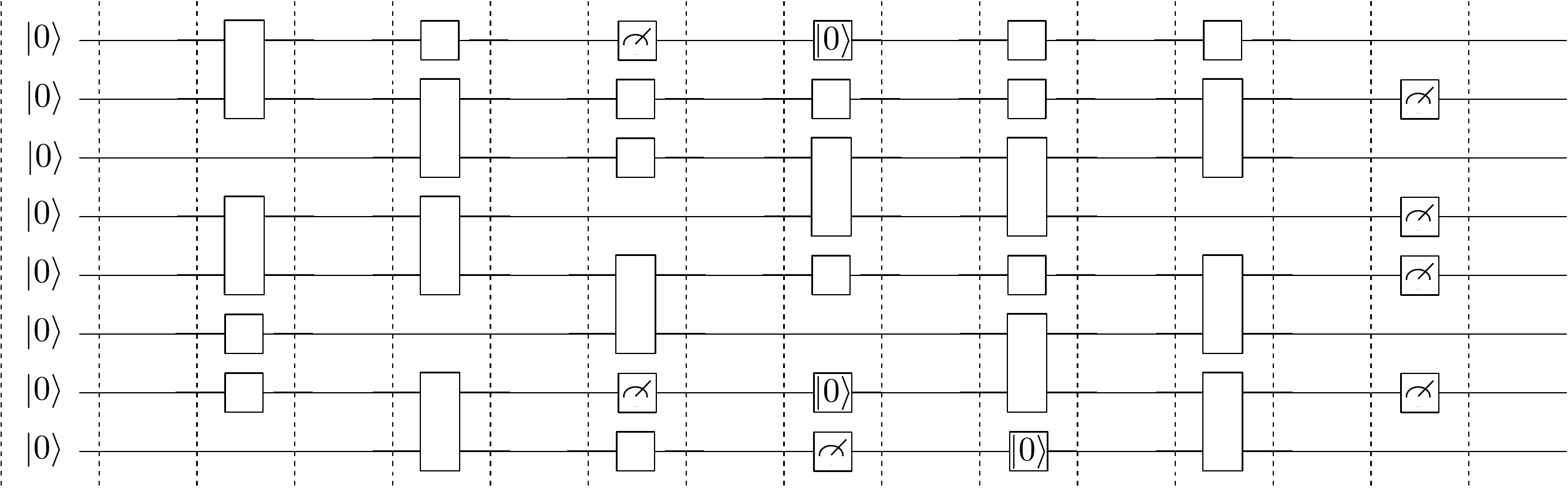}
\caption{The circuit~$\tilde{\cC}$ obtained by applying the inflation transformation\label{fig:inflated}}
\end{subfigure}\qquad 
\caption{Illustration of circuit inflation of depth $m=1$\label{fig:inflation}}
\end{figure}

\begin{lemma}[Circuit inflation]\label{lem:circuitinflation}
Let $\mathsf{C}=\mathsf{C}_D\circ\cdots\circ\cC_1$ be a product of $D$ subcircuits~$\cC_1,\ldots,\cC_D$. 
Let $m\in\mathbb{N}$.  Define the circuit
 \begin{align}
\tilde{\cC}^{(m)}&=\mathsf{id}^{\circ m}\circ \cC_D\circ\cdots \circ\mathsf{id}^{\circ m}\circ\cC_2\circ \mathsf{id}^{\circ m}\circ \cC_1\ ,\label{eq:insertmvavb}
 \end{align}
 Then $\tilde{\cC}^{(m)} \gtrsim \cC$.
 \end{lemma}
We call the circuit~$\tilde{\cC}^{(m)}$ defined by~\eqref{eq:insertmvavb} the result of applying a depth-$m$ inflation move to~$\cC$.
We often consider the case where each subcircuit~$\cC_t$, $t\in [D]$ is a depth-$1$ circuit.
In this case, the inflated circuit~$\tilde{\cC}^{(1)}$ is obtained by adding an idle (identity) layer after each layer of the original circuit~$\cC$, thereby doubling the depth.
More generally, we have $\mathsf{depth}(\tilde{\cC}^{(m)})=(m+1)\mathsf{depth}(\cC)$.

\begin{proof}
For $j\in [D]$, we write $D_j = \mathsf{depth}(\cC_j)$ and let
\begin{align}
\tilde{t}_j=\left(\sum_{r=1}^j D_j \right)+m(j-1)\qquad\textrm{ for }\qquad j\in [D]\ .
\end{align}
be the time slice associated with wires emanating from the subcircuit~$\cC_j$ in the circuit~$\tilde{\cC}$. 
Let
\begin{align}
\bigcup_{j=1}^D \mathsf{int}(\cC_j)&=\bigcup_{r=1}^{D} \{\tilde{t}_j-D_j+1,\tilde{t}_{j}-1\}\times [n]
\end{align}
denote the wires fully contained in the interior of one of the subcircuits~$\cC_1,\ldots,\cC_D$.
The remaining wires are
\begin{align}
\cW_{\tilde{\cC}}\backslash \bigcup_{t=1}^D \mathsf{int}(\cC_t)
&=\bigcup_{j=1}^D \left(\{\tilde{t}_j,\tilde{t}_j+1,\ldots,\tilde{t}_j+m\}\times [n]\right)
\end{align}
where $\left(\{\tilde{t}_j,\tilde{t}_j+1,\ldots,\tilde{t}_j+m\}\times [n]\right)$ denotes the wires between~$\cC_j$ and $\cC_{j+1}$ (respectively after $\cC_D$ for $j=D$).

Define
\begin{align}
t_j&=\sum_{r=1}^j D_j\qquad\textrm{ for }\qquad j\in [D]\ 
\end{align}
as the time slice associated with wires emanating from the circuit~$\cC_j$ in the circuit~$\cC$. 
Suppose $\tilde{F}\sim \cN^{\pauli}_{\tilde{\cC}}(p)$. Define a Pauli $F_{\mathsf{wait}}$ on~$\cW_{\cC}$
with support only within~$\{t_j\ |\ j\in [D]\}\times [n]$ by 
\begin{align}
F_{\mathsf{wait}}\left((t_j,q)\right)&= \prod_{k=0}^m \tilde{F}\left((\tilde{t}_j+k,q)\right)\qquad\textrm{ for all }\qquad j\in [D]\textrm{ and }q\in [n]\ .
\end{align}
The error $F_{\mathsf{wait}}$ is the ``aggregated'' error arising in the wait location (i.e., between the subcircuits).
We also define an error $F_0$ with support in the interior of the subcircuits~$\cC_1,\ldots,\cC_D$ by 
\begin{align}
F_0\left(t_j-D_j+k,q\right)&=\tilde{F}(\tilde{t}_j-D_j+k,q)\qquad\textrm{ for }\qquad k\in \{1,\ldots,D_j-1\}\
\end{align}
for $j\in [D]$. The error~$F_0$
encompasses all errors in the interior of the subcircuits. Defining
 \begin{align}
 F&:=F_{\mathsf{wait}}\cdot F_0\label{eq:fdefinitionproductmvd}
 \end{align}
it is easy to check that  $F\bowtie \cC = \tilde{F}\bowtie \tilde{\cC}^{(m)}$.

It remains to check that $F$ is local stochastic with strength $p^{(m)}$, for some $p^{(m)} > 0$.
By  Lemma~\ref{lem:pauli noise properties}~\eqref{item:pauli noise property 1},
 the error~$F_0$ satisfies~
 \begin{align}
 F_0\sim \cN^{\pauli}_{\cC}(p)\label{eq:fzeropaulicpm}
 \end{align}
since it is a restriction of the local stochastic Pauli error~$\tilde{F}$ which has strength~$p$. The error $F_{\mathsf{wait}}$ can be factored as
\begin{align}
F_{\mathsf{wait}}&=\prod_{k=0}^m F_{\mathsf{wait}}^{(k)}\label{eq:fproductwait}
\end{align}
where~$F_{\mathsf{wait}}^{(k)}$ has support only within~$\{t_j\ |\ j\in [D]\}\times [n]$ and is defined as
\begin{align}
F_{\mathsf{wait}}^{(k)}\left((t_j,q)\right)&= \tilde{F}\left((\tilde{t}_j+k,q)\right)\qquad\textrm{ for all }\qquad j\in [D]\textrm{ and }q\in [n]\ .
\end{align}
Clearly, each $F_{\mathsf{wait}}^{(k)}$ is obtained by restricting~$\tilde{F}$, hence we have 
\begin{align}
F_{\mathsf{wait}}^{(k)}\sim \cN^{\pauli}_{\cC}(p)\qquad\textrm{ for each }\qquad k\in \{0,\ldots,m\}\ \label{eq:localstchasticwaitk}
\end{align}
by  Lemma~\ref{lem:pauli noise properties}~\eqref{item:pauli noise property 1}. Combining Eqs.~\eqref{eq:fproductwait}  and~\eqref{eq:localstchasticwaitk} with Lemma~\ref{lem:pauli noise properties}~\eqref{item:pauli noise property 3}, we conclude that
\begin{align}
F_{\mathsf{wait}}\sim\cN^{\pauli}_{\cC}(p')\qquad\textrm{ where }\qquad p'=(m+1) p^{1/(m+1)}\ .\label{eq:localstochasticfwait}
\end{align}
Combining Eq.~\eqref{eq:fzeropaulicpm}, Eq.~\eqref{eq:localstochasticfwait} and the definition of $F$ (see Eq.~\eqref{eq:fdefinitionproductmvd}) with  Lemma~\ref{lem:pauli noise properties}~\eqref{item:pauli noise property 3}.
 gives
 \begin{align}
F\sim\cN^{\pauli}_{\cC}(p^{(m)})
\end{align}
with 
\begin{align}
p^{(m)}&= 2\max\{ p^{1/2},(m+1)^{1/2} p^{1/(2(m+1))} \}\\
&=2 (m+1)^{1/2}p^{\frac{1}{2(m+1)}}\ 
 \end{align}
 as claimed.
\end{proof}

Another convenient tool is the following (unitary) subcircuit substitution move, whose analysis is based on the circuit error cleaning Lemma~\ref{lem:cleaninglemma}. 
It is illustrated in Fig.~\ref{fig:rectanglessubstituted}.
\begin{figure}
\centering
\includegraphics[width=9cm]{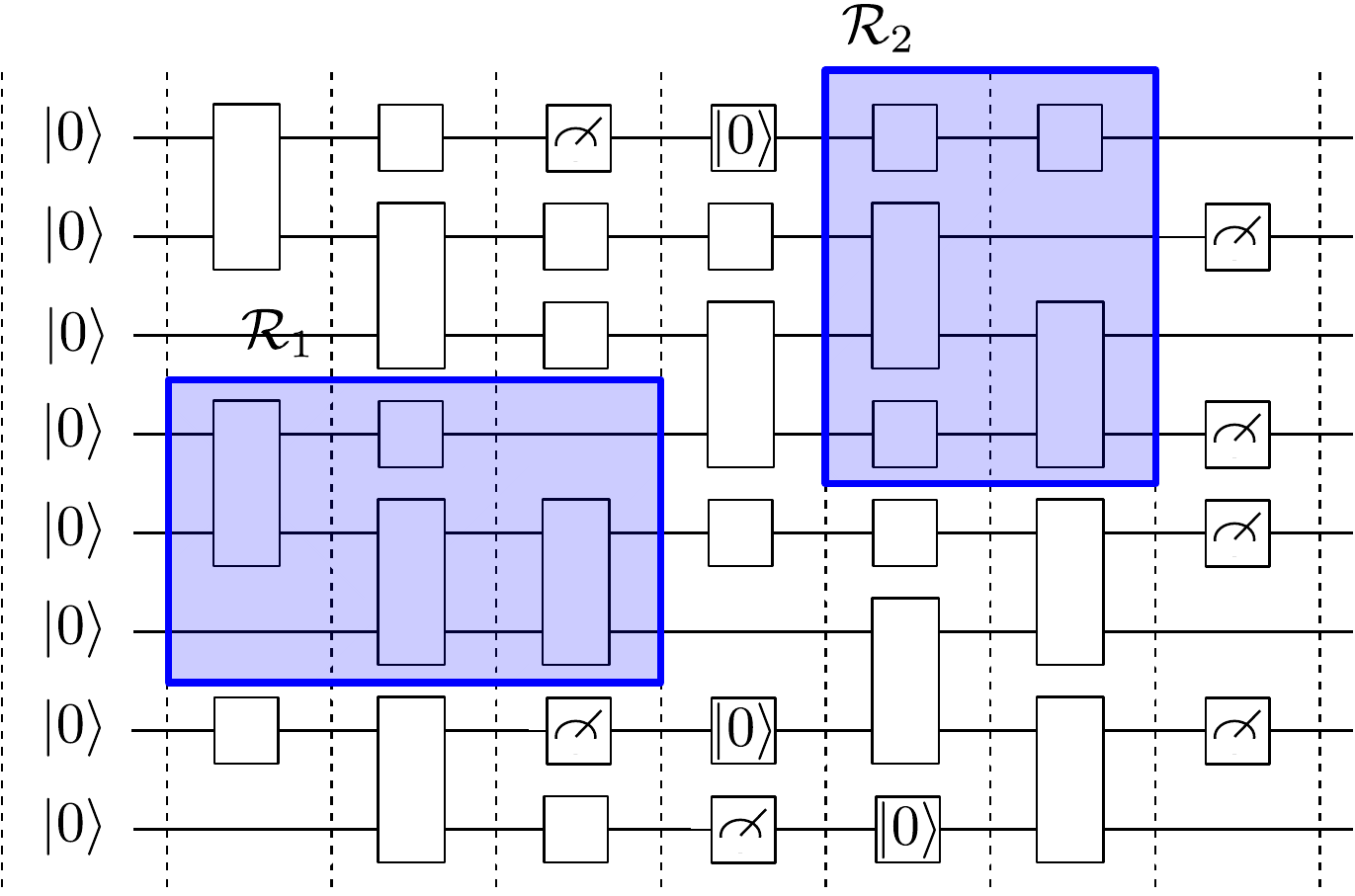}
\caption{The circuit~$\tilde{\cC}$ obtained by substituting unitary subcircuits in the circuit~$\cC$ defined in Fig.~\ref{fig:rectangles}\label{fig:rectanglessubstituted}}
\end{figure}

\begin{lemma}[Substitution of unitary subcircuits]\label{lem:subcircuitsubstitution}
Let $\cC$ be a depth-$D$ Clifford circuit. Let $\cR_1,\ldots,\cR_m$ be pairwise disjoint valid rectangles for~$\cC$ such that every 
subcircuit~$\cC^{\cR_j}$, $j\in [m]$ is unitary. For every $j\in [m]$, let $\tilde{\cC}^{\cR_j}$ be a unitary Clifford circuit defined on the qubits~$\Omega(\cR_j)$ associated with the rectangle~$\cR_j$. Assume further that $\tilde{\cC}^{\cR_j}$ is of depth $\Delta(\cR_j)$ and has the same action as $\cC^{\cR_j}$ as a unitary, i.e.,
\begin{align}
\cC^{\cR_j}&=\tilde{\cC}^{\cR_j}\qquad\textrm{ for every }j\in [m]\ .\label{eq:identicalactionsubcircuit}
\end{align}
Define a depth-$D$ Clifford circuit~$\tilde{\cC}$ by substituting the subcircuit~$\cC^{\cR_j}$ in $\cC$ by the circuit~$\tilde{\cC}^{\cR_j}$ for each $j\in [m]$. Assume that the maximal depth of a rectangle, $\Delta := \max_{j \in [m]} \Delta(\cR_j)$ is a constant, $\Delta = O(1)$. Then $\tilde{\cC} \gtrsim \cC$.
\end{lemma}
We call the circuit~$\tilde{\cC}$ obtained in this way the result of applying a depth-$\Delta$ substitution move to~$\cC$.
\begin{proof}
Let $\tilde{F}\sim\cN^{\pauli}_{\tilde{\mathsf{C}}}(p)$. We can apply  the circuit error cleaning Lemma~\ref{lem:cleaninglemma} to obtain a local stochastic error~$F\sim \cN^{\pauli}_{\mathsf{C}}(p')$ 
such that
\begin{align}
\tilde{F}\bowtie\tilde{\cC}&=F\bowtie\tilde{\cC}\label{eq:fbtcone}
\end{align}
such that $F$ has trivial action in the interior of each rectangle~$\cR_j$, $j\in [m]$, see Eq.~\eqref{eq:disjointfprimercj}. 
The latter property together with the assumption~\eqref{eq:identicalactionsubcircuit} means that we can replace~$\tilde{\cC}_{\cR_j}$ by $\cC_{\cR_j}$ for each $j\in [m]$ in the expression $F\bowtie\tilde{\cC}$.
It follows that
\begin{align}
F\bowtie\tilde{\cC}&=F\bowtie\cC\ .\label{eq:fbttwo} 
\end{align}
Combining Eqs.~\eqref{eq:fbtcone} and~\eqref{eq:fbttwo} concludes the proof.
\end{proof}
A useful consequence of 
the substitution lemma (Lemma~\ref{lem:subcircuitsubstitution})
and the inflation lemma (Lemma~\ref{lem:circuitinflation})
is the following result, which shows that every circuit is essentially equivalent to one in an ``alternating form''. By the latter, we mean that operation layers consisting of measurement and state preparation alternate with 
layers containing unitaries only.
\begin{lemma}[Circuits in alternating form]\label{lem:alternatingformcircuit}
Let $\cC=\cC_D\circ \cdots\circ \cC_1$ be a depth-$D$ Clifford circuit.
Then there is a circuit~$\tilde{\cC}=\tilde{\cC}_{2D}\circ\cdots  \circ \tilde{\cC}_1$ of depth~$2D$
such that the following holds:
\begin{enumerate}[(i)]
\item\label{it:alternatingone}
Each operation layer $\tilde{\cC}_t$ with odd~$t$ only consists of state preparation (of~$\ket{0}$), single-qubit measurements
and identities (idling).
\item\label{it:alternatingtwo}
Each operation layer~$\tilde{\cC}_t$ with even~$t$ only consists of one- and two-qubit unitaries.
\end{enumerate} 
Moreover, we have $\tilde{\cC} \gtrsim \cC$.
\end{lemma}
We say that a circuit $\tilde{\cC}$ satisfying~\eqref{it:alternatingone} and~\eqref{it:alternatingtwo} 
is in alternating form.

\begin{proof}
This follows immediately by applying a depth-$1$ inflation move, and two subcircuit substition moves of depth~$2$, see
Fig.~\ref{fig:alternatingcircuitproof} for an illustration. 

A noisy implementation~$\tilde{F}\bowtie\tilde{\cC}$ of~$\tilde{\cC}$ with $\tilde{F}\sim\cN^{\pauli}_{\tilde{\cC}}(p)$  is equivalent to a noisy execution of the circuit~$\cC'$ after the inflation move 
with local stochastic noise~$F'\sim\cN^{\pauli}_{\cC'}(p')$  of strength
\begin{align}
p'&=2\sqrt{2}p^{1/4}\ ,
\end{align}
see Lemma~\ref{lem:circuitinflation} with $m=1$, i.e., $F'\bowtie \cC'=\tilde{F}\bowtie\tilde{\cC}$. 
By Lemma~\ref{lem:subcircuitsubstitution} with $\Delta=1$,
there is $F''\sim \cN^{\pauli}_{\cC''}(p'')$ with
\begin{align}
p''&=4 (2p'^{1/2})^{1/2}=4\sqrt{2}p'^{1/4}
\end{align}
such that the circuit~$\cC''$ after the first substitution satisfies $F'\bowtie\cC'=F''\bowtie\cC''$.
It follows that
the final circuit~$\cC$ (obtained after another substitution) satisfies~$F''\bowtie\cC''=F\bowtie \cC$
with~$F\sim \cN^{\pauli}_{\cC}(p''')$ where 
\begin{align}
p'''&=4\sqrt{2}p''^{1/4}=8\cdot 2^{7/32}\cdot p^{1/64}<10 p^{1/64}\ .
\end{align}
The claim follows from this. 

\end{proof}

\begin{figure}
\centering
\begin{subfigure}{0.95\textwidth}
\centering
\includegraphics[width=7cm]{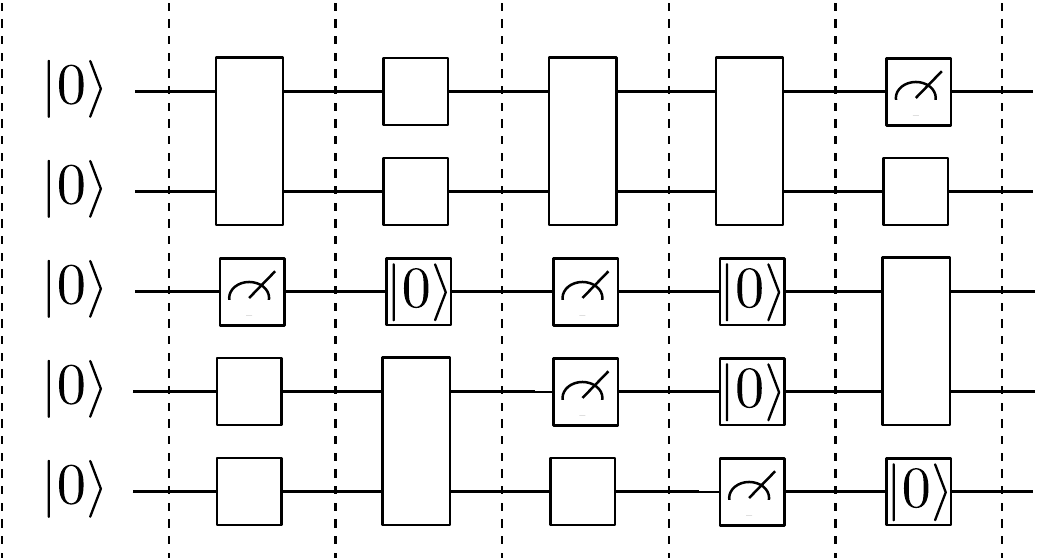}
\caption{The  circuit~$\cC$ with measurements, preparations and unitaries in any layer.}
\end{subfigure}
\begin{subfigure}{0.95\textwidth}
\centering
\includegraphics[width=14cm]{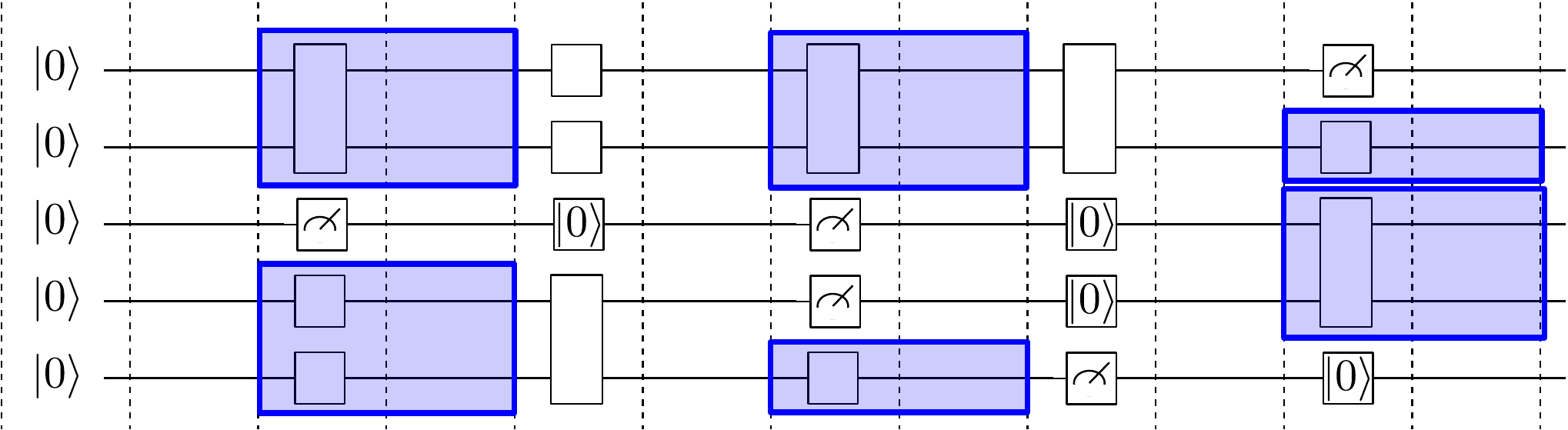}
\caption{The first substitution move applied to the inflated circuit. }
\end{subfigure}
\begin{subfigure}{0.95\textwidth}
\centering
\includegraphics[width=14cm]{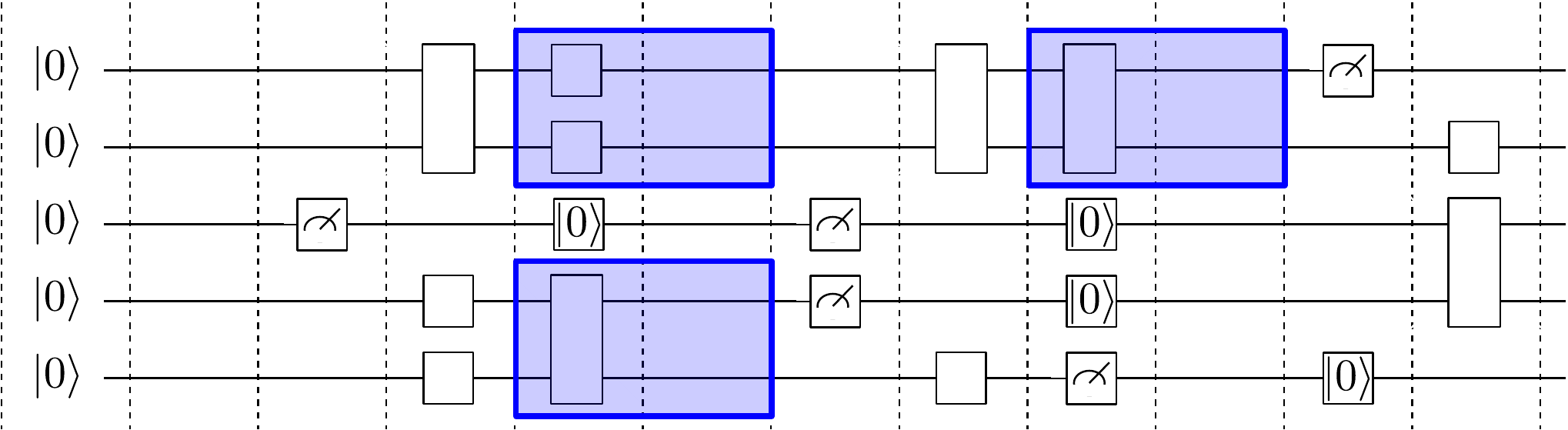}
\caption{The second substitution move applied to the inflated circuit.}
\end{subfigure}
\begin{subfigure}{0.95\textwidth}
\centering
\includegraphics[width=14cm]{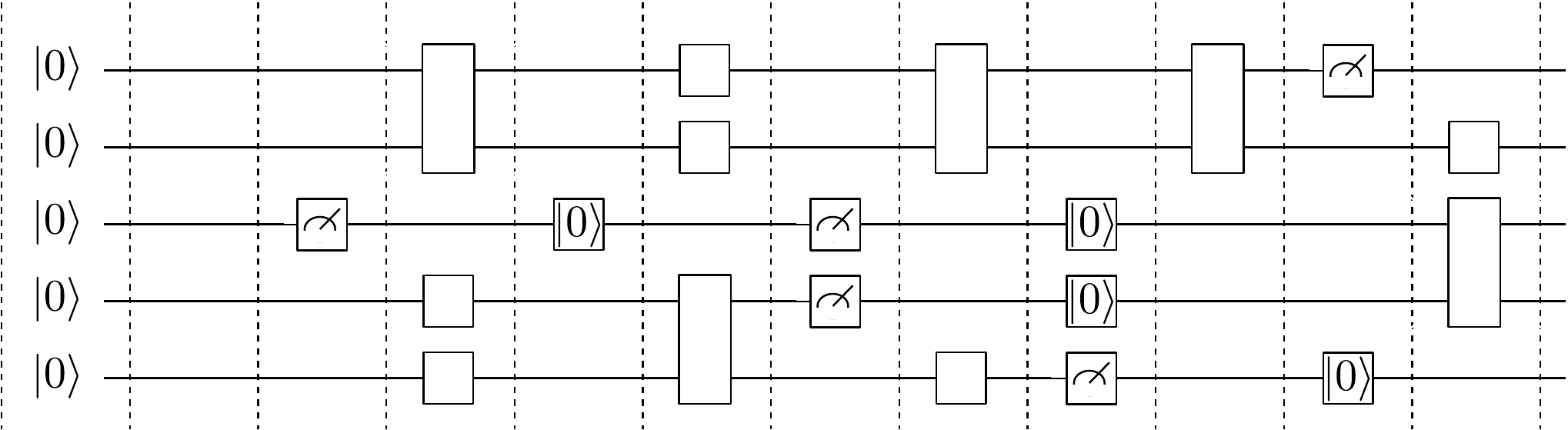}
\caption{The circuit~$\tilde{\cC}$ in alternating form.}
\end{subfigure}
\caption{Transforming a circuit~$\cC$ into an equivalent circuit~$\tilde{\cC}$ of alternating form.
In the substitution moves,  a subcircuit of the form~$\mathsf{id}\circ \cU$ (where $U$ is a one- or two-qubit unitary) is replaced by the circuit~$\cU\circ\mathsf{id}$ in every rectangle.\label{fig:alternatingcircuitproof}}
\end{figure}

\subsection{Circuit locality transformations and noise\label{sec:circuitlocalitytransformations}}
Let $G=(V,E)$ be a graph. We call a circuit~$\cC$
consisting of one- and two-qubit operations $G$-local if its qubit count is $N(\cC)=|V|$, and if the qubits can be assigned to the vertices of~$V$ in such a way that every two-qubit operation in~$\cC$ is between nearest neighbors (i.e., associated with an edge) on~$G$.

We use the term $1D$-local for a circuit~$\cC$
which is $\mathsf{P}_{N}$-local on the path graph~$\mathsf{P}_{N}$ for some $N\in\mathbb{N}$. 
We also say that a circuit~$\cC$ is $2D$-local if it is
$G_{\ell_X,\ell_Y}$-local for the rectangular grid graph
$G_{\ell_X,\ell_Y}:=\mathsf{P}_{\ell_X}\square \mathsf{P}_{\ell_Y}$ for some $\ell_X,\ell_Y\in\mathbb{N}$. 

In the following, we describe  a kind of ``normal form'' of a $G$-local circuit~$\cC$. 
Let~$\chi'=\chi'(G)$ be the edge chromatic number of~$G$, i.e., the minimum number of colors required to color the edges such that no two adjacent edges share the same color. 
We note that $\chi'(G)\leq d_{\max}(G)+1$  by Vizing's theorem~\cite{Vizing64}, where $d_{\max}(G)$ is the maximum degree of~$G$, and $\chi'(G)\leq d_{\max}(G)$ for any bipartite graph~$G$ according to Kőnig's line coloring theorem~\cite{konig1931}.
\begin{lemma}[Normal form of $G$-local circuit]\label{lem:circuitnormalform}
Let $\cC$ be a $G$-local circuit of depth~$D$.
Then there is 
a disjoint partition~$E=E_{1}\cup\cdots \cup E_{\chi'}$ and 
 a $G$-local circuit 
 \begin{align}
 \tilde{\cC}&=\left(\left(\tilde{\cC}_{D,\chi'}\circ \cdots\circ \tilde{\cC}_{D,1}\right)\circ
\tilde{\cC}^{\mathsf{pm}}_D\right)
\circ \cdots\circ 
 \left(\left(\tilde{\cC}_{1,\chi'}\circ \cdots\circ \tilde{\cC}_{1,1}\right)\circ \tilde{\cC}^{\mathsf{pm}}_1\right)
  \end{align}
  of depth
  $(\chi'+1)\cdot D$ with the same action as~$\cC$ (i.e., $\cC=\tilde{\cC}$) and the following properties:
\begin{enumerate}[(i)]
\item
For each $t\in [\tilde{D}]$, the depth-$1$ circuit~$\tilde{\cC}^{\mathsf{pm}}_t$
consists of single-qubit state preparation, measurement and identities only.
\item
For each $t\in [D]$, the circuit $\tilde{\cC}_{t,\alpha}$ is unitary, and local on the subgraph of~$G$ induced by the subset of edges~$E_\alpha$, $\alpha\in [\chi']$.
\item\label{it:normalformGlocalsteps}
The circuit~$\tilde{\cC}$ is obtained from~$\cC$ by application of a depth-$1$ inflation move, two subcircuit substitution moves of depth~$2$, a 
depth-$(\chi'-1)$ circuit inflation move, followed by a subcircuit substitution move.
 \end{enumerate}
\end{lemma}
\begin{proof}
We refer to~\cite[Lemma A.1]{brayietal20} for a similar proof. Choose any disjoint partition $E=E_{1}\cup \cdots \cup E_{\chi'}$ defined by an edge coloring of~$G$. We note that every color class~$E_\alpha$ is a matching, i.e., a set of edges where no two edges share a vertex.
Let 
\begin{align}
\cC'&=\cC'_{2D}\circ\cdots\circ \cC'_{1}\ 
\end{align}
be the circuit obtained by bringing~$\cC$ into alternating form, see Lemma~\ref{lem:alternatingformcircuit}. Recall
that each layer $\cC'_t$ for odd $t$ only consists of single-qubit state preparation and measurement, whereas each layer~$\cC_t$ with even~$t$ is unitary.

Consider the circuit
\begin{align}
\cC''&=\left(\left(\mathsf{id}^{\circ \chi'-1}\circ \cC'_{2D}\right)\circ\tilde{\cC}'_{2D-1}\right)\circ\cdots\circ \left(\left(\mathsf{id}^{\circ \chi'-1}\circ\cC'_{2}\right)\circ\cC'_{1}\right)
\end{align}
 obtained from~$\cC'$ with a depth-$(\chi'-1)$ inflation move between each pair of layers. 
Every even operation layer~$\cC'_{2t}$, $t\in [D]$ is a depth-$1$ $G$-local circuit, i.e., a product
\begin{align}
\cC'_{2t}&=\prod_{e\in \Omega_t} \cO_e
\end{align}
where $\Omega_t\subseteq E$ is a subset of edges and $\cO_e$ acts on  (one or both) the qubits defining~$e$.
It follows that
\begin{align}
\left(\mathsf{id}^{\circ \chi'-1}\circ \cC'_{2t}\right)&=\tilde{\cC}_{t,\chi'}\circ\cdots \circ\tilde{\cC}_{t,1}\ ,\label{eq:circuitsubstitutionappliedlem}
\end{align}
where
\begin{align}
\tilde{\cC}_{t,\alpha}&=\prod_{e\in \Omega_t\cap E_\alpha} \cO_e\qquad\textrm{ for }\qquad \alpha\in [\chi']\ 
\end{align}
is a depth-$1$ circuit.  We set
\begin{align}
\tilde{\cC}^{\mathsf{pm}}_t&= \cC'_{2t-1}\qquad\textrm{ for }\qquad t\in [D]\ .
\end{align}
The claim follows from this by applying the substitutions defined by Eq.~\eqref{eq:circuitsubstitutionappliedlem}.
\end{proof}

Let $G=(V,E)$ and $G'=(V',E')$ be graphs with the same vertex set we assume to be $V=V'=[n]$.
We are interested in emulating the behavior of a $G'$-local circuit~$\cC'$ by a $G$-local circuit~$\cC$.
This can typically be achieved by applying neighboring swaps to bring qubits close to each other.  We refer to~\cite{childsetal19} for a systematic treatment of associated optimization problems. The following quantity is of central interest to us.
\begin{definition}
Let $G=(V,E)$ and $G'=(V',E')$ be graphs with a common vertex set~$V=V'=[n]$.
Assume that $\pi\in S_n$ is such that it can be realized by a sequence of transpositions of neighboring vertices in the graph~$G$ (this is always the case for connected $G$). We denote by $\mathsf{pdepth}_G(\pi)$ the ``parallel depth'' of~$\pi$, i.e.,  the minimum number of layers of simultaneous disjoint  transpositions of neighboring vertices in the graph~$G$ required to factor~$\pi$. 
\end{definition}
We can then define a circuit~$U_\pi$ composed of $\mathsf{pdepth}_G(\pi)$ SWAP gates  which permutes the qubits according to~$\pi$.

Applied to circuits, it suffices to consider pairs of qubits associated with certain matchings of~$G'$; only these need to be brought close to each other. The following definitions are useful.
\begin{definition}
Let $G=(V,E)$ and $G'=(V',E')$ be graphs with vertex set~$V=V'=[n]$.
Let $E'=\bigcup_{\alpha=1}^r E'_\alpha$ be a disjoint partition of the edges of~$E'$ associated with an edge coloring.  Then we say that $G$ is compatible with $\left(G',\{E'_\alpha\}_{\alpha=1}^r\right)$ 
if for every $\alpha\in [r]$, there is a permutation~$\pi_\alpha\in S_n$ such that 
\begin{align}
\pi_\alpha(E'_\alpha)\subset E\qquad\textrm{ for every }\qquad \alpha\in [r]\ .\label{eq:pialphaealphaprime}
\end{align}
Here we write $\pi_\alpha(E'_\alpha)=\{\pi_\alpha(e)\ |\ e\in E'_\alpha\}$ where $\pi_\alpha(\{u,v\})=\{\pi_\alpha(u),\pi_\alpha(v)\}$.
\end{definition}
In other words, $G$ is compatible with $\left(G',\{E'_\alpha\}\right)$
if application of~$\pi_\alpha$ makes the endpoints of edges in~$E'_\alpha$ neighboring in the graph~$G$. In particular, if  $G$ is compatible with $\left(G',\{E'_\alpha\}_{\alpha=1}^r\right)$, 
then each $G'$-local depth-$1$ circuit~$\cC'_\alpha$
which is local on the subgraph induced by~$E'_\alpha$  has equivalent action to the circuit
\begin{align}
\cC_\alpha = \cU_{\pi_\alpha}^{-1} \circ \big( \cU_{\pi_\alpha}\circ \cC'_\alpha\circ \cU^{-1}_{\pi_\alpha} \big) \circ \cU_{\pi_\alpha} \ .\label{eq:equivalentcircuitupcmv}
\end{align}
Here we write~$\cU(\rho)=U\rho U^\dagger$ for the CPTP map defined by a unitary~$U$. 
With the factorization of~$U_{\pi_\alpha}$ into layers of SWAP gates and interpreting the term in square brackets as a $G$-local depth-1 circuit, we see that $\cC_\alpha$ is $G$-local and has depth
\begin{align}
\mathsf{depth}(\cC_\alpha)&=2\mathsf{pdepth}_G(\pi)+1\ .
\end{align}
Replacing~$\cC'_\alpha$ by~\eqref{eq:equivalentcircuitupcmv} can thus be seen as a subcircuit substitution.
Applying this reasoning to each circuit~$\cC'_{t,\alpha}$  in the normal form of a $G'$-local circuit~$\cC'$ (see Lemma~\ref{lem:circuitnormalform}) we obtain Lemma~\ref{lem:geometrychanges} below which transforms a $G'$-local circuit into a $G$-local circuit. The following definition will be convenient.

\begin{definition}[$\Delta$-reduction between graphs]
Given two graphs $G'=(V',E')$ and $G=(V,E)$ with common vertex set~$V'=V=[n]$, we say that
$G^\prime$ $\Delta$-reduces to $G$ if the following holds. Let $\chi'=\chi'(G')$ be the edge chromatic number of~$G'$.
Then there is a disjoint partition~$E'=\bigcup_{\alpha=1}^{\chi'}E'_\alpha$ associated with an edge coloring of~$G'$
such that $G$ is compatible with $\left(G',\{E'_\alpha\}_{\alpha=1}^{\chi'}\right)$
with associated permutations~$\{\pi_\alpha\}_{\alpha\in [\chi']}\subset S_n$ (see Eq.~\eqref{eq:pialphaealphaprime}) satisfying
\begin{align}
\max_{\alpha\in [r]}\mathsf{pdepth}_G(\pi_\alpha)&\leq \Delta\ .
\end{align}
\end{definition}
The main statement we need to change from one geometric layout to a different one is the following lemma.
\begin{lemma}[Geometry change]\label{lem:geometrychanges}
Let $G'=(V',E')$ and $G=(V,E)$ be such that $G^\prime$ $\Delta$-reduces to $G$.
Suppose $\cC'$ is a $G'$-local circuit of depth~$D'$. Then there is a $G$-local circuit~$\cC$ with the following properties. 
\begin{enumerate}[(i)]
\item
The circuit~$\cC$ has the same action as~$\cC'$.
\item
$\mathsf{depth}(\cC)=(\chi'(2\Delta+1)+1)D'$.
\item
The circuit~$\cC$ is obtained from~$\cC'$ by 
first applying the moves specified in Lemma~\ref{lem:circuitnormalform}~\eqref{it:normalformGlocalsteps}, and subsequently 
 applying  a depth-$2\Delta$ inflation move and a depth-$(2\Delta+1)$ substitution move.
\end{enumerate}
\end{lemma}
\begin{proof}
Let
\begin{align}
\cC''&=\left(\left(\cC'_{D',\chi'}\circ \cdots\circ \cC'_{D',1}\right)\circ
\cC'^{\mathsf{pm}}_{D'}\right)
\circ \cdots\circ 
 \left(\left(\cC'_{1,\chi'}\circ \cdots\circ \cC'_{1,1}\right)\circ \cC'^{\mathsf{pm}}_1\right)
  \end{align}
  be the depth-$[(\chi'+1)\cdot D']$ circuit obtained by 
  applying the moves given in Lemma~\ref{lem:circuitnormalform}~\eqref{it:normalformGlocalsteps}, i.e., the normal form of the $G'$-local circuit~$\cC'$.
  Applying a depth-$2\Delta$ inflation move to each of the $\cC_{t,\alpha}'$ layers gives the circuit
\begin{align}
\cC''&=\left(\left(\cC''_{D',\chi'}\circ \cdots\circ \cC''_{D',1}\right)\circ
\cC'^{\mathsf{pm}}_{D'}\right)
\circ \cdots\circ 
 \left(\left(\tilde{\cC}''_{1,\chi'}\circ \cdots\circ \tilde{\cC}''_{1,1}\right)\circ \cC'^{\mathsf{pm}}_1\right)
  \end{align}
  where for each $t\in [D']$ and $\alpha\in [\chi']$, the circuit
\begin{align}
\cC''_{t,\alpha}:=\mathsf{id}^{\circ 2\Delta}\circ \cC'_{t,\alpha}
\end{align}
is  unitary and of depth~$2\Delta+1$. In particular, this means that~$\cC''$ is of depth $(\chi'(2\Delta+1)+1)D'$.

We note that since for each $t\in [D']$, the operation  layer $\cC'^{\mathsf{pm}}_t$ is a tensor product of single-qubit operations (measurements, state preparation or identities), each layer~$\cC'^{\mathsf{pm}}_t$ is $G$-local. It remains to argue that the remaining operation layers can be realized by $G$-local operations. By construction,  for each~$t\in [D']$ and~$\alpha\in [\chi']$, the circuit~$\cC''_{t,\alpha}$
has non-trivial support only on the qubits associated with endpoints of edges belonging to~$E_\alpha$. By definition, it therefore has the same action as
\begin{align}
\mathsf{id}^{\circ 2\Delta}\circ \cC'_{t,\alpha}&= \cU_{\pi_\alpha}^{-1} \circ \big(\cU_{\pi_{\alpha}} \circ \cC'_{t,\alpha} \circ \cU_{\pi_\alpha}^{-1} \big) \circ \cU_{\pi_\alpha} = \cC_{t,\alpha} \ ,\label{eq:vdbadv}
\end{align}
where as in Eq.~\eqref{eq:equivalentcircuitupcmv}, each $\cC_{t,\alpha}$ can be realized as a $G$-local depth-$(2\Delta+1)$ circuit.
Applying a depth-$(2\Delta+1)$ substitution move to
$\cC''_{t,\alpha}$ by $\cC_{t,\alpha}$ thus gives the claim. 
\end{proof}

We often consider circuit families and associated graph families. We say one family of graphs $\Delta$-reduces to another if this is the case for each corresponding pair of graphs.  Following standard computer science terminology, we do not explicitly write out the family as  e.g., $G=\{G_n\}_{n\in \Gamma}$, $\Gamma\subseteq \mathbb{N}$.
\begin{corollary}[Compatible locality and local stochastic errors]\label{cor:compatiblev}
Let $G'=(V',E')$ and $G=(V,E)$ be (families of) graphs indexed by some subset of integers~$n$, where $V=V'=[n]$, such that $G^\prime$ $\Delta$-reduces to $G$. Assume that $\chi'=\chi'(G')=O(1)$
and  $\Delta=O(1)$  (i.e., a constant).
Let $\cC'$ be a depth-$D'$ $G'$-local circuit (where $D'=D'(n)$ can depend on~$n$). Then there is a $G$-local circuit $\cC$ of depth $\mathsf{depth}(\cC) = O(D')$ such that $\cC' \gtrsim \cC$.
\end{corollary}
\begin{proof}
This is an immediate consequence of Lemma~\ref{lem:geometrychanges} for geometry changes as well as
Lemmas~\ref{lem:circuitinflation} and~\ref{lem:subcircuitsubstitution} characterizing the transformation of local stochastic Pauli noise under circuit inflation and substitution.
\end{proof}

For our purposes, the most important  application of Corollary~\ref{cor:compatiblev} will be to translate between the path graph
$\pathgraph_{2r}=([2r],E(2r))$ with an even number $n=2r$ (where we assume $r$ to be even also) of vertices and edges
\begin{align}
E(2r)&=\{(j,j+1)\ |\ j\in [2r-1]\}\ ,
\end{align}
and the bilinear array $B'_{2r}=([2r],E'(2r))$ defined by the edges
\begin{align}
E'(2r)&=E_1'\cup E'_2\cup E'_3
\end{align}
where
\begin{align}
\begin{matrix}
E_1'&=&\left\{(2a-1,2a)\ |\ a\in [r-1]\right\}\\
E'_2&=&\left\{(4a-3,4a-1)\ |\ a\in [r/2]\right\}\cup \left\{(4a-2,4a)\ |\ a\in [r/2]\right\}\\
E'_3&=&\left\{(4a-1,4a+1)\ |\ a\in [r/2-1]\right\}\cup \left\{(4a,4a+2)\ |\ a\in [r/2-1]\right\}\ ,
\end{matrix}
\end{align}
see Fig.~\ref{fig:lineversusbilineararray}
for an illustration.
The disjoint partition $E'(2r)=E_1'\cup E'_2\cup E'_3$ defines a $3$-coloring of the edges of~$B'_{2r}$. 
In particular, with 
\begin{align}
\pi_1&=\mathsf{id}\\
\pi_2 &= \prod_{k=0}^{r/2-1}\left(4k+2\quad 4k+3\right)\\
\pi_3 &= \prod_{k=1}^{r/2-1}  \left(4k\quad 4k+1\right)
\end{align}
one can check that the path graph $\pathgraph_{2r}$ $\Delta$-reduces to the bilinear array graph~$B_{2r}$ with $\Delta=1$, see Fig.~\ref{fig:lineversusbilineararray}.

\subsection{Fault-tolerant implementation of prepare-and-measure Clifford circuits in 1D\label{ft:ftimplementationprepmeasure}}
The fact that the path graph~$\pathgraph_{2r}$ reduces to the bilinear array graph~$B_{2r}$
implies that  a fault tolerance construction on the bilinear array~$B_{2r}$ translates into a fault tolerance construction on the path graph~$\pathgraph_{2r}$ (which is ``truly'' 1D).

 In particular, invoking the fault tolerance construction of~\cite{stephens2007universal} for prepare-and-measure circuits (see Theorem~\ref{thm:ftbilineararrayconstruct}), we immediately obtain a robust implementation of a $1D$-local Clifford circuit of prepare-and-measure form, see Result~\ref{res:ftimplementation1Dprepmeas}. We state this as follows to clarify how the thresholds are related.

\begin{theorem}[Fault tolerance on the path graph]\label{thm:ftpathgraph}
Let $p_*>0$ be the fault tolerance threshold error strength of the bilinear qubit array construction (see Theorem~\ref{thm:ftbilineararrayconstruct}).
Then there are constants $(\Lambda_{\mathsf{1D}},\lambda_{\mathsf{1D}})$ such that the following holds for any $\varepsilon>0$.  Let  $\cC$ be a $1D$-local prepare-and-measure circuit~$\cC$ composed of Clifford unitaries. Then there a circuit~$\cC_{FT}$
whose output distribution is $\varepsilon$-close in variational distance 
to that of $\cC$ in the presence of local stochastic Pauli noise of strength 
\begin{align}
p&\leq (p_*/\Lambda_{\mathsf{1D}})^{1/\lambda_{\mathsf{1D}}}\ .
\end{align}
The circuit~$\cC_{FT}$ is  $1D$-local and the construction has a polylogarithmic overhead. 
\end{theorem}
\begin{proof}
Let $\cC^{\mathsf{bilinear}}$ be the fault-tolerant implementation of~$\cC$ given in Theorem~\ref{thm:ftbilineararrayconstruct}. We use the fact that the path graph $\pathgraph_{2R}$ $\Delta$-reduces to the bilinear array graph~$B_{2R}$ with~$\Delta=1$, see the remarks after Corollary~\ref{cor:compatiblev}. 
This means we can use Corollary~\ref{cor:compatiblev}
to obtain a $1D$-local circuit 
$\cC^\mathsf{1D}$ of depth
\begin{align}
\mathsf{depth}(\cC^\mathsf{1D})&=O\left(\mathsf{depth}(\cC^{\mathsf{bilinear}})\right)
\end{align}
which uses 
\begin{align}
N(\cC^\mathsf{1D})&= N(\cC^{\mathsf{bilinear}})
\end{align}
qubits and is such that $\cC^{\mathsf{1D}} \gtrsim \cC^{\mathsf{bilinear}}$. In particular, there exist constants $\Lambda_{\mathsf{1D}},\lambda_{\mathsf{1D}}>0$ such that every noisy implementation of $\cC^{\mathsf{1D}}$ under stochastic Pauli noise of strength $p$ is equivalent to a noisy implementation of $\cC^{\mathsf{bilinear}}$ under Pauli noise of strength at most $\Lambda_{\mathsf{1D}}p^{\lambda_{\mathsf{1D}}}$.
The claim follows from this and Theorem~\ref{thm:ftbilineararrayconstruct}. 
\end{proof}

\begin{figure}
\centering
\begin{subfigure}{0.12\textwidth}
\centering
\includegraphics[height=5cm]{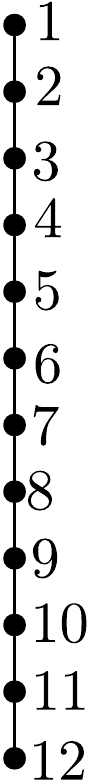}
\caption{The path graph~$\pathgraph_{2r}$\label{fig:linegraphfigure}}
\end{subfigure}\quad
\begin{subfigure}{0.3\textwidth}
\centering
\includegraphics[height=5cm]{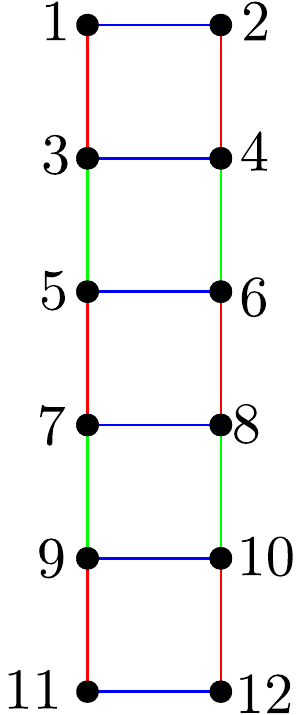}
\caption{The bilinear array $B_{2r}$ with the partition of edges~$E_1'$ (blue), $E'_2$ (red) and $E'_3$ (green) .\label{fig:bilineararray}}
\end{subfigure}\quad
\begin{subfigure}{0.25\textwidth}
\centering
\includegraphics[height=5cm]{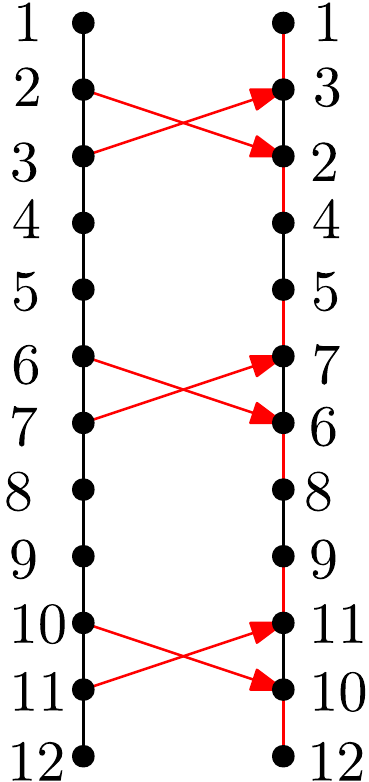}
\caption{The permutation $\pi_2$}
\end{subfigure}\quad
\begin{subfigure}{0.25\textwidth}
\centering
\includegraphics[height=5cm]{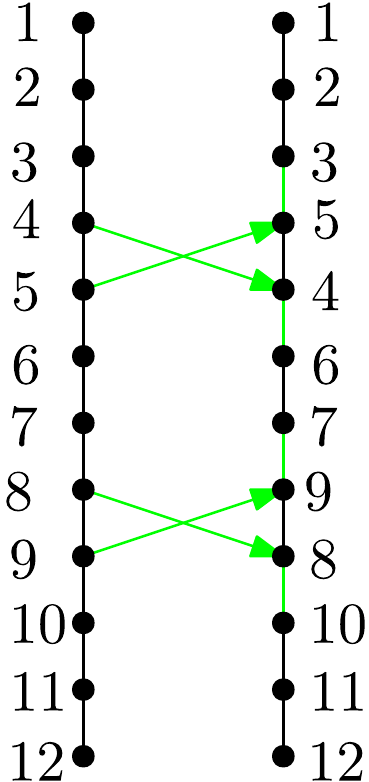}
\caption{The permutation $\pi_3$}
\end{subfigure}
\caption{The path graph~$\pathgraph_{2r}$ and the bilinear array graph $B_{2r}$ illustrated for $r=6$\label{fig:lineversusbilineararray}}
\end{figure}

\section{Robust, geometrically local implementation of circuits\label{sec:ftimplementcliff}}
In this section we explain how to fault-tolerantly implement circuits under certain structural restrictions. We then consider different geometries:
In  Section~\ref{sec:implementationbilinear}, we discuss geometrically local circuit implementation in bilinear arrays, before upgrading this to truly $1D$-local circuit implementation in Section~\ref{sec:ftimplementation1D}.
 
\subsection{Fault-tolerant circuit implementations in a bilinear array\label{sec:implementationbilinear}}
Applying the gadgets from Ref.~\cite{stephens2007universal}  (see Section~\ref{sec:bilineararrayfaulttolerance})
with Lemma~\ref{lem:failureprobencdec}, we obtain the following result. It applies to Clifford circuits~$\cC$ with any number of input respectively output qubits.
\begin{lemma}[Robust Clifford implementation in a bilinear array]\label{lem:bilinearimplementation}
Let $\mathsf{g}\geq 49$ be the constant from Theorem~\ref{thm:robustimplementationstochastic}. 
Let $\cC$ be a 1D-local (unitary) Clifford circuit on $N$-qubits. 
For $L \in \NN$  define 
\begin{align}
f^{\mathsf{bilinear}}_L(p):=
\left\{ \begin{array}{ll}
        p\cdot 4\mathsf{g} N + p_\ast (p / p_\ast)^{2^L} |\mathsf{Loc}(\cC)| & \text{for $p \leq p_\ast / 2$} \\
        1& \text{otherwise}
    \end{array}\right.\ ,\label{eq:fbilinearlp}
\end{align}
Then there is an adaptive Clifford circuit $\cC^{\mathsf{bilinear}}$ which $f_L$-robustly implements $\cC$ with the following properties:
\begin{enumerate}[(i)]
    \item The circuit $\cC^{\mathsf{bilinear}}$ is $B_{2R}$-local, where $B_{2R}$ is the bilinear array graph with $R$ qubits on each line (see Fig.~\ref{fig:bilineararray}) with
\begin{align}
R&=N e^{\theta(L)}\ .\label{eq:Rnelbound}
\end{align}
    \item The depth of $\cC^{\mathsf{bilinear}}$ is bounded as
\begin{align}\label{eq:dnelbound}
\mathsf{depth}(\cC^{\mathsf{bilinear}})&=\mathsf{depth}(\cC)e^{\Theta(L)}\ .
\end{align}
\end{enumerate}
Moreover, the $N_{\mathsf{in}}$ input qubits of $\cC^{\mathsf{bilinear}}$ are located on sites $Q_1,\dots,Q_{N_{\mathsf{in}}}$ which have pairwise distance  $e^{\Theta(L)}$ on the path graph. The $N_{\mathsf{out}}$ output qubits are similarly separated.
\end{lemma}
\begin{proof}
As in the proof of Theorem~\ref{thm:robustimplementationstochastic}, we define $\cC^{\mathsf{bilinear}}$ by Eq.~\eqref{eq:cftconstruction} using the gadgets for the $[[7,1,3]]$-code introduced in Ref.~\cite{stephens2007universal}. Since the latter are constructed 
for the bilineary array geometry, this ensures that the circuit~$\cC^{\mathsf{bilinear}}$
is $B_{2R}$-local on a bilinear array~$B_{2R}$. 
The function~$f^{\mathsf{bilinear}}_L$ introduced in Eq.~\eqref{eq:fbilinearlp}
 is obtained by specializing Eq.~\eqref{eq:flprobust} to the $[[7,1,3]]$-code and using that $N_{\mathsf{in}},N_{\mathsf{out}}\leq N$.

By Lemma~\ref{lem:failureprobencdec}
 and because $2R=N(\cC^{\mathsf{bilinear}})$ we have 
    \begin{align}
\begin{matrix}
7^L N&\leq    &     2R &\leq &\nmax^L N + \nmax^{L+1}(N_{\mathsf{out}} + N_{\mathsf{in}})\\
  \dmin^L\mathsf{depth}(\cC)+2\dmin^{L-1}& \leq& \mathsf{depth}(\cC^{\mathsf{bilinear}}) &\leq& \dmax^L \mathsf{depth}(\cC) + 2\dmax^{L+1}
\end{matrix}\ .
    \end{align}
The  bounds~\eqref{eq:Rnelbound} and~\eqref{eq:dnelbound} on $R$ and $\mathsf{depth}(\cC^{\mathsf{bilinear}})$ follow from this.

The statement about the distance between the input qubits follows from the fact that the $N_{\mathsf{in}}$ logical qubits, after encoding, are associated with contiguous blocks of physical qubits placed next to each other on the bilinear array by construction (see Ref.~\cite{stephens2007universal}). Here each block consists of the data qubit associated with one logical qubit and the auxiliary qubits used to act on it.  In particular, such a block consists of
\begin{align}
N(\cC^{\mathsf{bilinear}})&=
N\left(\left( \cdec{L\rightarrow 0} \right)^{\otimes N_{\mathsf{out}}} \circ \cC^{(L)} \circ \left(\cenc{0\rightarrow L}\right)^{\otimes N_{\mathsf{in}}}\right)\\
&=e^{\Theta(L)}
\end{align}
qubits, where we used 
Lemma~\ref{lem:qubitdepthoverheadsimulation}
 and Lemma~\ref{lem:qubitdepthoverheadsimulationb}. The claim follows from this.
\end{proof}

\subsection{Fault-tolerant implementation in 1D\label{sec:ftimplementation1D}}
Here we argue that the construction given in Lemma~\ref{lem:bilinearimplementation} can be used to obtain a fault-tolerant $1D$-local circuit implementation scheme. The corresponding parameters are  as follows.
\begin{lemma}\label{lem:implementation1D}
Let $\mathsf{g}\geq 49$ and $p_\ast>0$ be the constants from Theorem~\ref{thm:robustimplementationstochastic}. 
 and Lemma~\ref{lem:bilinearimplementation}, respectively.  
There are constants $\Lambda_{\mathsf{1D}},\lambda_{\mathsf{1D}} > 0$   such that the following holds.
Let $\cC$ be a 1D-local (unitary) Clifford circuit on $N$ qubits with any number of input and output qubits. 
 For any $L \in \NN$ define the function $f_L^{\mathsf{1D}} : [0,1] \rightarrow \RR$ by
\begin{align}
    f_L^{\mathsf{1D}}(p) &:= \left\{ \begin{array}{ll}
        4\mathsf{g} \Lambda_{\mathsf{1D}}p^{\lambda_{\mathsf{1D}}} N + p_\ast \left(\frac{\Lambda_{\mathsf{1D}}p^{\lambda_{\mathsf{1D}}}}{p_\ast} \right)^{2^L} |\mathsf{Loc}(\cC)|& \text{for $p\leq \left(\frac{p_\ast}{2\Lambda_{\mathsf{1D}}}\right)^{1/\lambda_{\mathsf{1D}}}$} \\
        1 & \text{otherwise}
    \end{array}\right.\  .\label{eq:upperboundfrobust1D}
\end{align} Then there is an adaptive Clifford circuit $\cC^{\mathsf{1D}}$ which $f_L^{\mathsf{1D}}$-robustly implements $\cC$ with the following properties:
\begin{enumerate}[(i)]
    \item The circuit $\cC^{\mathsf{1D}}$ is 1D-local, i.e., local with respect to the path graph $\pathgraph_{2R}$ (see Fig.~\ref{fig:linegraphfigure}) on $2R$ qubits, where $R=N e^{\Theta(L)}$.
    \item The depth of $\cC^{\mathsf{1D}}$ satisfies 
$\mathsf{depth}(\cC^{\mathsf{1D}})=e^{\Theta(L)}$.
\item
The $N_{\mathsf{in}}$ input qubits  of $\cC^{\mathsf{1D}}$ are located on sites $Q_1,\dots,Q_{N_{\mathsf{in}}}$ which have pairwise distance~$e^{\Theta(L)}$ on the path graph~$P_{\mathsf{2R}}$. The $N_{\mathsf{out}}$ output qubits are similarly separated.
\end{enumerate}
\end{lemma}

\begin{proof}
The proof is analogous to that of Theorem~\ref{thm:ftpathgraph}.
Starting from a robust implementation $\cC^{\mathsf{bilinear}}$ of~$\cC$ on the bilinear qubit array (see Lemma~\ref{lem:bilinearimplementation}), one uses the fact that the path graph $\pathgraph_{2R}$ $1$-reduces to the bilinear array graph~$B_{2R}$. This results in a circuit~$\cC^{\mathsf{1D}}$ which is $1D$-local and has  identical functionality as~$\cC^{\mathsf{bilinear}}$. 
 The bounds on circuit depth, qubit count, and the locations of the input and output qubits then follow directly from 
the bounds given for~$\cC^{\mathsf{bilinear}}$ in  Lemma~\ref{lem:bilinearimplementation}.  Following similar reasoning as in the proof of Theorem~\ref{thm:ftpathgraph}, it can be argued that 
$\cC^\mathsf{1D}$ $f^\mathsf{1D}_L$-robustly implements $\cC$ with 
\begin{align}
f^\mathsf{1D}_{L}(p)&=f^{\mathsf{bilinear}}_L(\Lambda_{\mathsf{1D}}p^{\lambda_{\mathsf{1D}}})\ 
\end{align}
as required.
\end{proof}

Lemma~\ref{lem:implementation1D}
implies the following general statement. 
\begin{theorem}[Fault-tolerant $1D$-local Clifford circuit implementation]
\label{thm:implementation1D}
Let $N\in\mathbb{N}$ and  $\varepsilon>0$ be arbitrary. 
Let $\mathsf{g}\geq 49$, $p_\ast > 0$, $\Lambda_{\mathsf{1D}}$ and $\lambda_{\mathsf{1D}}$ be the constants from Lemma~\ref{lem:bilinearimplementation}.
Define 
\begin{align}
p_0^{\mathsf{1D}}(N,\varepsilon):=\min\left\{
\left(\frac{p_\ast}{2\Lambda_{\mathsf{1D}}}\right)^{1/\lambda_{\mathsf{1D}}}, \left(\frac{\varepsilon}{8N\mathsf{g}\Lambda_{\mathsf{1D}}}\right)^{1/\lambda_{\mathsf{1D}}}\right\}\ ,\label{eq:pzerooneddef}
\end{align}
Let $\cC$ be a $1D$-local (unitary)  Clifford circuit~$\cC$ on $N$~qubits, and let 
\begin{align}
\upperboundL&\geq |\mathsf{Loc}(\cC)|
\end{align}
be an upper bound on the number of circuit locations of~$\cC$. 
 Then there is an adaptive Clifford circuit $\cC^{\mathsf{1D}}$
which implements the identical functionality as~$\cC$ except with probability~$\varepsilon$ in the presence of local stochastic Pauli noise of any strength~$p\leq p^{\mathsf{1D}}_0(N,\varepsilon)$.
The circuit $\cC^{\mathsf{1D}}$ is 1D-local, 
and the $N_{\mathsf{in}}$ input qubits of $\cC^{\mathsf{1D}}$ are located on sites $Q_1,\dots,Q_{N_{\mathsf{in}}}$ which have pairwise distance $\Theta\left(\mathsf{poly}\left(\log 1/\varepsilon, \log\upperboundL\right)\right)$ on the path graph. The $N_{\mathsf{out}}$ output qubits are similarly separated. The circuit~$\cC^{\mathsf{1D}}$ has a number of qubits and depth which satisfy 
\begin{align}     N(\cC^{\mathsf{1D}})&= 
 N(\cC)\cdot \Theta\left(\mathsf{poly}\left(\log 1/\varepsilon, \log\upperboundL\right)\right)\\
       \mathsf{depth}(\cC^{\mathsf{1D}})&= 
  \mathsf{depth}(\cC)\cdot \Theta\left(\mathsf{poly}\left(\log 1/\varepsilon, \log\upperboundL\right)\right)\ .
\end{align}
\end{theorem}
Setting $\upperboundL:=|\mathsf{Loc}(\cC)|$ 
and using that $|\mathsf{Loc}(\cC)|=\Theta(N(\cC)\cdot\mathsf{depth}(\cC))$ for any $1D$-local circuit~$\cC$ 
we have 
\begin{align}
\Theta\left(\mathsf{poly}\left(\log 1/\varepsilon, \log\upperboundL\right)\right)=\Theta\left(\mathsf{poly}\left(\log 1/\varepsilon, \log N(\cC),\log \mathsf{depth}(\cC)\right)\right)\ .
\end{align}
In particular, Theorem~\ref{thm:implementation1D} shows that the qubit and depth overhead of this construction is polylogarithmic (see  Result~\ref{res:resultft1dlocal}).
\begin{proof}
This is an immediate consequence of Lemma~\ref{lem:implementation1D} together with a suitable choice of~$L$: Let us write $p_0=p^{\mathsf{1D}}_0(N,\varepsilon)$ for brevity. 
Using Eq.~\eqref{eq:upperboundfrobust1D} 
we conclude that
\begin{align}
f_L^{\mathsf{1D}}(p_0)&\leq \varepsilon/2+p_* |\mathsf{Loc}(\cC)|\cdot \left(1/2\right)^{2^L}\ .
\end{align}
The remainder of the proof is analogous to that of 
Theorem~\ref{thm:robustimplementationstochastic}.
\end{proof}

\section{Fault tolerance with mid-circuit qubit resets}\label{sec:qubitresets}
In the following, we substitute Clifford gates  in the circuit~$\cC^{\mathsf{1D}}$ by gate teleportation circuits, see Section~\ref{sec:gateteleportationsubstitution}. Postponing Pauli corrections to the end of the circuit, we obtain a circuit~$\cC^{\mathsf{1D,res}}$ which robustly implements the original circuit~$\cC$ in~$1D$, see Section~\ref{sec:robustimplementationconcatenated}.
The circuit~$\cC^{\mathsf{1D,res}}$ has the additional property that each qubit reinitialized (i.e., reset to the state~$\ket{0}$) within a constant time (circuit depth). It combines both elements of the circuit model of computation (where Choi-Jamiolkowski states are prepared by regular circuits), as well as measurement-based computation (where operations are realized by  measurements).

\subsection{Gate teleportation substitution\label{sec:gateteleportationsubstitution}}
A  building block we use repeatedly is the fact that a one- or two-qubit Clifford unitary~$C$ can be implemented~\cite{gottesmanchuanggateteleport} by a
depth-$10$ adaptive Clifford circuit, the gate teleportation circuit shown in Fig.~\ref{fig:gateteleportationcircuit}. The latter first creates 
the Choi-Jamiolkowski state of $C$ (a $2$- respectively $4$-qubit state) and then performs two Bell measurements. 
The resulting state is the input state with the Clifford applied up to a Pauli correction determined by the measurement result.

In the following, we consider transforming a circuit by replacing one- and two-qubit Clifford gates by the teleportation circuits in Fig.~\ref{fig:gateteleportationtwoqubit}. For simplicity (to keep the substitution compatible with the remainder of the circuit), we replace suitable subcircuits of the right depth. Concretely, we assume that we are replacing a subcircuit acting on one or two qubits, which consists in the application of a one- or two-qubit Clifford gate and 9 ``idle'' steps (wait locations). 

In more detail, let $\cC=\cC_D\circ \cdots \circ\cC_1$ be a depth-$D$ Clifford circuit on $n$~qubits. 
Let $C$ be a single- or two-qubit Clifford unitary. 
A rectangle~$\cR=\cR(\Omega,t,\Delta)\subset \cW_{\cC}$
is called a valid Clifford subrectangle or simply a $C$-rectangle if the following holds: its depth is $\Delta=10$, $|\Omega|\in \{1,2\}$ depending on whether~$C$ is a one- or two-qubit gate, and 
\begin{align}
\cC_t\cap \cR=C\qquad\textrm{ and }\qquad \cC_{t+j}\cap \cR=\mathsf{id}\textrm{ for }j\in \{2,\ldots,10\}\ .
\end{align}

\begin{figure}
\centering
\begin{subfigure}{0.95\textwidth}
\centering
\includegraphics[width=0.9\textwidth]{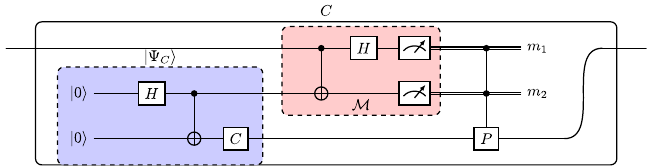}
\caption{The single-qubit gate teleportation circuit~$\cC_{\mathsf{tele}}(C)$  implementing a single-qubit Clifford unitary~$C$ \label{fig:gateteleportationsinglequbit}}
\end{subfigure}

\begin{subfigure}{0.95\textwidth}
\centering
\includegraphics[width=0.9\textwidth]{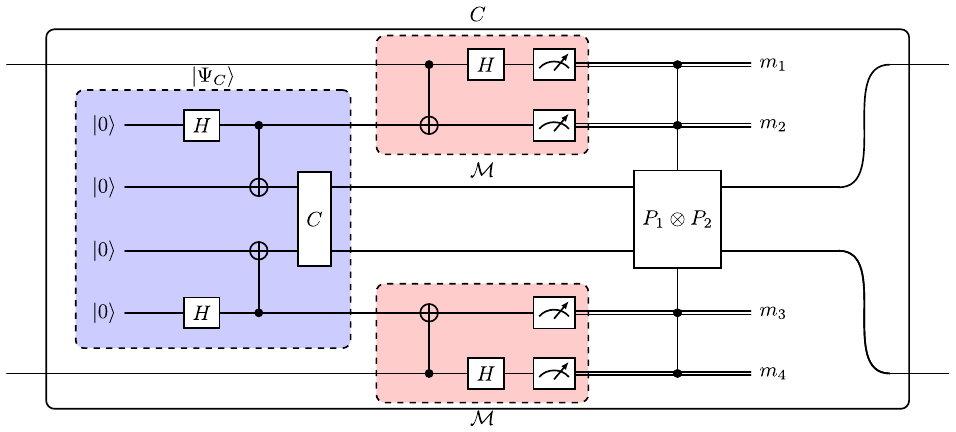}
\caption{The gate teleportation circuit~$\cC_{\mathsf{tele}}(C)$  implementing a two-qubit Clifford unitary~$C$ \label{fig:gateteleportationtwoqubit}}
\end{subfigure}
\caption{One- and two-qubit gate teleportation circuits implementing a single- respectively two-qubit Clifford unitary. We assume that the measured qubits are reinitialized in the state~$\ket{0}$ after measurement (not shown). This implies that these circuits use 
adaptive read-once Pauli corrections followed by two layers of $1D$-nearest-neighbor SWAP gates. Both circuits have depth~$10$. 
\label{fig:gateteleportationcircuit}}
\end{figure}

In this situation, we can replace the subcircuit~$\cC\cap \cR$ with the adaptive gate teleportation circuit~$\cC_{\mathsf{tele}}(C)$ in Fig.~\ref{fig:gateteleportationtwoqubit}. We note because the latter uses two respectively four auxiliary qubits, this creates a circuit~$\tilde{\cC}$ on $n+2$ or $n+4$~qubits.  However, since $\cC_{\mathsf{tele}}(C)$ discards (traces out) these auxiliary qubits, the resulting circuit~$\tilde{\cC}$ still only has an $n$-qubit output.  The auxiliary qubits are only used ``locally'' in space and time. In this sense, the derived circuit~$\tilde{\cC}$ essentially has the same input and output as~$\cC$, and statements such as~$\cC=\tilde{\cC}$ make sense.

The following lemma translates between the noise before and after applying several substitutions of this kind.
It is a consequence of the circuit error cleaning lemma for adaptive read-once rectangles (Lemma~\ref{lem:cleaninglemmaadaptive})
and the subcircuit substitution lemma for unitary subcircuits (Lemma~\ref{lem:subcircuitsubstitution}). 
\begin{lemma}[Gate teleportation substitution]\label{lem:gateteleportationsubstitution}
Let $\cC$ be a depth-$D$ circuit. Let $\cR_1,\ldots,\cR_m$ be pairwise disjoint valid  Clifford subrectangles, each realizing a one- or two-qubit Clifford~$C_j$, $j\in [m]$. 
Let $\tilde{\cC}$ be the circuit obtained by replacing the subcircuit~$\cC\cap \cR_j$ by
$\cC_{\mathsf{tele}}(C_j)$ for each $j\in [m]$ (see Fig.~\ref{fig:gateteleportationtwoqubit}). Then $\tilde{\cC} \gtrsim \cC$.
\end{lemma}
\begin{figure}
\centering
\includegraphics[width=0.8\textwidth]{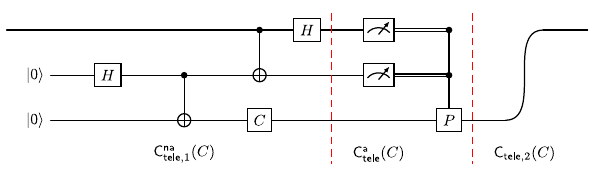} 
\caption{A convenient factorization of the one-qubit teleportation circuit. The qubits after the measurement are not shown. The last subcircuit~$\cC_{\mathsf{tele},2}$ applies (1D-)nearest neighbor SWAP gates in succession. 
~\label{fig:onequbitfactored}}
\end{figure}
\begin{proof}
Consider a  one-qubit Clifford unitary~$C$. (The argument for two-qubit Cliffords is analogous.) The circuit~$\cC_{\mathsf{tele}}(C)$ can be factored as depth-$10$ circuit as
\begin{align}
\cC_{\mathsf{tele}}(C)&=\cC^{\mathsf{na}}_{\mathsf{tele,2}}(C)\circ \cC^{\mathsf{a}}_{\mathsf{tele}}(C) \circ\cC^{\mathsf{na}}_{\mathsf{tele,1}}(C)\ ,
\end{align}
see Fig.~\ref{fig:onequbitfactored}.
Here $\cC^{\mathsf{na}}_{\mathsf{teleport,1}}(C)$ is a (non-adaptive) depth-$6$ circuit which prepares~$\ket{0}^{\otimes 2}$
and applies a unitary to it,  $\cC^{\mathsf{a}}_{\mathsf{teleport}}(C)$ is a depth-$2$ adaptive circuit which measures qubits in the computational basis and  then applies a Pauli operation depending on the measurement result. Finally,
$\cC^{\mathsf{na}}_{\mathsf{teleport,2}}(C)$ is a non-adaptive depth-$2$ circuit performing two SWAP gates. 

Correspondingly, we can partition each rectangle~$\cR_j=\cR_j(\Omega_j,t_j,10)$ in the derived circuit~$\tilde{\cC}$ as
\begin{align}
\cR_j&=\cR^{\mathsf{na},2}_j\cup \cR^{\mathsf{a}}_{j}\cup \cR^{\mathsf{na},1}_j\ ,
\end{align}
where for each $j\in [m]$, the rectangle $\cR^{\mathsf{a}}_{j}$ is an adaptive, read-once, linear rectangle, and where the rectangles share the wires
\begin{align}
\partial_{r}\cR^{\mathsf{na},1}&=\partial_\ell \cR^{\mathsf{a}}_j\label{eq:sharedwiresonev}\\
\partial_r\cR^{\mathsf{a}}_j&=\partial_\ell \cR^{\mathsf{na},2}_j\ .
\end{align}
Let $\tilde{F}\sim \cN^{\pauli}_{\tilde{C}}(
\tilde{p})$. By the cleaning Lemma~\ref{lem:cleaninglemma}
applied to the rectangles~$\{\cR^{\mathsf{na},1}_j\}_{j\in [m]}$, there is a local stochastic error~$\tilde{F}'\sim \cN^{\pauli}_{\tilde{C}}(\tilde{p}')$
of strength
\begin{align}
\tilde{p}'&=\Lambda'\tilde{p}^{\lambda'}\label{eq:tildepprimepequal}
\end{align}
for some constants $\Lambda',\lambda'>0$ 
such that
\begin{align}
\tilde{F}\bowtie\tilde{\cC}&=\tilde{F}'\bowtie\tilde{\cC}\ ,
\end{align}
and
\begin{align}
\supp(\tilde{F}')\cap \left(\partial_\ell \cR^{\mathsf{na}}_j\cup \mathsf{Int}(\cR^{\mathsf{na}}_j)\right)&=\emptyset\qquad\textrm{ for every }\qquad j\in [m]\ .
\end{align}
We next apply the cleaning Lemma~\ref{lem:cleaninglemmaadaptive}
to the local stochastic Pauli noise~$\tilde{F}'$ and the adaptive rectangles~$\cR^{\mathsf{a}}_{j}$, $j\in [m]$.
This results in local stochastic Pauli noise~$\tilde{F}''\sim \cN^{\pauli}_{\tilde{C}}(\tilde{p}')$ of strength
\begin{align}
\tilde{p}''&=5 (U\tilde{p}')^{1/(5V)}\ ,  
\label{eq:tibp8}
\end{align}
where for
the adaptive part of the one-qubit gate teleportation circuit~$\cC_{\mathsf{tele}}(C)$ we have $U=2$ and $V=2$. (If $C$ is a two-qubit Clifford, then $U=4$ and $V=2$.)
The noise~$\tilde{F}''$ has the property that
\begin{align}
\tilde{F}''\bowtie\tilde{\cC}&= \tilde{F'}\bowtie\tilde{\cC}
\end{align}
and
\begin{align}
\supp(\tilde{F}'')\cap
\left(\cR^{\mathsf{na}}_j \cup \mathsf{Int}(\cR^{\mathsf{a}}_j)\right)
&=\emptyset\ 
\end{align}
for every $j\in [m]$. 
Here we used Eq.~\eqref{eq:sharedwiresonev}.

We next apply the  cleaning Lemma~\ref{lem:cleaninglemma}
 to the rectangles~$\{\cR^{\mathsf{na},1}_j\}_{j\in [m]}$, obtaining local stochastic Pauli noise~$F\sim \cN^{\pauli}_{\tilde{C}}(p)$ of strength
 \begin{align}
 p &=\Lambda_2 (\tilde{p}'')^{\lambda_2}
 \end{align}
for some constants $\Lambda_2,\lambda_2>0$,
\begin{align}
F\bowtie\tilde{\cC}&=\tilde{F}''\bowtie\tilde{\cC}
\end{align}
and
\begin{align}
\supp(F)\cap \left(\partial_\ell \cR_j \cup \mathsf{Int}(\cR_j)\right)&=\emptyset\qquad\textrm{ for every }\qquad j\in [m]\ .
\end{align}
 In particular,~$F$ has the property that it does not act inside any Clifford teleportation circuit. Since  $\cC_{\mathsf{tele}}(C)=C$ for any two-qubit Clifford~$C$, this implies that 
\begin{align}
F\bowtie \tilde{\cC}&=F\bowtie{\cC}\ .
\end{align}
This implies the claim.
\end{proof}

\subsection{Robust implementations with mid-circuit resets}\label{sec:implementationresets}
\begin{figure}
\centering
\begin{subfigure}{0.45\textwidth}
\centering
\includegraphics[width=8cm]{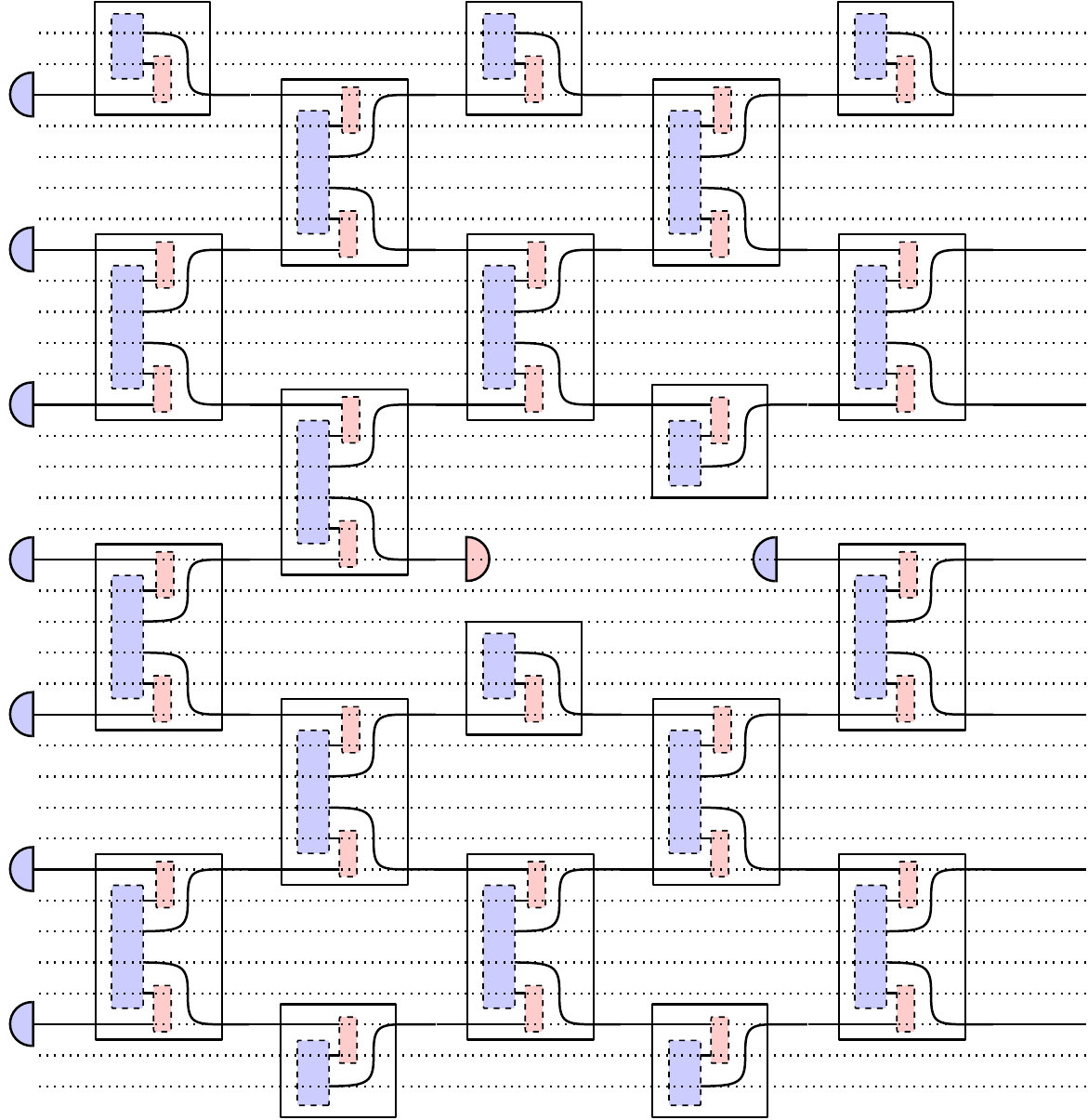}
\caption{The circuit $\cC^\mathsf{1D}_{\mathsf{post}}$\label{fig:qubitsdotted}}
\end{subfigure}
\caption{Construction of the circuit~$\cC^\mathsf{1D}_{\mathsf{post}}$ starting from a fault-tolerant $1D$-local implementation of a circuit~$\cC$.
The key building block are two- respectively one-qubit 
teleportation subcircuits, see Fig.~\ref{fig:twoqubitgateteleportationomitted}: Each blue structure corresponds to state preparation or a state preparation subcircuit.
Blue boxes correspond to Bell measurements. The lifespan of each qubit in this circuit is bounded by a constant.
The layer of adaptive Paulis at the end of the circuits are not shown, and SWAP-gates are only drawn for ``active'' qubits.
The qubits used are illustrated by dotted lines.
\label{fig:statepreparationcircuittime}}
\end{figure}

\begin{figure}
\centering
\includegraphics[width=0.6\textwidth]{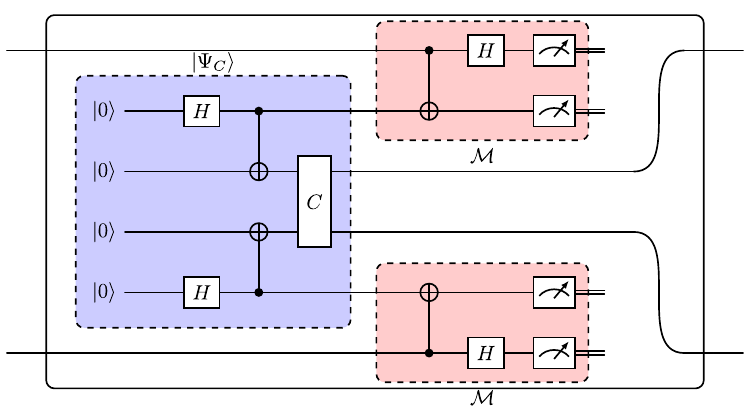}
\caption{The two-qubit gate teleportation circuit from Fig.~\ref{fig:gateteleportationtwoqubit} with the Pauli correction omitted. An analogous circuit is  used for one-qubit teleportation.\label{fig:twoqubitgateteleportationomitted}}
\end{figure}

We will now use the circuit $\cC^{\mathsf{1D}}$ to construct a new $1D$-local circuit $\cC^{\mathsf{1D}}_{\mathsf{post}}$, illustrated in Fig.~\ref{fig:statepreparationcircuittime}, which has the property that each qubit is only active for at most a constant amount of time. That is, the number of operation layers (time steps) between the preparation of a qubit in a computational basis state and subsequent measurement of the qubit is constant. We call the maximum number of operation layers between the preparation and measurement of any qubit the qubit lifespan of the circuit.

\begin{theorem}[Robust $1D$-implementation with mid-circuit resets]
\label{thm:1Dresets}
Let $N\in\mathbb{N}$ and  $\varepsilon>0$ be arbitrary. 
Let $\mathsf{g}\geq 49$, $p_\ast > 0$
be the constants used in Lemma \ref{lem:bilinearimplementation}.
Then there are constants~$\Lambda_{\mathsf{1D,res}},\lambda_{\mathsf{1D,res}}>0$ such that the following holds. Define 
\begin{align}
p_0^{\mathsf{1D,res}}(N,\varepsilon)&:=
 \min\left\{
\left(\frac{p_\ast}{2\Lambda_{\mathsf{1D,res}}}\right)^{1/\lambda_{\mathsf{1D,res}}
}, \left(\frac{\varepsilon}{8N\mathsf{g}\Lambda_{\mathsf{1D,res}}}\right)^{1/\lambda_{\mathsf{1D,res}}}\right\}\ . \label{eq:pzerooneDresdef}
\end{align}
Let $\cC$ be a $1D$-local (unitary) Clifford circuit~$\cC$ on $N$~qubits, and let
\begin{align}
\upperboundL&\geq |\mathsf{Loc}(\cC)|
\end{align}
be an upper bound on the number of circuit locations of~$\cC$. 
Then there is an adaptive Clifford circuit $\cC^{\mathsf{1D,res}}$ which  implements the same functionality as~$\cC$  except with probability~$\varepsilon$ in the presence of local stochastic Pauli noise of any strength~$p\leq p^{\mathsf{1D,res}}_0(N,\varepsilon)$. Furthermore, 
\begin{enumerate}[(i)]
    \item The circuit $\cC^{\mathsf{1D,res}}$ is non-adaptive except for its last operation layer, which consists of adaptive Pauli corrections.\label{it:nonadaptiveexceptlast}
    \item The qubit lifespan of $\cC^{\mathsf{1D,res}}$ is constant.\label{it:qubitlifespanconstant}
    \item The number of qubits and the depth of the circuit satisfy
    \begin{align}
        N(\cC^{\mathsf{1D,res}})& = N(\cC)\cdot \Theta\left(\mathsf{poly}(\log 1/\varepsilon, \log \upperboundL)\right)\\
       \mathsf{depth}(\cC^{\mathsf{1D,res}})& = \mathsf{depth}(\cC)\cdot \Theta\left(\mathsf{poly}(\log 1/\varepsilon, \log \upperboundL)\right)\ .
\end{align}
\item  The circuit $\cC^{\mathsf{1D,res}}$ is $1D$-local. The $N_{\mathsf{in}}$ input qubits of $\cC^{\mathsf{1D}}$ are located on sites $Q_1,\dots,Q_{N_{\mathsf{in}}}$ which have pairwise distance $\Theta\left(\mathsf{poly}\left(\log 1/\varepsilon, \log\upperboundL\right)\right)$ on the path graph. The $N_{\mathsf{out}}$ output qubits are similarly separated.
\end{enumerate}
\end{theorem}
\begin{proof}
Let $\cC^{\mathsf{1D}}$ be the 
circuit constructed in Theorem~\ref{thm:implementation1D}.
 The circuit~$\cC^\mathsf{1D}$ involves adaptive Pauli operations (e.g., in the error correction gadget).
We will apply a series of circuit transformations to give a sequence of circuits satisfying
\begin{align}
\cC^\mathsf{1D,res}\gtrsim\cC^\mathsf{1D,mb}\gtrsim \cC^\mathsf{1D,inflated} \gtrsim \cC^\mathsf{1D,post} \gtrsim\cC^\mathsf{1D}\ .\label{eq:successivecircuittransformationsconcatenated}
\end{align}
In more detail, these are constructed as follows:
\begin{description}
\item[$\cC^\mathsf{1D,post}$: ]
We first consider the derived circuit~$\cC^\mathsf{1D}_{\mathsf{post}}$ obtained by applying Lemma~\ref{lem:adaptivenonadaptiveClifford} to~$\cC^\mathsf{1D}$. 
The circuit~$\cC^\mathsf{1D,post}$ is non-adaptive except for its last layer. The latter consists of adaptive Pauli operations only, i.e., all corrections are postponed to the end of the circuit.

\item[$\cC^\mathsf{1D,inflated}$:]
We apply a depth-$9$ inflation move (see Lemma~\ref{lem:circuitinflation}) to $\cC^\mathsf{1D,post}$, obtaining an inflated circuit~$\cC^\mathsf{1D,inflated}$ whose (non-trivial) operation layers are sufficiently ``spread'' in time for the following step.

\item[$\cC^\mathsf{1D,mb}$:]
We transform 
the circuit~$\cC^\mathsf{1D,inflated}$  
into a ``measurement-based'' circuit~$\cC^\mathsf{1D,mb}$ using Lemma~\ref{lem:gateteleportationsubstitution}. 
That is, each 1- or 2-qubit Clifford unitary in every $10$th layer (including identity gates) is replaced by either a 1- or 2-qubit gate teleportation circuit. The circuit~$\cC^\mathsf{1D,mb}$ is  $1D$-local and adaptive.

\item[$\cC^\mathsf{1D,res}$:]
Again using Lemma~\ref{lem:adaptivenonadaptiveClifford}, 
we  transform the circuit~$\cC^\mathsf{1D,mb}$ into
a circuit $\cC^\mathsf{1D,res}$
with adaptivity only in its last layer.
\end{description}
We note that Properties~\eqref{it:nonadaptiveexceptlast} and~\eqref{it:qubitlifespanconstant}
of the circuit~$\cC^{\mathsf{1D,res}}$
 follow by construction.

By Lemma~\ref{lem:adaptivenonadaptiveClifford}, Lemma~\ref{lem:circuitinflation} and Lemma~\ref{lem:gateteleportationsubstitution}, 
we can obtain Eq.~\eqref{eq:successivecircuittransformationsconcatenated}. We conclude that $\cC^{\mathsf{1D,res}} \gtrsim \cC^{\mathsf{1D}}$, i.e., there are constants~$(\Lambda,\lambda)$ such that 
a noise execution~$\cE\bowtie \cC^{\mathsf{1D,res}}$ of $\cC^{\mathsf{1D,res}}$
with local stochastic Pauli noise~$\cE$ of strength~$p$ is equivalent to
a noisy~$\cE'\bowtie\cC^{\mathsf{1D}}$ of $\cC^{\mathsf{1D}}$ with local stochastic noise~$\cE'$ of strength  $p'=\Lambda p^\lambda$. Defining
\begin{align}
p_0^{\mathsf{1D,res}}(N,\varepsilon):=
(p_0^{\mathsf{1D}}(N,\varepsilon)/\Lambda)^{1/\lambda}\ ,
\end{align}
where $p_0^{\mathsf{1D}}(N,\varepsilon)$ is given Eq.~\eqref{eq:pzerooneddef},
it thus follows from Theorem~\ref{thm:implementation1D} that $\cE\bowtie \cC^{\mathsf{1D,res}}$ has identical action as $\cC$
except with probability at most~$\varepsilon$ for local stochastic Pauli noise~$\cE$ of strength~$p\leq p_0^{\mathsf{1D,res}}(N,\varepsilon)$. 

Recalling the definition of $p_0^{\mathsf{1D}}(N,\varepsilon)$ (see Eq.~\eqref{eq:pzerooneddef}), we obtain 
the expression
\begin{align}
p_0^{\mathsf{1D,res}}(N,\varepsilon)&:=
\frac{1}{\Lambda^{1/\lambda}}
\cdot \min\left\{
\left(\frac{p_\ast}{2\Lambda_{\mathsf{1D}}}\right)^{1/(\lambda_{\mathsf{1D}}
\lambda)}, \left(\frac{\varepsilon}{8N\mathsf{g}\Lambda_{\mathsf{1D}}}\right)^{1/(\lambda_{\mathsf{1D}}
\lambda)}\right\}\ .
\label{eq:pzerooneDres}
\end{align}
Setting
\begin{align}
\Lambda_{\mathsf{1D,res}}&:=\Lambda^{\lambda_{\mathsf{1D}}}\Lambda_{\mathsf{1D}}\\
\lambda_{\mathsf{1D,res}} &:=\lambda_{\mathsf{1D}}\lambda
\end{align}
it follows from Eq.~\eqref{eq:pzerooneDres} that 
$p_0^{\mathsf{1D,res}}(N,\varepsilon)$ has the form given in Eq.~\eqref{eq:pzerooneDresdef}.

Note that at each stage of circuit transformations the depth of the circuit and number of qubits is at most increased by a constant factor, so we have
\begin{align}
    N(\cC^{\mathsf{1D,res}}) &= O(N(\cC^{1D}))\\
\mathsf{depth}(\cC^{\mathsf{1D,res}}) &= O(\mathsf{depth}(\cC^{\mathsf{1D}}))\ ,
\end{align}
and hence the bounds on the qubit count and circuit depth follow from Theorem~\ref{thm:implementation1D}. A similar observation leads to the claim about the locations of the input and output qubits, following the corresponding claim of Theorem~\ref{thm:implementation1D}.
\end{proof}

\subsection{Robust constant-depth Clifford circuit implementation in 2D from space-time transformation}\label{sec:2dsimulation}
In the following, we consider robust implementations of Clifford circuits which can be realized in constant depth in~$2D$.  We pay particular attention to the physical location of corresponding input and output qubits. 

Our construction is as follows. The circuit $\cC^{\mathsf{1D,res}}$ constructed in Theorem~\ref{thm:1Dresets} immediately lends itself to parallelization by ``unfolding'' time into an additional spatial dimension. This procedure is illustrated in Fig.~\ref{fig:circuitunfolding}. Thanks to the constant lifespan property, this yields a $2D$-local circuit $\cC^{\mathsf{2D}}$ of constant depth.  The latter prepares a tensor product of Choi states which are then measured in the Bell basis. Importantly, the number of locations of the original circuit $\cC$ gives a bound on the spatial separation between the input and output systems of $\cC^{\mathsf{2D}}$.

\begin{theorem}[Robust constant-depth $2D$-implementation]
\label{thm:2Dimplementation}
Let $N\in\mathbb{N}$ and  $\varepsilon>0$ be arbitrary. Let $\cC$ be a $1D$-local (unitary) Clifford circuit~$\cC$ on $N$~qubits. 
Let 
\begin{align}
\upperboundL&\geq |\mathsf{Loc}(\cC)|\ 
\end{align}
be an upper bound on the number of circuit locations of~$\cC$. 
Then there is an adaptive Clifford circuit $\cC^{\mathsf{2D}}$ which  implements the same functionality as~$\cC$  except with probability~$\varepsilon$ in the presence of local stochastic Pauli noise of any strength~$p\leq p^{\mathsf{1D,res}}_0(N,\varepsilon)$, where the function $p^{\mathsf{1D,res}}_0$ is defined in Eq.~\eqref{eq:pzerooneDresdef}.  Furthermore, 
\begin{enumerate}[(i)]
   \item $\cC^{\mathsf{2D}}$ is $2D$-local, i.e., local with respect to the grid graph $G_{\ell_X,\ell_Y} := \mathsf{P}_{\ell_X}\square \mathsf{P}_{\ell_Y}$, where
        \begin{align}
            \ell_X &=N(\cC)\cdot\Theta \left(\mathsf{poly}(\log1/\varepsilon,\log \upperboundL)\right)\label{eq:ellXdist} \\
\ell_Y&=\mathsf{depth}(\cC)\cdot \Theta \left(\mathsf{poly}(\log1/\varepsilon,\log\upperboundL)\right)\ .\label{eq:ellYdist}
        \end{align}
        \item $\cC^{\mathsf{2D}}$ is of constant depth $\mathsf{depth}(\cC^{\mathsf{2D}}) = 7$, and is non-adaptive except for its last operation layer, which consists of adaptive Pauli corrections.
        \item The $N_{\mathsf{in}}$ input qubits of $\cC^{\mathsf{2D}}$ are  located in the slice $[\ell_X]\times \{1\}$ of $G_{\ell_X,\ell_Y}$ and have pairwise distance 
\begin{align}
\mathsf{dist}&=\Theta\left(\mathsf{poly}(\log1/\varepsilon,\log\upperboundL)\right)\ .\label{eq:distd}
\end{align} The $N_{\mathsf{out}}$ qubit outputs are located in the slice $[\ell_X]\times \{\ell_Y\}$ and also have pairwise distance~$\mathsf{dist}$.
\end{enumerate}
\end{theorem}

\begin{figure}
\centering
\begin{subfigure}{0.35\textwidth}
\centering
\includegraphics[width=0.95\textwidth]{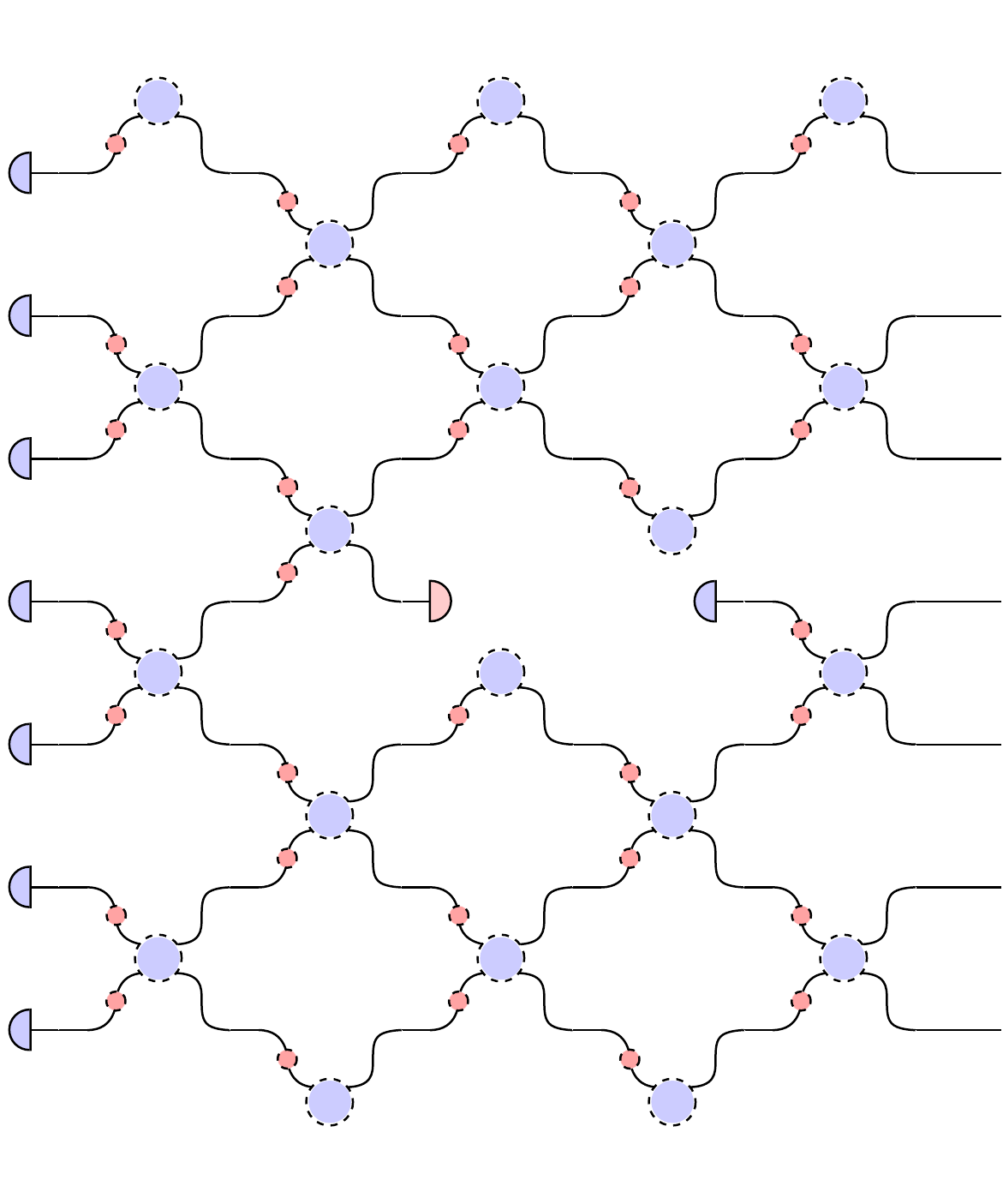}
\caption{Deformed circuit: We present each Choi state~$\ket{\Psi_C}$ by a circle.
Two-qubit Bell measurements now have two ingoing lines. This diagram is only used for illustrative purposes and the time ordering (left to right) is not strict. }
\end{subfigure}\qquad 
\begin{subfigure}{0.55\textwidth}
\centering
\includegraphics[width=0.95\textwidth]{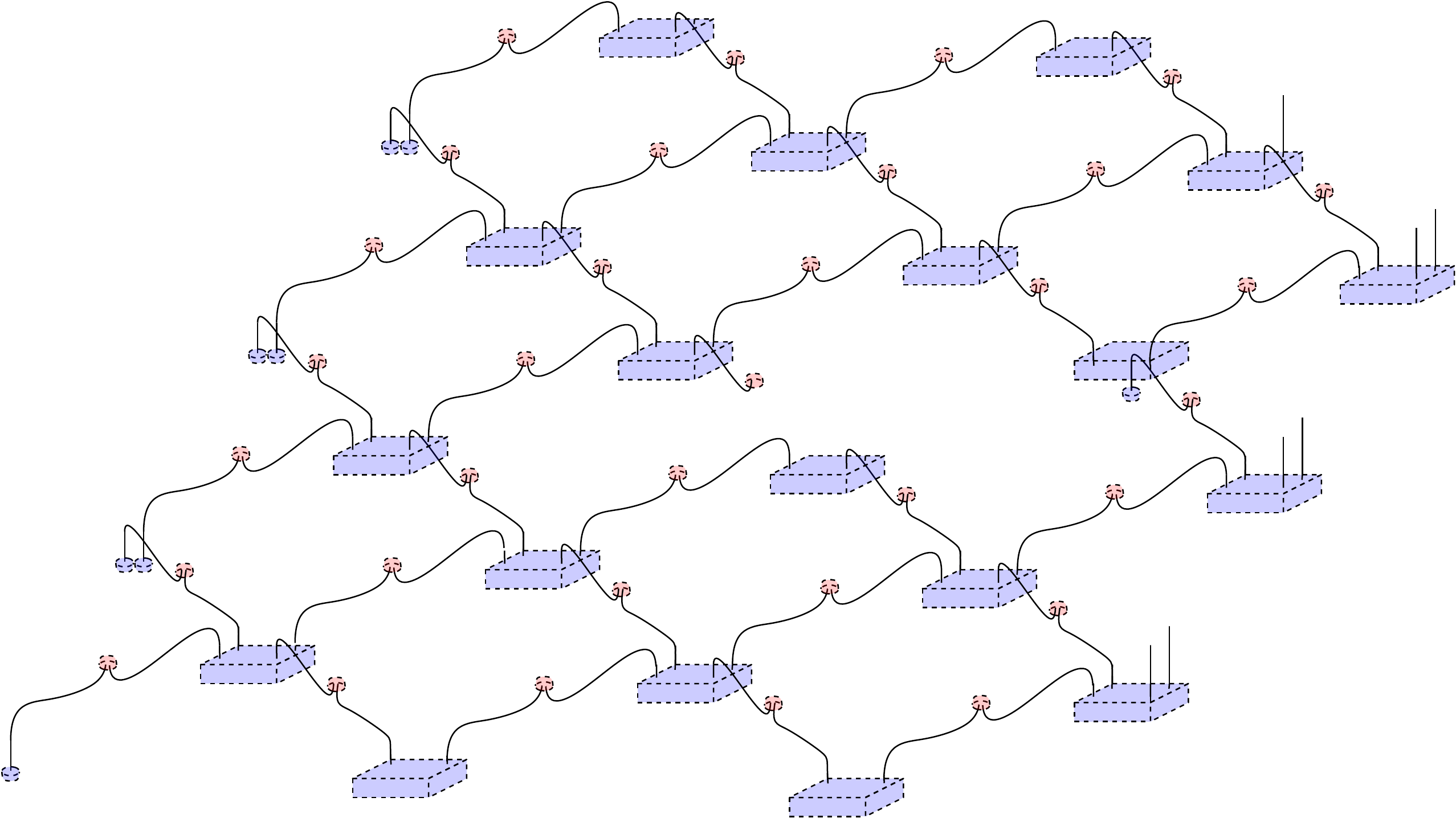}
\caption{The associated circuit~$\cC^{\mathsf{2D}}$ prepares single-qubit computational basis states and Choi states (blue) and applies single-qubit computational basis and Bell (two-qubit) measurements (red).  
Overall, the effect of the circuit is to measure the ``bulk'' qubits of  a $2D$-short-range entangled state in the computational basis. The post-measurement state of the qubits associated with (a subset of) dangling edges is the desired target state~$\ket{\Psi}$ up to a Pauli correction (efficiently) computable from the measurement outcomes.}
\end{subfigure}
\caption{Illustration of the ``unfolding'' of the $1D$-local circuit
$\cC^\mathsf{1D}_{\mathsf{post}}$ in space, resulting a constant-depth state preparation circuit~$\cC^{\mathsf{2D}}$. 
For concreteness, we illustrate the construction for an ideal circuit~$\cC$ with $N_{in}=0$ and $N_{out}=N$, i.e., the circuit starts with the state~$\ket{0^N}$ and applies a unitary $U$ to prepare the state~$\ket{\Psi}=U\ket{0^N}$. 
Residual (adaptive) Paulis depending on all measurement results at the end of the circuit are not shown.
The circuit~$\cC^{\mathsf{2D}}$ essentially prepares a $2D$-short range entangled state (by preparing a tensor product of Choi-states and applying ``Bell'' disentangling unitaries) and subsequently measures all ``bulk'' qubits in the computational basis. 
\label{fig:circuitunfolding}}
\end{figure}

\begin{proof}
We use the circuit $\cC^{\mathsf{1D,res}}$ constructed in Theorem~\ref{thm:1Dresets}, see Fig.~\ref{fig:statepreparationcircuittime}. The qubit lifespan of this circuit is constant, but qubits are reinitialized and reused. We construct the  circuit~$\cC^{\mathsf{2D}}$ by introducing a new qubit wherever a qubit is (re)initialized in the original circuit~$\cC^{\mathsf{1D,res}}$. This obviates the need for the two layers of swap gates in the teleportation circuit (see Fig.~\ref{fig:twoqubitgateteleportationomitted}), leading to a constant circuit depth of $7$.
Clearly, the constructed circuit~$\cC^{\mathsf{2D}}$ has identical action as~$\cC^{\mathsf{1D,res}}$. Furthermore, this is the case even in the presence of local stochastic Pauli noise as 
the set of circuit locations~$\mathsf{Loc}(\cC^{\mathsf{2D}})$ is a strict subset of those in the set~$\mathsf{Loc}(\cC^{\mathsf{1D,res}})$.

As illustrated in Fig.~\ref{fig:circuitunfolding}, the circuit~$\cC^{\mathsf{2D}}$ is $2D$-local when the qubits are appropriately laid out in the plane. That is, using that the circuit $\cC^{\mathsf{1D,res}}$ is $1D$-local, the qubits can be arranged on a rectangular lattice of side lengths $\ell_X,\ell_Y$, where
    \begin{align}
        \ell_X &= N(\cC^{\mathsf{1D,res}})\\
\ell_Y &= \mathsf{depth}(\cC^{\mathsf{1D,res}})\ .
\end{align}
The claim thus follows from Theorem~\ref{thm:1Dresets}.
\end{proof}

\section{Fault-tolerant one-shot entanglement generation in 2D \label{sec:longdistancegeneration}}

We now specialize the robust constant-depth $2D$-implementation from Theorem~\ref{thm:2Dimplementation} to concrete circuits generating entanglement.
 In Section~\ref{sec:entanglementgenerationlinearlattice}, we show that when this construction is applied to a two-qubit circuit mapping a product state to a Bell state, this gives a scheme for fault-tolerant one-shot entanglement generation in a geometry 
 associated with a grid graph of roughly equal side lengths. In Section~\ref{sec:polyloglattice} we then improve on this result 
 showing that one of the sides can be of only polylogarithmic length in the targeted distance. Finally, in Section~\ref{sec:localizableentanglementfinite}, we reinterpret these findings in terms of robust localizable long-range entanglement of ground states of a local Hamiltonian. We also argue that associated Gibbs states have long-range entanglement below a certain temperature. 

\subsection{Entanglement generation with linear lattice width\label{sec:entanglementgenerationlinearlattice}}

We now apply Theorem~\ref{thm:2Dimplementation}
to obtain fault-tolerant one-shot entanglement generation protocols. An immediate  consequence of Theorem~\ref{thm:2Dimplementation} is the following result. It is obtained by considering the two-qubit circuit~$\cC^{\mathsf{prep}}$ shown in Fig.~\ref{fig:Cprep}. It results in a circuit with an (essentially) square geometry of linear dimensions~$\Theta(R)\times \Theta(R)$ to establish  distance-$R$ entanglement.
We will further improve on this basic result below (see Theorem~\ref{thm:2Dentanglementgeneration}). 
\begin{lemma}[Robust single-shot entanglement generation on a square grid]\label{lem:squaregridstategeneration}
The following
 holds for any~$\varepsilon>0$ and~$R\geq 3$.
There is an adaptive Clifford circuit~$\cC^{\mathsf{prep}}_{\mathsf{robust}}$
with the following properties:
\begin{enumerate}[(i)]
\item
The circuit~$\cC^{\mathsf{prep}}_{\mathsf{robust}}$ is $2D$-local on the grid graph~$G_{\ell_X,\ell_Y} := \mathsf{P}_{\ell_X}\square \mathsf{P}_{\ell_Y}$, where 
\begin{align}
\ell_X&=\Theta(R)\\
\ell_Y&=\Theta(R)\ .
\end{align}
    \item $\cC^{\mathsf{prep}}_{\mathsf{robust}}$ is of constant depth and is non-adaptive except for its last operation layer, which consists of adaptive Pauli corrections.
    \item $\cC^{\mathsf{prep}}_{\mathsf{robust}}$ has no input qubits and two output qubits $Q_1,Q_2$ separated by distance~$\Theta(R)$.
    \item Under local stochastic Pauli noise of strength~
    \begin{align}
    p\leq p^{\mathsf{1D,res}}_0(2,\varepsilon)
    \end{align} where the function $p^{\mathsf{1D,res}}_0$ is defined in Eq.~\eqref{eq:pzerooneDresdef}, the circuit $\cC^{\mathsf{prep}}_{\mathsf{robust}}$ prepares a maximally entangled state $\ket{\Phi_{Q_1Q_2}}$ between the two distant output qubits, except with probability~$\varepsilon$. 
\end{enumerate}
\end{lemma} 
\begin{proof}
We apply Theorem~\ref{thm:2Dimplementation} and the remark following it to the circuit~$\cC=\cC^{\mathsf{prep}}$. 
We note that  the number of qubits and the circuit depth of~$\cC^{\mathsf{prep}}$ are
\begin{align}
N(\cC^{\mathsf{prep}})&=2\\
\mathsf{depth}(\cC^{\mathsf{prep}})&=3\ .
\end{align}
Set $K:=e^R$. Then $K\geq |\mathsf{Loc}(\cC^{\mathsf{prep}})|$. The claim thus follows from Theorem~\ref{thm:2Dimplementation}
and the fact that
\begin{align}
\mathsf{poly}(\log 1/\varepsilon,\log K)=\mathsf{poly}(\log 1/\varepsilon,R)\geq \Theta(R)\ .
\end{align}
\end{proof}

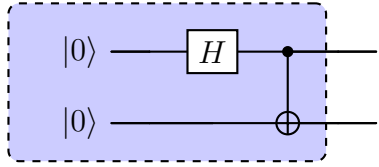
\begin{figure}
    \centering
        $\begin{aligned}
        \begin{quantikz}
\gategroup[wires=2, steps=6, background,
style={dashed, rounded corners, fill=blue!20, inner sep=6pt}]{{}}\setwiretype{n} & &
\lstick{\ket{0}} & \qw\setwiretype{q}& \gate{H} & \ctrl{1} & \qw &\qw \\
\setwiretype{n} & & \lstick{\ket{0}}& \qw\setwiretype{q} & \qw      & \targ{}  & \qw &\qw
\end{quantikz}
    \end{aligned}$
        \caption{The two-qubit circuit $\cC^{\mathsf{prep}}$, which prepares a maximally entangled state on two qubits.   This circuit acts on two qubits, with no input qubits and two output qubits: it acts by preparing both qubits in $\ket{0}$ states, applying a Hadamard gate to one, and then a $CNOT$ between the two.}
        \label{fig:Cprep}
\end{figure}

\subsection{Entanglement generation with polylogarithmic lattice width\label{sec:polyloglattice}}
Lemma~\ref{lem:squaregridstategeneration} 
matches the statement of Theorem~\ref{thm:2Dentanglementgeneration}
except for the fact that both~$\ell_X$ and $\ell_Y$ are of order~$\Theta(R)$.

The following shows that it in fact suffices if~$\ell_X=O(\mathsf{poly}(\log R))$. This is based on a  different construction, where entanglement is generated locally and one half of a Bell state is sent through a robust implementation of the identity circuit.
The statement of Theorem~\ref{thm:2Dentanglementgeneration} is a strengthening of Theorem~\ref{thm:2Dentanglementgeneration} giving an explicit expression for the threshold as a function of the failure probability~$\varepsilon$.

\begin{theorem}[Robust one-shot entanglement generation --- see Theorem~\ref{thm:2Dentanglementgeneration}]\label{thm:2Dentanglementgenerationgeneral}
Let $\varepsilon>0$ be arbitrary. 
Define
\begin{align}\label{eq:pthresepsilon}
p_{\mathsf{thres}}(\varepsilon):=\min\left\{\varepsilon/26,
p_0^{\mathsf{1D,res}}(1,\varepsilon/2)
\right\}\ 
\end{align}
where $p_0^{\mathsf{1D,res}}$ is the function introduced in Eq.~\eqref{eq:pzerooneDresdef}. Then  following holds for any $R\geq 3$. 
There exists an adaptive Clifford circuit $\cC^{\mathsf{Bell}}$ such that
\begin{enumerate}[(i)]
    \item $\cC^{\mathsf{Bell}}$ is $2D$-local, i.e., 
    $G_{\ell_X,\ell_Y}$-local with 
    \begin{align}
        \ell_X &= O(\mathsf{poly}(\log R))\ , \\
        \ell_Y &=\Theta(R)\ .
    \end{align}
    \item $\cC^{\mathsf{Bell}}$ is of constant depth and non-adaptive except for its last operation layer, which consists of adaptive Pauli corrections.
    \item $\cC^{\mathsf{Bell}}$ has no input qubits and two output qubits $Q_1,Q_2$ separated by a distance $\Theta(R)$ on the square lattice.
    \item In the presence of local stochastic noise of strength~$p\leq p_{\mathsf{thres}}(\varepsilon)$, the output of the circuit $\cC^{\mathsf{Bell}}$ is a maximally entangled state $\ket{\Phi_{Q_1Q_2}}$ between the two distant output qubits except with probability~$\varepsilon$. 
\end{enumerate}
\end{theorem}

\begin{figure}
    \centering
\begin{tikzpicture}
    \def\heights{0.6} 
    \def\prepheight{0.3}
    \def\outwireheight{.7}
    \def\midwireheight{.5}
    \def\prepwidth{1.2}
    \def\prepdepth{0.3}
    \def\waitwidth{.3}
    \def\waitdepth{.3}
    \def\mainwidth{3}
    \def\maindepth{5}
    \def\wirebuffer{0.25} 
    \def\labeldist{1.3} 
    \def\blueboxbuffer{0.1};

    \pgfmathsetmacro{\zx}{cos(30)}
    \pgfmathsetmacro{\zy}{sin(30)}

    \coordinate (C2D1) at (0,0);
    \coordinate (C2D2) at ($(C2D1) + (\mainwidth,0)$);
    \coordinate (C2D3) at ($(C2D2) + (\zx*\maindepth,\zy*\maindepth)$);
    \coordinate (C2D4) at ($(C2D3) + (-\mainwidth,0)$);
    \coordinate (C2D1L) at ($(C2D1) + (0,-\heights)$);
    \coordinate (C2D2L) at ($(C2D2) + (0,-\heights)$);
    \coordinate (C2D3L) at ($(C2D3) + (0,-\heights)$);

    \coordinate (C2Doutstart) at ($(C2D4) + (\wirebuffer - \wirebuffer*\zx,-\wirebuffer*\zy)$);
    \coordinate (C2Doutend) at ($(C2Doutstart) + (0,\outwireheight)$);

    \coordinate (wireinend) at ($(C2D1L) + (\wirebuffer + \wirebuffer*\zx,\wirebuffer*\zy)$);
    \coordinate (wireinstart) at ($(wireinend) + (0,-\midwireheight)$);
    
    \coordinate (Cprep3) at ($(wireinstart) + (\wirebuffer + \prepdepth*0.5*\zx,\prepdepth*0.5*\zy)$);
    \coordinate (Cprep2) at ($(Cprep3) + (-\prepdepth*\zx,-\prepdepth*\zy)$);
    \coordinate (Cprep1) at ($(Cprep2) + (-\prepwidth,0)$);
    \coordinate (Cprep4) at ($(Cprep1) + (\prepdepth*\zx,\prepdepth*\zy)$);
    \coordinate (Cprep3L) at ($(Cprep3) + (0,-\prepheight)$);
    \coordinate (Cprep2L) at ($(Cprep2) + (0,-\prepheight)$);
    \coordinate (Cprep1L) at ($(Cprep1) + (0,-\prepheight)$);
    
    \coordinate (wiremidstart) at ($(Cprep1) + (\wirebuffer + \prepdepth*0.5*\zx,\prepdepth*0.5*\zy)$);
    \coordinate (wiremidend) at ($(wiremidstart) + (0,\midwireheight)$);
    
    \coordinate (Cwait1L) at ($(wiremidend) + (-\waitwidth*0.5-\waitdepth*0.5*\zx,-\waitdepth*0.5*\zy)$);
    \coordinate (Cwait2L) at ($(Cwait1L) + (\waitwidth,0)$);
    \coordinate (Cwait3L) at ($(Cwait2L) + (\waitdepth*\zx,\waitdepth*\zy)$);
    \coordinate (Cwait4L) at ($(Cwait3L) + (-\waitwidth,0)$);
    \coordinate (Cwait1) at ($(Cwait1L) + (0,\heights)$);
    \coordinate (Cwait2) at ($(Cwait2L) + (0,\heights)$);
    \coordinate (Cwait3) at ($(Cwait3L) + (0,\heights)$);
    \coordinate (Cwait4) at ($(Cwait4L) + (0,\heights)$);

    \coordinate (wirewaitstart) at ($(wiremidend) + (0,\heights)$);
    \coordinate (wirewaitend) at ($(wirewaitstart) + (0,\outwireheight)$);

    \draw[thick,fill=white] (Cprep1) -- (Cprep2) -- (Cprep3) -- (Cprep4) -- cycle;
    \draw[thick,fill=white] (Cprep1) -- (Cprep1L) -- (Cprep2L) -- (Cprep2) -- cycle;
    \draw[thick] (Cprep2L) -- (Cprep3L) -- (Cprep3);

    \draw[thick] (wireinstart) -- (wireinend);
    \draw[thick] (wiremidstart) -- (wiremidend);

    \draw[thick,fill=white] (Cwait1) -- (Cwait2) -- (Cwait3) -- (Cwait4) -- cycle;
    \draw[thick,fill=white] (Cwait1) -- (Cwait1L) -- (Cwait2L) -- (Cwait2) -- cycle;
    \draw[thick] (Cwait2L) -- (Cwait3L) -- (Cwait3);

    \draw[thick,fill=white] (C2D1) -- (C2D2) -- (C2D3) -- (C2D4) -- cycle;
    \draw[thick,fill=white] (C2D1) -- (C2D1L) -- (C2D2L) -- (C2D2) -- cycle;
    \draw[thick] (C2D2L) -- (C2D3L) -- (C2D3);

    \draw[thick] (C2Doutstart) -- (C2Doutend);
    \draw[thick] (wirewaitstart) -- (wirewaitend);

    \draw[thick,<->] ($(C2D2) + (\labeldist,0)$) -- ($(C2D3) + (\labeldist,0)$);
    \draw[dashed] (C2D2) -- ($(C2D2) + (\labeldist,0)$);
    \draw[dashed] (C2D3) -- ($(C2D3) + (\labeldist,0)$);
    \node[font=\small,below right] at ($(C2D2) + (\labeldist + 0.5*\maindepth*\zx,0.5*\maindepth*\zy)$) {$\ell_Y = \Theta(R)$}; 

    \draw[thick,<->] ($(C2D4) + (\labeldist*\zx,\labeldist*\zy)$) -- ($(C2D3) + (\labeldist*\zx,\labeldist*\zy)$);
    \draw[dashed] (C2D4) -- ($(C2D4) + (\labeldist*\zx,\labeldist*\zy)$);
    \draw[dashed] (C2D3) -- ($(C2D3) + (\labeldist*\zx,\labeldist*\zy)$);
    \node[font=\small,above] at ($(C2D4) + (\labeldist*\zx+0.5*\mainwidth,\labeldist*\zy)$) {$\ell_X = \poly(\log R)$};

    \draw[thick,<->] ($(Cprep1) + (-\labeldist,0)$) -- ($(Cprep1L) + (-\labeldist,0)$);
    \draw[dashed] (Cprep1) -- ($(Cprep1) + (-\labeldist,0)$);
    \draw[dashed] (Cprep1L) -- ($(Cprep1L) + (-\labeldist,0)$);
    \node[font=\small,left] at ($(Cprep1L) + (-\labeldist,\prepheight*0.5)$) {$3$};

    \draw[thick,<->] ($(Cwait1) + (-\labeldist,0)$) -- ($(Cwait1L) + (-\labeldist,0)$);
    \draw[dashed] (Cwait1) -- ($(Cwait1) + (-\labeldist,0)$);
    \draw[dashed] (Cwait1L) -- ($(Cwait1L) + (-\labeldist,0)$);
    \node[font=\small,left] at ($(Cwait1L) + (-\labeldist,\heights*0.5)$) {$7$};
    
    \node[font=\small] at ($(C2D1) + (\mainwidth*0.5 + \maindepth*\zx*0.5,\maindepth*\zy*0.5)$) {$\cC^{\mathsf{2D}}$};
    \node[font=\small,left] at ($(Cwait1L) + (0,\heights*0.5)$) {$\cC^{\mathsf{wait}}$};
    \node[font=\small,right] at ($(Cprep3L) + (0,\prepheight*0.5)$) {$\cC^{\mathsf{prep}}$};

    \coordinate (blueboxbottomleft) at ($(Cprep1L) + (-\blueboxbuffer,-\blueboxbuffer)$);
    \coordinate (blueboxtopright) at ($(Cprep3) + (\blueboxbuffer,\blueboxbuffer)$);
    \draw[fill=blue,fill opacity=0.2,thick,dashed,rounded corners=5] (blueboxbottomleft) rectangle (blueboxtopright);
    
\end{tikzpicture}

    \caption{Construction of the circuit $\cC^{\mathsf{Bell}}$ for Theorem~\ref{thm:2Dentanglementgenerationgeneral}. A maximally entangled state between two qubits is prepared by the circuit $\cC^{\mathsf{prep}}$, see Fig.~\ref{fig:Cprep}. One qubit is left idling for time $O(1)$ by the circuit $\cC^{\mathsf{wait}}$, whilst the other is fault-tolerantly teleported over a distance $\Omega(R)$ by the $O(1)$-depth circuit $\cC^{\mathsf{2D}}$ (represented by a solid block) constructed via Theorem~\ref{thm:2Dimplementation}.}
    \label{fig:Cbell}
\end{figure}
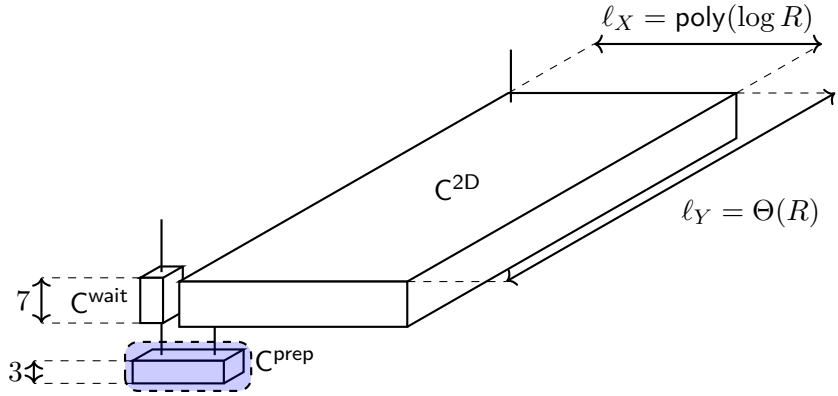

\begin{proof}
    Let $\cC$ define the circuit on $1$ qubit, with $N_{\mathsf{in}} = N_{\mathsf{out}} = 1$, which acts as the identity over time $R$. That is, $\cC$ consists of $R$ identity gates applied sequentially. We apply Theorem~\ref{thm:2Dimplementation} 
    with $\varepsilon/2$ in place of~$\varepsilon$
    to obtain an adaptive Clifford circuit $\cC^{\mathsf{2D}}$ which robustly implements $\cC$, and whose input and output qubits are separated by a distance $\Theta(\mathsf{depth}(\cC)) = \Theta(R)$.
    In particular,     using the robust implementation property of $\cC^{\mathsf{2D}}$ guaranteed by Theorem~\ref{thm:2Dimplementation}
 and with the 
    function $p^{\mathsf{1D,res}}_0$ defined in Eq.~\eqref{eq:pzerooneDresdef},
    this circuit satisfies
    \begin{align}
    \Pr\left[\cE^{\mathsf{2D}}\bowtie \cC^{\mathsf{2D}}\neq \mathsf{id}_{\cB(\mathbb{C}^2)}\right]&\leq \varepsilon/2\ \label{eq:localstochasticnoiseidentityup}
    \end{align}
    for any local stochastic Pauli noise $\cE^{\mathsf{2D}}$ of strength
        \begin{align}
    p\leq p^{\mathsf{1D,res}}_0(1,\varepsilon/2)\ .
    \end{align} 

    To construct $\cC^{\mathsf{Bell}}$, it remains to locally prepare a maximally entangled state between two qubits, and feed one of them into the circuit $\cC^{\mathsf{2D}}$, robustly teleporting it over a distance $\geq R$ (whilst the other qubit is left idling for the constant time it takes for $\cC^{\mathsf{2D}}$ to be executed).    The maximally entangled state of two qubits can be prepared
    by the two-qubit constant-depth circuit~$\cC^{\mathsf{prep}}$ in Fig.~\ref{fig:Cprep}.
    
 We write $\cC^{\mathsf{wait}}$ for the idling circuit applied to the qubit which is not fed into $\cC^{\mathsf{2D}}$. This circuit acts on one qubit, with single-qubit input and output, and simply applies the identity map a total of $\mathsf{depth}(\cC^{\mathsf{2D}})$ times.
 In particular,
 \begin{align}
 \mathsf{depth}(\cC^{\mathsf{wait}})&=\mathsf{depth}(\cC^{\mathsf{2D}})=O(1)\ .
  \end{align}
    The overall circuit $\cC^{\mathsf{Bell}}$ is then defined by
    \begin{align}
        \cC^{\mathsf{Bell}} := (\cC^{\mathsf{2D}}\otimes \cC^{\mathsf{wait}})\circ \cC^{\mathsf{prep}}\ ,
    \end{align}
    as illustrated in Fig.~\ref{fig:Cbell}. It remains to study the effect of local stochastic Pauli noise on this circuit. Let $\cE\sim \cN_{\cC^{\mathsf{Bell}}}^{\pauli}(p)$ be local stochastic noise of strength $p \in [0,1]$. We aim to bound
    \begin{align}
        \prob[\cE\bowtie \cC^{\mathsf{Bell}} \neq \cC^{\mathsf{Bell,ideal}}] \leq f(p)\ ,\label{eq:cbellrobustness}
    \end{align}
    for some function $f : [0,1] \rightarrow \RR$.

Let $\cE$ be local stochastic Pauli noise on~$\cC^{\mathsf{Bell}}$ of strength~$p\leq p_{\mathsf{thres}}(\varepsilon)$.     Let $\cE^{\restriction \mathsf{2D}}$ be the localization of the noise $\cE$ to the subcircuit $\cC^{\mathsf{2D}}$. Note that, by construction the subcircuits $\cC^{\mathsf{wait}}$ and $\cC^{\mathsf{prep}}$ contain a number of locations given by
\begin{align}
    |\mathsf{Loc}(\cC^{\mathsf{wait}})| + |\mathsf{Loc}(\cC^{\mathsf{prep}})| = 7 + 6 = 13 = K\ .
\end{align}
Therefore we can apply Lemma~\ref{lem:localstochasticnoiseqoutput}\eqref{it:circnoisycomposedclaim} to deduce that
    \begin{align}
        \prob\left[\cE\bowtie \cC^{\mathsf{Bell}} \neq ((\cE^{\restriction \mathsf{2D}} \bowtie \cC^{\mathsf{2D}})\otimes \cC^{\mathsf{wait}})\circ \cC^{\mathsf{prep}}\right] \leq Kp\leq \varepsilon/2
    \end{align}
    where we used the assumption that $p\leq p_{\mathsf{thres}}(\varepsilon)$. 
Observe that
    \begin{align}
    (\mathsf{id}_{\cB(\mathbb{C}^2)}\otimes \cC^{\mathsf{wait}})\circ \cC^{\mathsf{prep}}=\cC^{\mathsf{prep}}
    \end{align}
    (as a CPTP map) is equivalent to the (ideal) circuit~$\cC^{\mathsf{prep}}$ preparing the Bell state when considering the output qubits~$Q_1,Q_2$.  Thus 
    Eq.~\eqref{eq:localstochasticnoiseidentityup} implies that 
    \begin{align}
        \prob\left[((\cE^{\restriction \mathsf{2D}} \bowtie \cC^{\mathsf{2D}})\otimes \cC^{\mathsf{wait}})\circ \cC^{\mathsf{prep}} \neq \cC^{\mathsf{prep}} \right]\leq \varepsilon/2\ .
    \end{align}The claim then follows from the union bound. 
\end{proof}

As discussed after the statement of Theorem~\ref{thm:2Dentanglementgeneration}, this result immediately implies the existence of a state with robust long-range localizable entanglement.

\begin{corollary}[Robust long-range localizable entanglement in 2D]\label{cor:robustlocalizableentanglement}
    Given any distance $R>3$ and target error probability $\varepsilon > 0$, there exists a state $\ket{\Psi}$ on the $2D$ lattice $G_{\ell_X,\ell_Y}$, where $\ell_X,\ell_Y$ are as in Theorem~\ref{thm:2Dentanglementgenerationgeneral}, with the following properties.
    \begin{enumerate}[(i)]
        \item There is a constant-depth $G_{\ell_X,\ell_Y}$-local Clifford unitary $U$ such that
        \begin{align}
            \ket{\Psi} = U\ket{0^{\ell_X\times \ell_Y}}\ .
        \end{align}
        In particular, the state $\ket{\Psi}$ is $2D$-short range entangled.
        \item \label{it:localization}There exist two qubits $Q_1,Q_2 \in G_{\ell_X,\ell_Y}$ separated by distance $\Theta(R)$ such that, by measuring the qubits $G_{\ell_X,\ell_Y}\setminus \{Q_1,Q_2\}$ in the computational basis to obtain a  bit string $z \in \{0,1\}^{\times (\ell_X \ell_Y - 2)}$ and applying an effficiently computable Pauli correction $P(z)_{Q_1}$ to $Q_1$, a maximally entangled state $\ket{\Phi}_{Q_1Q_2}$ is prepared on the pair $Q_1Q_2$.
        \item The procedure of \eqref{it:localization} succeeds with probability at least $1-\varepsilon$ even when the initial state $\ket{\Psi}$ is corrupted with local stochastic Pauli noise of strength $p$, as long as $p\leq p_{\mathsf{thres}}(\varepsilon)$.
    \end{enumerate}
\end{corollary}

\begin{proof}
    This follows immediately from Theorem~\ref{thm:2Dentanglementgenerationgeneral}, letting $\ket{\Psi}$ be the state prepared by the unitary part of the circuit $\cC^{\mathsf{Bell}}$, i.e., before the final layers including measurement and adaptive correction.
\end{proof}

\subsection{Localizable entanglement at finite temperature\label{sec:localizableentanglementfinite}}

Notice that the state $\ket{\Psi}$ from Corollary~\ref{cor:robustlocalizableentanglement} is the unique ground state of the $2D$-local stabilizer Hamiltonian 
\begin{align}
H&=-\sum_{v\in V(G_{\ell_X,\ell_Y})} S_v\qquad\textrm{ where }\qquad S_v=U Z_vU^\dagger\textrm{ for every qubit } v\in V(G_{\ell_X,\ell_Y})\ .
\end{align}
Indeed, the constant-depth nature of the Clifford~$U$ ensures that each stabilizer generator $S_v$~has support contained in a constant-diameter region around the vertex~$v$.  Corollary~\ref{cor:robustlocalizableentanglement} thus states that 
the ground state~$\ket{\Psi}$ of the Hamiltonian~$H$ has long-range localizable entanglement (robust to local stochastic Pauli noise).

 This observation can be extended: In fact, for any temperature below a certain (constant) threshold temperature, the Gibbs state of~$H$ also has localizable entanglement. This statement is  obtained by translating temperature to (local stochastic Pauli) noise,
 similar to the analysis of~\cite{raussendorf2005long} for the 3D~cluster state.
 \begin{corollary}[Localizable entanglement at finite temperature]\label{cor:finitetemp}
    Let $\varepsilon > 0$ be arbitrary. Then there exists a constant threshold $\beta_{\mathsf{thres}}(\varepsilon) > 0$ such that the following holds for any $R \geq 3$. 
  Consider a system of qubits attached to the vertices of the grid graph $G_{\ell_X,\ell_Y}$, where $\ell_X,\ell_Y$ are as in Theorem~\ref{thm:2Dentanglementgenerationgeneral}.
    There exists a  (commuting) stabilizer Hamiltonian~$H$  with the following properties:
    \begin{enumerate}[(i)]
        \item The Hamiltonian $H$ is $2D$-local, i.e., consists of terms 
        supported on constant-diameter regions of~$G_{\ell_X,\ell_Y}$. 
        \item The Gibbs state $\rho_\beta=e^{-\beta H}/\tr(e^{-\beta H})$  at any inverse temperature~$\beta \geq  \beta_{\mathsf{thres}}(\varepsilon)$ satisfies the following:  There exist two qubits $Q_1,Q_2 \in G_{\ell_X,\ell_Y}$ separated by distance~$\Theta(R)$ such that, by measuring the qubits~$G_{\ell_X,\ell_Y}\setminus\{Q_1,Q_2\}$ of $\rho_\beta$ in the computational basis to obtain a string $z \in \{0,1\}^{\ell_X \ell_Y - 2}$ and applying an efficiently computable Pauli correction $P(z)_{Q_1}$ to $Q_1$, a maximally entangled state $\ket{\Phi}_{Q_1Q_2}$ is prepared on the pair $Q_1Q_2$ with probability at least~$1-\varepsilon$.
    \end{enumerate}
\end{corollary}
\noindent Succinctly, we say that the Gibbs state at temperature~$\beta\geq \beta_{\mathsf{thres}}(\varepsilon)$ has distance-$\Theta(R)$ localizable entanglement.

\begin{proof}
Let~$V=V(G_{\ell_X,\ell_Y})$ denote the vertices of the grid graph. Define the operators
\begin{align}
E_v&=U X_vU^\dagger\qquad\textrm{ for every }\qquad v\in V\ .\label{eq:evdefinitiona}
\end{align}
Then we have the commutation respectively anti-commutation relations
\begin{align}
\begin{matrix}
[E_v,E_w]&=&0\qquad&\textrm{ for all }\qquad v,w\\
[E_v,S_w]&=&0\qquad&\textrm{ for }\qquad v\neq w\\
\{E_v,S_v\}&=&0\qquad&\textrm{ for all}\qquad v\ .
\end{matrix}\label{eq:commutationrelationserror}
\end{align}
The ground state of $H$ has the form
\begin{align}
\proj{\Psi}&=\prod_{v\in V}\frac{1}{2}(I+S_v)\ ,\label{eq:gsgibbssto}
\end{align}
and for $\beta>0$, the Gibbs state~$\rho_\beta$ can be written as
\begin{align}
\rho_\beta 
&=\prod_{v\in V}\frac{1}{2}(I+\tanh(\beta) S_v)\\
&=\prod_{v\in V}\left((1-p)\frac{1}{2}(I+S_v)+p\frac{1}{2}(I-S_v)\right)
\textrm{ where }\qquad p:=(1-\tanh(\beta))/2\ 
.\label{eq:thermalstateeq}
\end{align}
It follows from Eq.~\eqref{eq:gsgibbssto}, Eq.~\eqref{eq:thermalstateeq} and Eq.~\eqref{eq:commutationrelationserror}
that
\begin{align}
\rho_\beta&:=\cE(\proj{\Psi})\qquad\textrm{ where }\qquad \cE:=\bigcirc_{v\in V}\cE_v\label{eq:rhobetaerrora}
\end{align}
where we introduced the quantum channels (CPTP maps)
\begin{align}
\cE_v(\rho)&=(1-p)\rho+ p E_v\rho E_v^\dagger \qquad\textrm{ for every }\qquad v\in V\ .\label{eq:rhobetaerrorb}
\end{align}
Eq.~\eqref{eq:rhobetaerrora}
expresses the  Gibbs state~$\rho_\beta$
as a noise-corrupted version of the ground state~$\proj{\Psi}$.
The corresponding noise channel~$\cE$ (defined by Eq.~\eqref{eq:rhobetaerrora} and~\eqref{eq:rhobetaerrorb}) can easily be seen to
be local stochastic Pauli noise. 

In more detail, it is the result of applying each Pauli error~$E_v$ (for $v\in V$)  independently at random with probability~$p$. In other words, defining Bernoulli-$p$ identically and independently distributed  random variables~$\{b_v\}_{v\in V}$, the (random) Pauli error which is applied is
\begin{align}
E&=\prod_{v\in V}E_v^{b_v}\\
&=U \left(\prod_{v\in V}Z_v^{b_v}\right)U^\dagger\ 
\end{align}
where we used the definition of~$E_v$ (see Eq.~\eqref{eq:evdefinitiona}).
That is, $E$ is the image under the constant-depth circuit~$U$ of independently and identically distributed Pauli-$Z$ noise on each qubit. It follows from arguments identical to those used in the proof of Lemma~\ref{lem:unitarycliffordcircuit} that there are constants~$(\Lambda,\lambda)$ such that 
$E$ is a local stochastic Pauli error of strength
\begin{align}
p'(\beta)&:=\Lambda p^{\lambda}=\Lambda \left(\frac{1-\tanh(\beta)}{2}\right)^{\lambda}\ .
\end{align}
According to Theorem~\ref{thm:2Dentanglementgenerationgeneral},
 and Eq.~\eqref{eq:rhobetaerrora}, a
Bell pair can be obtained with probability at least~$1-\varepsilon$ by local measurements and Pauli corrections applied to thermal state~$\rho_\beta$ if
\begin{align}
p'(\beta)&\leq  p_{\mathsf{thres}}(\varepsilon)\ .
\end{align}
The claim follows from this by setting
\begin{align}
\beta_{\mathsf{thres}}(\varepsilon)&:= \arctanh\left(1-2\left(p_{\mathsf{thres}}(\varepsilon)/\Lambda\right)^{1/\lambda}\right)\ .
\end{align}

\end{proof}
To our knowledge, Corollary~\ref{cor:finitetemp} gives the first example of a $2D$-local Hamiltonian whose thermal states possess long-ranged localizable entanglement below a constant threshold temperature. (We note that the $3D$ cluster state Hamiltonian has been shown to satisfy a similar property in Ref.~\cite{raussendorf2005long}.)
We note that a slight modification of the
argument used to establish Corollary~\ref{cor:finitetemp}  
also establishes that this localizable entanglement is robust to local stochastic Pauli noise (below a slightly different constant threshold temperature).  This result should be contrasted to the recent breakthrough  of Ref.~\cite{bakshi2024high}
showing that the Gibbs state of any local Hamiltonian on a constant-degree graph is separable above a constant threshold temperature (see Ref.~\cite{rouze2025efficient} for related work).

\textit{Acknowledgments.}
R.K. gratefully acknowledges support by the European Research Council under Grant No. 101001976 (project EQUIPTNT) and the Munich Quantum Valley, which
is supported by the Bavarian state government through the Hightech Agenda Bayern Plus. D.H. acknowledges support from the Novo Nordisk Foundation (Grant No. NNF20OC0059939 ‘Quantum for Life’).
\bibliographystyle{plain}
\bibliography{bib}
\end{document}

%% file: commands.tex
\usepackage[T1]{fontenc}
\usepackage{lmodern}
\usepackage[expansion=false]{microtype}
\usepackage{authblk}

\usepackage[margin=0.75in]{geometry}

\usepackage{graphicx}
\usepackage{svg}

\usepackage{amsmath,amssymb,amsfonts,amsthm}
\usepackage{mathtools}
\mathtoolsset{showonlyrefs}
\usepackage{dsfont}
\usepackage{pifont}

\usepackage{hyperref}
\usepackage{url}

\usepackage{cleveref}

\usepackage{ifthen}
\usepackage{enumerate}
\usepackage{tikz}
\usetikzlibrary{quantikz2, shapes.geometric}
\usepackage{caption}
\usepackage{subcaption}
\usepackage{chngcntr}

\theoremstyle{plain}

\newtheorem{theorem}{Theorem}
\newtheorem{result}{Step}
\newtheorem{lemma}[theorem]{Lemma}
\newtheorem{definition}[theorem]{Definition}

\newtheorem{corollary}[theorem]{Corollary}

\newtheorem{claim}[theorem]{Claim}

\newtheorem*{theorem*}{Theorem}
\newtheorem*{lemma*}{Lemma}
\newtheorem*{definition*}{Definition}
\newtheorem*{setup*}{Setup}
\newtheorem*{assumption*}{Assumption}
\newtheorem*{definitionnc*}{Definition}
\newtheorem*{claim*}{Claim}

\newtheorem*{example*}{Example}

\numberwithin{theorem}{subsection}
\numberwithin{equation}{section}  

\crefname{theorem}{Theorem}{Theorems}
\crefname{lemma}{Lemma}{Lemmas}
\crefname{definition}{Definition}{Definitions}
\crefname{proposition}{Proposition}{Propositions}
\crefname{corollary}{Corollary}{Corollaries}
\crefname{conjecture}{Conjecture}{Conjectures}
\crefname{setup}{Setup}{Setups}
\crefname{assumption}{Assumption}{Assumptions}
\crefname{claim}{Claim}{Claims}
\crefname{example}{Example}{Examples}

\DeclareMathOperator{\NN}{\mathbb{N}}
\DeclareMathOperator{\RR}{\mathbb{R}}

\DeclareMathOperator{\prob}{Pr}

\DeclareMathOperator{\pauli}{\mathbb{P}}

\DeclareMathOperator{\arctanh}{arctanh}
\DeclareMathOperator{\poly}{\mathsf{poly}}
\DeclareMathOperator{\localization}{\upharpoonright}

\newcommand{\appref}[1]{\hyperref[#1]{{Appendix~\ref*{#1}}}}

\newcommand{\be}{\begin{eqnarray} \begin{aligned}}
\newcommand{\ee}{\end{aligned} \end{eqnarray}}
\newcommand{\benn}{\begin{eqnarray*} \begin{aligned}}
\newcommand{\eenn}{\end{aligned} \end{eqnarray*}}

\newcommand*{\textfrac}[2]{{{#1}/{#2}}}

\newcommand*{\cB}{\mathcal{B}}
\newcommand*{\cC}{\mathsf{C}}

\newcommand*{\cE}{\mathcal{E}}
\newcommand*{\cF}{\mathcal{F}}

\newcommand*{\cH}{\mathcal{H}}
\newcommand*{\cI}{\mathcal{I}}

\newcommand*{\cL}{\mathcal{L}}
\newcommand*{\cM}{\mathcal{M}}
\newcommand*{\cN}{\mathcal{N}}
\newcommand*{\cO}{\mathcal{O}}
\newcommand*{\cP}{\mathcal{P}}

\newcommand*{\cR}{\mathcal{R}}

\newcommand*{\cU}{\mathcal{U}}

\newcommand*{\cW}{\mathcal{W}}
\newcommand*{\cX}{\mathcal{X}}
\newcommand*{\cY}{\mathcal{Y}}

\newcommand*{\nmin}{\mathsf{n}_{\min}}
\newcommand*{\nmax}{\mathsf{n}_{\max}}
\newcommand*{\dmin}{\mathsf{d}_{\min}}
\newcommand*{\dmax}{\mathsf{d}_{\max}}

\newcommand*{\upperboundL}{S}

\newcommand*{\tr}{\mathop{\mathrm{tr}}\nolimits}

\newcommand*{\supp}{\mathrm{supp}}

\newcommand{\bc}{\begin{center}}
\newcommand{\ec}{\end{center}}

\def\01{\{0,1\}}

\newcommand*{\CX}{\mathsf{CX}}

\newcommand*{\CNOT}{\mathsf{CNOT}}

\tikzset{ftlabel/.style={font=\scriptsize, xshift=-2pt, yshift=-3pt}}

\tikzset{
    ecbox/.style={draw, minimum width=0.8cm, minimum height=0.6cm, fill=white},
    ecboxs/.style={draw, minimum width=0.8cm, minimum height=0.6cm, fill=white, label={[ftlabel]above right:$s$}}
}
\newcommand{\EC}{\gate[style={ecbox}]{\text{\small$\mathrm{EC}$}}}
\newcommand{\ECs}{\gate[style={ecboxs}]{\text{\small$\mathrm{EC}$}}}

\newcommand{\filterwidthqz}{0.08}
\tikzset{
    filterbox/.style={draw, minimum width=\filterwidthqz cm, minimum height=0.5cm, fill=white, inner sep=0pt},
    filtr/.style={draw, minimum width=\filterwidthqz cm, minimum height=0.5cm, fill=white, inner sep=0pt, label={[ftlabel]above right:$r$}},
    filts/.style={draw, minimum width=\filterwidthqz cm, minimum height=0.5cm, fill=white, inner sep=0pt, label={[ftlabel]above right:$s$}},
    filtri/.style={draw, minimum width=\filterwidthqz cm, minimum height=0.5cm, fill=white, inner sep=0pt, label={[ftlabel]above right:$r_i$}},
    filtsum/.style={draw, minimum width=\filterwidthqz cm, minimum height=0.5cm, fill=white, inner sep=0pt, label={[font=\scriptsize, xshift=-2pt, yshift=-3pt]above right:{$s{+}\sum_i r_i$}}},
    filtzero/.style={draw, minimum width=\filterwidthqz cm, minimum height=0.5cm, fill=white, inner sep=0pt, label={[ftlabel]above right:$0$}}
}
\newcommand{\filtr}{\gate[style={filtr}]{}}
\newcommand{\filts}{\gate[style={filts}]{}}

\newcommand{\filtri}{\gate[style={filtri}]{}}
\newcommand{\filtsum}{\gate[style={filtsum}]{}}
\newcommand{\filtrzero}{\gate[style={filtzero}]{}}

\tikzset{
prepgate/.style={
thickness,
filling,
shape=semicircle,
rotate=90,
scale=1.6,
draw=black,
inner sep=2pt
}}
\DeclareExpandableDocumentCommand{\prepgate}{O{}{m}}{|[prepgate,#1]|}

\tikzset{
measgate/.style={
thickness,
filling,
shape=semicircle,
rotate=-90,
scale=1.6,
draw=black,
inner sep=2pt
}}
\DeclareExpandableDocumentCommand{\measgate}{O{}{m}}{|[measgate,#1]|}

\tikzset{
encgate/.style={
thickness,
filling,
shape=isosceles triangle,
rotate=180,
isosceles triangle apex angle=60,
scale=1.6,
draw=black,
inner sep=2pt
}}
\DeclareExpandableDocumentCommand{\encgate}{O{}{m}}{|[encgate,#1]|}

\tikzset{
decgate/.style={
thickness,
filling,
shape=isosceles triangle,
rotate=0,
isosceles triangle apex angle=60,
scale=1.6,
draw=black,
inner sep=2pt
}}
\DeclareExpandableDocumentCommand{\decgate}{O{}{m}}{|[decgate,#1]|}

\newcommand{\encwidth}{2pt}

\tikzset{
    ftgate/.style={draw, circle, minimum size=0.6cm, fill=white, inner sep=0pt},
    ftgates/.style={draw, circle, minimum size=0.6cm, fill=white, inner sep=0pt, label={[ftlabel]above right:$s$}}
}
\newcommand{\ftgate}[1]{\gate[style={ftgate}]{#1}}
\newcommand{\ftgates}[1]{\gate[style={ftgates}]{#1}}

\newcommand*{\pathgraph}{\mathsf{P}}

\newcommand{\cenc}[1]{\underset{#1}{\mathsf{Enc}}}
\newcommand{\cdec}[1]{\underset{#1}{\mathsf{Dec}}}
\newcommand{\enc}{\mathsf{Enc}}
\newcommand{\dec}{\mathsf{Dec}}

\newcommand*{\idec}{\mathsf{iDec}}